\documentclass[12pt]{article}
\usepackage[cp1251]{inputenc}
\usepackage[english,russian]{babel}
\usepackage{amsmath}
\usepackage{latexsym}
\usepackage{amssymb}
\usepackage{amsthm}
\usepackage{tocloft}
\usepackage{minitoc}
\usepackage{graphicx,color}
\usepackage{url}
\usepackage[unicode,bookmarksopen,bookmarksopenlevel=1,bookmarksnumbered]{hyperref}
\hypersetup{pdfstartview={XYZ null null 1.25}}
\newtheorem{theorem}{Theorem}[section]
\newtheorem{lemma}{Lemma}[section]
\newtheorem{note}{Note}[section]
\newtheorem{corollary}{Corollary}[theorem]

\def\hcorrection{\hspace{-0.3em}}
\def\Abstract#1{{\footnotesize\baselineskip=12pt\begin{quotation}\noindent\hcorrection{#1}\end{quotation}}}
\def\References#1{{\footnotesize\baselineskip=12pt\begin{thebibliography}{99}{#1}\end{thebibliography}}}

\def\psps{  \ \ \ \ \ \ \  }
\def\sps{, \ \ \  }
\def\spsd{. \ \ \  }
\begin{document}
\renewcommand{\proofname}{Proof.}
\renewcommand{\refname}{References}
\renewcommand{\figurename}{Fig.}
\renewcommand{\contentsname}{Contents/Содержание}
\renewcommand{\partname}{}
\renewcommand{\tablename}{Table}
\renewcommand{\listtablename}{Tables/Список таблиц}
\renewcommand{\listfigurename}{Figures/Список рисунков}
\doparttoc
\dopartlof
\dopartlot
\setcounter{parttocdepth}{4}
\addtocontents{toc}{\cftpagenumbersoff{part}}
\makeatletter
\@addtoreset{section}{part}
\@addtoreset{table}{part}
\@addtoreset{equation}{part}
\makeatother
\thispagestyle{empty}
\begin{center}
BASHKIR STATE UNIVERSITY\\
\rule{12cm}{.4pt} \\
Department of Mathematics
\end{center}
\vspace{5cm}
\begin{center}
АNDREEV Konstantin Vasil'evich\\
\vspace{1cm}

{\sf SPINOR FORMALISM AND THE GEOMETRY OF \\
     SIX-DIMENSIONAL RIEMANNIAN SPACES}\\

\vspace{1cm}
$ $\\
\vspace{1cm}
Thesis for the candidate\\
degree of physical and mathematical sciences\\
(Ph.D. Thesis in Physics and Mathematics)
\end{center}
\vspace{1cm}
\begin{flushleft}
\hspace*{6cm}
Scientific advisor: \\
\hspace*{6cm}
Ph.D in Physics and Mathematics,\\
\hspace*{6cm}
Associate Professor \`E. G. Ne\u\i fel'd .
\vspace{5cm}
\end{flushleft}
\begin{center}
UFA - 1997
\end{center}
\pagebreak

\newpage
\foreignlanguage{english}{
\part{English edition}
\parttoc
\partlot
\partlof
\newpage
\section{Introduction}$ $\\
     \indent The proposed thesis is a theoretical study of the 6-dimensional Riemannian space geometry devoted to some questions related to this geometry.\\
     \indent The study of some 6-dimensional Riemannian spaces is done using the corresponding 6-dimensional spinor formalism \cite{Cartan1}, \cite{Brauer1}, \cite{Penrose_1} and the Norden-Neifeld normalization theory \cite{Neifeld1}-\cite{Neifeld7},
     \cite{Norden1}-\cite{Norden6} that simplifies the important relations written down in the tensor form and that leads to the original results. \\
     \indent The subject choice is caused by the increased interest to such the spaces in these latter days. These spaces naturally appear in the Penrose spinor-twistor formalism \cite{Penrose1}-\cite{Penrose4}, \cite{Penrose_1}-\cite{Penrose_4}. Here the pseudo-Euclidean space $\mathbb R^6_{(2,4)}$, whose an isotropic cone allows to define the conformal pseudo-Euclidean Minkowski space, plays an important role. Moreover, twistors in the Penrose theory will be represented by spinors coordinated with the space $\mathbb R^6_{(2,4)}$. However, if in the monography  \cite{Penrose1}, twistors form the 4-dimensional complex vector bundle with the 4-dimensional real manifold as the base, in these thesis, the 6-dimensional complex analytic Riemannian space $\mathbb CV^6$ serves as the base. This leads to new results in the twistor theory. At the same time, a conformal pseudo-Euclidean space is associated with a complex analytic quadric $\mathbb CQ_6$ \cite{Hirtseburh1}, \cite{Hodzh1} that leads to the studying of properties of the group $SO(8,\mathbb C)$ \cite{Adams1}, \cite{Hughston1}, and hence the Cartan triality principle \cite{Cartan1}. In this case, the specified complex-Euclidean geometry appears as the intrinsic geometry of this normalized quadric. Writing out the derivational equations for this quadric \cite{Neifeld4}, it is possible to define the invariant equation under conformal transformations and hence the normalization replacement. This equation we call \emph{bitwistor equation} by analogy to the Penrose twistor equation. Solutions of this equation form pairs which can be interpreted as Rosenfeld null-pairs  \cite{Rosenfeld1} that leads to the 6-dimensional quadric and the Cartan triality principle. It is expedient to consider the following three manifolds, diffeomorphic to each other:
\begin{enumerate}
    \item the manifold of all points of the quadric $\mathbb CQ_6$;
    \item the manifold of I-family maximal planar generators $\mathbb C\mathbb P_3$ of the quadric
    $\mathbb CQ_6$;
    \item the manifold of II-family maximal planar generators $\mathbb C\mathbb P_3$  of the quadric
    $\mathbb CQ_6$.
\end{enumerate}
    A normalization of one of such the manifolds allows on this one to consider conformal (pseudo-)Euclidean connections which will be Weyl connections. This leads to a generalization of the triality principle to B-spaces in the Norden terminology.\\
    \indent Thus, the 6-dimensional spinor formalism is based on the Cartan's \cite{Cartan1} and  Brauer's \cite{Brauer1} works. The 4-dimensional spinor-twistor formalism and the twistor algebra are described in \cite {Penrose1}-\cite {Penrose4}, \cite{Penrose_1}-\cite{Penrose_4}. The Lie group's and Lie algebra's isomorphisms, associated with these formalisms, are considered in \cite{Besse1}, \cite{Shevale1}, \cite{Postnikov1}, \cite{Hochschild1}. In addition, a piece of the information on the Clifford algebras and the octaves is taken from \cite{Postnikov1} and \cite{Penrose1}. A piece of the information on quadrics and planar generators is given in \cite{Hirtseburh1} and \cite{Hodzh1}. A normalization of a maximal planar generator manifold is also described in \cite{Neifeld4}-\cite{Neifeld6}, \cite{Norden2}, \cite{Norden4}, \cite{Norden6}. Connections in bundles is induced according to \cite{Neifeld6}, \cite{Norden3}, \cite{Norden5}, \cite{Norden6}. A real space inclusion in a complex space is considered in \cite{Neifeld1}. The complex and real representations of (pseudo-)Riemannian spaces are carried out according to \cite{Koboyasi1}, \cite{Koboyasi2} and \cite{Lichnerowicz1}. Rosenfeld null-pairs were taken from \cite{Rosenfeld1}. The Klein correspondence is given in \cite{Klein1} and \cite{Klein2}. Physical applications of twistors can be found in \cite{Penrose1}, \cite{Dirac1} and \cite{Dirac2}.\\
    \indent The basic content, divided in four sections, will be considered below. For this purpose, some definitions is necessary to make preliminary.

\subsection{Basic definitions}$ $
    \indent These definitions are made according to \cite{Neifeld1}-\cite{Neifeld7}. Note, that all functions, involved in constructions, are assumed to be sufficiently smooth; all definitions, statements, constructions are a local. \emph{Complex analytic Riemannian space $\mathbb CV^n$} is a complex analytic manifold, in which each tangent space is equipped with an analytical quadratic metric. Such the metric is defined by means of a nondegenerate symmetric tensor $g_{\alpha\beta}$ whose the coordinates are analytic functions of the point coordinates. This tensor corresponds to the torsion-free complex Riemannian connection whose the coefficients are determined by the Christoffel symbols, and therefore these coefficients are analytic functions.\\
    \indent The tangent bundle $\tau^\mathbb C(\mathbb CV^n)$ to the manifold has fibers $\tau_x^\mathbb C\cong \mathbb C\mathbb R^n$, i.e., each fiber will be isomorphic to the complex n-dimensional Euclidean space whose the metric is determined by the value of the Euclidean metric tensor at the given point $x$. Let n=6. By $\Lambda^2\mathbb C^4$, we denote the bivector space of $\mathbb C^4$, and by $\Lambda$, we denote the corresponding bundle with the base $\mathbb CV^6$. Each fiber of this bundle is isomorphic to $\mathbb C\mathbb R^6$ ($\mathbb C\mathbb R^6\cong \Lambda^2\mathbb C^4$). It follows that in the six-dimensional case, the complex Riemannian space $\mathbb CV^6$ will be the base of the bundle $A^\mathbb C=\mathbb C^4(\mathbb CV^6)$. Then the canonical projection $p:\mathbb C^4_x\longmapsto x\in \mathbb CV^6$ maps the fiber $\mathbb C^4_x$ to the point $x$ of the base.\\
    \indent A real (pseudo-)Riemannian space $V^n_{(p,q)}$ is regarded as a real n-dimension surface in the space $\mathbb CV^n$. In a neighborhood U, this is an inclusion which is locally determined by the parametric equations
\begin{equation}
\label{e1v}
    z^\alpha=z^\alpha(u^i)\ (\alpha,\beta ,... ,i,j,g,h=\overline{1,n})\sps
\end{equation}
    where $z^\alpha$ are the complex coordinates of the base point $x$; the parameters $u^i$ are the local coordinates of the space $V^n_{(p,q)}$. By $\tau_x^\mathbb R(V^n_{(p,q)})\cong \mathbb R^n_{(p,q)}$, we denote the real tangent space to the surface (\ref{e1v}) at the point $x$. The partial derivatives $(\partial_i z^\alpha =:H_i{}^\alpha)$ define \emph{inclusion H} of the real tangent space $\tau_x^\mathbb R$ in the complex tangent space $\tau_x^\mathbb C$
\begin{equation}
\label{e2v}
    H:\tau_x^\mathbb R\longmapsto\tau_x^\mathbb C\sps
\end{equation}
\begin{equation}
\label{e3v}
\begin{array}{c}
    z^\alpha=z^\alpha(u^i(t))\sps V^\alpha:=\frac{dz^\alpha}{dt}=
    H_i{}^\alpha\frac{du^i}{dt}=:H_i{}^\alpha v^i\sps\\[2ex]
    \frac{du^i}{dt}\in\tau_x^\mathbb R\longmapsto \frac{dz^\alpha}{dt}\in\tau_x^\mathbb C\sps
\end{array}
\end{equation}
    where the differentiation is carried out along the real curve $\gamma (t)$ of the surface (\ref{e1v}). Since, the matrix $\parallel H_i{}^\alpha \parallel$ is a nonsingular Jacobian matrix then the operator $H^i{}_\alpha$ such that
\begin{equation}
\label{e4v}
    \left\{
\begin{array}{l}
    H^i{}_\alpha H_i{}^\beta=\delta_\alpha{}^\beta\sps\\
    H^i{}_\alpha H_j{}^\alpha=\delta_j{}^i
\end{array}
    \right.
\end{equation}
    exists. It follows that in the complex space, the operator $H_i{}^\alpha$ defines \emph{involution}
\begin{equation}
\label{e5v}
    S_\alpha {}^{\beta '}=H^i{}_\alpha\bar  H_i{}^{\beta '}\sps
\end{equation}
    where the coordinates $\bar H_i{}^{\beta '}$ are conjugated to the coordinates $H_i{}^\beta$ \cite{Neifeld1}. Therefor,
\begin{equation}
\label{e6v}
    v^i=H^i{}_\alpha V^\alpha=\overline{H^i{}_\alpha V^\alpha}
    \ \ \Rightarrow \ \ S_\alpha{}^{\beta '}V^\alpha=\bar V^{\beta '}\spsd
\end{equation}
    This is a necessary and sufficient condition for the vector $V^\alpha\in \tau^\mathbb C_x$ to be real. At the same time,
\begin{equation}
\label{e7v}
    S_\alpha{}^{\beta '}\bar S_{\beta '}{}^\gamma=\delta_\gamma{}^\alpha\spsd
\end{equation}
    The metric of $\tau_x^\mathbb R(V^n_{(p,q)})$ is defined with the help of the relation
\begin{equation}
\label{e8v}
    g_{\alpha\beta}V^\alpha V^\beta=\overline{g_{\alpha\beta}V^\alpha V^\beta}\sps
    \forall \bar V^{\beta '}=S_\alpha{}^{\beta '} V^\alpha\spsd
\end{equation}
    This is means that a real tensor of the space $\tau_x^\mathbb R(V^n_{(p,q)})$ is determined as a tensor conjugated under the action of the specified Hermitian involution
\begin{equation}
\label{e9v}
    g_{\alpha\beta}=S_\alpha{}^{\gamma '}S_\beta{}^{\delta '}
    \bar g_{\gamma '\delta '}\spsd
\end{equation}
    Therefor, the tensor
\begin{equation}
\label{e10v}
    g_{ij}:=H_i{}^\alpha H_j{}^\beta g_{\alpha\beta}=
    \overline{H_i{}^\alpha H_j{}^\beta g_{\alpha\beta}}
\end{equation}
    is the metric tensor of $\tau_x^\mathbb R(V^n_{(p,q)})\subset \tau_x^\mathbb C(\mathbb CV^n)$. A metric form significantly depends on a structure of the operator $H_i{}^\alpha$ and hence the involution $S_\alpha{}^{\beta '}$. The complex Riemannian connection on the space $\mathbb CV^n$ induces a connection
\begin{equation}
\label{e11v}
    \nabla_i:=H_i{}^\alpha\nabla_\alpha
\end{equation}
    on the real space such that
\begin{equation}
\label{e12v}
     \nabla_i H_j{}^\alpha=i\ b_{ij}{}^gH_g{}^\alpha\sps
     \nabla_i g_{jg}=2ib_{i(jg)}\sps b_{ijg}:=b_{ij}{}^hg_{hg}\spsd
\end{equation}
     Demand that the induced connection was the Riemannian one then
\begin{equation}
\label{e13v}
     b_{ijh}=b_{jih}\sps b_{ijh}=-b_{ihj}\ \ \Rightarrow\ \ \ b_{ijh}=0\sps
\end{equation}
     and hence
\begin{equation}
\label{e14v}
     \nabla_iS_\alpha{}^{\beta '}=0\ \ \Rightarrow\ \ \
     \nabla_\gamma S_\alpha{}^{\beta '}=0\spsd
\end{equation}\\
     For n = 6, the restriction of the base $\mathbb CV^6\longmapsto V^6_{(p,q)}$ allows to determine the bundle $A^\mathbb C(S)=\mathbb C^4(S)(V^6_{(p,q)})$ whose fibers should be supplied with an additional structure $s$. It will be shown that in the case of the even metric index (i.e., the number of minuses is even), this structure is determined by a Hermitian symmetric tensor, and  in the case of the odd metric index, this structure is determined by a Hermitian involution. For a pseudo-Riemannian space of the even index, equal to 4, the bundle $A^\mathbb C(S)$ is called \emph{twistor bundle} because its each fiber will be isomorphic to the vector space $\mathbb T$ \cite{Penrose1}.

\subsection{Second chapter}$ $
      \indent The main results of this chapter are based on the works \cite{Penrose1}-\cite{Penrose4}, where the 4-dimensional spinor formalism is developed. The connection operators are introduced in \cite{Norden1}. A piece of the information on Lie groups is taken from \cite{Besse1}, \cite{Shevale1}, \cite{Postnikov1}. The 6-dimensional spinor formalism, constructed in this chapter, are based on the three following isomorphisms:
\begin{enumerate}
     \item the isomorphism between the spaces $\mathbb C\mathbb R^6\cong\Lambda^2\mathbb C^4$;
     \item the isomorphism between the groups $SO(6,\mathbb C)\cong SL(4,\mathbb C)/\{\pm 1\}$;
     \item the isomorphism between the Lie algebras $so(6,\mathbb C)\cong sl(4,\mathbb C)$.
\end{enumerate}
     Explicitly, these isomorphisms are described as
\begin{enumerate}
     \item $r^\alpha=\frac{1}{2}\eta^\alpha{}_{ab}R^{ab}$,
     where $r^\alpha$ are the coordinates of a vector of $\mathbb C\mathbb R^6$,
     $R^{ab}$ are the coordinates of a bivector of $\Lambda^2\mathbb C^4$,
     $\eta^\alpha{}_{ab}$ are the coordinates of the connecting Norden operators;
     \item $K_\alpha{}^\beta =\frac{1}{4} \eta^\beta{}_{ab}
     \eta_\alpha{}^{cd}\cdot 2 S_c{}^aS_d{}^b$, where
     $K_\alpha{}^\beta$ is the coordinates of a transformation from the group
     $SO(6,\mathbb C)$, $S_a{}^b$ is the coordinates of a transformation from the group $SL(4,\mathbb C)$;
     \item $T^{\alpha\beta}=A^{\alpha\beta}{}_a{}^bT_b{}^a$,
     where $T^{\alpha\beta}$ are the coordinates of a bivector of $\Lambda^2 \mathbb C\mathbb R^6$,
     $T_a{}^b$ is the coordinates of a traceless operator in $\mathbb C^4$.
\end{enumerate}
     At the same time, \emph{connecting Norden operators} are defined with the help of the following relations
\begin{equation}
\label{e1r}
    g^{\alpha\beta}=1/4\cdot \eta^{\alpha}{}_{aa_1} \eta^{\beta}{}_{bb_1}
    \varepsilon^{aa_1bb_1}\sps
    \varepsilon^{aa_1bb_1}=
    \eta_{\alpha}{}^{aa_1} \eta_{\beta}{}^{bb_1}g^{\alpha\beta}\sps\\[2ex]
\end{equation}
    where $(\alpha,\beta,...=1,2,3,4,5,6;$ $a_1,b_1,a,b,e,f,k,l,m,n,...=1,2,3,4;$ $i,j,g,h=1,2,3,4,5,6)$.
    From this, it follows that the connecting Norden operators will satisfy the Clifford equation, and therefor they will define the full Clifford algebra which can be realized with the help of the matrix algebra of dimension $8 \times 8$.\\
    \indent Considering the inclusion $\mathbb R^6_{(p,q)}\subset \mathbb C\mathbb R^6$ in the case of the even metric index q, the decomposition
\begin{equation}
\label{e2r}
    s_{aa_1}{}^{b'b'{}_1}=s^{cb'}s^{c_1b'{}_1} \varepsilon_{cc_1aa_1}
    \sps \bar s_{a'b}=\pm s_{ba'}
\end{equation}
    and in the case of the odd metric index q, the decomposition
\begin{equation}
\label{e3r}
    s_{aa_1}{}^{b'b'{}_1}=2s_{\left[ a \right.}{}^{b'}s_{\left. a_1\right]}
    {}^{b'{}_1}\sps s_a{}_{b '}\bar s_{b '}{}^c=\pm \delta_a{}^c
\end{equation}
     can be obtained. The tensor $s_{aa_1}{}^{b'b'{}_1}$ is defined by means of the tensor Hermitian involution $S_\alpha{}^{\beta '}(=\frac{1}{4}\eta_\alpha{}^{aa_1}\bar\eta^{\beta '}{}_{b'b_1 '}s_{aa_1}{}^{b'b_1 '})$. As an indicative example, in the special basis, the inclusion $\mathbb R^6_{(2,4)}\subset \mathbb C\mathbb R^6$ is considered.\\
    \indent In the conclusion of the given chapter, \emph{generalized connecting Norden operators} are entered as analytical functions of the coordinates of a point $z^\gamma$ so that the equality
\begin{equation}
\label{e4r}
    g^{\alpha\beta}(z^\gamma)=1/4\cdot \eta^{\alpha}{}_{aa_1}(z^\gamma)
    \eta^{\beta}{}_{bb_1}(z^\gamma)\varepsilon^{aa_1bb_1}(z^\gamma)
\end{equation}
    is executed. This completes the construction of the required spinor formalism for the space $\mathbb CV^6$.\\
    \indent It should be noted that the spinor formalism is largely similar to the 4-dimensional Penrose spinor formalism in which the pseudo-Riemannian space $V^4_{(1,3)}$ is the base of the bundles $\mathbb C^2(V^4_{(1,3)})$ and $\mathbb C^4(V^4_{(1,3)})$. At Penrose, a vector in $\mathbb C^2 (V^4_{(1,3)})$ is called \emph{spinor}, and a vector in $\mathbb C^4(V^4_{(1,3)})$ is called \emph{twistor}. The main feature of twistors of this thesis is that \emph{twistor} is a vector of the bundle $\mathbb C^4(V^6_{(2,4)})$ that gives new results. This can lead to a new interpretation of the twistor physical exegesis which is described in the monography \cite{Penrose1}.

\subsection{Third chapter}$ $
     \indent The third chapter is devoted to the introduction of connections in the bundles with the complex base $\mathbb CV^6$. This procedure is carried out according to \cite{Neifeld6}, \cite{Norden3}, \cite{Penrose1}. The real and complex representations are taken from \cite{Koboyasi1}, \cite{Lichnerowicz1}. As the base, the complex analytic Riemannian space $\mathbb CV^6$ is considered. To carry out this procedure, we consider a complex analytic quadric $\mathbb CQ_6$ embedded in the projective space $\mathbb CP_7$
\begin{equation}
\label{e5r}
    G_{AB}X^AX^B=0\sps
\end{equation}
    where $(A,B,...=\overline{1,8})$. A manifold of maximal planar generators $\mathbb C \mathbb P_3$ of one of the two families is a complex six-dimensional manifold. Next, we consider \emph{harmonic normalization} which in local coordinates has the form
\begin{equation}
\label{e6r}
     X^a=X^a(u^\Lambda)\sps Y_b=Y_b(u^\Lambda)\sps
\end{equation}
     where $u^\Lambda$ are twelve real parameters $(\Lambda ,\Psi ,...=\overline{1,12})$. The first derivation equation of this normalized family has the form
\begin{equation}
\label{e7r}
     \nabla_\Lambda X_a=Y^bM_{\Lambda ab}\sps
     M_{\Lambda ab}=-M_{\Lambda ba}\spsd
\end{equation}
     For the transfer of binary indices, we use \emph{quadrivector $\varepsilon_{abcd}$} which is skew-symmetric in all its indices
\begin{equation}
\label{e8r}
     M_\Lambda{}^{ab}=\frac{1}{2}M_{\Lambda cd}\varepsilon^{abcd}\spsd
\end{equation}
    In addition, by means of the operators $M_\Lambda{}^{ab}$, a bivector of $\Lambda^2\mathbb C^4$ is associated with a vector of the tangent bandle
\begin{equation}
\label{e9r}
     V^{ab}=M_\Lambda{}^{ab}V^\Lambda\spsd
\end{equation}
     This defines the metric tensor
\begin{equation}
\label{e10r}
       G_{\Lambda\Psi}=\frac{1}{4}
       (M_\Lambda{}^{ab}M_\Psi{}^{cd}\varepsilon_{abcd}+
       \bar M_\Lambda{}^{a'b'}\bar M_\Psi{}^{c'd'}\varepsilon_{a'b'c'd'})
\end{equation}
        in the tangent bandle. Thus, the base will become the 12-dimensional pseudo-Riemannian space $V^{12}_{(6,6)}$ (the real representation of the manifold of the maximal planar generators) with the metric tensor $G_{\Lambda\Psi}$ and the complex structure $f_\Lambda{}^\Psi$ satisfying the following relation
\begin{equation}
\label{e11r}
       \bigtriangleup_\Lambda{}^\Psi=\frac{1}{2}
       (\delta_\Lambda{}^\Psi+if_\Lambda{}^\Psi)=\frac{1}{2}
       M_{\Lambda ab} M^{\Psi ab}
\end{equation}
       and defined on this manifold. As a fiber of $A^\mathbb C$, we consider the space $\mathbb C^4$ defined by the four basic points $X_a$ of a planar generator. The complex representation of the space $V^{12}_{(6,6)}$ is the space $\mathbb CV^6$ so that a mapping between the tangent spaces is done with the help of \emph{Neifeld operators $m_\alpha{}^\Lambda$}. In this case, the connection coefficients are determined  by the equation
\begin{equation}
\label{e12r}
     \nabla_\alpha m_\alpha{}^\Lambda=0\sps
     \bar\nabla_{\alpha '}\bar m_{\alpha '}{}^\Lambda=0\spsd
\end{equation}
     Then as the complex covariant derivative, we can take
\begin{equation}
\label{e13r}
     \nabla_\alpha=m_\alpha{}^\Lambda\nabla_\Lambda\sps
     \bar\nabla_{\alpha '}=\bar m_{\alpha '}{}^\Lambda\nabla_\Lambda\spsd
\end{equation}
     Then in this chapter, the properties of the torsion-free Riemannian connection, prolonged on fibers of $A^\mathbb C$, are established. It turns out, this prolongation is uniquely given by the requirement of the covariant constancy of the  quadrivector $\varepsilon_{abcd}$. Then using the inclusion operators, we can come to the real connection, but it is necessary to require the covariant constancy of the Hermitian involution.\\
     \indent This allows us to consider the conformally invariant bitwistor equation
\begin{equation}
\label{e14r}
      \nabla^{c\left(\right.d}X^{a\left.\right)}=0\spsd
\end{equation}
      Its solutions are associated with \emph{Rosenfeld null-pairs} which play an important role in further studies.

\subsection{Fourth chapter}$ $
      \indent The fourth chapter is devoted to the classification of \emph{Riemann curvature tensor} of the 6-dimensional (pseudo-)Riemannian spaces. In addition, the properties of bivectors of this spaces are investigated in this chapter.\\
      \indent We will show how to simplify the tensor record of the basic identities for the curvature tensor using the spinor formalism specified in the first chapter. In addition, it is shown that the classification of such a tensor can be reduced to the classification of \emph{curvature spin-tensor} of the space $\mathbb C^4$ such that
\begin{equation}
\label{e15r}
     R_{\alpha\beta\gamma\delta}=
     A_{\alpha\beta a}{}^b A_{\gamma\delta c}{}^d R_b{}^a{}_d{}^c\spsd
\end{equation}
     In this case, the curvature tensor satisfies the Bianchi identity
\begin{equation}
\label{e16r}
    R_{\alpha \beta \gamma \delta}+R_{\alpha \delta \beta \gamma}+
    R_{\alpha \gamma \delta \beta}=0
\end{equation}
    which have the spinor representation
\begin{equation}
\label{e17r}
     R_l{}^d{}_s{}^l=-\frac{1}{8}R\delta_s{}^d\spsd
\end{equation}
    As we can see, instead of the 105 equations from (\ref{e16r}), in (\ref{e17r}) the 16 ones only can be considered, 15 of which are significant. In the same way, we can construct the spinor analogue of the Weyl tensor
\begin{equation}
\label{e18r}
     C_{\alpha\beta\gamma\delta}=
     A_{\alpha\beta a}{}^b A_{\gamma\delta c}{}^d C_b{}^a{}_d{}^c\spsd
\end{equation}

    As a corollary from this theorem, the following fact is very interesting. An arbitrary simple isotropic bivector of the space $\Lambda^2\mathbb C^6$ defines a vector of the space $\mathbb C^4$ up to a factor. This will allow  to construct the geometric interpretation of \emph{isotropic twistor} in the space $\mathbb R^6_ {(2,4)}$. This interpretation in many respects similar to the exegesis of a spinor in the space $\mathbb R^4_{(1,3)}$: \emph{flag} consisted of \emph{flagpole} and  \emph{flag-plane}.\\
    \indent Finally, it is argued that in the space $\mathbb R^6_{(p, q)} $ of the even index q, any bivector can be reduced to \emph{canonical form} in some basis
\begin{equation}
\label{e19r}
    \frac{1}{2}R_{\alpha\beta}X^\alpha Y^\beta=
    R_{16}X^{\left[\right.1} Y^{6\left.\right]}+
    R_{23}X^{\left[\right.2} Y^{3\left.\right]}+
    R_{45}X^{\left[\right.4} Y^{5\left.\right]}\spsd
\end{equation}

\subsection{Fifth chapter}$ $
    \indent In the last chapter, the 8-dimensional complex space $\mathbb T^2$ constructs as the direct sum $\mathbb C^4\oplus \mathbb C^{*4}$ using a vector $X^a $ and a covector $Y_b$ from fibers of the bundles $A^\mathbb C$ and $A^{\mathbb C*}$ respectively
\begin{equation}
\label{e20r}
    X^A:=(X^a,Y_b)
\end{equation}
    and $X^A\in \mathbb T^2$. In this case, $X^a$ and $Y_b$ satisfy the following system
\begin{equation}
\label{e21r}
    \left\{
    \begin{array}{lcl}
    X^a & = & \dot X^a -ir^{ab}Y_b\sps \\
    Y_b & = & \dot Y_b\sps
    \end{array}
    \right.
\end{equation}
    where $r^{ab}$ are the coordinates of a bivector of $\Lambda^2\mathbb C\mathbb R^6$, and $\dot X^a$, $\dot Y_b$ are the values of $X^a$, $Y_b$ at the point O. In fact, the system (\ref{e21r}) can be regarded as \emph{bitwistor equation solutions}, and $\dot X^a$, $\dot Y_b$ will be its \emph{particular solutions}. Considering the locus of points, for which $X^a = 0$, we can come to Rosenfeld null-pairs and then we can formulate the following assertion.

\begin{theorem}(The triality principle for two B-cylinders).\\
In the projective space $\mathbb C \mathbb P_7$, there are two quadrics (two B-cylinders) with the following main properties:
\begin{enumerate}
\item The planar generator $\mathbb C \mathbb P_3$ of a one quadric will uniquely define the point R on the other quadric, and this mapping will be bijective.
\item The planar generator $\mathbb C \mathbb P_2$ of a one quadric will uniquely define the point R on the other quadric. But the point R of the second quadric can be associated to the manifold of planar generators $\mathbb C \mathbb P_2$ belonging to the same planar generator $\mathbb C \mathbb P_3$ of the first quadric.
\item The rectilinear generator $\mathbb C \mathbb P_1$ of a one quadric will uniquely define the rectilinear generator $\mathbb C \mathbb P_1$ of the other quadric, and this mapping will be bijective. And all the rectilinear generators belonging to the same planar generator $\mathbb C \mathbb P_3$ of the first quadric define the beam centered at R belonging to the second quadric.
     \end{enumerate}
\end{theorem}
     This allows us to introduce \emph{connecting operators} $\eta^A{}_{KL}$ such that
\begin{equation}
\label{e22r}
      r^A=\frac{1}{4}\eta^A{}_{KL}R^{KL}\sps
\end{equation}
      where $r^A$ are the coordinates of a vector of $\mathbb C\mathbb R^8$, and $R^{KL}$ are the coordinates of a spin-tensor of $\mathbb C\mathbb R^8$. Therefor, the operators $\eta^A{}_{KL}$ define the full Clifford algebra since this operators will satisfy the Clifford equation
\begin{equation}
\label{e23r}
      G_{AB}\delta_K{}^L=\eta_{AK}{}^R\eta_B{}^L{}_R+\eta_{BK}{}^R\eta_A{}^L{}_R\spsd
\end{equation}
      In this case, we will have the two metric tensors
\begin{equation}
\label{e24r}
      \varepsilon_{KLMN}=\eta^A{}_{KL}\eta_{AMN}\sps
      G_{AB}=\frac{1}{4}\eta_A{}^{KL}\eta_B{}^{MN}\varepsilon_{KM}\varepsilon_{LN}\sps
      \varepsilon_{(KL)(MN)}=\frac{1}{2}\varepsilon_{KL}\varepsilon_{MN}\spsd
\end{equation}
      With the first tensor we can raise and lower a pair of indices, and with the second we can make the specified operation with a single index. This imposes some severe constraints on the connecting operators such as
\begin{equation}
\label{e25r}
      \eta^A{}_{(MN)}=\frac{1}{8}\eta^A{}_{KL}\varepsilon^{KL}\varepsilon_{MN}\spsd
\end{equation}
      Such the connecting operators will determine \emph{structural constants} of the octonion algebra. Later on, this will lead to the double covering $Spin(8,\mathbb C)/\{\pm 1\}\cong SO(8,\mathbb C)$. Therefore, the operators $\eta^A{}_{KL}$ will be very similar to the connecting Norden operators $\eta^\alpha{}_{kl}$ in their properties.\\

\subsection{Conclusion} $ $
      \indent It should be noted that at the end of the thesis, Appendix is available
       in which all the necessary algebraic calculations are presented.\\
      \indent The main results of the dissertation were published in the press:
\foreignlanguage{russian}{
\begin{enumerate}
      \item "О бивекторах 6-мерных римановых пространств"\ [O bivektorakh 6-mernykh rimanovykh prostranstv].
            УТИС [UTIS], Уфа [Ufa], 1996, с. 59-61 [pp. 59-61];
      \item "О структуре тензора кривизны 6-мерных римановых пространств"\ [O structure tenzora krivizny 6-mernyx rimanovykh prostranstv].
             Вестник БГУ [Vestnik BGU], Уфа [Ufa], N2(I), 1996, с. 44-47 [pp. 44-47];
      \item "О твисторных расслоениях с 6-мерной базой"\ [O tvistornykh rassloeniyakh s 6-merno\u\i\ baso\u\i].
            МГС [MGS], Казань [Kazan'], 1997, с. 13 [p. 13];
      \item "О геометрии битвисторов"\ [O geometrii bitvistorov].
            РКСА [RKSA], Уфа [Ufa], 1997, с. 85-87 [pp. 85-87],
\end{enumerate}
      and reported at conferences:
\begin{enumerate}
      \item "Ленинские горы - 95"\ [Leninskie gory - 95], г. Москва [Moskva];
      \item "Чебышевские чтения - 96"\ [Chebyshevskie chteniya - 96], г. Москва [Moskva];
      \item "Лобачевские чтения - 97"\ [Lobachevskie chteniya - 97], г. Казань [Kazan'];
      \item many conferences and seminars in Ufa and seminars, held in Kazan (KSU, Department of Geometry).
\end{enumerate}
}
      \indent The author is grateful for the help in the preparation of the thesis to his supervisor Assoc. Prof. \`E. G. Ne\u\i fel'd  and the chair of geometry at KSU (Head of Department Professor B. N. Shapukov).

\subsection{Introduction for the English edition} $ $

      \indent  The Ph.D. thesis defence has taken place at 14:00, on December, 25th, 1997 at Kazan State University at the session of Dissertation Council (K 053.29.05) at Kazan State University located to the address: 18 Kremlevskaya str., Kazan, 420008, Russia. The Russian edition (on the pp. \pageref*{originb}-\pageref*{origine}) contains the original variant of the Ph. D. thesis with the corrected typing errors and the original numbering of the pages (pp. 1-121). Besides, two figures which have been lost at the thesis printing are included in the English edition. English translation of the Ph. D. thesis was executed in 2012. Transliteration on Cyrillic is given according to the scheme MR(new).\\

      Note that the constructed formalism for n=8  is \emph{initial induction step} for the construction of \emph{alternative-elastic group algebras} for n mod 8 =0 \cite{Andreev_2}. The spinor formalism for n=8 gives the opportunity to construct \emph{Lie operator analogues} for the spinor (and \emph{pair-spin}) bundle with the base: the space-time manifold and to transfer the metric into pair-spin fibers with the same base \cite{Andreev_0}, \cite{Andreev_1}. In addition, the spinor formalism allows us to construct \emph{hypercomplex Cayley-Dickson algebra generator} in the explicit form \cite{Andreev_2}. Moreover, the spinor formalism for even n can be constructed with the help of the particular solutions of the reduced Clifford equation \cite{Andreev_0}, \cite{Andreev_1}. Physical applications of this theory for n=6 can be found in \cite{Scharnhorst_1}.

\foreignlanguage{russian}{
\References{
     \bibitem{Andreev_0}
      К.В. Андреев [K.V. Andreev]: О спинорном формализме при четной размерности базового пространства [O spinornom formalizme pri chetno\u\i\ razmernosti bazovogo prostranstva]. ВИНИТИ - 298-B-11 [VINITI-298-V-11], июнь 2011 [iun' 2011]. [in Russian: On the spinor formalism for the base space of even dimension]
      \bibitem{Andreev_1}
      K.V. Andreev. On the spinor formalism for even n. [\href{http://arxiv.org/abs/1202.0941}{arXiv:1202.0941v2}].
      \bibitem{Andreev_2}
      K.V. Andreev. On the metric hypercomplex group alternative-elastic algebras for n mod 8 = 0. (\href{http://arxiv.org/abs/1110.4737}{arXiv:1110.4737v1}).
      \bibitem{Scharnhorst_1}
      K. Scharnhorst  and J.-W. van Holten: Nonlinear Bogolyubov-Valatin transformations: 2 modes. Annals of Physics (New York), 326(2011)2868-2933 [\href{http://arxiv.org/abs/1002.2737}{arXiv:1002.2737v3}, NIKHEF preprint NIKHEF/2010-005] (\href{http://dx.doi.org/10.1016/j.aop.2011.05.001}{DOI: 10.1016/j.aop.2011.05.001}).
}
}

\newpage
\section{Basic identities and formulas}
\Abstract{
    \indent This chapter is devoted to the study of algebraic properties of the double covering
    $$
    SO(6,\mathbb C)\cong SL(4,\mathbb C)/{\{\pm 1\}}\spsd
    $$
    On the basis of this isomorphism the elementary algebraic framework necessary for further investigations is constructed. To do this, various vector bundles with the base $\mathbb CV^6(\mathbb C\mathbb R^6)$ discuss. The tangent bundle $\tau^\mathbb C(\mathbb CV^6)$ which contains fibers isomorphic to $\mathbb C\mathbb R^6$ is isomorphic to $\Lambda$ with fibers isomorphic to $\Lambda^2 \mathbb C^4$ that follows from the existence of the connecting Norden operators
    $$
    r^\alpha=\frac{1}{2}\eta^\alpha{}_{aa_1}R^{aa_1}\sps
    $$
    where $(\alpha,\beta,...=1,2,3,4,5,6;a_1,b_1,a,b,...=1,2,3,4)$. In addition, we consider the bundle $A^\mathbb C$ with fibers isomorphic to $\mathbb C^4$ and the base $\mathbb CV^6(\mathbb C\mathbb R^6)$. From this, the existence of the operators $A_{\alpha\beta a}{}^b$
    $$
    T_{\alpha \beta}= A_{\alpha \beta d}{}^cT_c{}^d \sps
    T_a{}^a=0 \sps T_{\alpha \beta}=-T_{\alpha \beta}
    $$
    will imply. As follows from the results of the penultimate subsection of this chapter considering infinitesimal transformations, the resulting operators are an algebraic realization of the isomorphism between the Lie algebras
    $$
    so(6,\mathbb C)\cong sl(4,\mathbb C)\spsd
    $$
    Then we will study real inclusions with the help of the inclusion operator $H_i{}^\alpha$ and the involution $S_\alpha{}^{\beta '}$. The conjugation operation induced in the bundle $A^\mathbb C$ is divided into the two classes. In the first case (the space $V^6_{(p,q)}(\mathbb R^6_{(p,q)})$
    has the metric of the even index q), the conjugation is carried out by means of \emph{Hermitian polarity tensor $s^{aa '}$}
    $$
    \bar X^{a'}:=s^{aa'}X_a\spsd
    $$
    In the second case (q is odd), the conjugation is carried out by means of \emph{Hermitian involution tensor $s_a{}^{a'}$}
    $$
    \bar X^{a'}:=s_a{}^{a'}X^a\spsd
    $$
    The second subsection is just devoted to the elucidation of this fact which is proved by using the theorems from the monography \cite{Penrose1}.
 }

\subsection{\texorpdfstring{Bivectors of the space $\Lambda^2\mathbb C^4(\Lambda^2\mathbb R^4)$}{On bivectors}}
\subsubsection{Norden operators}$ $
    \indent It is known that one can establish the isomorphism between the complex Euclidean space $\mathbb C\mathbb R^6$ $(\mathbb R^6_{(3,3)})$ and the bivector space  $\Lambda^2\mathbb C^4(\Lambda^2\mathbb R^4)$. This isomorphism is determined by \emph{connecting Norden operators} \cite{Norden1} satisfying the following conditions
\begin{equation}
\label{e0}
    \frac{1}{2}\eta^\alpha{}_{aa_1}\eta_\beta{}^{aa_1}=\delta_\alpha{}^\beta\sps
    \eta^\alpha{}_{aa_1}\eta_\alpha{}^{bb_1}=\delta_{aa_1}{}^{bb_1}:=
    2\delta_{\left[\right.a}{}^{\left[\right.b}
    \delta_{a_1\left.\right]}{}^{b_1\left.\right]}
\end{equation}
    so that the following equations
\begin{equation}
\label{e1}
    r^{\alpha}=1/2\cdot \eta^{\alpha}{}_{aa_1}R^{aa_1}\sps
    R^{aa_1}=\eta_{\alpha}{}^{aa_1}r^{\alpha}\sps
\end{equation}
    where $(\alpha,\beta,...=1,2,3,4,5,6;$ $a_1,b_1,a,b,e,f,k,l,m,n,...=1,2,3,4;$ $i,j,g,h=1,2,3,4,5,6)$ and
\begin{equation}
\label{e2}
    \begin{array}{c}
    g^{\alpha\beta}=1/4\cdot \eta^{\alpha}{}_{aa_1} \eta^{\beta}{}_{bb_1}
    \varepsilon^{aa_1bb_1}\sps
    \varepsilon^{aa_1bb_1}=
    \eta_{\alpha}{}^{aa_1} \eta_{\beta}{}^{bb_1}g^{\alpha\beta}\sps\\[2ex]
    g_{\alpha\beta}=1/4\cdot \eta_{\alpha}{}^{aa_1} \eta_{\beta}{}^{bb_1}
    \varepsilon_{aa_1bb_1}\sps
    \varepsilon_{aa_1bb_1}=
    \eta^{\alpha}{}_{aa_1} \eta^{\beta}{}_{bb_1}g_{\alpha\beta}
    \end{array}
\end{equation}
     are executed. In this case, $R^{aa_1}$ are the coordinates of a bivector of the space $\Lambda^2\mathbb C^4$, and $r^{\alpha}$ are the coordinates of its image from $\mathbb C\mathbb R^6$; $g^{\alpha\beta}$ is the metric tensor of $\mathbb C\mathbb R^6$, its image is the spin-tensor (\emph{quadrivector}) $\varepsilon_{aa_1bb_1}$ antisymmetric in all indices.

\begin{note}
\label{note10}
    Note that with the help of the metric tensor $g_{\alpha\beta} $, defined on the space $\mathbb C\mathbb R^6$, we can raise and lower single indexs. Using the metric quadrivector $\varepsilon_{aa_1bb_1}$, defined in the bundle $\Lambda$, we can raise and lower pair skew-symmetric indices only, and there is no a metric spin-tensor with which one could do a similar operation with single indexes.
\end{note}
    It follows that there are the operators $A_{{\alpha\beta}d}{}^c$ such that
\begin{equation}
\label{e4}
    T_{\alpha \beta}= A_{\alpha \beta d}{}^cT_c{}^d \sps
    T_k{}^k=0 \sps T_{\alpha \beta}=-T_{\alpha \beta}\spsd
\end{equation}
    We will give a proof of this fact.
\begin{proof} For the Levi-Civita symbol, we have
\begin{equation}
\label{e7d}
    \begin{array}{c}
    \varepsilon^{abcd}\varepsilon_{klmn}=
    24\delta_{\left[\right.k}{}^a\delta_l{}^b\delta_m{}^c\delta_{n\left.\right]}{}^d\sps\\
    \varepsilon^{abcd}\varepsilon_{klmd}=
    6\delta_{\left[\right.k}{}^a\delta_l{}^b\delta_{m\left.\right]}{}^c\sps\\
    \varepsilon^{abcd}\varepsilon_{klcd}=
    4\delta_{\left[\right.k}{}^a\delta_{l\left.\right]}{}^b\sps\\
    \varepsilon^{abcd}\varepsilon_{kbcd}=
    6\delta_k{}^a\sps\\
    \varepsilon^{abcd}\varepsilon_{abcd}=
    24\sps\\
    \delta_{kl}^{ab}:=2\delta_{\left[\right.k}{}^a\delta_{l\left.\right]}{}^b\spsd
    \end{array}
\end{equation}
    Moreover, since $\varepsilon^{abcd}$ is the metric spin-tensor, it follows \cite[v.2, p. 65, eq. (6.2.19)(eng)]{Penrose1} that
\begin{equation}
\label{e8d}
    R_{ab}=\frac{1}{2}\varepsilon_{abcd}R^{cd}\spsd
\end{equation}
    Then, taking into account the formulas (\ref{e1}) and (\ref{e2}), for the tensor $T_{\alpha\beta}$, we have
    $$
    T_{\left[\alpha \beta \right]}=
    1/4\cdot\eta_\alpha{}^{aa_1}\eta_\beta{}^{bb_1}\cdot
    3/2(T_{a\left[ \right.a_1bb_1\left. \right]}-T_{\left[ \right. bb_1a
    \left.  \right]a_1})=
    $$
    $$
    =1/4\eta_\alpha{}^{aa_1}\eta_\beta{}^{bb_1}\cdot
    1/2(T_{ak}{}^{kd}\varepsilon_{a_1dbb_1}-T^{kd}{}_{ka_1}\varepsilon_
    {adbb_1})=
    1/2\eta_\alpha{}^{ca}\eta_\beta{}^{bk_1}\varepsilon_{adbk_1}\cdot
    $$
\begin{equation}
\label{e20}
    1/4(T^{kd}{}_{kc}-T_{kc}{}^{kd})
    =-1/2\eta_\alpha{}^{bk_1}\eta_\beta{}^{ca}\varepsilon_{adbk_1}\cdot
    1/4(T^{kd}{}_{kc}-T_{kc}{}^{kd})\spsd
\end{equation}
    Therefore, we must put
    $$
    A_{\alpha \beta d}{}^c:=\frac{1}{2}\eta_{\left[\alpha \right.}{}^{ca}\eta_{\left.
    \beta \right]}{}^{bk_1}\varepsilon_{dabk_1}
    \sps T_c{}^d:=1/4(T_{kc}{}^{kd}-T^{kd}{}_{kc})\sps
    $$
    $$
    T_{aa_1bb_1}-T_{bb_1aa_1}=4 \varepsilon_{bb_1c\left[\right.a_1}
    T_{a\left.\right]}{}^c
    $$
    to obtain the formula (\ref{e4}).
\end{proof}

    The most important relations for the operators $A_{\alpha\beta b}{}^a$ have the form
\begin{equation}
\label{e21}
    \begin{array}{c}
    A_{\alpha\beta d}{}^cA^{\alpha\beta}{}_r{}^s=\frac{1}{2}\delta_r{}^s\delta_d{}^c-
    2\delta_ r{}^c\delta_d{}^s\sps
    A_{\alpha\beta d}{}^cA^{\lambda\mu}{}_c{}^d=
    2\delta_{\left[\alpha \right.}{}^\mu\delta_{\left.\beta\right]}{}^\lambda\sps
    \\
    A_{\alpha\beta d}{}^cA_\gamma{}^\beta{}_r{}^s=
    (\eta_\alpha{}^{cs}\eta_\gamma{}_{rd}+
    \eta_\alpha{}^{ck}\eta_\gamma{}_{kr}\delta_d{}^s)+
    1/2(\eta_\alpha{}^{sn}\eta_\gamma{}_{rn}\delta_d{}^c+
    \\
    +\eta_\alpha{}^{ck}\eta_\gamma{}_{dk}\delta_r{}^s)-
    1/4g_{\alpha\gamma}\delta_r{}^s\delta_d{}^c\sps
    \\
    T_m{}^n=\frac{1}{2}A^{\alpha\beta}{}_m{}^nT_{\beta\alpha}
    \sps T_{\beta\alpha}=-T_{\alpha\beta}\spsd
    \\
    \end{array}
\end{equation}
    The proof of these formulas is rather cumbersome, and therefore it is given in Appendix (\ref{e1p})-(\ref{e5p}).\\
    \indent It should be noted that the connecting Norden operators define the Clifford algebra.
\begin{proof} Consider the identity
\begin{equation}
\label{e15v}
    \eta_{\left(\right.\alpha}{}^{a\left[\right.b}
    \eta_{\beta\left.\right)}{}^{cd\left.\right]}=
    -\eta_{\left(\right.\alpha}{}^{\left[\right.cb}
    \eta_{\beta\left.\right)}{}^{|a|d\left.\right]}=
    \eta_{\left(\right.\alpha}{}^{\left[\right.ab}
    \eta_{\beta\left.\right)}{}^{cd\left.\right]}\sps
\end{equation}
    the contraction of which with $\varepsilon_{bcdn}$ gives the relations
\begin{equation}
\label{e16v}
    \begin{array}{c}
    \eta_{\left(\right.\alpha}{}^{ab}\eta_{\beta\left.\right)}{}^{cd}
    \varepsilon_{bcdn}=\frac{1}{24}
    \eta_{\left(\right.\alpha}{}^{kl}\eta_{\beta\left.\right)}{}^{mn}
    \varepsilon_{klmn}\varepsilon^{abcd}\varepsilon_{bcdn}\sps\\[2ex]
    2\eta_{\left(\right.\alpha}{}^{ab}\eta_{\beta\left.\right)}{}_{bn}=
    -\frac{1}{4}
    \eta_{\left(\right.\alpha}{}^{kl}\eta_{\beta\left.\right)}{}^{mn}
    \varepsilon_{klmn}\delta_n{}^a\sps\\[2ex]
    \eta_{\left(\right.\alpha}{}^{ab}\eta_{\beta\left.\right)}{}_{nb}=
    \frac{1}{2}g_{\alpha\beta}\delta_n{}^a\sps
    \end{array}
\end{equation}
    where $ g_{\alpha\beta}$ is the same as in the formulas (\ref{e2}). We define
\begin{equation}
\label{e17v}
    \eta_\alpha:=\parallel \eta_\alpha{}^{aa_1} \parallel \sps
    \sigma_\alpha:=\parallel -\eta_\alpha{}_{aa_1} \parallel\sps
    \gamma_\alpha:=\sqrt{2}\left(
    \begin{array}{cc}
    0           & \sigma_\alpha \\
    \eta_\alpha & 0
    \end{array}
    \right)\sps
\end{equation}
    where $\lambda ,\psi,...=\overline{1,6}$. The operators $\gamma_\lambda$ satisfy the identity
\begin{equation}
\label{e18v}
    \gamma_\lambda\gamma_\psi+\gamma_\psi\gamma_\lambda=2g_{\lambda\psi}I
\end{equation}
    which follows from the last equation of (\ref{e16v}). The equality (\ref{e18v}) is the Clifford equation \cite[v.2, p. 441, eq. (B.1)(eng)]{Penrose1} such that $\gamma_1, \gamma_2, \gamma_3, \gamma_4, \gamma_5, \gamma_6$ can be represented by means of complex matrixes of dimension $8\times 8,\ g_{\alpha\beta}$ is the metric tensor (\ref{e2}), and I is the identity operator.\\
    \indent The opposite is true. Suppose we have the equation (\ref{e18v}). Then we can construct the element $\gamma_7$
\begin{equation}
\label{e18va}
    \gamma_7:=\gamma_1\gamma_2\gamma_3\gamma_4\gamma_5\gamma_6\sps
    (\gamma_7){}^2=I\spsd
\end{equation}
    In this case, since n = 6 (even), $\gamma_7$ anticommutes with every element $\gamma_\alpha\ (\alpha =\overline{1,6})$. This means that for $\gamma_\alpha $, the representation  (\ref{e17v})  is possible, and therefore the identity (\ref{e15v}) is executed.\\
    \indent It follows that the connecting Norden operators determine the full Clifford algebra which is formed by the finite sums
\begin{equation}
\label{e18vab}
     AI+B^\lambda\gamma_\lambda+C^{\lambda\mu}\gamma_\lambda\gamma_\mu+...
\end{equation}
    Dimension of this algebra is equal to $2^6=64$. Such the algebra can be represented by the full matrix algebra, elements of that have the dimension $8\times 8$ \cite[v. 2, p. 440-464(eng)]{Penrose1}.
\end{proof}

\subsubsection[Conjugation in the bundle A]{Conjugation in the bundle $A^\mathbb C(S)$}
\Abstract{
    \indent In this subsection the statement concerning an inclusion of real spaces in the complex one is formulated, the proof is given in the next subsection.
}

    Consider the 6-dimensional (pseudo-)Euclidean space $\mathbb R^6_{(p, q)}$ embedded in $\mathbb C\mathbb R^6$, the tangent space $\tau^\mathbb R(\mathbb R^6_{(p, q)})$ of which we will consider as a real subspace in $\mathbb C^4$. This will lead to the vector bundle $A^\mathbb C(S)$ with fibers isomorphic to $\mathbb C^4$. Besides, the bundle $A^\mathbb C(S)$ will be equipped with \emph{structure s}. We need to clarify the nature of this structure. To do this, we will consider a simple bivector of the $\tau^\mathbb C(\tau^\mathbb R)$. A necessary and sufficient condition of the simplicity for this bivector is expressed by the formula
\begin{equation}
\label{e5}
    R^{ab}=X^aY^b-X^bY^a\sps X^a\, Y^a\in \mathbb C^4\sps
\end{equation}
    where $i,j=\overline{1,6};a,b,c,d,k,l,m,n,...=\overline{1,4}$.\\

\begin{note}
\label{note11}
    Based on the formula (\ref{e1}), a simple bivector of the space $\Lambda^2 \mathbb C^4$ is uniquely associated with an isotropic vector of the space $\mathbb C\mathbb R^6$, and this mapping will be bijective. This follows from the relation
\begin{equation}
\label{e19v}
    \begin{array}{c}
    0=\varepsilon_{abcd}X^aY^bX^cY^d=
    \frac{1}{4}\varepsilon_{abcd}R^{ab}R^{cd}
    =\frac{1}{2}\eta_\alpha{}^{km} r^\alpha \eta^\beta{}_{km}r_\beta=
    r^\alpha r_\alpha=0\spsd
    \end{array}
\end{equation}
\end{note}
    Further, any bivector should be \emph{self-conjugated} with respect to the spin-tensor $s_{aa_1b'b'{}_1}$ (the notation introduced according to \cite{Penrose1})
\begin{equation}
\label{e6}
    s_{aa_1b'b'{}_1}=g_{ij}\eta^i{}_{b'b'{}_1}
    \eta^j{}_{aa_1}\sps
    g_{ij}=1/4\cdot\eta_i{}^{b'b'{}_1}\eta_j{}^{aa_1}
    s_{aa_1b'b'{}_1}\sps
\end{equation}
    $$
    \bar s_{b'b'{}_1aa_1}=s_{aa_1b'b'{}_1}\sps
    $$
    where $g_{ij}$ is the metric tensor (\ref{e10v}). The last equation expresses the Hermitian symmetry of the spin-tensor s. Such the tensor was introduced in \cite{Neifeld1}. In the case of the metric of the even index, the spin-tensor $s_{aa_1}{}^{b'b'{}_1}$ (the raising and lowering of double indices are carried out with the help of the metric quadrivector $\varepsilon_{aa_1bb_1}$) has the form
\begin{equation}
\label{e7}
    s_{aa_1}{}^{b'b'{}_1}=s^{cb'}s^{c_1b'{}_1} \varepsilon_{cc_1aa_1}
    \sps \bar s_{a'b}=\pm s_{ba'}\sps
\end{equation}
    and in the case of the odd index, we obtain
\begin{equation}
\label{e8}
    s_{aa_1}{}^{b'b'{}_1}=2s_{\left[ a \right.}{}^{b'}s_{\left. a_1\right]}
    {}^{b'{}_1}\sps s_a{}^{b '}\bar s_{b '}{}^c=\pm\delta_a{}^c\spsd
\end{equation}
     If $R^{ab}$ is simple and belongs to the tangent space $\tau^\mathbb R$ then for the vectors, defining the bivector, in the case of the even index metric, the identities \cite[v. 2, p. 63, eq. (6.2.13)(eng)]{Penrose1}
\begin{equation}
\label{e9}
    X_aX_{a'}s^{aa'}=0 \sps Y_aY_{a'}s^{aa'}=0 \sps X_aY_{a'}s^{aa'}=0
\end{equation}
    and in the case of the odd index metric, the identities
\begin{equation}
\label{e10}
    X_{a'}X^as_a{}^{a'}=0 \sps
    X_{a'}Y^as_a{}^{a'}=0 \sps
    Y_{a'}Y^as_a{}^{a'}=0
\end{equation}
     are executed. Thus, the structure s of the bundle $A^\mathbb C(S)$ is determined. In the case of the even index metric, the spin-tensor $s_{kk'}$  fulfills a role of the metric spin-tensor with which the help we can raise and lower single indexs, and in the case of the odd index metric, with the help of the spin-tensor $s_k{}^{k'}$, the identification between the primed (\emph{complex conjugate}) and unprimed spaces is carried out. The proof of this assertion is given in the next section.

\subsection{Spinor representation of special form tensors. The covering corresponding to this decomposition}
\subsubsection{\texorpdfstring{Theorem on the double covering $SL(4,\mathbb C)/\{\pm 1\}\cong SO(6, \mathbb C)}{Theorem on the double covering}$}
\Abstract {
     \indent Before the proof, we need to more thoroughly understand with the double covering $SL(4,\mathbb C)/\{\pm 1\}\cong SO(6,\mathbb C)$. Below, the explicit representation of this covering by means of the connecting Norden operators $\eta_\alpha{}^{ab}$ will be obtained. Using this representation, it is easier to understand how the real inclusion $\mathbb R^6_{(2,4)}\subset \mathbb C\mathbb R^6$ is occurs and hence how to construct the explicit representation of the involution operator in the spinor form. In addition, the results of this section will be useful in the study of the bivector structure of the space $\mathbb C\mathbb R^6$.
}

     By $K_\alpha{}^\beta$, we denote a transformation form the group  $SO(6,\mathbb C)$, and let $S_a{}^b$ be a transformation from the group $SL(4,\mathbb C)$. Then the following theorem will be true.
\begin{theorem}
\label{theorem5} Each transformation $K_\alpha{}^\beta$ corresponds to two and only two transformations $\pm S_a{}^b$($\pm S_{ab}$) such that
    $det\parallel S_c{}^d \parallel=1\ (det\parallel S_{cd} \parallel=1)$. And on the contrary, each transformation $\pm S_a{}^b$($\pm S_{ab}$) corresponds to one and only one transformation $K_\alpha{}^\beta$.
\end{theorem}

\begin{proof} Suppose, that there are the two transformations $\pm S_a{}^b (\pm S_ {ab})$ such that
     $$
     \varepsilon:=\varepsilon_{1234}\sps
     $$
\begin{equation}
\label{e1b}
    S_a{}^b S_{a_1}{}^{b_1} S_c{}^d S_{c_1}{}^{d_1}\varepsilon_{bb_1dd_1}=
    \varepsilon_{aa_1cc_1}\
    (S_{ab} S_{a_1b_1} S_{cd} S_{c_1d_1}\varepsilon^{bb_1dd_1}=
    \varepsilon^{-2}\varepsilon_{aa_1cc_1})\spsd
\end{equation}
     The last equation means that $det\parallel S_c{}^d \parallel=1\ (det\parallel S_{cd} \parallel=1)$. Define
\begin{equation}
\label{e2b}
    \begin{array}{c}
    K_\alpha{}^\beta:=\frac{1}{4}\eta_\alpha{}^{ab}\eta^\beta{}_{cd}
    \cdot 2\beta S_{\left[ a \right.}{}^c S_{\left. b \right]}{}^d\
    (K_\alpha{}^\beta:=\frac{1}{4}\eta_\alpha{}^{ab}\eta^\beta{}_{cd}
    \cdot\varepsilon  \beta S_{ak} S_{bl}\varepsilon^{klcd})\sps\\[2ex]
    \beta:=\pm 1\spsd
    \end{array}
\end{equation}
    Then on the basis of (\ref{e1b}) and (\ref{e2b}), we will obtain
\begin{equation}
\label{e3b}
    K_\alpha{}^\beta K_\gamma{}^\delta g_{\beta\delta}=g_{\alpha\gamma}\spsd
\end{equation}
    Thus, from (\ref{e1b}), the equation (\ref{e3b}) will follow.\\
    \indent If now, on the contrary, $K_\alpha{}^\beta$ of the form (\ref{e3b}) is set. Define
\begin{equation}
\label{e4b}
    K_{aa_1}{}^{bb_1}:=\eta^\alpha{}_{aa_1}\eta_\beta{}^{bb_1}K_\alpha{}^\beta\spsd
\end{equation}
    Thus, (\ref{e3b}) can be copied as
\begin{equation}
\label{e46ab}
    \frac{1}{4}K_{aa_1}{}^{bb_1}K_{cc_1}{}^{dd_1}
    \varepsilon_{bb_1dd_1}=\varepsilon_{aa_1cc_1}\spsd
\end{equation}
    The formula (\ref{e46ab}) means that the transformation $K_{aa_1}{}^{bb_1}$ should be \emph{regular}, i.e., $\forall r^{aa_1}\ne 0 \ \Rightarrow\ K_{aa_1}{}^{bb_1}r^{aa_1}\ne 0,\ K_{aa_1}{}^{bb_1}r_{bb_1}\ne 0$.\\
\begin{proof}
    Indeed, we will assume the contrary: $\exists r^{aa_1}\ne 0$ and $K_{aa_1}{}^{bb_1}r_{bb_1}=0$. It will mean that the transformation $K_{aa_1}{}^{bb_1}$ is \emph{singular}
\begin{equation}
\label{e46abc}
     0=r_{bb_1}\cdot\frac{1}{4} K_{aa_1}{}^{bb_1} K_{cc_1}{}^{dd_1}
     \varepsilon_{bb_1dd_1}=\varepsilon_{aa_1cc_1}r^{aa_1}\spsd
\end{equation}
   From this, the equality $r^{aa_1}=0$ will follows. Contradiction.
\end{proof}

   For further calculations we need the following lemma.

\begin{lemma}
\label{lemma1}
   Choose two non-zero vectors $r_1{}^\alpha ,r_2{}^\beta \in \tau^\mathbb C_x$, where $x$ is an arbitrary point at the base: the complex Euclidean space $\mathbb C\mathbb R^6$ equipped with the metric tensor $g_{\alpha\beta}$. Then the three following conditions are equivalent:
\begin{enumerate}
     \item  $r_1{}^\alpha(r_1)_\alpha=0$ ,
            $r_2{}^\alpha(r_2)_\alpha=0$ ,
            $r_1{}^\alpha(r_2)_\alpha=0$ ;
     \item  $r_1{}^\alpha=\frac{1}{2}\eta^\alpha{}_{aa_1} X^aY^{a_1}$,
            $r_2{}^\alpha=\frac{1}{2}\eta^\alpha{}_{aa_1} X^aZ^{a_1}$;
     \item  $(r_1)_\alpha=\frac{1}{2}\eta_\alpha{}^{aa_1} \tilde X_a\tilde Y_{a_1}$,
            $(r_2)_\alpha=\frac{1}{2}\eta_\alpha{}^{aa_1} \tilde X_a\tilde Z_{a_1}$;
    \end{enumerate}
    where the vectors $X^a,Y^a,Z^a$ belong to the fiber $\mathbb C^4_x$ of the bundle $A^\mathbb C$, and
    $\tilde X_a,\tilde Y_a,\tilde Z_a$ are covectors of the dual fiber.
\end{lemma}

\begin{proof} $1).\Rightarrow 2).$\\
    Consider the first equation of Condition 1). and define
\begin{equation}
\label{e1vv}
    r_1{}^{aa_1}:=\eta_\alpha{}^{aa_1}r_1{}^\alpha\sps
\end{equation}
    and then on the basis of (\ref{e2}), we obtain
\begin{equation}
\label{e2vv}
    \begin{array}{c}
    r_1{}^\alpha (r_1)_\alpha=g_{\alpha\beta}r_1{}^\alpha r_1{}^\beta=
    \frac{1}{4}\eta_\alpha{}^{aa_1}\eta_\beta{}^{cc_1}
    \varepsilon_{aa_1cc_1}\frac{1}{2}\eta^\alpha{}_{dd_1}r_1{}^{dd_1}
    \frac{1}{2}\eta^\beta{}_{kk_1}r_1{}^{kk_1}=\\[2ex]
    =\frac{1}{4}r_1{}^{dd_1}
    \delta_{\left[\right.d}{}^a\delta_{d_1\left.\right]}{}^{a_1}
    r_1{}^{kk_1}
    \delta_{\left[\right.k}{}^c\delta_{k_1\left.\right]}{}^{c_1}
    \varepsilon_{aa_1cc_1}=
    \frac{1}{2}
    r_1{}^{aa_1}(r_1)_{aa_1}=0\spsd
    \end{array}
\end{equation}
    Define
\begin{equation}
\label{e3vv}
    pf(r):=r^\alpha r_\alpha=\frac{1}{2}r^{aa_1}r_{aa_1}\spsd
\end{equation}
    Then from (\ref{e16v}), the equalities
\begin{equation}
\label{e4vv}
    r^{ab}r_{bc}=-pf(r)\delta_c{}^d\ \ \Leftrightarrow\ \ \
    3r^{a\left[\right.b}r^{cd\left.\right]}=pf(r)\varepsilon^{abcd}
\end{equation}
    will follow, and for the vector $r_1{}^\alpha$, we obtain
\begin{equation}
\label{e5vv}
    r_1{}^{ab}r_1{}^{cd}=r_1{}^{ac}r_1{}^{bd}-r_1{}^{bc}r_1{}^{ad}\spsd
\end{equation}
    Since, $r_1{}^{cd}$ is a non-zero bivector then some covectors $A_c,B_d$, that $r_1{}^{cd}A_cB_d\ne 0$, $r_1{}^{cd}A_cB_d\in \mathbb R$, exist. Put
\begin{equation}
\label{e6vv}
    P^a:=\sqrt{2}r_1{}^{ak}A_k/\sqrt{(r_1{}^{cd}A_cB_d)}\sps
    Q^a:=\sqrt{2}r_1{}^{bk}B_k/\sqrt{(r_1{}^{cd}A_cB_d)}\spsd
\end{equation}
    Then from (\ref{e5vv}), the equality
\begin{equation}
\label{e7vv}
     r_1{}^{ab}=P^{\left[\right.a}Q^{b\left.\right]}
\end{equation}
     will follow. In this case, $P^a, Q^a$ are linearly independent. It is also possible to obtain the expansion for $r_2{}^{ab}$
\begin{equation}
\label{e8vv}
     r_2{}^{ab}=R^{\left[\right.a}S^{b\left.\right]}
\end{equation}
     from the second equation of Condition 1). such that the vectors $R^a,S^a$ will be also linearly independent. From the third equation of Condition 3). implies the following relation
\begin{equation}
\label{e9vv}
    \begin{array}{c}
    0=r_1{}^\alpha(r_2)_\alpha=\frac{1}{4}
    \varepsilon_{abcd}r_1{}^{ab}r_2{}^{cd}=
    \frac{1}{4}\varepsilon_{abcd}P^aQ^bR^cS^d=0\spsd
    \end{array}
\end{equation}
    This means that the vectors $P^a,Q^b,R^c,S^d$ are linearly dependent
\begin{equation}
\label{e10vv}
    \alpha P^a+\beta Q^a+\gamma R^a+\delta S^a=0\sps
    |\alpha|+|\beta|+|\gamma|+|\delta|\ne 0\spsd
\end{equation}
    In this case, either $\alpha \ne 0 $ or $\beta \ne 0$.
    Otherwise, $(\alpha = \beta = 0)$, and the vectors  $R^c,S^c$ would be linearly dependent. For definiteness, let $\alpha \ne 0$. Then again, either $\gamma \ne 0 $ or $\delta \ne 0$. Put
\begin{equation}
\label{e11vv}
    \begin{array}{c}
    X^a:=P^a+(\beta/\alpha) Q^a=-((\gamma/\alpha) R^a+(\delta/\alpha) S^a)\sps
    Y^a:=Q^a\sps\\[2ex]
    Z^a:=\left\{
    \begin{array}{l}
    \ (\alpha/\delta) R^a,\ \delta\ne 0,\gamma =0\sps \\
    -(\alpha/\gamma) S^a,\ \gamma \ne 0\spsd
    \end{array}
    \right.
    \end{array}
\end{equation}
    Thus, from (\ref{e11vv}), Condition 2). of the lemma implies. \\
    $2).\Rightarrow 1).$\\
    It is verified directly, for example,
\begin{equation}
\label{e12vv}
    r_1{}^\alpha(r_2)_\alpha=\frac{1}{4}\varepsilon_{aa_1bb_1}
    X^aY^{a_1}X^bZ^{b_1}=0\spsd
\end{equation}
    \indent In the same way, the equivalences $1).\Rightarrow 3).$ and $3).\Rightarrow 1).$ can be proved. These implications are possible because of the metric tensor presence in the tangent bundle and the metric quadrivector presence in the bundle $A^\mathbb C$.
\end{proof}
    Take two non-zero isotropic vectors
\begin{equation}
\label{e13vv}
    \begin{array}{c}
    r_1{}^\alpha=\frac{1}{2}\eta^\alpha{}_{aa_1} M^aN^{a_1}\sps
    r_2{}^\alpha=\frac{1}{2}\eta^\alpha{}_{aa_1} M^aL^{a_1}
    \end{array}
\end{equation}
    and two non-zero isotropic covectors
\begin{equation}
\label{e14vv}
    (\tilde r_1)_\alpha=\frac{1}{2}\eta_\alpha{}^{aa_1} \tilde M_a\tilde N_{a_1}\sps
    (\tilde r_2)_\alpha=\frac{1}{2}\eta_\alpha{}^{aa_1} \tilde M_a\tilde L_{a_1}
\end{equation}
    satisfying Condition 2). and Condition 3). of Lemma \ref{lemma1} respectively. We act on (\ref{e13vv}), (\ref{e14vv}) with the orthogonal transformation $K_\alpha{}^\beta$ and obtain
\begin{equation}
\label{e15vv}
    \begin{array}{c}
    r_3{}^\alpha:=K_\beta{}^\alpha r_1{}^\beta\sps
    r_4{}^\alpha:=K_\beta{}^\alpha r_2{}^\beta\sps\\[2ex]
    (\tilde r_3)_\alpha:=K_\alpha{}^\beta (\tilde r_1)_\beta\sps
    (\tilde r_4)_\alpha:=K_\alpha{}^\beta (\tilde r_2)_\beta\spsd
    \end{array}
\end{equation}
    Then from Condition 1). of Lemma \ref{lemma1}  with the account (\ref{e3b}) and (\ref{e46ab}), the equations
\begin{equation}
\label{e16vv}
    \begin{array}{c}
    r_3{}^\alpha (r_3)_\alpha=K_\alpha{}^\beta K_\gamma{}^\delta
    g_{\beta\delta} r_1{}^\alpha r_1{}^\gamma =
    r_1{}^\alpha (r_1)_\alpha=0\sps\\[2ex]
    r_4{}^\alpha (r_3)_\alpha=r_2{}^\alpha (r_1)_\alpha=0\sps
    r_4{}^\alpha (r_4)_\alpha=r_2{}^\alpha (r_2)_\alpha=0\sps\\[2ex]
    \tilde r_4{}^\alpha (\tilde r_3)_\alpha=\tilde r_2{}^\alpha (\tilde r_1)_\alpha=0\sps
    \tilde r_4{}^\alpha (\tilde r_4)_\alpha=\tilde r_2{}^\alpha (\tilde r_2)_\alpha=0\sps\\[2ex]
    \tilde r_3{}^\alpha (\tilde r_3)_\alpha=\tilde r_1{}^\alpha \tilde r_1{}_\alpha=0
    \end{array}
\end{equation}
    will follow. Since, the transformation $K_{aa_1}{}^{bb_1}$ is regular then the vectors and covectors (\ref{e15vv}) are non-zero elements, and hence from Condition 2). and Condition 3). of Lemma \ref{lemma1}, we obtain
\begin{equation}
\label{e17vv}
    \begin{array}{c}
    r_3{}^\alpha=\frac{1}{2}\eta^\alpha{}_{aa_1} X^a Y^{a_1}\sps
    r_4{}^\alpha=\frac{1}{2}\eta^\alpha{}_{aa_1} X^a Z^{a_1}\\[2ex]
    (\tilde r_3)_\alpha=\frac{1}{2}\eta_\alpha{}^{aa_1}\tilde X_a\tilde  Y_{a_1}\sps
    (\tilde r_4)_\alpha=\frac{1}{2}\eta_\alpha{}^{aa_1}\tilde X_a\tilde  Z_{a_1}\spsd
    \end{array}
\end{equation}
    Consider the identity
\begin{equation}
\label{e18vv}
    r_3{}^{\left[\right.\alpha}r_4{}^{\beta \left.\right]}=
    K_{\left[\right.\gamma}{}^{\left[\right.\alpha}
    K_{\delta \left.\right]}{}^{\beta \left.\right]}
    r_1{}^{\left[\right.\gamma}r_2{}^{\delta \left.\right]}\spsd
\end{equation}
    We rewrite it using the formulas (\ref{e4}) and (\ref{e21})
\begin{equation}
\label{e19vv}
    \begin{array}{c}
    A^{\alpha\beta}{}_a{}^b\cdot\frac{1}{4}X^aY^{b_1}X^cZ^{c_1}
    \varepsilon_{cc_1bb_1}=\\
    =\frac{1}{4}A^{\gamma\delta}{}_r{}^s M^rN^{k_1}M^l L^{l_1}\varepsilon_{ll_1sk_1}
    A_{\gamma\delta c}{}^d A^{\alpha\beta}{}_a{}^b\cdot
    \frac{1}{8}(K_{dm}{}^{ak}K^{cm}{}_{bk}-K^{cmak}K_{dmbk})\sps\\[2ex]
    X^aY^{b_1}X^cZ^{c_1}\varepsilon_{cc_1bb_1}=\\
    =2 \delta_r{}^d
    \delta_c{}^s M^rN^{k_1}M^lL^{l_1}\varepsilon_{k_1ll_1s}
    \cdot\frac{1}{8}(K_{dm}{}^{ak}K^{cm}{}_{bk}-K^{cmak}K_{dmbk})\sps\\[2ex]
    X^a(Y^{b_1}X^cZ^{c_1}\varepsilon_{cc_1bb_1})=\\
    =M^d(N^{k_1}M^lL^{l_1}\varepsilon_{k_1ll_1c})
    \cdot\frac{1}{4}(K_{dm}{}^{ak}K^{cm}{}_{bk}-K^{cmak}K_{dmbk})\spsd
    \end{array}
\end{equation}
    Define
\begin{equation}
\label{e20vv}
    \begin{array}{c}
    T_b:=Y^{b_1}X^cZ^{c_1}\varepsilon_{cc_1bb_1}\sps
    P_c:=N^{k_1}M^lL^{l_1}\varepsilon_{k_1ll_1c}\sps\\[2ex]
    \tilde K_d{}^c{}_b{}^a:=
    \frac{1}{8}(K^{cmak}K_{dmbk}-K_{dm}{}^{ak}K^{cm}{}_{bk})
    \end{array}
\end{equation}
   such that the equations
\begin{equation}
\label{e21vv}
    X^cT_c=0\sps
    M^cP_c=0\sps
    \tilde K_c{}^c{}_b{}^a=0\sps\tilde K_d{}^c{}_b{}^b=0\sps
\end{equation}
\begin{equation}
\label{e22vv}
    K_{aa_1}{}^{bb_1}K_{cc_1}{}^{dd_1}-
    K_{aa_1}{}^{dd_1}K_{cc_1}{}^{bb_1}=
    8\varepsilon_{cc_1k\left[\right.a_1}
    \tilde K_{a\left.\right]}{}^k{}_r{}^{\left[\right.b}
    \varepsilon^{b_1\left.\right]rdd_1}
\end{equation}
    are executed. Whence,
\begin{equation}
\label{e23vv}
    X^aT_b=-2 M^dP_c\tilde K_d{}^c{}_b{}^a\spsd
\end{equation}
    In the same way, from the identity
\begin{equation}
\label{e24vv}
    (\tilde r_3)_{\left[\right.\gamma}
    (\tilde r_4)_{\delta\left.\right]}=
    K_{\left[\right.\gamma}{}^{\left[\right.\alpha}
    K_{\delta\left.\right]}{}^{\beta\left.\right]}
    (\tilde r_1)_{\left[\right.\alpha}
    (\tilde r_2)_{\beta\left.\right]}\sps
\end{equation}
    determining
\begin{equation}
\label{e25vv}
    \tilde T^b:=\tilde Y_{b_1}\tilde X_c\tilde Z_{c_1}\varepsilon^{cc_1bb_1}\sps
    \tilde P^b:=\tilde N_{k_1}\tilde M_l\tilde Z_{l_1}\varepsilon^{k_1ll_1b}\sps
\end{equation}
    we can obtain
\begin{equation}
\label{e26vv}
    \tilde X_d\tilde T^c=-2 \tilde M_a\tilde P^b\tilde K_d{}^c{}_b{}^a\spsd
\end{equation}
     We now find \emph{homogeneous solution} satisfying the equations (\ref{e23vv}) and (\ref{e26vv})
\begin{equation}
\label{e27vv}
    \left\{
    \begin{array}{c}
    (\tilde K_{\mbox{homogeneous}})_d{}^c{}_b{}^aM^dP_c=0\sps \\
    (\tilde K_{\mbox{homogeneous}})_d{}^c{}_b{}^a\tilde M_a\tilde P^b=0\sps
    \end{array}
    \right.
    \ \ \Leftrightarrow \ \ \
    \left\{
    \begin{array}{c}
    M^dP_d=0\sps  \\
    \tilde M_a\tilde P^a=0\spsd
    \end{array}
    \right.
\end{equation}
    These two systems should coincide identically since the left system is valid for each $M^a,\tilde M_a,P_a,\tilde P^a$ satisfying the right system. This is possible only when
\begin{equation}
\label{e28vv}
    (\tilde K_{\mbox{homogeneous}})_d{}^c{}_b{}^a=\alpha \delta_d{}^c
    \delta_b{}^a \sps \alpha \in \mathbb C\spsd
\end{equation}
    Next, we consider \emph{particular solution} of the equation (\ref{e23vv}), for example. This solution should be regular that means that we can not satisfy the condition
\begin{equation}
\label{e29vv}
    \exists\ M^d\ne 0 , P_c\ne 0\sps{\mbox{   that   }}\
    (\tilde K_{\mbox{particular}})_d{}^c{}_b{}^aM^dP_c=0
\end{equation}
    (The condition (\ref{e29vv}) is equivalent to the singularity of the transformation $K_{aa_1}{}^{bb_1}$ (see (\ref{e22vv}))). To solve (\ref{e23vv}), we need the following lemma.

\begin{lemma}
\label{lemma2}
    Let {\bf {A,B,C,...}} be collective indices. Then the three following condition on $\lambda_{\mbox{\scriptsize\bf AB}}{}^{\mbox{\scriptsize\bf Q}}$ are equivalent:
\begin{enumerate}
    \item $\lambda_{\mbox{\scriptsize\bf AB}}{}^{\mbox{\scriptsize\bf Q}}
          \xi_{\mbox{\scriptsize\bf Q}}$ can be represented as $\rho_{\mbox{\scriptsize\bf A}}
          \xi_{\mbox{\scriptsize\bf B}}$ for each $\xi_{\mbox{\scriptsize\bf Q}}$;
    \item $\lambda_{{\mbox{\scriptsize\bf A}_1}\left[\right. {\mbox{\scriptsize\bf B}_1}}{}^
                {\left(\right.{\mbox{\scriptsize\bf  Q}_1}}
           \lambda_{|{\mbox{\scriptsize\bf A}_2}|{\mbox{\scriptsize\bf B}_2}\left.\right]}{}^
                {{\mbox{\scriptsize\bf  Q}_1}\left.\right)}$=0;
    \item $\lambda_{\mbox{\scriptsize\bf AB}}{}^{\mbox{\scriptsize\bf Q}}$
          can be represented either as $\alpha_{\mbox{\scriptsize\bf A}}
          \varphi_{\mbox{\scriptsize\bf B}}{}^{\mbox{\scriptsize\bf Q}}$,
          or as $\theta_{\mbox{\scriptsize\bf A}}{}^
          {\mbox{\scriptsize\bf Q}}\beta{\mbox{\scriptsize\bf B}}$.\\
\end{enumerate}
\end{lemma}

\begin{proof}
    It is given on the page 160 of \cite[v. 1, Pr. (3.5.8)(eng)]{Penrose1}. As in its proof, the metric tensor did not the participate then this lemma is
    true for any arrangement of indices: or top, or bottom.\\
\end{proof}

    We apply Lemma \ref{lemma2} to the equation (\ref{e23vv}) and obtain the 2 variants:
\begin{equation}
\label{e30vv}
    \begin{array}{ccccc}
    a). & (\tilde K_{\mbox{particular}})_d{}^c{}_b{}^aP_c=A^aB_{bd}\sps &
    b). & (\tilde K_{\mbox{particular}})_d{}^c{}_b{}^aP_c=A_d{}^aB_b
    \end{array}
\end{equation}
    \indent First. Suppose that Item a). and Item b). are performed simultaneously. We use one more lemma.

\begin{lemma}
\label{lemma3}
     From $\psi_{\mbox{\scriptsize\bf AB}}\varphi_{\mbox{\scriptsize\bf C}}=
     \chi_{\mbox{\scriptsize\bf A}}\theta_{\mbox{\scriptsize\bf BC}}$,
     the execution of the identities $\psi_{\mbox{\scriptsize\bf AB}}=
     \chi_{\mbox{\scriptsize\bf A}}\xi_{\mbox{\scriptsize\bf B}}$,
     $\theta_{\mbox{\scriptsize\bf BC}}=\xi_{\mbox{\scriptsize\bf B}}
     \varphi_{\mbox{\scriptsize\bf C}}$ follows for some
     $\xi_{\mbox{\scriptsize\bf B}}$.
\end{lemma}

\begin{proof}
     It is given on page 160 of the monography \cite[v. 1, (3.5.6)(eng)]{Penrose1}. And just as in the previous lemma, the location of the index is not significant.
\end{proof}
     We apply this lemma to the equation (\ref{e30vv}) that will give
\begin{equation}
\label{e31vv}
    (\tilde K_{\mbox{particular}})_d{}^c{}_b{}^aP_c=A_dB^aC_b\spsd
\end{equation}
    But there is the vector $M^d \ne 0$, that $M^dP_d = 0$ and $M^dA_d = 0$, then from (\ref{e31vv}), the statement (\ref{e29vv}) follows that is impossible. From this, we conclude that at the same time, a). and b). from (\ref{e30vv}) can not be executed.\\
    \indent Second. Now, we apply Lemma \ref{lemma2} to the equation (\ref{e30vv}). This will give the four variants:
\begin{equation}
\label{e32vv}
    \begin{array}{ccccc}
    I). & a). & (\tilde K_{\mbox{particular}})_d{}^c{}_b{}^a=A^{ac}B_{db}\sps &
          b). & (\tilde K_{\mbox{particular}})_d{}^c{}_b{}^a=C^aD^c{}_{db}\sps\\[2ex]
    II).& a). & (\tilde K_{\mbox{particular}})_d{}^c{}_b{}^a=S_d{}^aE_b{}^c\sps &
          b). & (\tilde K_{\mbox{particular}})_d{}^c{}_b{}^a=U_d{}^{ac}V_b\spsd
    \end{array}
\end{equation}
     Items b). in the both cases disappear as they lead to a singular transformation (see the explanation after the formula (\ref{e31vv})). \\
     \indent For definiteness, we will consider Item II).a). We contract \emph{common solution}
\begin{equation}
\label{e33vv}
    \tilde K_d{}^c{}_b{}^a=S_d{}^aE_b{}^c+\alpha\delta_d{}^c
    \delta_b{}^a
\end{equation}
     of the equation (\ref{e23vv}) with $\delta_c{}^d$ and using (\ref{e21vv}), we obtain
\begin{equation}
\label{e34vv}
    0=S_k{}^aE_b{}^k+4\alpha\delta_b{}^a\ \ \Rightarrow\ \ \
    E_b{}^k=-4\alpha(S^{-1})_b{}^k
\end{equation}
     (the transformation $S_k{}^a$ is nondegenerate since otherwise the transformation $\tilde K_d{}^c{}_b{}^a$ will be singular
     that would entail the singular transformation $K_{aa_1}{}^{bb_1}$). Therefore,
\begin{equation}
\label{e35vv}
    \tilde K_d{}^c{}_b{}^a=(-\alpha)(4S_d{}^a(S^{-1})_b{}^c-
    \delta_d{}^c\delta_b{}^a)\spsd
\end{equation}
    Contract (\ref{e22vv}) with $\varepsilon_{dd_1pp_1}K_{ss_1}{}^{pp_1}$ that will give with the account (\ref{e46ab})
\begin{equation}
\label{e36vv}
    \begin{array}{c}
    K_{aa_1}{}^{bb_1}\varepsilon_{ss_1cc_1}-
    K_{cc_1}{}^{bb_1}\varepsilon_{ss_1aa_1}=
    8\varepsilon_{cc_1k\left[\right.a_1}
    \tilde K_{a\left.\right]}{}^k{}_r{}^{\left[\right.b}
    K_{ss_1}{}^{b_1\left.\right]r}\spsd
    \end{array}
\end{equation}
    Contract (\ref{e36vv}) with $\varepsilon^{ss_1cc_1}$ using the formulas (\ref{e7d})
\begin{equation}
\label{e37vv}
    5K_{aa_1}{}^{bb_1}=8\tilde
    K_{\left[\right.a}{}^k{}_{\left|\right.r}{}^{\left[\right.b}
    K_{k\left.\right|a_1\left.\right]}{}^{b_1\left.\right]r}
\end{equation}
    and substitute (\ref{e35vv}) in (\ref{e37vv})
\begin{equation}
\label{e38vv}
    \begin{array}{c}
    5K_{aa_1}{}^{bb_1}=(-8\alpha)
    K_{k\left[\right.a_1}{}^{\left[\right.b_1|r|}
    (4S_{a\left.\right]}{}^{b\left.\right]}(S^{-1})_r{}^k-
    \delta_{a\left.\right]}{}^{|k|}\delta_r{}^{b\left.\right]})\sps\\[2ex]
    K_{aa_1}{}^{bb_1}=\frac{32\alpha}{5+8\alpha}
    K_{k\left[\right.a_1}{}^{r\left[\right.b_1}
    S_{a\left.\right]}{}^{b\left.\right]}(S^{-1})_r{}^k
    \end{array}
\end{equation}
     ($\alpha\ne 0,\alpha\ne \pm 5/8$; otherwise, the transformation $\tilde K_a{}^k{}_r{}^b$ is singular). Put
\begin{equation}
\label{e39vv}
     K_{aa_1}{}^{bb_1}:=
     2M_{\left[\right.a_1}{}^{\left[\right.b_1}
     S_{a\left.\right]}{}^{b\left.\right]}
\end{equation}
    and obtain
\begin{equation}
\label{e40vv}
    \begin{array}{c}
     K_{aa_1}{}^{bb_1}:=
    \frac{32\alpha}{5+8\alpha}
    S_{\left[\right.a}{}^{\left[\right.b}
    (M_{a_1\left.\right]}{}^{b_1\left.\right]}+
    \frac{1}{2}
    S_{a_1\left.\right]}{}^{b_1\left.\right]}M_k{}^r(S^{-1})_r{}^k)=
    2M_{\left[\right.a_1}{}^{\left[\right.b_1}
    S_{a\left.\right]}{}^{b\left.\right]}\spsd
    \end{array}
\end{equation}
    Define
\begin{equation}
\label{e41vv}
     M_k{}^r:=\beta S_k{}^r\Rightarrow
     \beta=\frac{8\alpha}{5-8\alpha}M_k{}^r(S^{-1})_r{}^k\Rightarrow \alpha=\frac{1}{8}\spsd
\end{equation}
    Then from (\ref{e40vv}), the definition
\begin{equation}
\label{e42vv}
    \begin{array}{c}
     K_{aa_1}{}^{bb_1}:=
     2\beta S_{\left[\right.a_1}{}^{\left[\right.b_1}
     S_{a\left.\right]}{}^{b\left.\right]}
    \end{array}
\end{equation}
    follows. Substituting (\ref{e42vv}) in (\ref{e46ab}), we find out that $\beta =\pm 1$. \\
    \indent Similarly, Item I).a). shall be considered. In this case, the transformation $K_{aa_1}{}^{bb_1}$
     has the form
\begin{equation}
\label{e42vva}
    K_{aa_1}{}^{bb_1}:=
    \beta S_{ac}S_{a_1c_1}\varepsilon\varepsilon^{cc_1bb_1}\spsd
\end{equation}
    Note that the factor $\varepsilon$ can be included in the definition of $S_{ac}$. \\
    \indent In this way, from (\ref{e3b}), we can really come to (\ref{e1b}) that completes the proof of the inverse path of the theorem. Therefore, the transformation $S_a{}^b$($S_a{}_b$) will match to one and only one transformation $K_\alpha{}^\beta$, and conversely, each transformation $K_\alpha{}^\beta$ will correspond to two and only two transformations $\pm S_a{}^b$ ($\pm S_a{}_b$), that $det\parallel S_c{}^d \parallel=1\ (det\parallel S_{cd} \parallel=1)$.\\

    Find out what a transformation corresponds to the special transformation $K_\alpha{}^\beta$. For this purpose, let's consider the following identity
    $$
    K_\alpha{}^\beta K_\gamma{}^\delta K_\lambda{}^\mu
    K_\nu{}^\chi K_\pi{}^\omega K_\sigma{}^\xi
    e_{\beta\delta\mu\chi\omega\xi}=\pm
    e_{\alpha\gamma\lambda\nu\pi\sigma}\sps
    $$
\begin{equation}
\label{e5b}
    e_{\beta\delta\mu\chi\omega\xi}=
    e_{[\beta\delta\mu\chi\omega\xi]}\sps
    e:=e_{123456}\spsd
\end{equation}
     At the same time, $e_{\beta\delta\mu\chi\omega\xi}$ is the 6-vector skew-symmetric in all indices. Consequently, we can get the record
\begin{equation}
\label{e20v}
    \begin{array}{c}
    K_1{}^\beta K_2{}^\delta K_3{}^\mu
    K_4{}^\chi K_5{}^\omega K_6{}^\xi
    e_{\beta\delta\mu\chi\omega\xi}=\pm e_{123456} \ \ \Leftrightarrow \ \ \
    det\parallel K_\alpha{}^\beta \parallel=\pm 1
    \end{array}
\end{equation}
     equivalent to (\ref{e5b}). If $K_\alpha{}^\beta$ is the special transformation then in (\ref{e5b}), the sign ''+'' is chosen that
     means that $det\parallel K_\alpha{}^\beta \parallel=1$. In otherwise case (the non-special transformation), the sign ''-'' is chosen. Since for a 4-vector, there are the identities
    $$
    e_{\alpha\beta\gamma\delta}=
    e_{[\alpha\beta\gamma\delta]}\sps
    $$
\begin{equation}
\label{e6b}
    e_{\alpha\beta\gamma\delta}=A_{\alpha\beta b}{}^{a}
    A_{\gamma\delta d}{}^{c} e_a{}^b{}_c{}^d\sps
\end{equation}
    $$
    e_a{}^b:=\frac{1}{3}e_a{}^k{}_k{}^b\sps
    $$
    following from (\ref{e4}) then, using its symmetries, we can obtain the expansion
     $$
    B_{\alpha\beta\gamma\delta r}{}^k:=
    A_{\alpha\beta r}{}^d A_{\gamma\delta d}{}^k+
    A_{\alpha\beta c}{}^k A_{\gamma\delta r}{}^c\sps
    $$
\begin{equation}
\label{e7b}
    e_{\alpha\beta\gamma\delta}:=B_{\alpha\beta\gamma\delta r}{}^k e_k{}^r\sps
\end{equation}
    $$
    e_k{}^k=0
    $$
    (the proof is given in Appendix (\ref{e56p}) - (\ref{e58p})). In turn, using these formulas, we can obtain the expansion for the 6-vector
    $$
    e_{\alpha\beta\gamma\delta\lambda\mu}=
    A_{\alpha\beta b}{}^aA_{\gamma\delta d}{}^c
    A_{\lambda\mu l}{}^k \ e_a{}^b{}_c{}^d{}_k{}^l\sps
    $$
    $$
    e_a{}^b{}_c{}^d{}_k{}^l=\frac{i}{8}
    (2((4\delta_k{}^b\delta_c{}^l-\delta_k{}^l\delta_c{}^b)\delta_a{}^d
        +(4\delta_k{}^d\delta_a{}^l-\delta_k{}^l\delta_a{}^d)\delta_c{}^b)-
    $$
\begin{equation}
\label{e8b}
        -(4\delta_k{}^b\delta_a{}^l-\delta_k{}^l\delta_a{}^b)\delta_c{}^d-
        (4\delta_k{}^d\delta_c{}^l-\delta_k{}^l\delta_c{}^d)\delta_a{}^b)
\end{equation}
    (the proof is given in Appendix (\ref{e59p}) - (\ref{e68p})). From (\ref{e8b}), the identity
\begin{equation}
\label{e78p}
       \begin{array}{c}
       e_{\alpha\gamma\lambda\nu\pi\sigma}=2
       \eta_{\left[\right.\alpha}{}^{bb_1}\eta_{\gamma |dd_1|}
       \eta_{\lambda}{}^{mm_1}\eta_{\nu |xx_1|}
       \eta_{\pi}{}^{rr_1}\eta_{\sigma\left.\right]ss_1}\cdot
       i \delta_r{}^d\delta_m{}^s\delta_b{}^x
       \delta_{b_1}{}^{d_1}\delta_{m_1}{}^{x_1}\delta_{r_1}{}^{s_1}=\\
       =\frac{1}{4}\eta_{\left[\right.\alpha}{}^{bb_1}\eta_{\gamma}{}^{dd_1}
       \eta_{\lambda}{}^{mm_1}\eta_{\nu}{}^{xx_1}
       \eta_{\pi}{}^{rr_1}\eta_{\sigma\left.\right]}{}^{ss_1}\cdot
       i \ \varepsilon_{rb_1dd_1}\varepsilon_{mr_1ss_1}\varepsilon_{bm_1xx_1}=\\
       =i(A_{\alpha\gamma b}{}^aA_{\lambda\nu a}{}^cA_{\pi\sigma c}{}^b+A_{\alpha\gamma b}{}^aA_{\lambda\nu c}{}^bA_{\pi\sigma a}{}^c)
       \end{array}
\end{equation}
    will imply. From (\ref{e8b}) and (\ref{e78p}), applying the definition (\ref{e5b}) to the special (non-special) transformations, we obtain
\begin{equation}
\label{e79p}
       \begin{array}{c}
       S_a{}^b S_{a_1}{}^{b_1} S_c{}^d S_{c_1}{}^{d_1}\varepsilon_{bb_1dd_1}=
       \varepsilon_{aa_1cc_1}\ (
       S_{ab} S_{a_1b_1} S_{cd} S_{c_1d_1}\varepsilon^{bb_1dd_1}=
       \varepsilon^{-2}\varepsilon_{aa_1cc_1})
       \end{array}
\end{equation}
    that gives the identity (\ref{e1b}). It follows that (\ref{e42vv}) corresponds to the special transformation and (\ref{e42vva}) corresponds to the non-special transformation $K_\alpha{}^\beta$. \\
    \indent Finally, the transformations $S_a{}^b$ and $iS_a{}^b$ belong to the same group $SL(4,\mathbb C)$. This means that in the formula (\ref{e42vv}), we can consider only the case when $\beta = +1$. Therefore, the group $SL(4,\mathbb C)$ is the double covering of the connected identity component of the group $SO(6,\mathbb C)$ (we denote it as $SO^e(6,\mathbb C)$).
\end{proof}

\subsubsection{\texorpdfstring{Real representation of the double covering $SL(4,\mathbb C)/\{\pm 1\}\cong SO(6, \mathbb C)$ in the presence of the involution $S_\alpha{}^{\beta '}$}{Real representation of the double covering}}

\begin{theorem}
\label{theorem6}
     Suppose that in the six-dimensional complex Euclidean space $\mathbb C\mathbb R^6$,  the involution
\begin{equation}
\label{e28b}
     S_\alpha{}^{\beta '}\bar S_{\beta '}{}^\gamma =
     \delta_\alpha{}^\gamma\sps
     S_\alpha{}^{\beta '}S_\gamma {}^{\delta '}\bar g_{\beta '\delta '}=
     g_{\alpha\gamma}
\end{equation}
     is given. Define
\begin{equation}
\label{e29b}
     \begin{array}{ccccc}
     & s_{aba'b'} & = & \bar\eta_{\beta 'a'b'}\eta^{\alpha}{}_{ab}
     S_\alpha{}^{\beta '} &
     \end{array}
\end{equation}
     then the relations
\begin{equation}
\label{e29ba}
     s_{aba'b'}=\bar s_{a'b'ab}\sps
     s_{ab}{}^{a'b'}\bar s_{a'b'}{}^{cd}=2\delta_{ab}^{cd}
\end{equation}
     are executed, and there are two and only two decompositions
\begin{equation}
\label{e30b}
     \begin{array}{ccccccc}
     I). &  s_{ab}{}^{a'b'}   =
     2s_{\left[\right.a}{}^{a'}  s_{b\left.\right]}{}^{b'}\sps
       & s_a{}^{b'}\bar s_{b'}{}^c=\pm\delta_a{}^c\sps\\[2ex]
     II). &    s_{ab}{}_{a'b'}    =
     2s_{\left[\right.a}{}_{|a'|}s_{b\left.\right]}{}_{b'}\sps
      &  s_{ab'}=\pm \bar s_{b'a}\spsd &
     \end{array}
\end{equation}
     In addition, for real inclusions, the identities
\begin{equation}
\label{e22v}
     \begin{array}{cccc}
     I).  &\bar \eta_i{}^{a'b'}=\eta_j{}^{cd}s_c{}^{a'}s_d{}^{b'} &
     II).&\bar \eta_i{}^{a'b'}=\eta_{jcd}s^{ca'}s^{db'}\sps \\
     & \bar A_{ij a'}{}^{b'}=A_{ijc}{}^d\bar s_{a'}{}^cs_d{}^{b'} &
     & \bar A_{ij a'}{}^{b'}=-A_{ijc}{}^d s_{da'}s^{cb'}\\
     \end{array}
\end{equation}
     will be true.
\end{theorem}

\begin{proof}
     The proof of the expansion (\ref{e30b}) is carried out as in the previous theorem. All changes are confined to the replacement of the transformation $K_\alpha{}^\beta$ on the transformation $S_\alpha{}^{\beta '}$ so that
\begin{equation}
\label{e23v}
       S_\alpha{}^{\beta '}S_\gamma{}^{\delta '} \bar g_{\beta '\delta '}=
       g_{\alpha\gamma}
\end{equation}
     is the analog of (\ref{e3b}) that will give the equation
\begin{equation}
\label{e2vvv}
    \begin{array}{c}
    s_a{}^{b'} s_{a_1}{}^{{b '}_1} s_c{}^{d'} s_{c_1}{}^{{d'}_1}
    \bar \varepsilon_{b'{b'}_1d'{d'}_1}=
    \varepsilon_{aa_1cc_1}\sps\\[2ex]
    (s_{ab'} s_{a_1{b'}_1} s_{cd'} s_{c_1{d'}_1}
    \bar\varepsilon^{b'{b'}_1d'{d'}_1}=
    \varepsilon_{aa_1cc_1})
    \end{array}
\end{equation}
     similar to (\ref{e1b}) (the relevant factor is included in the definition of the spin-tensor s). From (\ref{e28b}) and (\ref{e23v}), it is possible to obtain
\begin{equation}
\label{e3vvv}
     S_\alpha{}^{\beta '}=\bar S^{\beta '}{}_\alpha\spsd
\end{equation}
     From this, the equation
\begin{equation}
\label{e4vvv}
     \bar s_{a'b'ab}=\bar \eta{}^{\alpha '}{}_{a'b'}
     \eta_{\beta ab} \bar S_{\alpha '}{}^\beta=
     \bar \eta{}_{\beta 'a'b'}
     \eta^\alpha{}_{ab} \bar S^{\beta '}{}_\alpha=
     s_{aba'b'}
\end{equation}
     will follow.\\
     \indent Note that in the tangent bundle $\tau^\mathbb C (\tau^\mathbb R)$ there is the metric tensor $g_{\alpha\beta} (g_{ij})$ with help of which single indices can be raised and lowered. In the bundle $A^\mathbb C(S)$, the similar role is carried out by means of the quadrivector $\varepsilon_{abcd}$. The tensor $\bar g_{\alpha '\beta '}$($\bar\varepsilon_{a'b'c'd'}$) is one, the coordinates of which are conjugated to the coordinates of the tensor $g_{\alpha\beta}(\varepsilon_{abcd})$.\\
     \indent Consider the identities following from (\ref{e28b})
\begin{equation}
\label{e5vvv}
     \begin{array}{c}
     S_\alpha{}^{\beta '}\bar S_{\beta '}{}^\gamma =
     \delta_\alpha{}^\gamma\sps\\[2ex]
     \frac{1}{4}
     \eta_\alpha{}^{aa_1}\bar\eta^{\beta '}{}_{b'{b'}_1}
     s_{aa_1}{}^{b'{b'}_1}
     \frac{1}{4}
     \bar\eta_{\beta '}{}^{d'{d'}_1}\eta^\gamma{}_{cc_1}
     \bar s_{d'{d'}_1}{}^{cc_1}=
     \frac{1}{4}
     \eta_\alpha{}^{aa_1}\eta^\gamma{}_{cc_1}
     2\delta_{\left[\right.a}{}^c\delta_{a_1\left.\right]}{}^{c_1}\sps\\[2ex]
     \frac{1}{2}s_{aa_1}{}^{b'{b'}_1}\bar s_{b'{b'}_1}{}^{cc_1}=
     \delta_{aa_1}^{cc_1}\spsd
\end{array}
\end{equation}
     We now investigate Item II). From the last identity of (\ref{e5vvv}), we obtain
\begin{equation}
\label{e6vvv}
     \begin{array}{c}
     s_{ac'}s_{a_1{c'}_1}\bar \varepsilon^{c'{c'}_1b'{b'}_1}
     \cdot
     \bar s_{b'd}\bar s_{{b'}_1d_1} \varepsilon^{dd_1ff_1}=
     2\delta_{ff_1}^{cc_1}\sps\\[2ex]
     s_{ac'}s_{a_1{c'}_1}
     \bar s_{b'd}\bar s_{{b'}_1d_1} \bar\varepsilon^{c'{c'}_1b'{b'}_1}=
     \varepsilon_{aa_1dd_1}\spsd
\end{array}
\end{equation}
     We define $s^{kl '}$ as follows
\begin{equation}
\label{e7vvv}
      s^{kl'}s_{km'}=\delta_{m'}{}^{l'}
\end{equation}
     such that
\begin{equation}
\label{e8vvv}
     s^{kl'}s^{k_1{l'}_1}s^{mn'}s^{m_1{n'}_1}
     \varepsilon_{kk_1mm_1}=
     \bar \varepsilon^{l'{l'}_1n'{n'}_1}\spsd
\end{equation}
    Multiply (\ref{e6vvv}) by $s^{ak'} s^{a_1{k'}_1}s^{dn'}s^{d_1{n'}_1}$ and obtain
\begin{equation}
\label{e9vvv}
     s^{dn'}\bar s_{b'd}s^{d_1{n'}_1}\bar s_{{b'}_1d_1}
     \bar \varepsilon^{k'{k'}_1b'{b'}_1}=
     \bar \varepsilon^{k'{k'}_1n'{n'}_1}
\end{equation}
    taking into account (\ref{e8vvv}). Define
\begin{equation}
\label{e10vvv}
     \bar N_{b'}{}^{n'}:=s^{dn'}\bar s_{b'd}
\end{equation}
    then (\ref{e9vvv}) can be rewritten as
\begin{equation}
\label{e11vvv}
     \bar N_{\left[\right.b'}{}^{n'}
     \bar N_{{b'}_1\left.\right]}{}^{{n'}_1}=
     \delta_{\left[\right.b'}{}^{n'}
     \delta_{{b'}_1\left.\right]}{}^{{n'}_1}\spsd
\end{equation}
     From this, the identity
\begin{equation}
\label{e12vvv}
     \bar N_{b'}{}^{n'}=s^{dn'}\bar s_{b'd}=n\delta_{b'}{}^{n'}=
     ns^{dn'}s_{db'}\sps n^2=1
\end{equation}
     will follow (the proof is given in Appendix (\ref{e69p}) - (\ref{e74p})). Therefore, from (\ref{e7vvv}), the relation
\begin{equation}
\label{e13vvv}
     \bar s_{b'd}=\pm s_{db'}
\end{equation}
     follows. Similarly, we analyze Item I). From the identity (\ref{e5vvv}), the equation
\begin{equation}
\label{e14vvv}
     s_{\left[\right.a}{}^{b'}
     \delta_{a_1\left.\right]}{}^{{b'}_1}
     \bar s_{b'}{}^c\bar s_{{b_1}'}{}^{c_1}=
     \delta_{\left[\right.a}{}^c
     \delta_{a_1\left.\right]}{}^{c_1}
\end{equation}
     follows. Define
\begin{equation}
\label{e15vvv}
     N_a{}^c:=s_a{}^{b'}\bar s_{b'}{}^c
\end{equation}
     and obtain
\begin{equation}
\label{e16vvv}
     \bar N_{\left[\right.a}{}^c
     \bar N_{a_1\left.\right]}{}^{c_1}=
     \delta_{\left[\right.a}{}^c
     \delta_{a_1\left.\right]}{}^{c_1}\spsd
\end{equation}
    From here, the relation
\begin{equation}
\label{e17vvv}
     N_a{}^c=n\delta_a{}^c=s_a{}^{b'}\bar s_{b'}{}^c \sps n^2=1
\end{equation}
     will follow, defining the following equation
\begin{equation}
\label{e17vvva}
     s_a{}^{b'}\bar s_{b'}{}^c=\pm\delta_a{}^c\spsd
\end{equation}
     And we will need to prove (\ref{e22v}) only. We use the inclusion operator $H_i{}^\alpha$ and the involution $S_\alpha{}^{\beta '}$ defined by the formula (\ref{e5v}). For Item II)., we have
\begin{equation}
\label{e24v}
      \begin{array}{c}
      \bar \eta_i{}^{a'b'}=\bar H_i{}^{\alpha '}
      \bar \eta_{\alpha '}{}^{a'b'}=
      \bar H_i{}^{\alpha '}\bar S_{\alpha '}{}^\beta
      \eta_{\beta cd} s^{ca'}s^{db'}=\\[2ex]
      = H_i{}^\beta \eta_{\beta cd} s^{ca'}s^{db'}=
      \eta_{i cd} s^{ca'}s^{db'}\sps
      \end{array}
\end{equation}
\begin{equation}
\label{e25v}
      \begin{array}{c}
      \bar A_{ij a'}{}^{b'}=\bar \eta_{\left[\right. i}{}^{b'k'}
      \bar \eta_{j\left.\right] a'k'}=
      H_i{}^{\gamma}H_j{}^\delta \eta_{\left[\right. \gamma |bk|}
      \eta_{\delta\left.\right]}{}^{ak} s_{aa'}s^{bb'}=
      -A_{ij b}{}^a s_{aa'}s^{bb'}\spsd
      \end{array}
\end{equation}
      For Item I)., the proof is such
\begin{equation}
\label{e24va}
      \begin{array}{c}
      \bar \eta_i{}^{a'b'}=\bar H_i{}^{\alpha '}
      \bar \eta_{\alpha '}{}^{a'b'}=
      \bar H_i{}^{\alpha '}\bar S_{\alpha '}{}^\beta
      \eta_\beta{}^{cd} s_c{}^{a'}s_d{}^{b'}=\\[2ex]
      = H_i{}^\beta \eta_\beta{}^{cd} s_c{}^{a'}s_d{}^{b'}=
      \eta_i{}^{cd} s_c{}^{a'}s_d{}^{b'}\sps
      \end{array}
\end{equation}
\begin{equation}
\label{e25va}
      \begin{array}{c}
      \bar A_{ij a'}{}^{b'}=\bar \eta_{\left[\right. i}{}^{b'k'}
      \bar \eta_{j\left.\right] a'k'}=
      H_i{}^{\gamma}H_j{}^\delta \eta_{\left[\right. \gamma}{}^{ck}
      \eta_{\delta\left.\right] dk} s_c{}^{b'}\bar s_{a'}{}^d=
      \eta_{\left[\right. i}{}^{ck} \eta_{j\left.\right]}{}_{dk}
      s_c{}^{b'}\bar s_{a'}{}^d=
      A_{ij d}{}^c s_c{}^{b'}\bar s_{a'}{}^d\spsd
      \end{array}
\end{equation}
\end{proof}

\subsubsection{\texorpdfstring{Inclusion $\mathbb R^6_{(2,4)}\subset \mathbb C\mathbb R^6$ in the special basis}{Inclusion example}}$ $

     \indent Let us now consider an inclusion of the real space $\mathbb R^6_{(2,4)}$ in the complex space $\mathbb C\mathbb R^6$  as example. In this case, we have the opportunity to carry out the identification of upper primed indexes with lower unprimed indices using the spin-tensor $s_{aa '}$. Consider the identities
\begin{equation}
\label{e31b}
    \begin{array}{c}
    K_i{}^j=\bar K_i{}^j\sps
    K_i{}^j:=H_i{}^\alpha H^j{}_\beta K_\alpha{}^\beta\sps\\[2ex]
    \eta_j{}^{ab}K_i{}^j\eta^i{}_{cd}=
    \eta_j{}^{ab}\bar K_i{}^j\eta^i{}_{cd}\sps\\[2ex]
    2S_{\left[\right.c}{}^a S_{d\left.\right]}{}^b=\frac{1}{4}
    \eta_j{}^{ab}\bar\eta^{jm'n'}2\bar S_{m'}{}^{k'}\bar S_{n'}{}^{l'}
    \bar\eta_{ik'l'}\eta^{i}{}_{cd}\sps
    \end{array}
\end{equation}
\begin{equation}
\label{e31bb}
    \begin{array}{c}
    S_{\left[\right.c}{}^aS_{d\left.\right]}{}^b=\frac{1}{4}
    s^{abm'n'}\bar S_{\left[\right.m'}{}^{k'}\bar S_{n'\left.\right]}{}^{l'}
    s_{cdk'l'}\sps\\[2ex]
    S_{\left[\right.c}{}^aS_{d\left.\right]}{}^b=
    s^{am'}s^{bn'}\bar S_{\left[\right.m'}{}^{k'}\bar S_{n'\left.\right]}{}^{l'}
    s_{ck'}s_{dl'}\sps\\[2ex]
    s^{lk'}\bar S_{k'}{}^{m'}s_{am'}
    S_{\left[\right.c}{}^aS_{d\left.\right]}{}^b
    s_{bn'}\bar S_{r'}{}^{n'}s^{sr'}=
    \delta_{\left[\right.c}{}^a\delta_{d\left.\right]}{}^b\spsd
    \end{array}
\end{equation}
    Define
\begin{equation}
\label{e31bc}
    N_c{}^l:=s^{lk'}\bar S_{k'}{}^{m'}s_{am'}S_c{}^a
\end{equation}
    and obtain
\begin{equation}
\label{e31bd}
    N_{\left[\right.c}{}^aN_{d\left.\right]}{}^b=
    \delta_{\left[\right.c}{}^a\delta_{d\left.\right]}{}^b\spsd
\end{equation}
    From this, the equation
\begin{equation}
\label{e32b}
    N_c{}^l=s^{lk'}\bar S_{k'}{}^{m'}s_{am'}S_c{}^a=n\delta_c{}^l\sps
    n=\pm 1
\end{equation}
    will follow (the proof is given in Appendix (\ref{e69p}) - (\ref{e74p})). Choosing the sign ''+'' in (\ref{e32b}), we obtain the transformation of the group isomorphic to the group $SU(2,2)$ which will, as is evident from the above, the double covering of the connected identity component of the group $SO^e(2,4)$. This component is determined by the following conditions
\begin{equation}
\label{e33b}
    1).\ \ det\| K_\alpha{}^\beta\|=1\ \ \alpha ,\beta=\overline{1,6}\sps
    2).\ \ det\| K_\alpha{}^\beta\|>0\ \ \alpha ,\beta=\overline{1,2}\spsd
\end{equation}
    If in the (\ref{e32b}), ''-'' is chosen then in 2). from (\ref{e33b}), the sign  will change to the opposite. Next, to better understand how this works in the practice, we use a representation of the obtained results in the special basis. For this, we define the basis of $\mathbb C\mathbb R^6$ as follows
\begin{equation}
\label{e26v}
    \begin{array}{lccc}
    t^\alpha & =(1,0,0,0,0,0)\sps &
    v^\alpha & =(0,1,0,0,0,0)\sps \\
    w^\alpha & =(0,0,i,0,0,0)\sps &
    x^\alpha & =(0,0,0,i,0,0)\sps \\
    y^\alpha & =(0,0,0,0,i,0)\sps &
    z^\alpha & =(0,0,0,0,0,i)\spsd
    \end{array}
\end{equation}
    Let in this basis, the matrix of the metric tensor $g_{\alpha\beta}$ has the form
\begin{equation}
\label{e27v}
     \parallel g_{\alpha\beta} \parallel=
     \left(
       \begin{array}{cccccc}
       1 & 0 & 0 & 0 & 0 & 0\\
       0 & 1 & 0 & 0 & 0 & 0\\
       0 & 0 & 1 & 0 & 0 & 0\\
       0 & 0 & 0 & 1 & 0 & 0\\
       0 & 0 & 0 & 0 & 1 & 0\\
       0 & 0 & 0 & 0 & 0 & 1\\
       \end{array}
    \right)\spsd
\end{equation}
    We define the real representation of the inclusion $\mathbb R^6_{(2,4)} \subset \mathbb C \mathbb R^6$ with the help of the operator $H_i{}^\alpha$
\begin{equation}
\label{e28v}
     \parallel H_i{}^\alpha \parallel=
     \left(
       \begin{array}{cccccc}
       1 & 0 & 0 & 0 & 0 & 0\\
       0 & 1 & 0 & 0 & 0 & 0\\
       0 & 0 & i & 0 & 0 & 0\\
       0 & 0 & 0 & i & 0 & 0\\
       0 & 0 & 0 & 0 & i & 0\\
       0 & 0 & 0 & 0 & 0 & i\\
       \end{array}
    \right)\sps
     \parallel H^i{}_\alpha \parallel=
     \left(
       \begin{array}{cccccc}
       1 & 0 &  0 &  0 &  0 &  0\\
       0 & 1 &  0 &  0 &  0 &  0\\
       0 & 0 & -i &  0 &  0 &  0\\
       0 & 0 &  0 & -i &  0 &  0\\
       0 & 0 &  0 &  0 & -i &  0\\
       0 & 0 &  0 &  0 &  0 & -i\\
       \end{array}
    \right)\spsd
\end{equation}
    Then (\ref{e26v}) is a self-conjugate basis with respect to the involution $S_\alpha{}^{\beta '}$
\begin{equation}
\label{e29v}
     \parallel S_\alpha{}^{\beta '} \parallel=
     \left(
       \begin{array}{cccccc}
       1 & 0 &  0 &  0 &  0 &  0\\
       0 & 1 &  0 &  0 &  0 &  0\\
       0 & 0 & -1 &  0 &  0 &  0\\
       0 & 0 &  0 & -1 &  0 &  0\\
       0 & 0 &  0 &  0 & -1 &  0\\
       0 & 0 &  0 &  0 &  0 & -1\\
       \end{array}
    \right)\spsd
\end{equation}
    Therefore, in the space $\mathbb R^6_{(2,4)}$, the induced metric tensor $g_{ij}$ would have the matrix
\begin{equation}
\label{e30v}
     \parallel g_{ij} \parallel=
     \left(
       \begin{array}{cccccc}
       1 & 0 &  0 &  0 &  0 &  0\\
       0 & 1 &  0 &  0 &  0 &  0\\
       0 & 0 & -1 &  0 &  0 &  0\\
       0 & 0 &  0 & -1 &  0 &  0\\
       0 & 0 &  0 &  0 & -1 &  0\\
       0 & 0 &  0 &  0 &  0 & -1\\
       \end{array}
    \right)=\parallel H_i{}^\alpha H_j{}^\beta g_{\alpha\beta} \parallel
\end{equation}
    in the basis
\begin{equation}
\label{e31v}
    \begin{array}{lcccccc}
    t^i  &= H^i{}_\alpha  t^\alpha  & =(1,0,0,0,0,0) & \sps &
    v^i  &= H^i{}_\alpha  v^\alpha  & =(0,1,0,0,0,0) \sps\\[2ex]
    w^i  &= H^i{}_\alpha  w^\alpha  & =(0,0,1,0,0,0) & \sps &
    x^i  &= H^i{}_\alpha  x^\alpha  & =(0,0,0,1,0,0) \sps\\[2ex]
    y^i  &= H^i{}_\alpha  y^\alpha  & =(0,0,0,0,1,0) & \sps &
    z^i  &= H^i{}_\alpha  z^\alpha  & =(0,0,0,0,0,1)\spsd
    \end{array}
\end{equation}
    We will define the vector basis in the bundle $A^\mathbb C(S)$ as
    $$
    X^a=(1,0,0,0) \sps
    Y^a=(0,1,0,0) \sps
    $$
\begin{equation}
\label{e51}
    Z^a=(0,0,1,0) \sps
    T^a=(0,0,0,1)\sps
\end{equation}
    $$
    \varepsilon_{abcd}X^aY^bZ^cT^d=1\sps \varepsilon=1\spsd
    $$
    Then in the bases (\ref{e30v}) and (\ref{e51}), the decomposition
    $$
    R^{ab}=2(R^{12}X^{\left[ \right.a}Y^{b\left. \right]}+
             R^{13}X^{\left[ \right.a}Z^{b\left. \right]}+
             R^{14}X^{\left[ \right.a}T^{b\left. \right]}+
    $$
\begin{equation}
\label{e52}
             +R^{23}Y^{\left[ \right.a}Z^{b\left. \right]}+
             R^{24}Y^{\left[ \right.a}T^{b\left. \right]}+
             R^{34}Z^{\left[ \right.a}T^{b\left. \right]})=
\end{equation}
    $$
    \begin{array}{c}
    =\frac{1}{\sqrt{2}}
    (R^{12}+R^{34})\cdot {\sqrt{2}}(X^{\left[ \right.a}Y^{b\left. \right]}+
                    Z^{\left[ \right.a}T^{b\left. \right]})+\\ \\
    +\frac{1}{\sqrt{2}}
    (R^{12}-R^{34})\cdot {\sqrt{2}}(X^{\left[ \right.a}Y^{b\left. \right]}-
                    Z^{\left[ \right.a}T^{b\left. \right]})+\\ \\
    +\frac{1}{\sqrt{2}}
    (R^{13}+R^{24})\cdot {\sqrt{2}}(X^{\left[ \right.a}Z^{b\left. \right]}+
                    Y^{\left[ \right.a}T^{b\left. \right]})+\\ \\
    +\frac{i}{\sqrt{2}}
    (R^{13}-R^{24})\cdot (-i{\sqrt{2}})(X^{\left[ \right.a}Z^{b\left. \right]}-
                    Y^{\left[ \right.a}T^{b\left. \right]})+\\ \\
    +\frac{-i}{\sqrt{2}}
    (R^{14}+R^{23})\cdot i{\sqrt{2}}(X^{\left[ \right.a}T^{b\left. \right]}+
                    Y^{\left[ \right.a}Z^{b\left. \right]})+\\ \\
    +\frac{-i}{\sqrt{2}}
    (R^{14}-R^{23})\cdot i{\sqrt{2}}(X^{\left[ \right.a}T^{b\left. \right]}-
                    Y^{\left[ \right.a}Z^{b\left. \right]})=\\ \\
    =(Tt^i+Vv^i+Ww^i+Xx^i+Yy^i+Zz^i)\cdot
    \eta_i{}^{ab}=r^i\eta_i{}^{ab}
    \end{array}
    $$
\begin{table}
\caption{Matrix form of the spin-tensor s for the real inclusions.}
\label{table1}
\begin{center}
\begin{tabular}{|c|c|c|c|c|}
      \hline
       & Space & s &
      s in the special basis & Isomprphism\\
      \hline
      1  &
      $
      \begin{array}{c}
      R^6   \\
      \ll ++++++ \gg
      \end{array}
      $  &
      $s_{kk'}$ &
      $\left(
      \begin{array}{cccc}
      1 & 0 & 0 & 0\\
      0 & 1 & 0 & 0\\
      0 & 0 & 1 & 0\\
      0 & 0 & 0 & 1
      \end{array}
      \right)
      $&
      $\small{SU(4)/\{\pm 1\} \cong SO^e(6)}$\\
      \hline
      2  &
      $
      \begin{array}{c}
      R^6_{(1,5)}   \\
      \ll +----- \gg
      \end{array}
      $  &
      $s_k{}^{k'}$ &
      $\left(
      \begin{array}{cccc}
      0 & 1 & 0 & 0\\
     -1 & 0 & 0 & 0\\
      0 & 0 & 0 & 1\\
      0 & 0 &-1 & 0
      \end{array}
      \right)
      $&
      $\small{SL(2,H)/\{\pm 1\} \cong SO^e(1,5)}$\\
      \hline
      3  &
      $
      \begin{array}{c}
      R^6_{(2,4)}   \\
      \ll ++---- \gg
      \end{array}
       $ &
       $s_{kk'}$ &
       $\left(
       \begin{array}{cccc}
       0 & 0 & 1 & 0\\
       0 & 0 & 0 & 1\\
       1 & 0 & 0 & 0\\
       0 & 1 & 0 & 0
       \end{array}
       \right)
       $&
      $\small{SU(2,2)/\{\pm 1\} \cong SO^e(2,4)}$\\
      \hline
       4  &
      $
      \begin{array}{c}
      R^6_{(3,3)}   \\
      \ll +++--- \gg
      \end{array}
       $ &
       $s_k{}^{k'}$ &
       $\left(
       \begin{array}{cccc}
       1 & 0 & 0 & 0\\
       0 & 1 & 0 & 0\\
       0 & 0 & 1 & 0\\
       0 & 0 & 0 & 1
       \end{array}
       \right)
       $&
      $\small{SL(4,R)/\{\pm 1\} \cong SO^e(3,3)}$\\
      \hline
\end{tabular}
\end{center}
\end{table}
    takes place. Therefore, we can put
\begin{equation}
\label{e32v}
    \begin{array}{cccccc}
    v^i\eta_i{}^{ab} & := & {\sqrt{2}}(X^{\left[ \right.a}Y^{b\left. \right]}+
                    Z^{\left[ \right.a}T^{b\left. \right]})\sps &
    w^i\eta_i{}^{ab} & := & {\sqrt{2}}(X^{\left[ \right.a}Y^{b\left. \right]}-
                    Z^{\left[ \right.a}T^{b\left. \right]})\sps \\[2ex]
    y^i\eta_i{}^{ab} & := & {\sqrt{2}}(X^{\left[ \right.a}Z^{b\left. \right]}+
                    Y^{\left[ \right.a}T^{b\left. \right]})\sps &
    x^i\eta_i{}^{ab} & := & -{\sqrt{2}i}(X^{\left[ \right.a}Z^{b\left. \right]}-
                    Y^{\left[ \right.a}T^{b\left. \right]})\sps \\[2ex]
    z^i\eta_i{}^{ab} & := & {\sqrt{2}i}(X^{\left[ \right.a}T^{b\left. \right]}+
                    Y^{\left[ \right.a}Z^{b\left. \right]})\sps &
    t^i\eta_i{}^{ab} & := & {\sqrt{2}i}(X^{\left[ \right.a}T^{b\left. \right]}-
                    Y^{\left[ \right.a}Z^{b\left. \right]})
    \end{array}
\end{equation}
    that will define the connecting Norden operators $\eta_i{}^{aa_1}$ in these bases as
\begin{equation}
\label{e56}
    \begin{array}{crcrcrcr}
    \eta_2{}^{12}= & \frac{1}{\sqrt{2}}\sps &
    \eta_2{}^{34}= & \frac{1}{\sqrt{2}}\sps &
    \eta_3{}^{12}= & \frac{1}{\sqrt{2}}\sps &
    \eta_3{}^{34}= & -\frac{1}{\sqrt{2}}\sps \\
    \eta_1{}^{14}= & \frac{i}{\sqrt{2}}\sps &
    \eta_1{}^{23}= & -\frac{i}{\sqrt{2}}\sps  &
    \eta_6{}^{14}= & \frac{i}{\sqrt{2}}\sps &
    \eta_6{}^{23}= & \frac{i}{\sqrt{2}}\sps \\
    \eta_5{}^{13}= & \frac{1}{\sqrt{2}}\sps  &
    \eta_5{}^{24}= & \frac{1}{\sqrt{2}}\sps  &
    \eta_4{}^{13}= & -\frac{i}{\sqrt{2}}\sps  &
    \eta_4{}^{24}= & \frac{i}{\sqrt{2}}\spsd
    \end{array}
\end{equation}
    From (\ref{e32v}),  the following conditions
\begin{equation}
\label{e33v}
    \begin{array}{cc}
    T=\frac{i}{\sqrt{2}}(R^{23}-R^{14})\sps &
    V=\frac{1}{\sqrt{2}}(R^{12}+R^{34})\sps \\[2ex]
    W=\frac{1}{\sqrt{2}}(R^{12}-R^{34})\sps &
    X=\frac{i}{\sqrt{2}}(R^{13}-R^{24})\sps \\[2ex]
    Y=\frac{1}{\sqrt{2}}(R^{13}+R^{24})\sps &
    Z=\frac{-i}{\sqrt{2}}(R^{14}+R^{23})\sps\\[4ex]
    R^{12}=\frac{1}{\sqrt{2}}(V+W)\sps  &
    R^{13}=\frac{1}{\sqrt{2}}(Y-iX)\sps \\[2ex]
    R^{14}=\frac{i}{\sqrt{2}}(T+Z)\sps  &
    R^{23}=\frac{i}{\sqrt{2}}(Z-T)\sps  \\[2ex]
    R^{24}=\frac{1}{\sqrt{2}}(Y+iX)\sps &
    R^{34}=\frac{1}{\sqrt{2}}(V-W)
    \end{array}
\end{equation}
    imply so that the inverse values $\eta^i{}_{aa_1}$ have the form
\begin{equation}
\label{e57}
    \begin{array}{crcrcrcr}
    \eta^2{}_{12}= & \frac{1}{\sqrt{2}} \sps &
    \eta^2{}_{34}= & \frac{1}{\sqrt{2}} \sps &
    \eta^3{}_{12}= & \frac{1}{\sqrt{2}} \sps &
    \eta^3{}_{34}= & -\frac{1}{\sqrt{2}}\sps \\
    \eta^1{}_{14}= & -\frac{i}{\sqrt{2}}\sps &
    \eta^1{}_{23}= & \frac{i}{\sqrt{2}} \sps &
    \eta^6{}_{14}= & -\frac{i}{\sqrt{2}}\sps &
    \eta^6{}_{23}= & -\frac{i}{\sqrt{2}}\sps \\
    \eta^5{}_{13}= & \frac{1}{\sqrt{2}} \sps &
    \eta^5{}_{24}= & \frac{1}{\sqrt{2}} \sps &
    \eta^4{}_{13}= & \frac{i}{\sqrt{2}} \sps &
    \eta^4{}_{24}= & -\frac{i}{\sqrt{2}}\spsd
    \end{array}
\end{equation}
    And moreover, the equalities
\begin{equation}
\label{e35v}
    \begin{array}{ccccccccccc}
    \overline{R^{23}} & = & -R^{23} & = &  R_{41} & \sps &
    \overline{R^{34}} & = &  R^{34} & = &  R_{12}   \sps\\[2ex]
    \overline{R^{12}} & = &  R^{12} & = &  R_{34} & \sps &
    \overline{R^{13}} & = &  R^{24} & = &  R_{31}
    \end{array}
\end{equation}
    will be true. At the performance of (\ref {e35v}), choose the covector basis coordinated with the basis (\ref{e51})  as follows
    $$
    X_a=s_{aa'}\bar X^{a'}=(0,0,1,0) \sps
    Y_a=s_{aa'}\bar Y^{a'}=(0,0,0,1) \sps
    $$
\begin{equation}
\label{e34v}
    Z_a=s_{aa'}\bar Z^{a'}=(1,0,0,0) \sps
    T_a=s_{aa'}\bar T^{a'}=(0,1,0,0)\spsd
\end{equation}
    This determines the Hermitian polarity
\begin{equation}
\label{e36v}
       \parallel s_{aa'} \parallel=
       \left(
       \begin{array}{cccc}
       0 & 0 & 1 & 0\\
       0 & 0 & 0 & 1\\
       1 & 0 & 0 & 0\\
       0 & 1 & 0 & 0
       \end{array}
       \right)
\end{equation}
    by which the bundle $A^\mathbb C(S)$ (its base is $\mathbb R^6_{(2,4)}$) is endowed. It follows that the pffafian of the bivector $R^{ab}$ has the form
\begin{equation}
\label{e53}
    \begin{array}{c}
    pf(R):=\frac{1}{2}R^{ab}R_{ab}=\\
    =2(R^{12}R^{34}-R^{13}R^{24}+R^{14}R^{23})=\\
    =T^2+V^2-W^2-X^2-Y^2-Z^2\spsd
    \end{array}
\end{equation}
    In the special basis for the remaining inclusion cases, the matrix form of the spin-tensor s is given in the table \ref{table1}. \\

\subsubsection{Infinitesimal transformation}$ $

    \indent Suppose we have $K_\alpha{}^\beta (\lambda) $: a one-parameter family satisfying the condition
\begin{equation}
\label{e34b}
    g_{\alpha\delta}=K_\alpha{}^\beta(\lambda)
    K_\delta{}^\gamma(\lambda)g_{\beta\gamma}\sps
    K_\alpha{}^\beta(0)=\delta_\alpha{}^\beta\spsd
\end{equation}
    The infinitesimal transformation, corresponding to it, is defined as
\begin{equation}
\label{e35b}
     T_\delta{}^\gamma=\left.\left[\frac{d}{d\lambda}
     K_\delta{}^\gamma(\lambda)\right]\right|_{\lambda =0}\spsd
\end{equation}
     Then from (\ref{e34b}), the equation
\begin{equation}
\label{e36b}
     T_{\alpha\beta}=-T_{\beta\alpha}
\end{equation}
     follows. According to \cite[v. 1, p. 176(eng)]{Penrose1}, from (\ref{e34b}), the equation (\ref{e36b}) follows, and from (\ref{e36b}) with the help of the exponent
\begin{equation}
\label{e37b}
     K_\delta{}^\gamma(\lambda ):=exp(\lambda T_\delta{}^\gamma)\sps
\end{equation}
     we can obtain (\ref{e34b}).\\
     \indent Suppose also that an one-parameter family $S_a{}^b(\lambda)$
\begin{equation}
\label{e38b}
      S_a{}^b(\lambda )S_c{}^d(\lambda )S_{a_1}{}^{b_1}(\lambda )
      S_{c_1}{}^{d_1}(\lambda )\varepsilon_{bb_1dd_1}=
      \varepsilon_{aa_1cc_1}\sps
      S_a{}^b(0)=\delta_a{}^b
\end{equation}
     is given. We will differentiate it assuming
\begin{equation}
\label{e39b}
     T_a{}^b:=\left.\left[\frac{d}{d\lambda}
     S_a{}^b(\lambda)\right]\right|_{\lambda =0}\sps
\end{equation}
     and we will obtain
\begin{equation}
\label{e40b}
     \varepsilon_{b\left[\right.a_1cc_1}T_{a\left.\right]}{}^b=0\ \ \
     \Leftrightarrow\ \ T_a{}^a=0\spsd
\end{equation}
    The opposite is true. Let
\begin{equation}
\label{e41b}
     S_a{}^b(\lambda ):=exp(\lambda T_a{}^b)\spsd
\end{equation}
     Then the following identity
\begin{equation}
\label{e42b}
      \begin{array}{c}
      S_a{}^b(\lambda )S_c{}^d(\lambda )S_{a_1}{}^{b_1}(\lambda )
      S_{c_1}{}^{d_1}(\lambda )\varepsilon_{bb_1dd_1}=\\ \\
      =exp(\lambda T_a{}^b)exp(\lambda T_c{}^d)exp(\lambda T_{a_1}{}^{b_1})
      exp(\lambda T_{c_1}{}^{d_1})\varepsilon_{bb_1dd_1}=\\ \\
      =det(exp(\lambda T_a{}^b))\varepsilon_{aa_1cc_1}=
      exp(\lambda tr(T_a{}^b))\varepsilon_{aa_1cc_1}=\varepsilon_{aa_1cc_1}
      \end{array}
\end{equation}
     is satisfied. Since
\begin{equation}
\label{e43b}
       K_\alpha{}^\beta(\lambda)=\frac{1}{4}\eta_\alpha{}^{aa_1}
       \eta^\beta{}_{bb_1}2S_{\left[\right.a}{}^b(\lambda)
       S_{a_1\left.\right]}{}^{b_1}(\lambda)
\end{equation}
       then, differentiating with respect to $\lambda$, setting $\lambda = 0$, and lowering the superscript with the help of the metric tensor $g_{\alpha\beta}$, we obtain
\begin{equation}
\label{e44b}
      T_{\alpha\beta}=\frac{1}{2}\eta_\alpha{}^{aa_1}
      \eta_{\beta bb_1}(T_a{}^b\delta_{a_1}{}^{b_1}+
      T_{a_1}{}^{b_1}\delta_a{}^b)=A_{\alpha\beta b}{}^aT_a{}^b\spsd
\end{equation}
      Now, the purpose of this subsection is visible. In fact, (\ref{e44b}) is an algebraic interpretation of the isomorphism between the Lie algebras
\begin{equation}
\label{e45b}
      so(6,\mathbb C)\cong sl(4,\mathbb C)\sps
\end{equation}
      and the definition (\ref{e4}), given at the beginning of this chapter, is quite justified.

\subsection{Generalized Norden operators}$ $
      \indent If the complex analytic Riemannian space $\mathbb CV^6$, which will be the base of the tangent bundle $\tau^\mathbb C$ and bundle $\Lambda$, is set then there is the tensor $g_{\alpha\beta} (z^\gamma)$ which is the metric tensor. This tensor is analytic on $z^\gamma$, where $z^\gamma$ are the coordinates of a base point. The tensor value at the point $O (z_o^\gamma)$ is denoted as $\tilde{\dot g}_ {\alpha\beta}$
\begin{equation}
\label{e42v}
      \tilde{\dot g}_{\alpha\beta}:=
      g_{\alpha\beta}(z_o^\gamma)\spsd
\end{equation}
      Since, the tensor $\tilde{\dot g}_{\alpha\beta}$ has a symmetric matrix, it can be reduced to the diagonal form by means of a nonsingular transformation $\dot P_\alpha{}^\gamma$
\begin{equation}
\label{e43v}
      \dot g_{\alpha\beta}=
      \dot P_\alpha{}^\gamma\dot P_\beta{}^\delta
      \tilde{\dot g}_{\gamma\delta}\sps
      \dot P_\alpha{}^\gamma:=
      P_\alpha{}^\gamma(z_o^\delta)\sps
\end{equation}
      where $P_\alpha{}^\gamma (z^\delta)$ are analytic functions of the point coordinates. But for the tensor $\dot g_{\alpha\beta}$, the following relations
\begin{equation}
\label{e44v}
    \dot g^{\alpha\beta}=1/4\cdot \dot\eta^{\alpha}{}_{aa_1} \dot\eta^{\beta}{}_{bb_1}
    \dot\varepsilon^{aa_1bb_1}\sps
    \dot\varepsilon^{aa_1bb_1}=
    \dot\eta_{\alpha}{}^{aa_1} \dot\eta_{\beta}{}^{bb_1}\dot g^{\alpha\beta}
\end{equation}
    are executed, where $\dot\eta^{\alpha}{}_{aa_1}$ are the connecting Norden operators satisfying the relation (\ref{e0}). Then from (\ref{e43v}),  the equation
\begin{equation}
\label{e45v}
      g_{\alpha\beta}(z_o^\gamma):=
      \dot g_{\alpha\beta}
      (\dot P^{-1}){}_\gamma{}^\alpha
      (\dot P^{-1}){}_\delta{}^\beta=
      (\dot P^{-1}){}_\gamma{}^\alpha
      (\dot P^{-1}){}_\delta{}^\beta
      \dot\eta_{\alpha}{}^{aa_1} \dot\eta_{\beta}{}^{bb_1}
      \dot\varepsilon_{aa_1bb_1}
\end{equation}
      follows. We define \emph{generalized connecting Norden operators} as
\begin{equation}
\label{e46v}
      \eta_{\alpha}{}^{aa_1}(z_o^\delta):=
      (\dot P^{-1}){}_\gamma{}^\alpha\dot\eta_{\alpha}{}^{aa_1}\sqrt{\varepsilon^{-1}(z_o^\delta)}\sps
      \varepsilon_{abcd}(z_o^\delta)=\varepsilon(z_o^\delta)\dot\varepsilon_{abcd}\sps \varepsilon_{1234}(z_o^\delta)=\varepsilon(z_o^\delta)\spsd
\end{equation}
      As the root we can take any of the two options. Generally speaking, $\ll\dot{} \gg$ can be omitted since all calculations are valid for an arbitrary point O, and at the same time, the functions $P_\alpha{}^\gamma (z^\delta),\ \varepsilon(z_o^\delta)$ are analytic. Then from (\ref{e44v}), the identities
\begin{equation}
\label{e47v}
    \begin{array}{c}
    g^{\alpha\beta}(z^\delta)=1/4\cdot \eta^{\alpha}{}_{aa_1}(z^\delta)
    \eta^{\beta}{}_{bb_1}(z^\delta)
    \varepsilon^{aa_1bb_1}(z^\delta)\sps\\[2ex]
    \varepsilon^{aa_1bb_1}(z^\delta)
    =\eta_{\alpha}{}^{aa_1}(z^\delta)
    \eta_{\beta}{}^{bb_1}(z^\delta)g^{\alpha\beta}(z^\delta)
    \end{array}
\end{equation}
    will follow. Below, we shall use the generalized connecting Norden operators.

\section{\texorpdfstring{Connections in the bundle  $A^\mathbb C$ with the base $\mathbb CV^6$}{Connections in bundles A}}
\Abstract{
    \indent This chapter is devoted to the depiction of the two approaches to the introduction of a connection in the bundle $A^\mathbb C$. The first is described in the monography \cite{Penrose1}, and the second is determined with the help of the Norden-Neifeld normalization theory. In the first subsection, the two definitions of a connection in the bundles according to these theories are considered.\\
    \indent In the second section, we consider a normalization of a maximal planar generator manifold for a quadric $\mathbb CQ_6$ embedded in the projective space $\mathbb C\mathbb P_7$. This manifold is diffeomorphic to the one of all points of this quadric. Considering the derivation equations of the normalized family of maximal planar generators, we arrive to the Norden connecting operator definition in terms of Neifeld operators. If in the bundle $A^\mathbb C$ we consider the quadrivector $\varepsilon_{abcd}$ as the metric tensor antisymmetric in all indices then the metric tensor $G_{\Lambda\Psi}$ is induced on the base, and therefor the maximal planar generator manifold transforms to the real pseudo-Riemannian space $V^{12}_{(6,6)}$ with the complex structure $f_\Lambda{}^\Psi$. We can move to the complex representation of our manifold taking the space $\mathbb CV^6$ as the base. A 4-dimensional cone generator of the 8-dimensional space $\mathbb C\mathbb R^8$ (i.e., in the projective geometry, this will just be a 3-dimensional generator of the quadric $\mathbb CQ_6\subset \mathbb C \mathbb P_7$) corresponds to a fiber of the bundle $A^\mathbb C$ with the base $\mathbb CV^6$. Then we obtain that the torsion-free Riemannian connection, introduced by the formulas
    $$
    \nabla_\alpha g_{\beta\gamma}=0\sps
    \bar\nabla_{\alpha '} \bar g_{\beta '\gamma '}=0\sps
    $$
    can be uniquely prolonged to the equiaffine connection in the bundle $A^\mathbb C(\mathbb CV^6)$
    $$
    \nabla_\alpha\varepsilon_{abcd}=0\sps
    \bar\nabla_{\alpha '}\bar \varepsilon_{a'b'c'd'}=0\sps
    $$
    where $\alpha,\beta,...=1,2,3,4,5,6$. The existence and the uniqueness of such the connections are proved in this chapter.\\
    \indent Next, the real torsion-free  Riemannian connection, induced by the inclusion $V^6_{(p,q)}\subset\mathbb CV^6$, is described. Such the connection must be coordinated with the involution, i.e., the following relation
    $$
    \nabla_\alpha S_\beta{}^{\gamma '}=0\sps
    \bar\nabla_{\alpha '} S_\beta{}^{\gamma '}=0
    $$
    must be satisfied. Then using the results of the first chapter, we introduce either a Hermitian polarity or a Hermitian involution in the bundle $A^\mathbb C$. The specified structure must be a covariant constant. The bitwistor equation
    $$
    \nabla^{a\left(\right.b}X^{c\left.\right)}=0
    $$
    is obtained from the results of these subsections. This equation is a conformal invariant and an invariant under a normalization transformation. Solutions of this equation will be discussed in the next chapter.
}

\subsection{Connection in a bundle}$ $
    \indent Let a bundle R with the base $V^{2n}_{(n,n)}$ and fibers isomorphic to $\mathbb C^k$ be given. We define \emph{covariant derivative operator} acting on the bundle R along a vector field X as a mapping between two smooth sections of the fiber $\mathbb C^k_x$
\begin{equation}
\label{e1w}
     \begin{array}{c}
     \nabla_X s:x\longmapsto \nabla_X s(x)\sps
     \end{array}
\end{equation}
     where $s(x)$ is section. If $X=\frac{\partial}{\partial x^i}$ then this will give the decomposition
\begin{equation}
\label{e2w}
     \nabla_\frac{\partial}{\partial x^i} s=\nabla_i s\sps
\end{equation}
     where $i,j,k,...=\overline{1,2n}$. The operator $\nabla_i$ must satisfy the following relations (which incidentally can be put in its definition)
\begin{equation}
\label{e3w}
     \begin{array}{c}
     \nabla_i(X^a+Y^a)=\nabla_iX^a+\nabla_iY^a\sps\\
     \nabla_i(fX^a)=f\nabla_iX^a+X^a\nabla_if\sps\\
     \nabla_i(X_aY^a)=Y^a\nabla_iX_a+X_a\nabla_iY^a\sps\\
     \nabla_i \bar X^{a'}=\overline{\nabla_i  X^a}\sps
     \nabla_i \bar X_{a'}=\overline{\nabla_i  X_a}\sps\\
     \nabla_ik=0\sps\\
     \nabla_i(g+h)=\nabla_ig+\nabla_ih\sps\\
     \nabla_i(gh)+g\nabla_ih+h\nabla_ig\sps
     \end{array}
\end{equation}
     where $a,b,c,...,f=\overline{1,n}$. In this case, k, g, h are analytical functions, $k = const$; $X^a,\ Y^a$ are vectors of the fiber $\mathbb C^k_x$, and $X_a,\ Y_a$ are covectors of the dual space ${\mathbb C^*}^k_x$. In the basis $s_a(x)$ of the fiber $\mathbb C^k_x$, the section $s(x)$ can be decomposed as
\begin{equation}
\label{e3wa}
     s=s^as_a
\end{equation}
     so that the connection coefficients are determined from the following equation
\begin{equation}
\label{e3wb}
     \nabla_is_a=\Gamma_{ia}{}^cs_c\spsd
\end{equation}
     Then the differentiation can be accomplished as follows
\begin{equation}
\label{e3wc}
     \nabla_i X^a=\partial_i X^a+\Gamma_{ic}{}^aX^c\spsd
\end{equation}
     The repeated covariant derivative is written down as
\begin{equation}
\label{e3wd}
     \nabla_i\nabla_jX^a=\partial_i\nabla_jX{}^a-
     \Gamma_{ij}{}^k\nabla_kX^a+\Gamma_{ic}{}^aX_j{}^c\spsd
\end{equation}
      By $\Gamma$, we denote the connection defined by means of $\Gamma_{ij}{}^k$ in the tangent bundle. The tensor $T_{ij}{}^k$ defining with the help of the relation
\begin{equation}
\label{e4w}
     2\nabla_{\left[\right.i}\nabla_{j\left.\right]}f=
     T_{ij}{}^k\nabla_kf
\end{equation}
     is called \emph{torsion tensor of the connection $\Gamma$}. The tensor $R_{ijk}{}^l$ defining with the help of the following condition
\begin{equation}
\label{e4wa}
     (2\nabla_{\left[\right.i}\nabla_{j\left.\right]}-T_{ij}{}^k\nabla_k)
     X^i=R_{ijk}{}^lX^k
\end{equation}
      is called \emph{curvature tensor of the connection $\Gamma$}. If the torsion is equal to zero then the operator $\nabla_i$ is called \emph{symmetric covariant derivative operator}.\\
      \indent Let $\nabla_i$ be a symmetric covariant derivative operator, and $\tilde\nabla_i$ is an arbitrary covariant derivative operator. Then
\begin{equation}
\label{e5w}
     (\tilde\nabla_i-\nabla_i)f=0\sps
\end{equation}
     and we can define the tensor $Q_{ib}{}^a$ called \emph{strain tensor}
\begin{equation}
\label{e6w}
     \begin{array}{cc}
     (\tilde\nabla_i-\nabla_i)X^a=Q_{ib}{}^a X^b\sps &
     (\tilde\nabla_i-\nabla_i)X_a=-Q_{ia}{}^b X_b\sps\\
     (\tilde\nabla_i-\nabla_i)X^{a'}=Q_{ib'}{}^{a'} X^{b'}\sps &
     (\tilde\nabla_i-\nabla_i)X_{a'}=-Q_{ia}{}^{b'} X_{b'}\spsd\\
     \end{array}
\end{equation}
     If $R=\tau^\mathbb R(V^{2n}_{(n,n)})$ is the tangent bundle then the torsion of the operator $\tilde\nabla_i$ has the form
\begin{equation}
\label{e7w}
     \tilde T_{ij}{}^k=2Q_{[ij]}{}^k\sps
\end{equation}
     where $Q_{ij}{}^k$ is \emph{strain tensor} in the tangent bundle.

\subsubsection{\texorpdfstring{Normalization (\emph{spinor normalization}) of the quadric $\mathbb CQ_6$ in $\mathbb CP_7$}{Spinor normalization}}$ $

     \indent Consider a nonsingular quadric $\mathbb CQ_6$ embedded in the projective space $\mathbb C\mathbb P_7$. It can be described by means of the equation
\begin{equation}
\label{e1aa}
     G_{AB}X^AX^B=0\ \ \ \Leftrightarrow\ \ \  (X,X)=0\
     (A,B,...=\overline{1,8})\spsd
\end{equation}
     Based on the Cartan triality principle \cite[p. 119(eng)]{Cartan1}, the manifold of quadric points is diffeomorphic to a manifold of 3-dimensional planar generators representing one and the same family (so we have the three manifolds are isomorphic to each other). Basic points of these generators
\begin{equation}
\label{e2aa}
     X_a=(X_a{}^A)\psps (a,b,...,i,j,...,p,q,...=\overline{1,4})
\end{equation}
     determine the equation
\begin{equation}
\label{e3aa}
      (X_a,X_b)=0\spsd
\end{equation}
      We define a planar generator with the help of its \emph{matrix coordinate}  $Z=(Z_a^p)$ \cite{Rosenfeld3}
\begin{equation}
\label{e4aa}
     X_a:=A_a+B_pZ^p_a\sps (A_a,B_p):=d_{ap}\sps B^a:=d^{ap}B_p
\end{equation}
     then from (\ref{e3aa}), the equation
\begin{equation}
\label{e5aa}
     Z_{ab}=-Z_{ba}\sps Z_{ab}:=d_{ap}Z_a^p
\end{equation}
     follows. This means that $X_a$ depend on the 6 complex parameters. As is well known \cite{Neifeld2}, \cite{Neifeld4}, \emph{spinor normalization} of a maximal planar generator manifold is defined by means of the giving of such a real differential correspondence between maximal planar generators of the quadric
\begin{equation}
\label{e6aa}
     f:\ \mathbb C\mathbb P_3(X_a)\rightarrow \mathbb C\mathbb P_3(Y_p)
\end{equation}
      that the generator $\mathbb C\mathbb P_3(X_a)$ corresponds to the plane $\mathbb C\mathbb P_3(Y_p)$, which does not intersect the first. For the six-dimensional quadric, these planar generators must belong to one of the two family. We will require that the normalization was \emph{harmonic} \cite[p. 209]{Norden6}. In the local coordinates, the normalization is determined by the parametric equations
\begin{equation}
\label{e7aa}
      X_a=X_a(u^\Lambda )\sps Y_a=Y_a(u^\Lambda )\ \
      (\Lambda ,\Psi ,...=\overline{1,12})\spsd
\end{equation}
      In this case, the relations
\begin{equation}
\label{e8aa}
     (X_a,X_b)=0\sps (Y_p,Y_q)=0\sps (X_a,Y_p)=c_{ap}
\end{equation}
     are executed. Due to the nondegeneracy of $c_{ap}$, we can define
\begin{equation}
\label{e9aa}
     Y^a:=c^{ap}Y_p\sps c^{ap}c_{pb}=\delta_b^a\sps (X_a,Y^b)=\delta_b^a\spsd
\end{equation}

\subsubsection{Neifeld operators}$ $
     \indent  \emph{Derivation equations} of the normalized family of maximal planar generators have the form \cite{Neifeld2}, \cite{Neifeld4}
\begin{equation}
\label{e10aa}
     \left\{
     \begin{array}{lcl}
     \nabla_\Lambda X_a & = & Y^b M_{\Lambda ab}\sps\\
     \tilde \nabla_\Lambda Y^b & = & X_a N_\Lambda{}^{ab}\spsd
     \end{array}
     \right.
\end{equation}
     Then from (\ref{e8aa}), the equalities
\begin{equation}
\label{e11aa}
     M_{\Lambda (ab)}=0\sps N_\Lambda{}^{(ab)}=0\sps
     \Gamma_{\Lambda a}{}^c=\tilde \Gamma_{\Lambda a}{}^c
     \end{equation}
     will follow, where $\Gamma_{\Lambda a}{}^c$ are the coefficients of the conformal torsion-free pseudo-Euclidean connection in the complex vector bundle whose the base is the maximal planar generator manifold. Note that the complex vector bundle is the metrizable in the sense that in it we can set a field of the metric quadrivector $\varepsilon_{abcd}$. Since the normalization is harmonic then the connection, defined above, is \emph{equiaffine}: the quadrivector $\varepsilon_{abcd}$ is a covariant constant. It allows to use $\varepsilon_{abcd}$ for the transfer of indexes. The operators $M_\Lambda{}^{ab}$ are \emph{connecting operators} so that each bivector of the fiber associates with the real vector of the tangent bundle
\begin{equation}
\label{e12aa}
     V^{ab}:=M_\Lambda{}^{ab}V^\Lambda\spsd
\end{equation}
     This correspondence is bijective. It follows that we can determine
\begin{equation}
\label{e13aa}
       \left\{
       \begin{array}{ccc}
       M^{\Lambda ab}M_{\Lambda cd} & = & \delta_{cd}^{ab}\sps\\
       \bar M^{\Lambda a'b'}M_{\Lambda cd} & = & 0\sps
       \end{array}
       \right.
       det
       \left\|
       \begin{array}{c}
       M_{\Lambda ab} \\
       \bar M_{\Lambda a'b'}
       \end{array}
       \right\|
       \ne 0\spsd
\end{equation}
     Then the operator
\begin{equation}
\label{e14aa}
       \bigtriangleup_\Lambda{}^\Psi=\frac{1}{2}
       (\delta_\Lambda{}^\Psi+if_\Lambda{}^\Psi)=\frac{1}{2}
       M_{\Lambda ab} M^{\Psi ab}
\end{equation}
       is \emph{Norden affinor} \cite{Norden1} such that
\begin{equation}
\label{e15aa}
       f_\Lambda{}^\Psi M^{\Lambda cd}=-iM^{\Psi cd}\sps
\end{equation}
       where $f_\Lambda{}^\Psi$ is the operator of the complex structure
\begin{equation}
\label{e16aa}
      f^2=-E\spsd
\end{equation}
      Let's define \emph{Neifeld operators} $m_\alpha{}^\Lambda$ as
\begin{equation}
\label{e17aa}
       \left\{
       \begin{array}{ccc}
       m_\alpha{}^\Lambda m^\beta{}_\Lambda & = & \delta_\alpha{}^\beta\sps\\
       m_\alpha{}^\Lambda \bar m^{\beta '}{}_\Lambda & = & 0\sps
       \end{array}
       \right.
       det
       \left\|
       \begin{array}{c}
       m^\alpha{}_\Lambda \\
       \bar m^{\alpha '}{}_\Lambda
       \end{array}
       \right\|
       \ne 0
\end{equation}
      according to \cite{Neifeld1}, and then
\begin{equation}
\label{e18aa}
       \bigtriangleup_\Lambda{}^\Psi=\frac{1}{2}
       (\delta_\Lambda{}^\Psi+if_\Lambda{}^\Psi)=
       m^\alpha{}_\Lambda m_\alpha{}^\Psi\ \
       (\alpha ,\beta ,...=\overline{1,6})
\end{equation}
      is the same Norden affinor. At the same time,
\begin{equation}
\label{e19aa}
       f_\Lambda{}^\Psi m_\alpha{}^\Lambda=-im_\alpha{}^\Psi\spsd
\end{equation}
      This means that we have the following decomposition
\begin{equation}
\label{e20aa}
       m_\alpha{}^\Lambda=\frac{1}{2}\eta_\alpha{}^{ab}
       M^\Lambda{}_{ab}
\end{equation}
      which defines the connecting Norden operators $\eta_\alpha{}^{ab}=-\eta_\alpha{}^{ba}$. For an arbitrary tensor $A_{\Lambda\Psi}$, we will have the following decomposition
\begin{equation}
\label{e21aa}
       \left\{
       \begin{array}{ccc}
       a_{\alpha\beta} & = & m_\alpha{}^\Lambda m_\beta{}^\Psi
       A_{\Lambda\Psi}\sps\\
       a_{\alpha '\beta} & = & \bar m_{\alpha '}{}^\Lambda m_\beta{}^\Psi
       A_{\Lambda\Psi}\sps
       \end{array}
       \right.
       \left\{
       \begin{array}{ccc}
       a_{abcd} & = & M^\Lambda{}_{ab} M^\Psi{}_{cd}
       A_{\Lambda\Psi}\sps\\
       a_{a'b'cd} & = & \bar M^\Lambda{}_{a'b'} M^\Psi{}_{cd}
       A_{\Lambda\Psi}\spsd\\
       \end{array}
       \right.
\end{equation}
      In this case, the metric quadrivector will correspond to the metric tensor $G_{\Lambda\Psi}$ so that
\begin{equation}
\label{e22aa}
       \left\{
       \begin{array}{ccl}
       g_{\alpha \beta} & = & m_\alpha{}^\Lambda m_\beta{}^\Psi
       G_{\Lambda\Psi}\sps\\
       g_{\alpha '\beta} & =  & 0\sps
       \end{array}
       \right.
       \left\{
       \begin{array}{ccl}
       \varepsilon_{abcd} & = & M^\Lambda{}_{ab} M^\Psi{}_{cd}
       G_{\Lambda\Psi}\sps\\
       \varepsilon_{a'b'cd} & =  & 0\spsd
       \end{array}
       \right.
\end{equation}
      The inverse relationships have the form
       $$
       G_{\Lambda\Psi}=\frac{1}{4}
       (M_\Lambda{}^{ab}M_\Psi{}^{cd}\varepsilon_{abcd}+
       \bar M_\Lambda{}^{a'b'}\bar M_\Psi{}^{c'd'}\bar\varepsilon_{a'b'c'd'})\sps
       $$
\begin{equation}
\label{e23aa}
       \eta^\alpha{}_{ab}=m^\alpha{}_\Lambda M^\Lambda{}_{ab}
       \sps
       \bar \eta^{\alpha '}{}_{a'b'}=
       \bar m^{\alpha '}{}_\Lambda \bar M^\Lambda{}_{a'b'}\sps
\end{equation}
       $$
       \eta^{\alpha '}{}_{ab}=\bar m^{\alpha '}{}_\Lambda M^\Lambda{}_{ab}\equiv 0
       \sps
       \bar \eta^{\alpha }{}_{a'b'}=
       m^\alpha {}_\Lambda \bar M^\Lambda{}_{a'b'}\equiv 0\spsd
       $$
      The last pair of the equations appears due to the analyticity of $M^\Lambda{}_{ab}$. From this, taking into account (\ref{e13aa}), (\ref{e17aa}), (\ref{e20aa}), the equations
\begin{equation}
\label{e24aa}
       \begin{array}{c}
       \eta^\alpha{}_{ab}\eta_\alpha{}^{cd}=
       M^\Lambda{}_{ab}M_\Lambda{}^{cd}=\delta_{ab}^{cd}\sps
       \frac{1}{4}\eta_\alpha{}^{ab}\eta^\beta{}_{cd}\delta_{ab}^{cd}=
       \delta_\alpha{}^\beta
       \end{array}
\end{equation}
      will follow. Thus, the maximal planar generator manifold is equipped with the metric tensor $G_{\Lambda\Psi}$, and therefore this manifold is diffeomorphic to the pseudo-Riemannian real space $V^{12}_{(6,6)}$ with the complex structure $f_\Lambda{}^\Psi$.

\subsubsection{Real and complex representations of the connection} $ $

      \indent Let us construct a more general connection. We say that two connections are equivalent if they define the same parallel transport along any curve of the base. The complex and real representations are carried out according to \cite[p. 169-178(rus)]{Lichnerowicz1}.
\begin{theorem}
\label{theorem111}
      Let $V^{2n}_{(n,n)}$ be a real pseudo-Riemannian space with the complex structure and $\mathbb CV^n$ be a complex analytic Riemannian space. Let $\mathbb CV^n$ be the complex representation of $V^{2n}_{(n,n)}$. Then the two following definitions are equivalent (specify the same connection)
       \begin{enumerate}
       \item In the tangent bundle $\tau^\mathbb R(V^{2n}_{(n,n)})$, there is the torsion-free Riemannian connection such that the tensor $m_\alpha{}^\Lambda$ is a covariant constant
\begin{equation}
\label{e20w}
        \nabla_\Lambda G_{\Theta\Psi}=0\sps
\end{equation}
\begin{equation}
\label{e21w}
       \nabla_\Lambda m_\alpha{}^\Psi=0\sps
       \nabla_\Lambda \bar m_{\alpha '}{}^\Psi=0\spsd
\end{equation}
       \item In the tangent bundle $\tau^\mathbb C(\mathbb CV^n)$, there is the torsion-free Riemannian connection such that the tensor $m_\alpha{}^\Lambda$ is a covariant constant
\begin{equation}
\label{e22w}
       \left\{
       \begin{array}{l}
       \nabla_\alpha g_{\beta\gamma}=0\sps\\
       \bar\nabla_{\alpha '} g_{\beta\gamma}=0\sps
       \end{array}
       \right.
       \left\{
       \begin{array}{l}
       \nabla_\alpha \bar g_{\beta '\gamma '}=0\sps\\
       \bar\nabla_{\alpha '} \bar g_{\beta '\gamma '}=0\sps
       \end{array}
       \right.
\end{equation}
\begin{equation}
\label{e23w}
       \left\{
       \begin{array}{l}
       \nabla_\beta m_\alpha{}^\Psi=0\sps\\
       \bar\nabla_{\beta '} m_\alpha{}^\Psi=0\sps
       \end{array}
       \right.
\end{equation}
       \end{enumerate}
      and the definition
\begin{equation}
\label{e24w}
       \nabla_\alpha:=m_\alpha{}^\Lambda{}\nabla_\Lambda\sps
       \nabla_{\alpha '}:=\bar m_{\alpha '}{}^\Lambda\nabla_\Lambda
\end{equation}
      is made.
\end{theorem}

\begin{proof}$ $\\
      \indent First. Let Connection 1). exists then we multiply (\ref{e21w}) by $m_\beta{}^\Lambda$ and obtain
\begin{equation}
\label{e25w}
        \begin{array}{c}
        \nabla_\Lambda m_\alpha{}^\Psi=0\sps\\
        0=m_\beta{}^\Lambda\nabla_\Lambda m_\alpha{}^\Psi=\nabla_\beta m_\alpha{}^\Psi\sps\\
        0=\bar m_{\beta '}{}^\Lambda\nabla_\Lambda m_\alpha{}^\Psi=\nabla_{\beta '} m_\alpha{}^\Psi\\
        \end{array}
\end{equation}
      taking into account the definition (\ref{e23w}). Inverse. Assume that (\ref{e23w}) is executed then obtain
\begin{equation}
\label{e26w}
        \begin{array}{c}
        \nabla_\alpha:=m_\alpha{}^\Lambda\nabla_\Lambda\sps\\
        m^\alpha{}_\Psi\nabla_\alpha=
        \bigtriangleup_\Psi{}^\Lambda\nabla_\Lambda
        \ \ \Leftrightarrow\ \ \
        \bar m^{\alpha '}{}_\Psi\bar\nabla_{\alpha '}
        =\bar\bigtriangleup_\Psi{}^\Lambda\nabla_\Lambda
        \end{array}
\end{equation}
      taking into account the definitions (\ref{e17aa}) and (\ref{e18aa}). We will combine these two equations and obtain
\begin{equation}
\label{e27w}
        \begin{array}{c}
        \nabla_\Lambda=
        (\bigtriangleup_\Psi{}^\Lambda+
        \bar\bigtriangleup_\Psi{}^\Lambda)\nabla_\Lambda=
        m^\alpha{}_\Lambda\nabla_\alpha+
        \bar m^{\alpha '}{}_\Lambda\bar \nabla_{\alpha '}\spsd
        \end{array}
\end{equation}
      Then from the condition (\ref{e23w}), the identity
\begin{equation}
\label{e28w}
        \nabla_\Lambda m^\beta{}_\Psi{}=
        m^\alpha{}_\Lambda{}\nabla_\alpha m^\beta{}_\Psi+
        \bar m^{\alpha '}{}_\Lambda\bar \nabla_{\alpha '} m^\beta{}_\Psi{}=0
\end{equation}
      follows.\\
      \indent Second. Because, from (\ref{e21w}) or (\ref{e23w}), the covariant constancy of the complex structure operator will follow due to the performance of (\ref{e19aa}) then the existence of the affine connection in the tangent bundle $\tau^\mathbb C(\mathbb CV_n)$ will imply according to \cite[v. 2, pp. 135-139(rus)]{Koboyasi1}. Considering the torsion-free  Riemannian connection, we find that if we know Connection 1)., we can define Connection 2). with the help of the condition (\ref{e21w}) rewriting as
\begin{equation}
\label{e29w}
        \begin{array}{c}
        \Gamma_{\Lambda\alpha}{}^\beta:=
        \Gamma_{\Lambda\Theta}{}^\Psi m_\alpha{}^\Theta m^\beta{}_\Psi+
        m^\beta{}_\Psi\partial_\Lambda m_\alpha{}^\Psi\sps
        \bar \Gamma_{\Lambda\alpha '}{}^{\beta '}:=
        \Gamma_{\Lambda\Theta}{}^\Psi \bar m_{\alpha '}{}^\Theta \bar m^{\beta '}{}_\Psi+
        \bar m^{\beta '}{}_\Psi\partial_\Lambda \bar m_{\alpha '}{}^\Psi\spsd
        \end{array}
\end{equation}
      And if we know Connection 2). then we can define Connection 1). rewriting the condition (\ref {e23w}) as
\begin{equation}
\label{e30w}
\begin{array}{l}
        \Gamma_{\beta\phantom{'}\Theta}{}^\Psi:=
        \Gamma_{\beta\phantom{'}\alpha}{}^\gamma m^\alpha{}_\Theta m_\gamma{}^\Psi+
        \bar \Gamma_{\beta\phantom{'} \alpha '}{}^{\gamma '}
        \bar m^{\alpha '}{}_\Theta \bar m_{\gamma '}{}^\Psi-
        m^\alpha{}_\Theta\partial_{\beta\phantom{'}} m_\alpha{}^\Psi-
        \bar m^{\alpha '}{}_\Theta\partial_{\beta\phantom{'}} \bar m_{\alpha '}{}^\Psi\sps\\
        \Gamma_{\beta'\Theta}{}^\Psi:=
        \Gamma_{\beta'\alpha}{}^\gamma m^\alpha{}_\Theta m_\gamma{}^\Psi+
        \bar \Gamma_{\beta' \alpha '}{}^{\gamma '}
        \bar m^{\alpha '}{}_\Theta \bar m_{\gamma '}{}^\Psi-
        m^\alpha{}_\Theta\partial_{\beta '} m_\alpha{}^\Psi-
        \bar m^{\alpha '}{}_\Theta\partial_{\beta '} \bar m_{\alpha '}{}^\Psi\spsd
\end{array}
\end{equation}
      At the same time, the equations
\begin{equation}
\label{e31w}
        \begin{array}{c}
        \Gamma_{\Lambda\Theta}{}^\Psi:=
        \Gamma_{\beta\Theta}{}^\Psi m^\beta{}_\Lambda+\Gamma_{\beta '\Theta}{}^\Psi \bar m^{\beta'}{}_\Lambda\sps
        \Gamma_{\beta\Theta}{}^\Psi=\Gamma_{\Lambda\Theta}{}^\Psi m_\beta{}^\Lambda\sps
        \Gamma_{\beta '\Theta}{}^\Psi=\Gamma_{\Lambda\Theta}{}^\Psi \bar m_{\beta '}{}^\Lambda\sps\\
        \partial_\beta = m_\beta{}^\Psi\partial_\Psi\sps
        \bar \partial_{\beta '}= \bar m_{\beta '}{}^\Psi\partial_\Psi\sps
        \partial_\Lambda = m^\beta{}_\Lambda\partial_\beta +\bar m^{\beta '}{}_\Lambda\bar \partial_{\beta '}
        \end{array}
\end{equation}
      are executed.\\
      \indent Third. From (\ref{e17aa}) and (\ref{e18aa}), the equations
\begin{equation}
\label{e32w}
        \begin{array}{c}
        g_{\alpha\beta}=G_{\Psi\Lambda} m_\alpha{}^\Psi m_\beta{}^\Lambda\sps\\
        G_{\Theta\Upsilon}\bigtriangleup_\Lambda{}^\Theta
        \bigtriangleup_\Psi{}^\Upsilon=m^\alpha{}_\Lambda
        m^\beta{}_\Psi g_{\alpha\beta}\sps\\
        \frac{1}{2}(G_{\Lambda\Psi}+i
        G_{\Theta\left(\right.\Lambda} f_{\Psi\left.\right)}{}^\Theta)=
        m^\alpha{}_\Lambda m^\beta{}_\Psi g_{\alpha\beta}\sps\\
        \frac{1}{2}(G_{\Lambda\Psi}-i
        G_{\Theta\left(\right.\Lambda} f_{\Psi\left.\right)}{}^\Theta)=
        \bar m^{\alpha '}{}_\Lambda \bar m^{\beta '}{}_\Psi{}
        \bar g_{\alpha '\beta '}\sps\\
        G_{\Lambda\Psi}=m^\alpha{}_\Lambda m^\beta{}_\Psi g_{\alpha\beta}+
        \bar m^{\alpha '}{}_\Lambda \bar m^{\beta '}{}_\Psi{}
        g_{\alpha '\beta '}
        \end{array}
\end{equation}
      will follow. Therefore, from the conditions (\ref{e22w}), the equation (\ref{e20w}) will follow too. Conversely, if  (\ref{e20w}) is executed then we have
\begin{equation}
\label{e33w}
        \begin{array}{c}
        m_\alpha{}^\Lambda m_\beta{}^\Psi m_\gamma{}^\Theta{}
        \nabla_\Lambda G_{\Psi\Theta}=0\ \ \Leftrightarrow\ \ \
        \nabla_\alpha g_{\beta\gamma}=0\sps\\
        \bar m_{\alpha '}{}^\Lambda m^\Psi{}_\beta m_\gamma{}^\Theta
        \nabla_\Lambda G_{\Psi\Theta}=0\ \ \Leftrightarrow\ \ \
        \nabla_{\alpha '} g_{\beta\gamma}=0\spsd
        \end{array}
\end{equation}
      \indent Fourth. Since,  Connection 1). is the unique, then Connection 2). is the unique too.\\
\end{proof}
      Note that for an analytical connection, the analytical conditions have the form $\Gamma_{\alpha '\beta}{}^\gamma\equiv 0$, $\bar \partial_{\alpha '} m_\beta{}^\Lambda \equiv 0$ $\Rightarrow$ $\bar\nabla_{\alpha '}m_\beta{}^\Lambda \equiv 0$.
\begin{theorem}
\label{theorem112}
      Let the real pseudo-Riemannian space $V^{12}_{(6,6)}$ be given as the base of the bundles. Then the two torsion-free connections, given in the bundles $\tau^\mathbb R (V^{12}_{(6,6)})$ and $A^\mathbb C$, are equivalent:
      \begin{enumerate}
      \item the Riemannian connection, defined in the bundle $\tau^\mathbb R (V^{12}_{(6,6)})$ by means of the condition
\begin{equation}
\label{e34w}
            \nabla_\Lambda G_{\Psi\Theta}=0\spsd
\end{equation}
      \item the Riemannian connection, defined in the bundle $A^\mathbb C$ by means of the conditions
\begin{equation}
\label{e35w}
            \nabla_\Lambda \varepsilon_{abcd}=0\sps
            \nabla_\Lambda \bar\varepsilon_{a'b'c'd'}=0\spsd
\end{equation}
      Thus, Connection 2). is uniquely determined from the conditions
\begin{equation}
\label{e36w}
            \nabla_\Lambda M_\Psi{}^{ab}=0\sps
            \nabla_\Lambda \bar M_\Psi{}^{a'b'}=0\spsd
\end{equation}
      \end{enumerate}
\end{theorem}

\begin{proof}
      In the tangent bundle, the torsion-free Riemannian connection, defined with the help of the condition (\ref{e34w}), always exists and it is the unique. We will rewrite the first condition (\ref{e36w}) as
\begin{equation}
\label{e37w}
      \nabla_\Lambda M_\Psi{}^{aa_1}=
      \partial_\Lambda M_\Psi{}^{aa_1} -
      \Gamma_{\Lambda\Psi}{}^{\Theta}M_\Theta{}^{aa_1}+
      \Gamma_{\Lambda c}{}^aM_\Psi{}^{ca_1}+
      \Gamma_{\Lambda c}{}^{a_1}M_\Psi{}^{ac}=0\spsd
\end{equation}
      Multiply this equation by $M^\Psi{}_{ac_1}$
\begin{equation}
\label{e38w}
      \Gamma_{\Lambda c_1}{}^{a_1}=-\frac{1}{2}
      (M^\Psi{}_{ac_1}\partial_\Lambda M_\Psi{}^{aa_1}
      -\Gamma_{\Lambda\Psi}{}^\Theta M_\Theta{}^{aa_1} M^\Psi{}_{ac_1}+
      \Gamma_{\Lambda a}{}^a\delta_{c_1}{}^{a_1})\spsd
\end{equation}
      Besides, from the conditions (\ref{e35w}) and (\ref{e13aa}), the equation
\begin{equation}
\label{e39w}
      \begin{array}{c}
      \frac{1}{24} M_\Psi{}^{ab} M^\Psi{}^{cd}
      \partial_\Lambda(M^\Theta{}_{ab} M_\Theta{}_{cd})=
      \frac{1}{24}\varepsilon^{abcd}\partial_\Lambda\varepsilon_{abcd}=\\[2ex]
      =\frac{1}{24}\varepsilon^{abcd}(\nabla_\Lambda\varepsilon_{abcd}+
      4\Gamma_{\Lambda\left[\right. a}{}^k\varepsilon_{|k|bcd\left.\right]})=
      \frac{1}{6}\varepsilon^{abcd}\Gamma_{\Lambda a}{}^k
      \varepsilon_{kbcd}=\Gamma_{\Lambda k}{}^k
      \end{array}
\end{equation}
      follows. On this basis, the equation (\ref{e38w}) can be put in the definition of Connection 2).\\
      \indent Suppose that in the bundle $A^\mathbb C$, there is another operator of the symmetric covariant derivative $\tilde\nabla_\Lambda$ such that
\begin{equation}
\label{e40w}
       \tilde\nabla_\Lambda\varepsilon_{abcd}=0\ \ \Rightarrow\ \ \
       (\tilde\nabla_\Lambda-\nabla_\Lambda)\varepsilon_{abcd}=0
       \ \ \Leftrightarrow\ \ \ Q_{\Lambda k}{}^k=0\sps
\end{equation}
      where the tensor $Q_{\Lambda a}{}^b$ is the strain tensor defined in the bundle $A^\mathbb C$. Let the tensor  $Q_{\Lambda\Psi}{}^\Theta$ be the strain tensor in the tangent bundle $\tau^\mathbb R(V^{12}_{(6,6)})$. Consider the action of these operators on a bivector $R^{ab}=M_\Psi{}^{ab}r^\Psi$
\begin{equation}
\label{e41w}
      \begin{array}{c}
      (\tilde\nabla_\Lambda-\nabla_\Lambda)R^{ab}=
      (Q_{\Lambda k}{}^a\delta_t{}^b-
       Q_{\Lambda k}{}^b\delta_t{}^a)R^{kt}
      =M_\Psi{}^{ab}(\tilde\nabla_\Lambda-\nabla_\Lambda)r^\Psi=
       M_\Psi{}^{ab}Q_{\Lambda \Theta}{}^\Psi r^\Theta\sps\\
      (\tilde\nabla_\Lambda-\nabla_\Lambda)R_{ab}=-M^\Theta{}_{ab}Q_{\Lambda \Theta}{}^\Psi r_\Psi\spsd
      \end{array}
\end{equation}
      From this, the identities
\begin{equation}
\label{e42w}
      \begin{array}{c}
      M_\Psi{}^{ab}Q_{\Lambda \Theta}{}^\Psi r^\Theta=
      2Q_{\Lambda\left[\right. k}{}^a\delta_{t\left.\right]}{}^b R^{kt}\sps\\
      M_\Psi{}^{ab}Q_{\Lambda \Theta}{}^\Psi r^\Theta=
      2Q_{\Lambda\left[\right. k}{}^a\delta_{t\left.\right]}{}^b M_\Theta{}^{kt} r^\Theta\sps\\
      M_\Psi{}^{ab}Q_{\Lambda \Theta}{}^\Psi =M_\Theta{}^{kt}
      2Q_{\Lambda\left[\right. k}{}^a\delta_{t\left.\right]}{}^b=
      2M_\Theta{}^{k\left[\right.b}Q_{\Lambda k}{}^{a\left.\right]}\sps\\
      Q_{\Lambda \Theta\Psi} =M_\Theta{}^{kb}M_\Psi{}_{ab}
      Q_{\Lambda k}{}^a+\bar M_\Theta{}^{k'b'}\bar M_\Psi{}_{a'b'}
      \bar Q_{\Lambda k'}{}^{a'}=\\
      =-M_\Psi{}^{kb}M_\Theta{}_{ab}
      Q_{\Lambda k}{}^a-\bar M_\Psi{}^{k'b'}\bar M_\Theta{}_{a'b'}
      \bar Q_{\Lambda k'}{}^{a'}
      \end{array}
\end{equation}
      will follow. Whence, we obtain
\begin{equation}
\label{e43w}
      Q_{\Lambda\Theta\Psi}=-Q_{\Lambda\Psi\Theta}\spsd
\end{equation}
      In the absence of the torsion, we have
\begin{equation}
\label{e44w}
      Q_{\Lambda\Theta\Psi}=Q_{\Theta\Lambda\Psi}
\end{equation}
      then
\begin{equation}
\label{e45w}
      Q_{\Lambda\Theta\Psi}=0\sps
\end{equation}
      and this means the uniqueness of Connection 2).
\end{proof}

\begin{corollary}
\label{corollary113}
      Let as the base of the bundle, the complex analytic Riemannian space $\mathbb CV^6$ be given. Then the two torsion-free connections, given in the bundles $\tau^\mathbb C(\mathbb CV^6)$ and $A^\mathbb C$, are equivalent:
      \begin{enumerate}
      \item the Riemannian analytic connection, defined in the bundle $\tau^\mathbb C(\mathbb CV^6)$ by means of the conditions
\begin{equation}
\label{e46w}
            \nabla_\alpha g_{\beta\gamma}=0\sps
            \bar\nabla_{\alpha '} g_{\beta '\gamma '}=0\spsd
\end{equation}
      \item the Riemannian analytic connection, defined in the bundle $A^\mathbb C$ by means of the conditions
\begin{equation}
\label{e47w}
            \nabla_\alpha\varepsilon_{abcd}=0\sps
            \bar\nabla_{\alpha '}\bar \varepsilon_{a'b'c'd'}=0\spsd
\end{equation}
      Thus, Connection 2). is uniquely determined by means of the conditions
\begin{equation}
\label{e49w}
            \nabla_\alpha \eta_\beta{}^{ab}=0\sps
            \bar\nabla_{\alpha '} \bar\eta_{\beta '}{}^{a'b'}=0\spsd
\end{equation}
      \end{enumerate}
\end{corollary}

\begin{proof}
      The proof follows from Theorem \ref{theorem111}, Theorem \ref{theorem112}, the analyticity (\ref{e23aa}), and the equation (\ref{re19aa}). In particular, the analyticity of $\eta_\beta{}^{ab}$ means $\partial_{\alpha '}\eta_\beta{}^{ab}\equiv 0$, and from (\ref{e19aa}), the equality $\Gamma_{\alpha '\beta}{}^\gamma\equiv 0$ implies.
\end{proof}

\subsubsection{\texorpdfstring{Involution in $\mathbb C\mathbb P_7$}{Involution in the projective space}}$ $

      \indent Suppose now that in $\mathbb C\mathbb P_7$  an involution is given in the sense of \cite{Neifeld1}
\begin{equation}
\label{e25aa}
       \bar S_{A'}{}^B S_B{}^{D'}=\delta_{A'}{}^{D'}
\end{equation}
      then the condition of the reality of the point $X^A$ takes the form
\begin{equation}
\label{e26aa}
       S_A{}^{B'}\bar X^A=X^{B'}\spsd
\end{equation}
      We require that this involution defines an embedding of a real quadric in the complex one. This is equivalent to that the tensor, defining the quadric, is self-conjugated with respect to this involution. Then maximum planar generators of the real quadric must satisfy the conditions
\begin{equation}
\label{e27aa}
       1).\ \ \bar S_{A'}{}^B\bar X_{a'}^{A'}s_a{}^{a'}=X_a^B\sps
       2).\ \ \bar S_{A'}{}^B\bar X_{a'}^{A'}s^{aa'}=X^{aB}\spsd
\end{equation}
      Here the spin-tensors $s_a{}^{a'}$ and $s^{aa'}$ define \emph{Hermitian involution} and \emph{Hermitian polarity} in the complex bundle respectively. These two cases arise from the fact that in the bundle with fibers isomorphic to $\mathbb C^4$, we do not have a tensor with which the help single indices can be raised and lowered. The first case means that the generator itself and the conjugate one belong one and the same family, in the second case, these generators belong to the two different families. Therefor,  all possible cases of real inclusions are exhausted that follows from the results of the second chapter. From (\ref{e30b}), (\ref{e25aa})-(\ref{e27aa}), the identities
\begin{equation}
\label{e28aa}
       1).\ \ s_a{}^{a'}\bar s_{a'}{}^b=\pm\delta_a{}^b\sps
       2).\ \ s_{aa'}\bar s^{a'b}=\pm\delta_a{}^b
\end{equation}
      will follow. Next, we consider Item 2). only as the most interesting from the standpoint of physics \cite[v.2, p.68(eng)]{Penrose1}. Item 1). is treated similarly and we omit it. Then
\begin{equation}
\label{e29aa}
       \bar X^{b'}=\bar s^{b'a}X_a\spsd
\end{equation}
      Therefore, we can write down the expression
\begin{equation}
\label{e30aa}
     (X_a,\bar X^{b'})=0\sps (Y_p,\bar Y^{q'})=0\sps
     (X_a,\bar Y_{b'})=s_{ab'}
\end{equation}
      equivalent to (\ref {e8aa}), (\ref{e9aa}). Put
\begin{equation}
\label{e31aa}
     \left\{
     \begin{array}{lcl}
     \nabla_\Lambda \bar X^{a'} & = & \bar Y_{b'} \tilde M_\Lambda{}^{a'b'}\sps\\
     \nabla_\Lambda \bar Y_{b'} & = & \bar X^{a'} \tilde N_{\Lambda a'b'}
     \end{array}
     \right.
\end{equation}
      then from (\ref{e30aa}), the identities
\begin{equation}
\label{e32aa}
     \tilde M_{\Lambda a'b'}=-\frac{1}{2}s_{ca'}s_{db'}
     \varepsilon^{cdab}M_{\Lambda ab}\sps
     \nabla_\Lambda s_{ab'}=0
\end{equation}
      will follow. Therefore, by means of the equation
\begin{equation}
\label{e33aa}
     S_\Lambda{}^\Theta=\frac{1}{2}
     (M_{\Lambda ab} \bar M^\Theta{}_{c'd'}
     \bar s^{c'a}\bar s^{d'b}+
     \bar M_{\Lambda a' b'} M^\Theta{}_{cd}
     s^{ca'}s^{db'})\sps
\end{equation}
    we define the real involution
\begin{equation}
\label{e34aa}
     S_\Lambda{}^\Theta S_\Theta{}^\Psi=\delta_\Lambda{}^\Psi\sps
     \bar M_{\Lambda a'b'}=-S_\Lambda{}^\Psi\tilde M_{\Psi a'b'}\sps
     S_\Lambda{}^\Theta f_\Theta{}^\Lambda=
     -f_\Lambda{}^\Theta S_\Theta{}^\Lambda\spsd
\end{equation}
      In addition, we can define the complex representation of the involution according to \cite{Neifeld1}
\begin{equation}
\label{e38aa}
       \left\{
       \begin{array}{ccl}
       S_\alpha{}^\beta & =  & 0 \sps\\
       S_\alpha{}^{\beta '} & = & m_\alpha{}^\Lambda \bar m^{\beta '}{}_\Psi
       S_\Lambda{}^\Psi\sps
       \end{array}
       \right.
       S_\alpha{}^{\beta '}\bar S_{\beta '}{}^\gamma=\delta_\alpha{}^\gamma\spsd
\end{equation}

\subsubsection{Riemannian connection compatible with the involution}
\begin{corollary}
\label{corollary114}
      Let the complex analytic Riemannian space $\mathbb CV^6$  be given as the base of the bundles. Then the two torsion-free connections, given in the bundles $\tau^\mathbb C(\mathbb CV^6)$ and $A^\mathbb C(S)$ are equivalent:
      \begin{enumerate}
      \item the Riemannian real connection, defined in the bundle $\tau^\mathbb C(\mathbb CV^6)$ by means of the conditions
\begin{equation}
\label{e50w}
            \nabla_\alpha g_{\beta\gamma}=0\sps
            \nabla_\alpha S_\gamma{}^{\beta'}=0
\end{equation}
      (such the connection we call \texttt{connection compatible with the involution}).
      \item the Riemannian real connection, defined in the bundle $A^\mathbb C(S)$ by means of the conditions
\begin{equation}
\label{e51w}
            \nabla_\alpha\varepsilon_{abcd}=0\sps
            \nabla_\alpha s_{aba'b'}=0\spsd
\end{equation}
      Thus, Connection 2). is uniquely determined by means of the condition
\begin{equation}
\label{e52w}
            \nabla_\alpha \eta_\beta{}^{ab}=0\spsd
\end{equation}
      \end{enumerate}
\end{corollary}

\begin{proof}
      Under the conditions of Corollary \ref{corollary113}, we consider Connection 1). specified by means of the conditions (\ref{e49w}) then from the reality of the conditions, the equation
\begin{equation}
\label{e54w}
      S_\beta{}^{\gamma '}\bar\partial_{\gamma '}=\partial_\beta
\end{equation}
     will follow. Therefor, from the covariant constancy of the involution tensor,  we give
\begin{equation}
\label{e55w}
      \nabla_\gamma=S_\gamma{}^{\beta '}\nabla_{\beta '}
\end{equation}
      that will define the real connection. If we put
\begin{equation}
\label{e56w}
      s_{aba'b'}:=\eta^\alpha{}_{ab}\bar \eta_{\beta 'a'b'}S_\alpha{}^{\beta '}\sps
\end{equation}
      then from (\ref{e52w}) and (\ref{e50w}), the equation
\begin{equation}
\label{e57w}
       \nabla_\alpha s_{aba'b'}=0
\end{equation}
      will follow.
\end{proof}

\begin{corollary}
\label{corollary115}
      Let the real Riemannian space $V^6_{(2,4)}$ be given as the base of the bundles. Then the two torsion-free connections, given in the bundles $\tau^\mathbb R(V^6_{(2,4)})$ and $A^\mathbb C(S)$, are equivalent:
      \begin{enumerate}
      \item the Riemannian real connection, defined in the bundle $\tau^\mathbb R(\mathbb V^6_{(2,4)})$ by means of the conditions
\begin{equation}
\label{e58w}
            \nabla_i g_{jk}=0\spsd
\end{equation}
      \item the Riemannian real connection, defined in the bundle $A^\mathbb C(S)$ by means of the conditions
\begin{equation}
\label{e59w}
            \nabla_i\varepsilon_{abcd}=0\sps
            \nabla_is_{ab'}=0\spsd
\end{equation}
      Thus, Connection 2). is uniquely determined by means of the condition
\begin{equation}
\label{e60w}
            \nabla_i \eta_j{}^{ab}=0\spsd
\end{equation}
      \end{enumerate}
\end{corollary}

\begin{proof}
      It follows from preceding Corollary \ref{corollary114} under the condition of the covariant constancy of the inclusion operator $H_i{}^\alpha$ which will determine the appropriate connection. We will only prove the covariant constancy of the Hermitian polarity tensor. Since,
\begin{equation}
\label{e61w}
    \nabla_\alpha s_{abc'd'}=
    \nabla_\alpha s_{\left[ \right. a
    \left| c' \right|} s_{b \left. \right] d'}=0\spsd
\end{equation}
      Deploying this equation by the Leibniz rule and contracting with $s^{ac'}$, we get
\begin{equation}
\label{e62w}
    \nabla_\alpha s_{bd'}=-1/2s_{bd'}s^{ac'}\nabla_\alpha s_{ac'}\spsd
\end{equation}
      After the contraction with $s^{bd'}$ of this equation, we finally obtain
\begin{equation}
\label{e63w}
    s^{ac'}\nabla_\alpha s_{ac'}=0\sps \nabla_\alpha s_{ac'}=0\spsd
\end{equation}
\end{proof}

\subsubsection{Bitwistor equation}
      \indent From (\ref{e10aa}), (\ref{e13aa}), (\ref{e20aa}), assuming
\begin{equation}
\label{e51aa}
       \nabla_{ab}:=\eta^\alpha{}_{ab}\nabla_\alpha\sps
\end{equation}
       we obtain
\begin{equation}
\label{e52aa}
       \nabla_\alpha X_a=Y^b\eta_{\alpha ab}\ \ \Leftrightarrow\ \
       \nabla_{cd} X_a=Y^b\varepsilon_{cdab}
\end{equation}
      such that the equations
\begin{equation}
\label{e53aa}
       \nabla_{c\left(\right.d}X_{a\left.\right)}=0\sps
       \nabla^{c\left(\right.d}X^{a\left.\right)}=0\sps
\end{equation}
      the last of which we called \emph{bitwistor equation}, are executed. Using this equation, we can investigate the conformal structure of the spaces $\mathbb C\mathbb R^6$. It should be noted that the bitwistor equation does not change under conformal transformations of the metric and this equation is invariant under normalization transformations in the sense of \cite{Neifeld2}, \cite{Neifeld4}.
\begin{proof}
      Indeed, suppose that a conformal transformation of the metric has the form
\begin{equation}
\label{e144}
        g_{\alpha\beta}\longmapsto\hat g_{\alpha\beta}=
        \Omega^2g_{\alpha\beta}\spsd
\end{equation}
      Then from the equation
\begin{equation}
\label{e145}
        \hat\nabla_\alpha\hat\varepsilon_{abcd}=
        \nabla_\alpha\varepsilon_{abcd}=0\sps
\end{equation}
      the condition
\begin{equation}
\label{e146}
        0=\hat\nabla_\alpha(\Omega^2\varepsilon_{abcd})=
        \varepsilon_{abcd}(2\Omega\nabla_\alpha\Omega-
        \Omega^2\Theta_{\alpha k}{}^k)=0
\end{equation}
        will follow. Let's put
\begin{equation}
\label{e147}
        B_\alpha:=\frac{1}{2}\Theta_{\alpha k}{}^k\spsd
\end{equation}
        Because $\nabla_\alpha$ and $\hat\nabla_\alpha$ are symmetrical operators then in the tangent bundle $\tau^\mathbb C(\mathbb CV^6)$, the equation
\begin{equation}
\label{e148}
        Q_{\alpha\beta\gamma}=Q_{\beta\alpha\gamma}
\end{equation}
        is executed. Then in the bundle $A^\mathbb C$, the relations
\begin{equation}
\label{e149}
        \Theta_{cc_1a}{}^b=B_{ca}\delta_{c_1}{}^b-B_{c_1a}\delta_c{}^b
        \sps
        B_{ab}=-B_{ba}
\end{equation}
      are executed too. Then
\begin{equation}
\label{e150}
        B_\alpha=\frac{1}{2}\eta_\alpha{}^{ab}B_{ab}=
        \Omega^{-1}\nabla_\alpha\Omega\spsd
\end{equation}
      Put
\begin{equation}
\label{e151}
        \hat X^c=X^c\spsd
\end{equation}
      Then, according to
\begin{equation}
\label{e152}
        \hat\nabla_{ab}X^c=\nabla_{ab}X^c+
        2B_{\left[\right. a|k|}\delta_{\left.b\right]}{}^cX^k\sps
\end{equation}
      we obtain
\begin{equation}
\label{e153}
        \hat\nabla^{a\left(b\right.}\hat X^{c\left.\right)}=
        \Omega^{-2}\nabla^{a\left(\right.b}X^{c\left.\right)}\spsd
\end{equation}
      This means that the bitwistor equation is a conformal invariant.
\end{proof}

\section{\texorpdfstring{Theorems on the curvature tensor. The canonical form of bivectors of 6-dimensional
$\mbox{(pseudo-)}$ Euclidean spaces $\mathbb R^6_{(p,q)}$ with the even index q}{Theorems on the curvature tensor. The canonical form of bivectors}}$ $
\Abstract{
    \indent Since the connection, introduced in the tangent space of $\mathbb CV_6$, satisfies the equations
    $$
    \nabla_\alpha g_{\gamma\delta}=0\sps
    \bar\nabla_{\alpha '} \bar g_{\gamma '\delta '}=0
    $$
    and the connection in the bundle $A^\mathbb C$ is determined from the relations
    $$
    \nabla_\alpha\eta_\beta{}^{ab}=0\sps
    \bar \nabla_{\alpha '}\bar\eta_{\beta '}{}^{a'b'}=0
    $$
    then we can choose a nonholonomic special basis such that the metric $g_{\gamma\delta}$ will have the diagonal form with $\ll$+1$\gg$ on the main diagonal in this basis, and the coordinates of the generalized connecting Norden operators will be some constants (similar to the formulas (\ref{e56})-(\ref{e57})). From this, it follows that the coordinates of the operators $A_{\alpha\beta a}{}^b$ will be some constants in the basis. Then the curvature tensor with the help of the operators $A_{\alpha\beta a}{}^b$ can be represented as
    $$
    R_{\alpha\beta\gamma\delta}=
    A_{\alpha\beta a}{}^bA_{\gamma\delta c}{}^d R_b{}^a{}_d{}^c\spsd
    $$
    In this case, knowing the structure of the spin-tensor $R_b{}^a{}_d{}^c$, we can restore the structure of the curvature tensor. But the study of the structure of the spin-tensor $R_b{}^a{}_d{}^c$ is facilitated by the fact that it does not almost contain non-significant components. In the 4-dimensional case, similar tensor, called  \emph{curvature spinors} \cite{Penrose1}, greatly simplify the classification of the curvature tensor of the 4-dimensional space. This classification was undertaken for the first time by Petrov with the help of direct tensor methods. Therefore, it is expected that it will be easier to classify the spin-tensor $R_b{}^a{}_d{}^c$ rather than to deal with the classification of the tensor $R_{\alpha\beta\gamma\delta}$. The first part of this chapter is devoted to the interaction between the tensor and the spin-tensor.\\
    \indent The canonical form of a skew-symmetric bilinear form for the even index q of the metric of the space $\mathbb R^6_{(p,q)}$ is discussed in the third subsection. This form in the special basis has the representation
    $$
    \frac{1}{2}R_{\alpha\beta}X^\alpha Y^\beta=
    R_{16}X^{\left[\right.1} Y^{6\left.\right]}+
    R_{23}X^{\left[\right.2} Y^{3\left.\right]}+
    R_{45}X^{\left[\right.4} Y^{5\left.\right]}\spsd
    $$
    In addition, such the fact as a correspondence between vectors of $\mathbb C^4$ and isotropic simple bivectors belonging to the isotropic cone $K_6$ of $\mathbb C\mathbb R^6$ is established.  This correspondence will determine a vector up to the factor $re^{i\Theta}\in\mathbb C$. Based on this correspondence, we can talk about the geometric interpretation of \emph{isotropic twistor} of $\mathbb C^4$  (in the sense of $s_{aa'}X^aX^{a'}=0$) into the space $R^6_{(2,4)}$. To implement it, we need to learn how to compare isotropic vectors belonging to the cone $K_6$. Therefore, using the stereographic projection, we have an invariant (\emph{coordinate-independent}) way to define a vector, tangent to $K_6$ and applied at the point P. Its norm, taken with the sign ''-'', is associated to the isotropic vector K with the beginning at the apex of the cone and the ending at the point P. This norm is called \emph{extension} of the vector K. Then we can choose the unit vector k and compare all isotropic vectors with this vector. In this case, the ambiguity of the correspondence can be removed as follows:
\begin{enumerate}
    \item \emph{flagpole}: r is the extension of any isotropic vector determined by the specified isotropic simple bivector (the isotropic 2-plane) belonging to the cone $K_6$;
    \item \emph{flag-plane}: the plane \foreignlanguage{russian}{П} is spanned on the flagpole and the vector orthogonal to the flagpole. $\Theta$ is an angle of the rotation of the 3-half-plane \foreignlanguage{russian}{П} around the flagpole.
\end{enumerate}
    The resulting interpretation is similar to the correspondence between spinors and isotropic vectors of the Minkowski space discussed in the monography \cite{Penrose1}.
}

\subsection{Theorem on bitensors of the 6-dimensional space}$ $

    \indent Before to pass to the investigation of curvature tensor properties in the space $\mathbb CV^6$, we consider the following theorem.
\begin{theorem}
\label{theorem1}
     The classification of a bitensor, possessing  the properties
\begin{equation}
\label{e11}
    R_{\alpha \beta \gamma \delta}=R_{\left[ \alpha \beta \right]
    \left[\gamma \delta] \right.} \sps
    R_{\alpha \beta \gamma \delta}=R_{\gamma \delta \alpha \beta}\sps
    R_{\alpha \beta \gamma \delta}+R_{\alpha \delta \beta \gamma}+
    R_{\alpha \gamma \delta \beta}=0
\end{equation}
    and belonging to the tangent bundle $\tau^\mathbb C(\mathbb CV^6)$ over the analytic Riemannian space $\mathbb CV^6$, can be reduced to the classification of a spin-tensor $R_a{}^b{}_c{}^d$ of the 4-dimensional complex spinor space such that
\begin{equation}
\label{e12}
    R_{\alpha \beta \gamma \delta}= A_{\alpha \beta d}{}^c
    A_{\gamma \delta r}{}^sR_c{}^d{}_s{}^r\spsd
\end{equation}
    Besides, the equations
\begin{equation}
\label{e13}
    R_k{}^k{}_s{}^r=R_s{}^r{}_k{}^k=0 \sps  R_c{}^d{}_s{}^r=
    R_s{}^r{}_c{}^d
\end{equation}
    are executed. The decomposition
\begin{equation}
\label{e14}
    R_c{}^d{}_s{}^r=C_c{}^d{}_s{}^r-P_{cs}{}^{dr}-\frac{1}{40}\cdot
    R(3\delta_s{}^d \delta_c{}^r-2\delta_s{}^r \delta_c{}^d)
\end{equation}
    corresponds to the decomposition of the tensor $R_{\alpha\beta}{}^{\gamma\delta}$
\begin{equation}
\label{e14a}
    R_{\alpha\beta}{}^{\gamma\delta}=C_{\alpha\beta}{}^{\gamma\delta}+
    R_{\left[\alpha \right.}{}^{\left[ \gamma \right.}
    g_{\left.\beta \right]}{}^{\left. \delta \right]}-
    1/10Rg_{\left[\alpha \right.}{}^{\left[ \gamma \right.}
    g_{\left.\beta \right]}{}^{\left. \delta \right]}
\end{equation}
    on the irreducible components not resulted by orthogonal transformations. These components will satisfy the following relations
\begin{equation}
\label{e15}
    P_{cs}{}^{rd}=-4(R_{\left[c \right.}{}^{\left[r \right.}
    {}_{\left. s \right]}{}^{\left. d\right]}+
    R_k{}^{\left[r \right.}{}_{\left[c \right.}{}^{\left|k\right|}
    \delta_{\left. s\right]}{}^{\left.d\right]})\sps
\end{equation}
\begin{equation}
\label{e16}
    C_c{}^d{}_s{}^r=R_{\left(c \right.}{}^{\left(d \right.}{}_
    {\left. s\right)}{}^{\left. r\right)}+
    \frac{1}{40}\cdot R\delta_{\left(s \right.}{}^d\delta_{\left. c\right)}{}^r\sps
    C_c{}^d{}_s{}^r=C_{\left(c \right.}{}^{\left(d \right.}{}_
    {\left. s\right)}{}^{\left.r\right)}\sps
\end{equation}
\begin{equation}
\label{e17}
    R=R_\beta{}^\beta=-2\cdot R_k{}^r{}_r{}^k\sps
    P_{kc}{}^{kd}=1/2\cdot R\delta_c{}^d\sps
\end{equation}
\begin{equation}
\label{e18}
    R_l{}^d{}_s{}^l=-\frac{1}{8}\cdot R\delta_s{}^d\sps
\end{equation}
    the last of which is equivalent to the Bianchi identity (\ref{e11}).
\end{theorem}

\begin{proof}
    On the basis of (\ref{e1}), we have the following equality
\begin{equation}
\label{e22}
    R_{\alpha\beta\gamma\delta}=1/16\cdot\eta_{\alpha}{}^{aa_1}
    \eta_{\beta}{}^{bb_1}\eta_{\gamma}{}^{cc_1}
    \eta_{\delta}{}^{dd_1}R_{aa_1bb_1cc_1dd_1}\spsd
\end{equation}
    Put
\begin{equation}
\label{e23}
    R_c{}^d{}_s{}^r:=\frac{1}{4}R_{ck}{}^{dk}{}_{st}{}^{rt}
    \sps  R_{\beta\gamma}=\frac{1}{4}\cdot\eta_{\beta}{}^{cs}\eta_{\gamma}{}_{rd}
    \cdot P_{cs}{}^{rd}\spsd
\end{equation}
    From this, taking into account (\ref{e4}),  the formula (\ref{e12}) follows
\begin{equation}
\label{e12a}
    R_{\alpha \beta \gamma \delta}= A_{\alpha \beta d}{}^c
    A_{\gamma \delta r}{}^sR_c{}^d{}_s{}^r\spsd
\end{equation}
    Then the equation
\begin{equation}
\label{e7p}
       \begin{array}{c}
       R_{\beta\delta}=R_{\alpha\beta}{}^{\alpha}{}_\delta=
       A_{\alpha\beta d}{}^cA^\alpha{}_{\delta r}{}^s
       R_c{}^d{}_s{}^r=
       (\eta_\beta{}^{cs}\eta_{\delta rd}+
       \eta_\beta{}^{ck}\eta_{\delta kr}\delta_d{}^s)
       R_c{}^d{}_s{}^r=\\ \\
       =\frac{1}{4}
       \eta_\beta{}^{cs}\eta_{\delta rd}\cdot 4
       (R_{\left[\right.c}{}^{\left[\right.d}
       {}_{s\left.\right]}{}^{r\left.\right]}-
       R_{\left[\right.c}{}^k
       {}_{|k|}{}^{\left[\right.r}
       \delta_{s\left.\right]}{}^{d\left.\right]})
       \end{array}
\end{equation}
    is executed. Put
\begin{equation}
\label{e8p}
       \begin{array}{c}
       P_{cs}{}^{rd}:=
       -4(R_{\left[\right.c}{}^{\left[\right.r}
       {}_{s\left.\right]}{}^{d\left.\right]}-
       R_{\left[\right.c}{}^k
       {}_{|k|}{}^{\left[\right.r}
       \delta_{s\left.\right]}{}^{d\left.\right]})
       \end{array}
\end{equation}
    then
\begin{equation}
\label{e9p}
       R_{\beta\delta}=\frac{1}{4}
       \eta_\beta{}^{cs}\eta_{\delta rd}P_{cs}{}^{rd}\sps
\end{equation}
    and this proves the formula (\ref{e15}). Since, the scalar curvature is given by
 \begin{equation}
\label{e10p}
       \begin{array}{c}
       R=R_\beta{}^\beta=\frac{1}{4}
       \eta_\beta{}^{cc_1}\eta^\beta{}_{aa_1}P_{cc_1}{}^{aa_1}=
       \frac{1}{4}\varepsilon_{aa_1}{}^{cc_1}P_{cc_1}{}^{aa_1}=
       \frac{1}{2}P_{aa_1}{}^{aa_1}=-2R_k{}^r{}_r{}^k\sps
       \end{array}
\end{equation}
    and, moreover, the equation
\begin{equation}
\label{e11p}
       \begin{array}{c}
       P_{ks}{}^{kd}=
       -4(R_{\left[\right.k}{}^{\left[\right.k}
       {}_{s\left.\right]}{}^{r\left.\right]}+
       R_{\left[\right.r}{}^k
       {}_{|k|}{}^{\left[\right.r}
       \delta_{s\left.\right]}{}^{d\left.\right]})=\\ \\
       =-4(-\frac{1}{2}R_k{}^d{}_s{}^k+
       \frac{1}{4}(R_k{}^r{}_r{}^k\delta_s{}^d+
       4R_k{}^d{}_s{}^k-2R_k{}^d{}_s{}^k))=
       -R_k{}^r{}_r{}^k\delta_s{}^d=\frac{1}{2}R\delta_s{}^d
       \end{array}
\end{equation}
    is executed then the formulas (\ref{e17}) will actually be true.\\
    \indent The Bianchi identity (\ref{e11}) can be rewritten as
\begin{equation}
\label{e26}
    (A_{\alpha\beta d}{}^cA_{\gamma\delta}{}_r{}^s+
    A_{\alpha\gamma d}{}^cA_{\delta\beta}{}_r{}^s+
    A_{\alpha\delta d}{}^cA_{\beta\gamma}{}_r{}^s)\cdot
    R_c{}^d{}_s{}^r=0\spsd
\end{equation}
    Contracting this equation with $A^{\alpha\beta}{}_t{}^lA^{\gamma\delta}{}_m{}^n$ and taking into account (\ref{e21}), we will obtain
    $$
    4R_k{}^l{}_m{}^k\delta_t{}^n+
    4R_r{}^n{}_t{}^r\delta_m{}^l-
    2R_k{}^l{}_t{}^k\delta_m{}^n-
    2R_k{}^n{}_m{}^k\delta_t{}^l-
    $$
\begin{equation}
\label{e27}
    -2R_k{}^r{}_r{}^k\delta_t{}^n\delta_m{}^l+
    R_r{}^k{}_k{}^r\delta_m{}^n\delta_t{}^l=0\spsd
\end{equation}
     Contracting this equation with $\delta_n{}^t$, we will obtain the formula (\ref{e18}). In this case, all 15 significant equations are stored (all calculations are given in Appendix (\ref{e16p})-(\ref{e18p})).\\
     \indent Put
    $$
    C_{\alpha\beta}{}^{\gamma\delta}:=
    A_{\alpha\beta d}{}^cA^{\gamma\delta}{}_r{}^sC_c{}^d{}_s{}^r\sps
    $$
\begin{equation}
\label{e24}
    C_{\alpha\beta}{}^{\gamma\delta}:=
    R_{\alpha\beta}{}^{\gamma\delta}-
    R_{\left[\alpha \right.}{}^{\left[ \gamma \right.}
    g_{\left.\beta \right]}{}^{\left. \delta \right]}+
    1/10Rg_{\left[\alpha \right.}{}^{\left[ \gamma \right.}
    g_{\left.\beta \right]}{}^{\left. \delta \right]}\spsd
\end{equation}
    From (\ref{e4})~, (\ref{e15})~, (\ref{e17}), the  equations
    $$
    R_{\left[\alpha \right.}{}^{\left[ \gamma \right.}
    g_{\left.\beta \right]}{}^{\left. \delta \right]}=
    A_{\alpha\beta d}{}^cA^{\gamma\delta}{}_r{}^s
    \cdot\frac{1}{4}(P_{sc}{}^{dr}-1/2R\delta_s{}^d\delta_c{}^r+
    \frac{1}{4}R\delta_s{}^r\delta_c{}^d)\sps
    $$
\begin{equation}
\label{e25}
    g_{\left[\alpha \right.}{}^{\left[ \gamma \right.}
    g_{\left.\beta \right]}{}^{\left. \delta \right]}=
    A_{\alpha\beta d}{}^cA^{\gamma\delta}{}_r{}^s
    \cdot\frac{1}{4}(1/2\delta_s{}^r\delta_c{}^d-2\delta_s{}^d\delta_c{}^r)
\end{equation}
    will follow, and therefor the decompositions (\ref{e14}), (\ref{e16}) are executed (all calculations are given in Appendix (\ref{e12p}) - (\ref{e20p})).\\
\end{proof}

\subsubsection{Corollaries from the theorem}
\begin{corollary}
\label{corollary1}
\begin{enumerate}
    \item
    The simplicity conditions of a bivector of the 6-dimensional space $\mathbb C\mathbb R^6$ can be written down as
\begin{equation}
\label{e28a}
    p^{\left[\alpha\beta\right.}p^{\gamma\delta\left.\right]}=0\spsd
\end{equation}
    The coordinates of the bivector can be associated to the traceless complex matrix $4 \times 4$ such that
\begin{equation}
\label{e28}
    p_l{}^dp_s{}^l-1/4(p_l{}^kp_k{}^l)\delta_s{}^d=0\spsd
\end{equation}
    \item
    A simple bivector of the space $\mathbb C\mathbb R^6$, constructed on isotropic vectors ($p^{\alpha \beta} p_{\alpha \beta}=0$), can be associated to \texttt{degenerate Rosenfeld null-pair}: a covector and a vector of the space $\mathbb C^4$, the contraction of which is zero. In this case, the given vector and covector are determined to within a complex factor.
\end{enumerate}
\end{corollary}

\begin{proof}
   \indent 1). A bivector is simple if and only if there is the decomposition
\begin{equation}
\label{e1www}
    p^{\alpha\beta}=X^\alpha Y^\beta-Y^\alpha X^\beta\spsd
\end{equation}
   Therefore, if (\ref{e1www}) is satisfied then the formula (\ref{e28a}) will be true. \\
   \indent Conversely, if the condition (\ref{e28a}) is executed then it can be written down as
\begin{equation}
\label{e2www}
    p^{\alpha\beta}p^{\gamma\delta}-
    p^{\alpha\gamma}p^{\beta\delta}+
    p^{\beta\gamma}p^{\alpha\delta}=0\spsd
\end{equation}
    We contract this equation with nonzero covectors $T_\delta$ and $Z_\gamma$, that $p^{\gamma\delta}Z_\gamma T_\delta \ne 0$,
\begin{equation}
\label{e3www}
    p^{\alpha\beta}=\frac{1}{p^{\lambda\mu}Z_\lambda T_\mu}
    (p^{\alpha\gamma}Z_\gamma p^{\beta\delta}T_\delta-
     p^{\beta\gamma}Z_\gamma p^{\alpha\delta}T_\delta)\spsd
\end{equation}
    Put
\begin{equation}
\label{e4www}
    X^\alpha:=\frac{1}{p^{\lambda\mu}Z_\lambda T_\mu}p^{\alpha\gamma}Z_\gamma\sps
    Y^\beta:=\frac{1}{p^{\lambda\mu}Z_\lambda T_\mu}p^{\beta\delta}T_\delta
\end{equation}
    whence the condition (\ref{e1www}) will follow. Since the tensor $R_{\alpha\beta\gamma\delta}=p_{\alpha\beta}
    p_{\gamma\delta}$ satisfies the conditions of Theorem \ref{theorem1} then the formula (\ref{e28}) is a direct consequence of the Bianchi identity (\ref{e18}). \\
    \indent 2). Under the condition of Item 1)., we add the isotropy condition
\begin{equation}
\label{e5www}
    p^{\alpha\beta}p_{\alpha\beta}=0
\end{equation}
    which in view of the formulas (\ref{e21}), takes the form
 \begin{equation}
\label{e6www}
    \begin{array}{c}
    A^{\alpha\beta}{}_a{}^b A_{\alpha\beta c}{}^d
    p_b{}^a p_d{}^c=0\sps\\
    p_b{}^a p_a{}^b=0\spsd
    \end{array}
\end{equation}
   It follows that there exist nonzero $X^a$ and $Y_b$, that
\begin{equation}
\label{e7www}
    p_a{}^b=X^aY_b\sps X^aY_a=0\spsd
\end{equation}
    This formula can be viewed as a consequence of Lemma \ref{lemma1} of the second chapter (it is enough to consider the bivector $p^{\alpha\beta}=r_1{}^{\left[\right.\alpha}r_2{}^{\beta\left.\right]}$, where $r_1{}^\alpha$ and $r_2{}^\alpha$ are the same as in the lemma). In this case, $X^a$ and $Y_b$ are defined up to the transformation
\begin{equation}
\label{e8www}
    X^a\longmapsto e^\phi X^a\sps
    Y_b\longmapsto e^{-\phi} Y_b\spsd
\end{equation}
\end{proof}

    Note that the pair $(X^a,Y_b)$ is \emph{Rosenfeld null-pair}. In the space $\mathbb C\mathbb P^4='\mathbb C^4/'\mathbb C$ (where $'\mathbb C^s=\mathbb C^s/{0}$), $X^a$ will determine the point, and $Y_b$ will determine the plane with the incidence condition
\begin{equation}
\label{e9www}
     X^aY_a=0\spsd
\end{equation}
    Therefore, we can define the space $\mathbb C$\foreignlanguage{russian}{П}${}^4='\mathbb C^*{}^4/'\mathbb C$ and the dual space $\mathbb C\mathbb P^4$. Then the space $\mathbb C\mathbb P^4\times \mathbb C$\foreignlanguage{russian}{П}${}^4$ is \emph{Rosenfeld null-pair space}. It should be noted that such the spaces have been studied for the first time in \cite{Sintsov1} and \cite{Kotelnikov1}.\\

\begin{corollary}
\label{corollary2}
    In the case of the reality of the bitensor from Theorem \ref{theorem1}, for the even index metric, the condition
\begin{equation}
\label{e29}
    R_{ab'cd'}=\bar R_{b'ad'c}
\end{equation}
    is imposed, and for the odd index metric, the condition
\begin{equation}
\label{e30}
    R_a{}^{b'}{}_c{}^{d'}=
    \bar R_a{}^{b'}{}_c{}^{d'}
\end{equation}
     is imposed.
\end{corollary}
\begin{proof}
    It is based on the properties of the inclusion tensor $s_{...}{}^{...}$.
\end{proof}

\subsection{Basic properties and identities of the curvature tensor}$ $

    \indent As an example, let's consider basic properties of the curvature tensor of the Riemannian space $\mathbb CV^6$. Because in a nonholonomic basis, the operators $A_{\alpha\beta a}{}^b$ are constant then we can get all properties of the curvature tensor considering the spin-tensor $R_a{}^b{} _c{}^ d$. The curvature tensor of the space $\mathbb CV_6$ in a neighborhood U satisfies Theorem \ref{theorem1}. We set
\begin{equation}
\label{e21p}
       \begin{array}{c}
       \Box_a{}^d:=\frac{1}{2}
       (\nabla_{ak}\nabla^{dk}-
       \nabla^{dk}\nabla_{ak})\sps\\ \\
       \Box_{\alpha\beta}:=
       2\nabla_{\left[\right.\alpha}\nabla_{\beta\left.\right]}\spsd
       \end{array}
\end{equation}
    Due to the covariant constancy of the  generalized connecting Norden operators, we have
\begin{equation}
\label{e22p}
       \begin{array}{c}
       \nabla_{\left[\right.\alpha}\nabla_{\beta\left.\right]}=
       \frac{1}{4}
       \eta_{\left[\right.\alpha}{}^{aa_1}\eta_{\beta\left.\right]}{}^{bb_1}
       \nabla_{aa_1}\nabla_{bb_1}=\\ \\
       =\frac{1}{4}
       \eta_\alpha{}^{aa_1}\eta_\beta{}^{bb_1}\cdot\frac{3}{2}
       (\nabla_{a\left[\right.a_1}\nabla_{bb_1\left.\right]}-
       \nabla_{\left[\right.bb_1}\nabla_{a\left.\right]a_1})=\\ \\
       =\frac{1}{4}
       \eta_\alpha{}^{aa_1}\eta_\beta{}^{bb_1}\cdot\frac{3}{2}
       (\delta_{\left[\right.a_1}{}^k\delta_b{}^n
       \delta_{b_1\left.\right]}{}^{n_1}\nabla_{ak}\nabla_{nn_1}-
       \delta_{\left[\right.b}{}^n\delta_{b_1}{}^{n_1}
       \delta_{a\left.\right]}{}^k\nabla_{nn_1}\nabla_{ka_1})=
       \end{array}
\end{equation}
       $$
       \begin{array}{c}
       =\frac{1}{4}
       \eta_\alpha{}^{aa_1}\eta_\beta{}^{bb_1}\cdot\frac{1}{4}
       (\varepsilon_{a_1bb_1d}\varepsilon^{knn_1d}\nabla_{ak}\nabla_{nn_1}-
       \varepsilon_{bb_1ad}\varepsilon^{nn_1kd}\nabla_{nn_1}\nabla_{ka_1})=\\ \\
       =\frac{1}{4}
       \eta_\alpha{}^{aa_1}\eta_\beta{}^{bb_1}\cdot\frac{1}{4}
       (\varepsilon_{a_1bb_1d}\varepsilon^{kd}{}_{nn_1}\nabla_{ak}\nabla^{nn_1}+
       2\varepsilon_{bb_1a_1d}\nabla^{kd}\nabla_{ka})=\\ \\
       =\frac{1}{4}
       \eta_\alpha{}^{aa_1}\eta_\beta{}^{bb_1}\cdot\frac{1}{4}
       \varepsilon_{a_1bb_1d}(2\delta_{\left[\right.n}{}^k
       \delta_{n_1\left.\right]}{}^d\nabla_{ak}\nabla^{nn_1}+
       2\nabla^{kd}\nabla_{ka})=
       \end{array}
       $$
       $$
       \begin{array}{c}
       =\frac{1}{4}
       \eta_\alpha{}^{aa_1}\eta_{\beta a_1d}
       (\nabla_{an}\nabla^{nd}+
       \nabla^{kd}\nabla_{ka})=
       A_{\alpha\beta d}{}^a\cdot\frac{1}{4}
       (\nabla_{ak}\nabla^{dk}-
       \nabla^{dk}\nabla_{ak})\spsd
       \end{array}
       $$
     Therefor,
\begin{equation}
\label{e23p}
       \begin{array}{c}
       \Box_{\alpha\beta}=A_{\alpha\beta d}{}^a\Box_a{}^d\spsd
       \end{array}
\end{equation}
    We formulate the few basic statements concerning the operator $\Box_a{}^d$:
\begin{enumerate}
    \item
    From the Ricci identity
\begin{equation}
\label{e46s}
     \Box_{\alpha\beta}k^{\gamma\delta}=
     R_{\alpha\beta\lambda}{}^\gamma k^{\lambda\delta}+
     R_{\alpha\beta\lambda}{}^\delta k^{\gamma\lambda}\sps
     \Box_{\alpha\beta}r^\gamma=
     R_{\alpha\beta\lambda}{}^\gamma r^{\lambda}
\end{equation}
     will follow the identities ($k^{\alpha\beta}=-k^{\beta\alpha}$)
\begin{equation}
\label{e47s}
     \Box_a{}^bk_d{}^c=R_a{}^b{}_m{}^ck_d{}^m-R_a{}^b{}_d{}^nk_n{}^c\sps
     \Box_a{}^br^{cc_1}=R_a{}^b{}_m{}^cr^{mc_1}+R_a{}^b{}_m{}^{c_1}r^{cm}\spsd
\end{equation}
    And from this, we obtain finally
\begin{equation}
\label{e47}
    \Box_a{}^bX^c=R_a{}^b{}_m{}^cX^m\sps
    \Box_a{}^bX_c=-R_a{}^b{}_c{}^mX_m
\end{equation}
    (the proof is given in Appendix (\ref{e25p})-(\ref{e42p})).
    \item
    The differential Bianchi identity
\begin{equation}
\label{e48s}
    \nabla_{\left[\right.\alpha}R_{\beta\gamma\left.\right]\delta\lambda}=0
\end{equation}
    has the form
\begin{equation}
\label{e48}
    \nabla_{\left[cm\right.}R_{t\left.\right]}{}^k{}_r{}^s=
    \delta_{\left[m\right.}{}^k\nabla_{c\left|n\right|}R_{t\left.
    \right]}{}^n{}_r{}^s
\end{equation}
    (the proof is given in Appendix (\ref{e43p})-(\ref{e50p})).
    \item
    Contracting (\ref{e48}) with $\delta_k{}^c$, we obtain
\begin{equation}
\label{e49}
    \nabla_{c\left(\right.m}R_{t\left.\right)}{}^c{}_r{}^s=0\sps
\end{equation}
    and the contraction of (\ref{e49}) with $\delta_s{}^m$ will give
\begin{equation}
\label{e50}
    \nabla_{cm}R_t{}^c{}_r{}^m=1/8\nabla_{rt}R
\end{equation}
    which is equivalent to the well-known equation
\begin{equation}
\label{e50s}
    \nabla^\alpha(R_{\alpha\beta}-1/2Rg_{\alpha\beta})=0\spsd
\end{equation}
\end{enumerate}

\subsection{\texorpdfstring{The canonical form of bivectors of the 6-dimensional $\mbox{(pseudo-)}$ Euclidean space $\mathbb R^6_{(p, q)}$ with the metric of the even index q}{The canonical form of bivectors}}

\begin{theorem}
\label{theorem7} (On the canonical form of a bivector.)
    For the space $\mathbb R^6_{(p, q)}$ with the metric of the even index q = 0,6, a nondegenerate skew-symmetric bilinear form can be reduced to \texttt{canonical form} in some basis
\begin{equation}
\label{e1g}
    \frac{1}{2}R_{\alpha\beta}X^\alpha Y^\beta=
    R_{16}X^{\left[\right.1} Y^{6\left.\right]}+
    R_{23}X^{\left[\right.2} Y^{3\left.\right]}+
    R_{45}X^{\left[\right.4} Y^{5\left.\right]}\spsd
\end{equation}
\end{theorem}

\begin{proof}
    In the case of the space $\mathbb R^6_{(p, q)}$ with the metric of the even index q, the equations
\begin{equation}
\label{e2g}
    R_{\alpha\beta}=A_{\alpha\beta a}{}^b R_b{}^a\sps
    R_b{}^a=-\bar R^a{}_b\sps R_a{}^a=0
\end{equation}
     are executed that means that $iR_b{}^a$ is an Hermitian tensor in the case of q = 0,6. Therefore, the matrix of the tensor $R_b{}^a$ is reduced to the diagonal form by transformations of the group isomorphic to $SU(4)$. These transformations correspond to ones of the orthogonal group $SO^e(6,\mathbb R)$. It follows that the matrix of the tensor $R_b{}^a$ in the special basis has the form
\begin{equation}
\label{e3g}
    \begin{array}{c}
    R_b{}^a=
    \left(
    \begin{array}{cccc}
    \lambda_1 & 0 & 0 & 0\\
    0 & \lambda_2  & 0 & 0\\
    0 & 0 &\lambda_3  & 0\\
    0 & 0 & 0 & \lambda_4 \\
    \end{array}
    \right)\sps\\ \\
    \begin{array}{ccc}
    q=0,6, & \lambda_1 ,\lambda_2, \lambda_3, \lambda_4 \in iR, &
    \lambda_1+\lambda_2+\lambda_3+\lambda_4=0\spsd\\
    \end{array}
    \end{array}
\end{equation}
    At the same time, the two equalities
\begin{equation}
\label{e4g}
    \tilde R_b{}^a=S_b{}^cR_c{}^d\bar S^a{}_d\sps
    S_a{}^b\bar S^c{}_b=\delta_a{}^c
\end{equation}
    are executed.

\begin{proof}
    We consider a transformation $K_\alpha{}^\beta$ from  the group $SO^e(6,\mathbb R)$ and a spinor representation $S_a{}^b$ corresponding to $K_\alpha{}^\beta$ from the group $SU(4)$
\begin{equation}
\label{e5g}
    \begin{array}{c}
    K_\alpha{}^\beta K_\gamma{}^\delta R_{\beta\delta}=
    -A^{\beta\delta}{}_a{}^b K_{\left[\right.\alpha}{}^{ak}
    K_{\beta\left.\right]bk}A_{\beta\delta}{}_c{}^d R_d{}^c
    =-\frac{1}{2}(\frac{1}{2}\delta_a{}^b\delta_c{}^d-
    2\delta_a{}^d\delta_c{}^b)K_{\left[\right.\alpha}{}^{ak}
    K_{\beta\left.\right]bk}R_d{}^c=\\
    =K_{\left[\right.\alpha}{}^{ak}K_{\beta\left.\right]bk}R_a{}^b=
    \frac{1}{2}A_{\alpha\beta c}{}^dK_{dr}{}^{ak}K^{cr}{}_{bk}R_a{}^b
    =\frac{1}{8}K_{dr}{}^{ak}K_{mn}{}^{sl}\varepsilon_{slbk}\varepsilon^{mncr}
    R_a{}^b A_{\alpha\beta c}{}^d=\\=\tilde R_{\alpha\beta}=
    A_{\alpha\beta c}{}^d \tilde R_d{}^c\spsd
    \end{array}
\end{equation}
    Multiply the both sides (\ref{e5g}) by $A^{\alpha\beta}{}_p{}^t$ and obtain
\begin{equation}
\label{e6g}
    \begin{array}{c}
    \frac{1}{8}K_{pr}{}^{ak}K_{mn}{}^{sl}\varepsilon_{slbk}\varepsilon^{mntr}
    R_a{}^b=\tilde R_p{}^t\sps\\
    \frac{1}{2}S_{\left[\right.p}{}^aS_{r\left.\right]}{}^k
    S_{\left[\right.m}{}^sS_{n\left.\right]}{}^l
    \varepsilon_{slbk}\varepsilon^{mntr}R_a{}^b=
    \frac{1}{2}S_{\left[\right.p}{}^a(S_{r\left.\right]}{}^k
    S_m{}^sS_n{}^l\varepsilon_{kslb})\varepsilon^{mntr}R_a{}^b=\\
    =\frac{1}{2}S_{\left[\right.p}{}^a((S^{-1})_{|b|}{}^q
    \varepsilon_{r\left.\right]mnq})\varepsilon^{mntr}R_a{}^b
    =\frac{1}{4}(6S_p{}^a(S^{-1}){}_b{}^t-4S_r{}^a(S^{-1}){}_b{}^q
    \delta_{\left[\right.q}{}^t\delta_{p\left.\right]}{}^rR_a{}^b)=\\
    =S_p{}^aR_a{}^b(S^{-1})_b{}^t=S_p{}^aR_a{}^b\bar S^t{}_b=\tilde R_p{}^t\spsd
    \end{array}
\end{equation}
\end{proof}

    Using the special basis in the case of q = 0,6, we found the corresponding coordinates of the cobivector from $\mathbb R^6_{(p,q)}$
\begin{equation}
\label{e7g}
    \begin{array}{c}
    R_{16}=A_{16a}{}^bR_b{}^a= -R_{61} \sps \\
    R_{23}=A_{23a}{}^bR_b{}^a= -R_{32} \sps \\
    R_{45}=A_{45a}{}^bR_b{}^a= -R_{54} \spsd \\
    \end{array}
\end{equation}
\end{proof}

    Note that for the case q = 2,4, a similar statement can be formulated. But here there will be some difficulties associated with the problem of the diagonalization since in this case, in the special basis, the matrix of the Hermitian polarity tensor will be different from the identity.

\subsection{\texorpdfstring{Geometric representation of a twistor in $\mathbb R^6_{(2,4)}$}{Geometric representation of a twistor}}
\subsubsection{Stereographic projection}
\Abstract{
    \indent This part defines the notion \emph{extension of the isotropic vector} for the space $\mathbb R^6_{(2,4)}$ with the metric of the index equal to 4. It will be shown how to choose a vector of \emph{unit extension}. Then vectors collinear to this vector will differ from the latter by a real factor r: \emph{extension of the flagpole}.
    }

    Let the metric of the space $\mathbb R^6_{(2,4)}$ has the form
\begin{equation}
\label{e78}
    dS^2=dT^2+dV^2-dW^2-dX^2-dY^2-dZ^2\sps
\end{equation}
    and let a cross-section of the light cone $K_6$
\begin{equation}
\label{e79}
    T^2+V^2-W^2-X^2-Y^2-Z^2=0
\end{equation}
     be set by the plane V+W=1. Let's consider the stereographic projection of this section on the plane (V=0,W=1) with the pole $N(0,\frac{1}{2},\frac{1}{2},0,0,0)$ so that the point P(T,V,W,X,Y,Z) corresponds to the point $p(t,0,1,x,y,z)$ in the plane (V=0,W=1). Then
\begin{equation}
\label{e80}
    T/t=X/x=Y/y=Z/z=-\frac{(V-\frac{1}{2})}{\frac{1}{2}}\spsd
\end{equation}
    We make the substitution
\begin{equation}
\label{e82}
    \varsigma=-ix+y \sps
    \omega=-i(t+z)  \sps
    \eta=i(z-t)
\end{equation}
    and obtain
\begin{equation}
\label{e83}
    \varsigma =\frac{-iX+Y}{2V-1}\sps
    \eta=\frac{-i(T+Z)}{2V-1}    \sps
    \omega=\frac{i(Z-T)}{2V-1}   \spsd
\end{equation}
    Therefore, the metric, induced in the cross-section, has the form
 \begin{equation}
\label{e84}
    d s^2:=dT^2-dX^2-dY^2-dZ^2=-
    \frac{ d\varsigma d\bar\varsigma+d\omega d\eta}
    {(\varsigma\bar\varsigma+\eta\omega)^2}
\end{equation}
    (the proof of this fact is given in Appendix (\ref{e81p})-(\ref{e86p})). Put
\begin{equation}
\label{e85}
    X:=
    \left(
    \begin{array}{ccc}
    \omega & \varsigma \\
    -\bar\varsigma & \eta
    \end{array}
    \right) \sps
    dX:=
    \left(
    \begin{array}{ccc}
    d\omega & d\varsigma \\
    -d\bar\varsigma & d\eta
    \end{array}
    \right) \sps
    \frac{\partial}{\partial X}:=
    \left(
    \begin{array}{ccc}
    \frac{\partial}{\partial\omega}   &
    \frac{\partial}{\partial\varsigma} \\
    -\frac{\partial}{\partial\bar\varsigma} &
    \frac{\partial}{\partial\eta}
    \end{array}
    \right)\spsd
\end{equation}
    Then (\ref{e84}) can be rewritten as
\begin{equation}
\label{e86}
    d s^2=-\frac{det(dX)}{(det(X))^2}\sps
    \bar X^T+X=0\spsd
\end{equation}
    Consider the linear-fractional group L
\begin{equation}
\label{e87}
    {\tilde X}=(AX+B)(CX+D)^{-1} \sps
    S:=
    \left(
    \begin{array}{ccc}
     A & B \\
     C & D
    \end{array}
    \right) \sps
    detS=1\spsd
\end{equation}
    The reality condition, imposed on X ($X^*+X=0)$, gives the linear-fractional unitary subgroup  $LU(2,2)$ so that the matrix S from (\ref{e87}) satisfies the identity
\begin{equation}
\label{e87a}
    S^*\hat E S=\hat E\sps
    \hat E:=\left(
    \begin{array}{cc}
    0 & E \\
    E & 0
    \end{array}
    \right)\spsd
\end{equation}
    Further, we define the special basis \cite[v. 2, p. 65, eq. (6.2.18), p. 306, eq. (9.3.7)(eng)]{Penrose1}
\begin{equation}
\label{e89}
    \begin{array}{ccccc}
    R^{12} & = & 1/\sqrt{2}(V+W)  & = &  \omega^0\xi^1-\omega^1\xi^0 \sps\\
    R^{34} & = & 1/\sqrt{2}(V-W)  & = &  \bar\pi^0\bar\eta^1-\bar\pi^1\bar\eta^0\sps\\
    R^{14} & = & i/\sqrt{2}(T+Z)  & = &  \omega^0\bar\eta^0-\xi^0\bar\pi^0\sps\\
    R^{23} & = & i/\sqrt{2}(Z-T)  & = &  \xi^1\bar\pi^1-\omega^1\bar\eta^1\sps\\
    R^{24} & = & 1/\sqrt{2}(Y+iX) & = & \omega^1\bar\eta^0-\xi^1\bar\pi^0\sps\\
    R^{13} & = & 1/\sqrt{2}(Y-iX) & = & \xi^0\bar\pi^1-\omega^0\bar\eta^1\spsd
    \end{array}
\end{equation}
    This relations are the remarkable because it shows as the bivector $R^{ab}$ is expressed in terms of its spinor components. Define
\begin{equation}
\label{e88}
    \begin{array}{c}
    X:=YZ^{-1} \sps
    {\tilde Y}=AY+BZ \sps
    {\tilde Z}=CY+DZ \sps
    Y=\left(
    \begin{array}{cc}
    \omega^0 & \xi^0 \\
    \omega^1 & \xi^1
    \end{array}
    \right)\sps
    Z=\left(
    \begin{array}{cc}
    \bar\pi_0 & \bar \eta_0 \\
    \bar\pi_1 & \bar \eta_1
    \end{array}
    \right)\sps
    \end{array}
\end{equation}
    then from (\ref{e89}), the equations
\begin{equation}
\label{e99}
    R:=\parallel R^{ab}\parallel=
    \left(
    \begin{array}{ccc}
    (detY)J & YZ^{-1}(detZ)J \\
    -(YZ^{-1}(detZ)J)^T & (detZ)J
    \end{array}
    \right)=
    \left(
    \begin{array}{c}
    Y\\
    Z
    \end{array}
    \right)J\left(
    \begin{array}{cc}
    Y^T & Z^T
    \end{array}
    \right)\sps
\end{equation}
    $$
    J:=\left(
    \begin{array}{cc}
    0  & E \\
    -E & 0
    \end{array}
    \right)\sps
    \tilde R=SRS^T
    $$
    will follow. Determine
\begin{equation}
    \tilde S:=\tilde IS\tilde I^{-1}\sps
    \tilde I:=\frac{1}{\sqrt 2}
    \left(
    \begin{array}{ccc}
    E & E  \\
    -E & E
    \end{array}
    \right) \sps
\end{equation}
    then
\begin{equation}
\label{e101}
    \tilde S^*\tilde E\tilde S=\tilde E\spsd
\end{equation}
    The matrixes S form the group isomorphic $SU(2,2)$ so the matrixes $\tilde S$ form the group $SU(2,2)$. A transformation from the group $LU(2,2)$  is called \emph{twistor transformation}. Due to the double covering of the connected identity component of the group $SO(2,4)$ (which is denoted as $SO^e(2,4)$) by the group $SU(2,2)$ and due to the double covering of the conformal group $C^{\uparrow 4}_+(1,3)$ \cite[v. 2, p. 304, eq. (9.2.10)(eng)] {Penrose1} by the group $SO^e(2,4)$,  the existence of the isomorphisms
\begin{equation}
\label{e102}
    \begin{array}{c}
    SU(2,2)/\{\pm 1 ;\pm i\}
    \cong LU(2,2)\cong C^{\uparrow 4}_+(1,3) \cong
    SO^e(2,4)/\{\pm 1\}
    \end{array}
\end{equation}
    will imply. This means that the group $LU(2,2)$ exhausts all conformal transformations of the group $\mathbb C^{\uparrow 4}_+(1,3)$. The matrix S is restored up to a factor $\lambda $, that $\lambda^4=1$ (det(S)=1), and whence we obtain the ambiguity. Sine we have the equalities
\begin{equation}
\label{e90}
    Y=AX+B\ \ \Rightarrow \ \ dX=AdY \sps
    Y=X^{-1}\ \ \Rightarrow \ \ dX=-X^{-1}dXX^{-1}\sps
\end{equation}
    where A and B is some constant matrixes, then
\begin{equation}
\label{e92}
    \tilde Z^*\ d\tilde X\ \tilde Z=
    Z^*\ dX\ Z\spsd
\end{equation}
    This equation is an invariant under the group LU(2,2). The proof of this fact is considered in Appendix (\ref{e87p})-(\ref{e100p}). Other invariant can be obtained with the help of the identity
\begin{equation}
\label{e93}
    Y=AX+B\ \ \Rightarrow \ \
    \frac{\partial}{\partial X}=A^T\frac{\partial}{\partial Y} \sps
    Y=X^{-1}\ \ \Rightarrow \ \ \frac{\partial}{\partial X}
    =-Y^T\frac{\partial}{\partial Y}Y^T\sps
\end{equation}
    where A and B is also some constant matrixes. This invariant will have the form
\begin{equation}
\label{e95}
    \tilde Z^{-1}\frac{\partial}{\partial
    \tilde X^T}\tilde Z^{*\ -1}=
     Z^{-1}\frac{\partial}{\partial X^T} Z^{*\ -1}
\end{equation}
    (the proof of this fact is given in Appendix (\ref{e101p})-(\ref{e114p})). This means that there is a real vector $\tilde L$ tangent to the hyperboloid resulting with the help of the cross-section of the cone $K_6$ by the plane V + W = 1. This vector is an invariant under transformations of a basis from the group LU(2,2) (i.e., coordinate-independent in the tangent space to this hyperboloid). The vector $\tilde L$ is uniquely determined by the matrix
\begin{equation}
\label{e96}
    \begin{array}{c}
    \hat L:=\frac{1}{\sqrt{2}}
    (Z^{-1}\frac{\partial}{\partial X^T} Z^{*\ -1}-
    \bar Z^{-1}\frac{\partial}{\partial X} Z^{T\ -1})=\\ \\
    =\left(
    \begin{array}{cc}
    0  & 1 \\
    -1 & 0
    \end{array}
    \right)
    (\frac{\partial}{\partial\omega}(-\bar\eta^0\pi^0+\eta^0\bar\pi^0)+
    \frac{\partial}{\partial\eta}(-\bar\eta^1\pi^1+\eta^1\bar\pi^1)+\\ \\
    +\frac{\partial}{\partial\xi}(-\bar\eta^1\pi^0+\eta^0\bar\pi^1)+
    \frac{\partial}{\partial\bar\xi}(\bar\eta^0\pi^1-\eta^1\bar\pi^0))
    \cdot \frac{1}{(det(Z))^2 \sqrt{2}}:=
    \left(
    \begin{array}{cc}
    0  & 1 \\
    -1 & 0
    \end{array}
    \right)\tilde L\spsd
    \end{array}
\end{equation}
    For the metric (\ref{e86}), the norm of this vector will be such
\begin{equation}
\label{e96a}
    \parallel\tilde L\parallel =-\frac{1}{2(det(Y))^2}=-\frac{1}{(V+W)^2}\spsd
\end{equation}
     An isotropic vector k is called a vector of \emph{first type unit extension} \cite[v. 1, p. 36, eq. (1.4.16)(eng)]{Penrose1} in the case when k is sets the point belonging to the cross-section of the isotropic cone by the plane V + W = 1. Then $\parallel \tilde L \parallel =- 1$ and any isotropic vector K collinear with k is defined as
\begin{equation}
\label{e96aa}
    K=(-\parallel\tilde L\parallel)^\frac{1}{2}k\spsd
\end{equation}
    However, when V=-W we will obtain a vector with the infinite first type extension. To learn to distinguish between them, it is necessary to set a cross-section of the cone $K_6$ by the plane T+Z=1 and  to enter a vector $\tilde{\tilde L}$ with the norm
\begin{equation}
\label{e96ab}
    \parallel\tilde{\tilde L}\parallel =-\frac{1}{(T+Z)^2}
\end{equation}
     in the same way. An isotropic vector k is called a vector of \emph{second type unit extension} in the case when k is set the point belonging to the cross-section of the isotropic cone by the plane T + Z = 1 and the first type extension will not be finite. We define \emph{extension} of the vector K as
\begin{enumerate}
     \item \emph{first type extension} if such the extension is finite;
     \item \emph{second type extension} if the first type extension is infinite.
\end{enumerate}
     Note that the vector $\tilde L$ is not coordinate-independent in the space $\mathbb R^6_{(2,4)}$ although it is an invariant of the tangent space to the hyperboloid resulting the cross-section of the cone $K_6$ by the plane V+W=1. Our next task is to find an invariant in the space $\mathbb R^6_{(2,4)}$.
\vspace{1cm}

\subsubsection{The geometric twistor picture in the 6-dimensional space}$ $

     \indent Now there is a possibility to represent \emph{isotropic twistor} in the space $\mathbb R^6_{(2,4)}$ visually. Let's consider a pair of vectors of an equal extension in $\mathbb R^6_{(2,4)}$
\begin{equation}
\label{e59}
     K^\alpha=\eta^\alpha{}_{ab}\ iT^{\left[\right. a}X^{b\left.\right]}\sps
     N^\alpha=\eta^\alpha{}_{ab}\ T^{\left[\right. a}Z^{b\left.\right]}\spsd
\end{equation}
     From Lemma \ref{lemma1} of the second chapter, the conditions
\begin{equation}
\label{e61}
    K^\alpha K_\alpha=0\sps
    N^\alpha K_\alpha=0\sps
    N^\alpha N_\alpha=0
\end{equation}
    will follow. We choose a vector $Y^a$ in such a way to satisfy the conditions
\begin{equation}
\label{e62}
    Y^aY_a=0\sps
    Y^aX_a=0\sps
    Y^aZ_a=0\sps
\end{equation}
\begin{equation}
\label{e66}
    \varepsilon^{abcd}X_aY_bZ_cT_d=X^cZ_cY^dT_d=1\sps
    \varepsilon^{abcd}=24X^{\left[\right.a}Y^bZ^cT^{d\left.\right]}\spsd
\end{equation}
    Thus, we will obtain the vector basis $X^a, Y^a, Z^a, T^a$ (recall that $Y_a=s_{aa '}Y^{a'}$)
\begin{equation}
\label{e63}
    \begin{array}{c}
    Y^aY_a=0\sps
    Y^aX_a=0\sps
    Y^aZ_a=0\sps
    X^aX_a=0\sps
    X^aT_a=0\sps\\
    Z^aZ_a=0\sps
    Z^aT_a=0\sps
    T^aT_a=0\spsd
    \end{array}
\end{equation}
    Whence,
\begin{equation}
\label{e67}
    \varepsilon_{abcd}T^{\left[\right.c}Y^{d\left.\right]}=
    -2X_{\left[\right.a}Z_{b\left.\right]}\spsd
\end{equation}
    Therefore, the vectors
\begin{equation}
\label{e68}
    L^\alpha=\eta^\alpha{}_{ab}(-T^{\left[\right.a}Y^{b\left.\right]}+
                                X^{\left[\right.a}Z^{b\left.\right]})\sps
    M^\alpha=\eta^\alpha{}_{ab}(-i)(T^{\left[\right.a}Y^{b\left.\right]}+
                                X^{\left[\right.a}Z^{b\left.\right]})
\end{equation}
    satisfy to the parities
    $$
    L^\alpha=\bar L^\alpha\sps
    M^\alpha=\bar M^\alpha\sps
    $$
    $$
    L^\alpha K_\alpha=0\sps
    L^\alpha M_\alpha=0\sps
    M^\alpha K_\alpha=0\sps
    L^\alpha N_\alpha=0\sps
    M^\alpha N_\alpha=0
    $$
\begin{equation}
\label{e69}
    L^\alpha L_\alpha=-2\sps
    M^\alpha M_\alpha=-2
\end{equation}
    Now we can construct the threevector
\begin{equation}
\label{e70}
    P^{\alpha\beta\gamma}=6K^{\left[\right.\alpha}
                           N^\beta L^{\gamma\left.\right]}\spsd
\end{equation}
    Knowing $K^\alpha$, we know $T^a$ and $X^a$ up to
\begin{equation}
\label{e71}
    X^a\longmapsto\lambda_1 X^a+\mu_1 T^a\sps
    T^a\longmapsto\nu_1 X^a+\xi_1 T^a\sps
    det\left(
    \begin{array}{cc}
    \lambda_1 & \mu_1 \\
    \nu_1  & \xi_1
    \end{array}
    \right)=1\spsd
\end{equation}
    And if we know $N^\alpha$ then an arbitrariness in a choice of $T^a$ and $Z^a$ will be such
\begin{equation}
\label{e72}
    Z^a\longmapsto\lambda_2 Z^a+\mu_2 T^a\sps
    T^a\longmapsto\nu_2 Z^a+\xi_2  T^a\sps
    det\left(
    \begin{array}{cc}
    \lambda_2 & \mu_2 \\
    \nu_2  & \xi_2
    \end{array}
    \right)=1\spsd
\end{equation}
    Therefor, $\nu_1=\nu_2=0$ and $\xi_2=\xi_1$. For $Y^a$, we obtain
\begin{equation}
\label{e73}
    Y^a  \longmapsto  \alpha X^a+\beta Y^a+\gamma Z^a+\delta T^a\spsd
\end{equation}
    If we require that two such bases (similar (\ref{e63})) are related by a transformation from the group LU(2,2) then we obtain
    $$
    \begin{array}{cclccl}
    X^a & \longmapsto & \tau^{-1} X^a+\mu T^a\ , \ &
    T^a & \longmapsto & \tau T^a\sps\\
    Z^a & \longmapsto & \tau^{-1} Z^a+\chi T^a\ , \ &
    Y^a & \longmapsto & -\bar\chi X^a+\tau Y^a-\bar\mu Z^a+\delta T^a\sps
    \end{array}
    $$
\begin{equation}
\label{e74}
    \bar\chi\mu+\bar\mu\chi+\tau\bar\delta+\bar\tau\delta=0\sps
    \tau\bar\tau=1\spsd
\end{equation}
    Whence,
\begin{equation}
\label{e75}
    \begin{array}{cclcccl}
    X^{\left[\right.a}T^{b\left.\right]}& \longmapsto &
    X^{\left[\right.a}T^{b\left.\right]}& \Leftrightarrow &
    K^\alpha & \longmapsto &  K^\alpha \sps\\
    Z^{\left[\right.a}T^{b\left.\right]}& \longmapsto &
    Z^{\left[\right.a}T^{b\left.\right]}& \Leftrightarrow &
    N^\alpha & \longmapsto &  N^\alpha \sps
    \end{array}
\end{equation}
    $$
    \begin{array}{ccl}
    T^{\left[\right.a}Y^{b\left.\right]} & \longmapsto &
    -\tau\bar\chi T^{\left[\right.a}X^{b\left.\right]}+
    \tau^2 T^{\left[\right.a}Y^{b\left.\right]}-
    \bar\mu\tau T^{\left[\right.a}Z^{b\left.\right]}\sps\\
    X^{\left[\right.a}Z^{b\left.\right]} & \longmapsto &
    \tau^{-2} X^{\left[\right.a}Z^{b\left.\right]}+
    \tau^{-1}\chi X^{\left[\right.a}T^{b\left.\right]}+
    \tau^{-1}\mu T^{\left[\right.a}Z^{b\left.\right]}
    \end{array}
    $$
    Define
\begin{equation}
\label{e76}
    \tau=:e^{i\Theta}
\end{equation}
    then
    $$
    L^\alpha \psps \longmapsto \psps
    L^\alpha \cos(2\Theta) +M^\alpha\sin(2\Theta)-
    $$
    $$
    -i(\bar\chi\tau-\bar\tau\chi)K^\alpha+
    (\mu\bar\tau+\tau\bar\mu)N^\alpha)
    $$
    $$
    M^\alpha \psps \longmapsto \psps
    M^\alpha \cos(2\Theta) -L^\alpha\sin(2\Theta)+
    $$
\begin{equation}
\label{e77}
    +(\bar\chi\tau+\bar\tau\chi)K^\alpha-
    i(\mu\bar\tau-\tau\bar\mu)N^\alpha)\\
\end{equation}
    Therefor, the 3-half-plane, spanned by $K^\alpha, N^\alpha, L^\alpha$, is coordinate-independent in the space $\mathbb R^6_{(2,4)}$. Thus, our design can be presented as follows. The first type extension of the vectors $K^\alpha$ and $N^\alpha $ should be the same. $K^\alpha$ and $N^\alpha$ determine \emph{flagpole}: the set of vectors with:
\begin{enumerate}
    \item the first type extension is equal to the first type extension of the vector $K^\alpha$;
    \item the start coinciding with the beginning of the vector $K^\alpha$.
\end{enumerate}
    $K^\alpha$, $N^\alpha$, $L^\alpha$ determine the 3-half-plane which we call \emph{flag-plane}. Thus, knowing $K^\alpha$ and $N^\alpha$, we know the twistor $T^a$ up to the phase $\Theta$. In turn, in the 2-plane $(L^\alpha, M^\alpha)$, $2\Theta$ is an angle of the rotation of the flag (3-half-plane $P^{\alpha\beta\gamma}$) around the flagpole $(N^\alpha,K^\alpha)$. Therefore, a rotation of the flag on $2\pi$ will lead to the twistors $-T^a$, and only a rotation on $4\pi$ will return our design to the original state. In addition, collinear twistors can be distinguished from each other using the concept \emph{extension of the vector} for the $K^\alpha$ so that under the transformation $T^a\longmapsto rT^a,\ Y^a\longmapsto r^{-1}Y^a\ (r\in R\backslash \{0\})$, the flagpole is multiplied by r and the flag-plane remains unchanged. Finally, it should be noted that the mentioned geometrical structure is uniquely determined by the twistor $T^a$. In the case of the infinite first type extension of the vector $K^\alpha$, we consider the cone $K_4\subset K_6$ on which the vector $N^\alpha$ lays. But non-zero vectors $K^\alpha$ and $N^\alpha$ have the finite second type extension giving the geometric interpretation of a spinor on the isotropic cone $K_4$.

\section{The theorem on two quadrics}
\Abstract{
    \indent In this chapter, we study the common solution
    $$
    \left\{
    \begin{array}{lcl}
    X^a & = & \dot X^a -ir^{ab}\dot Y_b\sps \\
    Y_b & = & \dot Y_b
    \end{array}
    \right.
    $$
    of the bitwistor equation leading to the Rosenfeld null-pair
    $$
    X^A:=(X^a,Y_b)\spsd
    $$
    $\dot X^a,\dot Y_b$ will be some partial solutions of the bitwistor equation (\ref{e53aa}). We are interested in the locus of points defined by the equation
    $$
    X^a=0\spsd
    $$
    It is shown that the solutions of this equation lead to the two quadrics, for which the modified triality principle is just. It is proved that the modified triality principle is the generalization of the Cartan triality principle and the Klein correspondence that allows to realize it explicitly with the help of operators $\eta_A{}^{KL}$ which is the generalization of the connecting Norden operators $\eta_\alpha{}^{ab}$. The connecting operators $\eta_A{}^{KL}$ satisfy the Clifford equation that leads to the Cayley numbers. A proof of the generalized theorem is an example of an application of the 6-dimensional spinor formalism, developed above, which is closely associated with the 4-dimensional spinor formalism \cite{Penrose1}.
}

\subsection{Solutions of the bitwistor equation}$ $
    \indent Let us consider the bundle $A^\mathbb C$ with fibers isomorphic to $\mathbb C^4$ and the base $\mathbb CV^6$ which is an analytic complex space with the quadratic metric. The equation
\begin{equation}
\label{e143}
    \nabla^{a\left(b\right.}X^{c\left.\right)}=0\psps
    (a,b,...=\overline{1,4})
\end{equation}
    is called \emph{bitwistor equation} ($X^c$ are analytic functions). By the above, the bitwistor equation is a conformal invariant. Furthermore, the integrability condition of the equation (\ref{e143}) has the form
\begin{equation}
\label{e154}
    \frac{1}{2}\varepsilon_{akmn}\nabla^{m\left(\right.n}
    \nabla^{d|k|}X^{c\left.\right)}=\frac{5}{6}C_a{}^c{}_l{}^d
    X^l=0
\end{equation}
    (the proof of this fact is given in Appendix (\ref{e2500}) - (\ref{e2503})). We restrict ourselves by the case of the conformal space (see (\ref{e24}))
\begin{equation}
\label{e155}
    C_a{}^c{}_l{}^d=0\spsd
\end{equation}
    This means that the space $\mathbb CV^6 $ is conformal to the space $\mathbb C\mathbb R^6$ which, without loss of generality, we shall consider below. If $X^c$ is a solution of (\ref{e143}) then the spin-tensor $\nabla^{ab}\nabla^{cd}X^r$ is antisymmetric on $rd$, it is also antisymmetric on $br$ in view of that the space is flat and the derivatives are commute. In addition, there is the antisymmetry on the pairs $ab,cd,rc,ra$. This means that $\nabla^{ab}\nabla^{cd}X^r$ is antisymmetric on $abcdr$ and hence it equals to zero. We fix the point O in the space $\mathbb C\mathbb R^6$: O is the origin of coordinates. All other points describe by vectors $r^a$ with the beginning at O then for $r^\alpha\ne 0$, we have
\begin{equation}
\label{e156}
    \nabla_\alpha r^\beta=\delta_\alpha{}^\beta\spsd
\end{equation}
    Therefore, $\nabla^{ab}X^c$ is a constant antisymmetric on $abc$ that follows from (\ref{e154}). We set
\begin{equation}
\label{e157}
    \nabla^{cd}X^a=
    -i\varepsilon^{cdab}Y_b\spsd
\end{equation}
    Integrating this equation, we obtain \emph{common solution}
\begin{equation}
\label{e158}
    \left\{
    \begin{array}{lcl}
    X^a & = & \dot X^a -ir^{ab}\dot Y_b\sps \\
    Y_b & = & \dot Y_b\spsd
    \end{array}
    \right.
\end{equation}
    Here $r^{ab}$ is the bivector from the formula (\ref{e1}), where the factor i is chosen for the convenience (it is clarified under the considering the real case). $\dot X^a$ is a constant vector field whose the value matches the one of the field $X^a$ at the point O. For the space $\mathbb R^6_{(2,4)}$, we have
\begin{equation}
\label{e159}
    Y_k=s_{km'}\bar Y^{m'} \sps \bar Y^{m'}=\overline{Y^m}\spsd
\end{equation}
    In addition, the radius vector $r^\gamma$($=\frac{1}{2}\eta^\gamma{}_{ab}R^{ab}$) satisfies the following relations
\begin{equation}
\label{e160}
    \frac{1}{2}r^{ab}r_{ab}=pf(r)\sps
    r^{ab}r_{bc}=-\frac{1}{2}pf(r)\delta_c{}^a\spsd
\end{equation}

\subsection{Rosenfeld null-pairs}$ $

    \indent Denote by $\mathbb A^{\mathbb C*}$ the spinor 4-dimensional complex vector space. Such the space is dual to the space $\mathbb A^{\mathbb C}$. Then the 8-dimensional complex space $\mathbb T^2$ is formed as the direct sum $\mathbb A^{\mathbb C}\oplus\mathbb A^{\mathbb C*}$. That is, if $X^a$ $(a, b ,...=\overline{1,4})$ are the coordinates of a vector in $A^{\mathbb C}$, and $Y_b$ are the coordinates of a covector in $A^{\mathbb C*}$ then
\begin{equation}
\label{e161}
    X^A:=(X^a,Y_b)\psps (A,B,...=\overline{1,8})
\end{equation}
    are the coordinates of a vector in $\mathbb T^2$. The transformation (\ref{e158}) is a linear one not keeping the structure of the direct sum. We will consider the bivector coordinates $r^{ab}$ as the ones in the complex affine space $\mathbb C \mathbb A_6$. We are interested in a set of points defined by the equation
\begin{equation}
\label{e162}
    X^a=0\psps \Leftrightarrow \psps \dot X^a=ir^{ab}\dot Y_b\spsd
\end{equation}
    This is a system of 4 linear equations with 6 unknowns. To determine its rank, we consider \emph{homogeneous equation}
\begin{equation}
\label{e163}
    r^{ab}\dot Y_b=0,
\end{equation}
    which has nontrivial solutions if and only if the bivector is simple
\begin{equation}
\label{e164}
    \varepsilon_{abfc}r^{ab}r^{cd}=-pf(r)\delta_f{}^d=0\sps
\end{equation}
     and therefore this bivector can be represented as
\begin{equation}
\label{e165}
     r^{ab}_{\mbox{\scriptsize{homogeneous}}}=\dot P^a\dot Q^b-\dot P^b\dot Q^a\sps
\end{equation}
     where $P^a$ and $Q^a$ are defined up to linear combinations of them. From this, it follows that
\begin{equation}
\label{e166}
     \dot P^a\dot Y_a=0\sps
     \dot Q^a\dot Y_a=0\spsd
\end{equation}
     Denote by $X^a$, $S^a$, $Z^a$ those solutions of the equation
\begin{equation}
\label{e167}
     X^aY_a = 0\sps
\end{equation}
     that form a basis. Then our solution (\ref{e165}) takes the form
\begin{equation}
\label{e168}
    r^{ab}_{\mbox{\scriptsize{homogeneous}}}=
    \lambda_1\dot S^{\left[\right.a}\dot X^{b\left.\right]}+
    \lambda_2\dot X^{\left[\right.a}\dot Z^{b\left.\right]}+
    \lambda_3\dot S^{\left[\right.a}\dot Z^{b\left.\right]}
\end{equation}
    and hence determines a 3-dimensional subspace in the bivector space. From here, \emph{common solution}
\begin{equation}
\label{e169}
    r^{ab}=
    r^{ab}_{\mbox{\scriptsize{particular}}}+
    \lambda_1\dot S^{\left[\right.a}\dot X^{b\left.\right]}+
    \lambda_2\dot X^{\left[\right.a}\dot Z^{b\left.\right]}+
    \lambda_3\dot S^{\left[\right.a}\dot Z^{b\left.\right]}
\end{equation}
    of the equation (\ref{e162}) is obtained, where $r^{ab}_{\mbox{\scriptsize{particular}}}$ is an arbitrary bivector being \emph{particular solution} of (\ref{e162}).\\

\subsection{\texorpdfstring{Construction of the quadrics $\mathbb CQ_6$ and $\mathbb C\tilde Q_6$}{Construction of the two quadrics}}$ $

    \indent The space $\mathbb T^2$ will be the complex space in which the scalar square of a vector is determined by the quadratic form
\begin{equation}
\label{e170}
    \varepsilon_{AB}X^AX^B=2X^aY_a
\end{equation}
    in the sense of (\ref{e161}) so that the matrix of the spin-tensor $\varepsilon_{AB}$ has the form
\begin{equation}
\label{e172}
    \parallel \varepsilon_{AB}\parallel=\left(
    \begin{array}{cc}
    0 & \delta_a{}^c \\
    \delta^b{}_d  & 0
    \end{array}
    \right)
\end{equation}
    in the special basis. The form (\ref{e170}) is invariant under the transformation (\ref{e158})
\begin{equation}
\label{e171}
    X^aY_a=(\dot X^a-ir^{ab}\dot Y_b)\dot Y_a=\dot X^a\dot Y_b\spsd
\end{equation}
    For the fixed $r^{ab}$, the equation (\ref{e162}) will define the 4-dimensional subspace in $\mathbb T^2$ which will be the 4-dimensional planar generator of the cone
\begin{equation}
\label{e173}
    \varepsilon_{AB}X^AX^B=0\spsd
\end{equation}
    Thus, in the projective space $\mathbb C \mathbb P_7$, we can consider the quadric $\mathbb CQ_6$ defined by the equation (\ref{e173}). 4 basis points of the generator satisfy the condition
\begin{equation}
\label{e174}
    \varepsilon_{AB}X^A_iX^B_j=0\psps (i,j,...=\overline{1,4})\spsd
\end{equation}
    Put
\begin{equation}
\label{e175}
    X^A_1:=(\dot X^a,\dot Y_b)\sps
    X^A_2:=(\dot Z^a,\dot T_b)\sps
    X^A_3:=(\dot L^a,\dot N_b)\sps
    X^A_4:=(\dot K^a,\dot M_b)\spsd
\end{equation}
    On the basis of the common solution (\ref{e169}), each point of the quadric $\mathbb CQ_6$  can be associated to the 3-dimensional isotropic plane of the space $\mathbb C \mathbb A_6$. The point (t,v,w,x,y,z) of the space $\mathbb C \mathbb A_6$ can be represented by the line  $(\lambda T, \lambda V, \lambda U, \lambda S, \lambda W, \lambda X,$ $ \lambda Y, \lambda Z)$ of the space $\mathbb C\mathbb R^8$ having the metric
\begin{equation}
\label{e176}
    dL^2=dT^2+dV^2+dU^2+dS^2+dW^2+dX^2+dY^2+dZ^2\spsd
\end{equation}
    This line will be a generator of the isotropic cone $\mathbb CK_8$
\begin{equation}
\label{e177}
    T^2+V^2+U^2+S^2+W^2+X^2+Y^2+Z^2=0\spsd
\end{equation}
    The intersection of the 7-plane
\begin{equation}
\label{e1771}
    U-iS=1
\end{equation}
    with the cone $\mathbb CK_8 $ has the induced metric
\begin{equation}
\label{e1772}
    d\tilde L^2=dT^2+dV^2+dW^2+dX^2+dY^2+dZ^2\spsd
\end{equation}
    This space has the form of a paraboloid in $\mathbb CK_8$, and it is identical to the space $\mathbb C\mathbb R^6$
\begin{equation}
\label{e1773}
    U=1+iS=\frac{1}{2}(1-T^2-V^2-W^2-X^2-Y^2-Z^2)\spsd
\end{equation}
    Every generator of this cone (a set of points belonging to $\mathbb CK_8$ with the constant ratio  T:V:U:S:W:X:Y:Z), not lying on the hyperplane $U=iS$, intersects the paraboloid in the single point. Every generator of the cone, lying on the hyperplane $U=iS$, corresponds to the point belonging to the infinity of the space $\mathbb C\mathbb R^6$. Thus, straight lines of $\mathbb C\mathbb R^8$ passing through the origin of $\mathbb C\mathbb R^8$ correspond to points of the projective space $\mathbb C \mathbb P_7$. The stereographic projection of this section on the plane (S=0, U=1) with the pole $N(0,0,\frac{1}{2},\frac{i}{2},0,0,0,0)$ maps the point P(T,V,U,S,W,X,Y,Z) of the hyperboloid to the point p(t,v,1,0,w,x,y,z) of the plane (S=0, U=1)
\begin{equation}
\label{e17731}
    \begin{array}{c}
    \lambda T=t\sps
    \lambda V=v\sps
    \lambda W=w\sps
    \lambda X=x\sps
    \lambda Y=y\sps
    \lambda Z=z\sps\\
    \lambda =\frac{1}{U+iS}\sps
    \lambda U=\frac{1}{2}(1-t^2-v^2-w^2-x^2-y^2-z^2)=-\lambda iS+1\sps
    pf(r)=-\frac{U-iS}{U+iS}\spsd
    \end{array}
\end{equation}
    All generators of the same cone $\mathbb CK_8$ form the quadric $\mathbb C\tilde Q_6 $ in the projective space $\mathbb C \mathbb P_7$
\begin{equation}
\label{e179}
    G_{AB}R^AR^B=0\spsd
\end{equation}

\subsection{\texorpdfstring{Correspondence  $\mathbb CQ_6\longmapsto \mathbb C\tilde Q_6$}{Correspondence: CQ to C'Q}}

\begin{enumerate}
\item
    \indent On the basis of (\ref{e169}),
\begin{equation}
\label{e178}
    r^{ab}=r^{ab}_{\mbox{\scriptsize{particular}}}+r^{ab}_{\mbox{\scriptsize{homogenous}}}=
    r^{ab}_{\mbox{\scriptsize{particulare}}}+
    \lambda_1\dot S^{\left[\right.a}\dot X^{b\left.\right]}+
    \lambda_2\dot X^{\left[\right.a}\dot Z^{b\left.\right]}+
    \lambda_3\dot S^{\left[\right.a}\dot Z^{b\left.\right]}\spsd
\end{equation}
    Therefor, by this equation, the 4-dimensional planar generator of the cone $\mathbb CK_8$ is determined. The equations (\ref{e174}), (\ref{e175}) will define the system
\begin{equation}
\label{e182}
    \left\{
    \begin{array}{lcl}
    ir^{ab}\dot Y_b & = & \dot X^a\sps\\
    ir^{ab}\dot T_b & = & \dot Z^a\sps\\
    ir^{ab}\dot N_b & = & \dot L^a\sps\\
    ir^{ab}\dot M_b & = & \dot K^a
    \end{array}
    \right.
\end{equation}
    with the conditions
\begin{equation}
\label{e183}
    \begin{array}{c}
    \dot X^a\dot Y_a=0\sps
    \dot Z^a\dot T_a=0\sps
    \dot L^a\dot N_a=0\sps
    \dot K^a\dot M_a=0\sps\\
    \dot X^a\dot T_a=-\dot Z^a\dot Y_a\sps
    \dot X^a\dot N_a=-\dot L^a\dot Y_a\sps
    \dot X^a\dot M_a=-\dot K^a\dot Y_a\sps\\
    \dot Z^a\dot N_a=-\dot L^a\dot T_a\sps
    \dot Z^a\dot M_a=-\dot K^a\dot T_a\sps
    \dot K^a\dot N_a=-\dot L^a\dot M_a\spsd
    \end{array}
\end{equation}
    Thus, from the 16 equations with the 6 unknowns $r^{ab}$, only 6 from them  will be significant (the 10 communication conditions (\ref{e183})).  Then the 3-dimensional planar generator $\mathbb C \mathbb P_3$ belonging to the quadric $\mathbb C \mathbb Q_6$ will uniquely define the point of $\mathbb C \mathbb A_6$ and hence the point of the quadric $\mathbb C \tilde Q_6$.
\item
    If from the system (\ref{e182}), we know only the one equation
\begin{equation}
\label{e184}
    ir^{ab}\dot Y_b=\dot X^a
\end{equation}
    with the condition
\begin{equation}
\label{e185}
    \dot X^a\dot Y_a=0
\end{equation}
    then from the 4 equations, only 3 from them will be significant (the 1 communication condition (\ref{e185})). This means that the point of the quadric $\mathbb CQ_6$  will uniquely define the 3-dimensional planar generator $\mathbb C \mathbb P_3$ belonging to the quadric $\mathbb C \tilde Q_6$. This is follows from (\ref{e169}).

\begin{figure}\center
\includegraphics[width=0.7\textwidth]{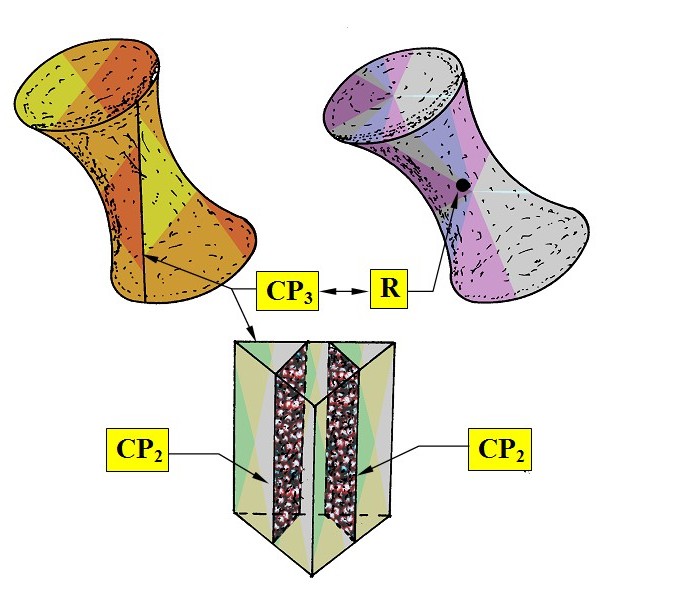}
\caption{Correspondence  $\forall CP_2\subset CP_3 \leftrightarrow R$}
\vspace{1cm}
\includegraphics[width=0.7\textwidth]{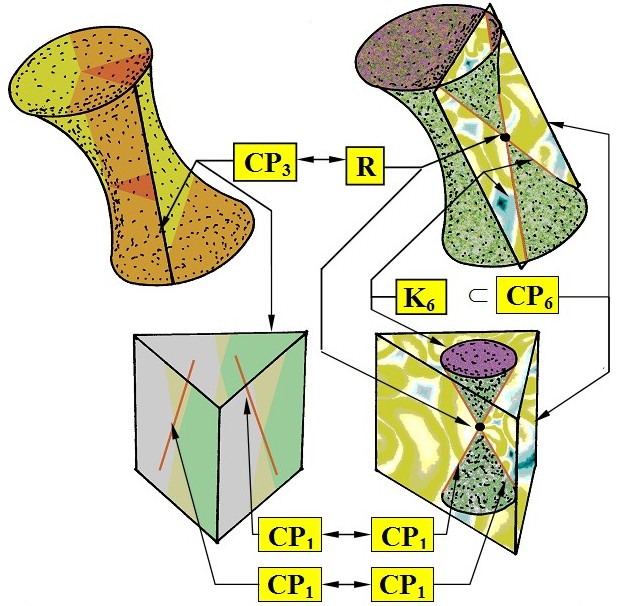}
\caption{Correspondence $CP_3\supset CP_1\leftrightarrow CP_1\subset K_6$}
\end{figure}

\item
    If from the system (\ref{e182}), we know only the two equations
\begin{equation}
    \left\{
    \begin{array}{lcl}
    ir^{ab}Y_b & = & X^a\sps\\
    ir^{ab}T_b & = & Z^a
    \end{array}
    \right.
\end{equation}
    with the condition
\begin{equation}
\label{e187}
    \begin{array}{c}
    \dot X^a\dot Y_a=0\sps
    \dot Z^a\dot T_a=0\sps
    \dot X^a\dot T_a=-\dot Z^a\dot Y_a
    \end{array}
\end{equation}
    then from the 8 equations, only 5 from them will be significant (the 6 unknowns and the 3 communication conditions (\ref{e187})). This means that the rectilinear generator  $\mathbb C\mathbb P_1$ of the quadric $\mathbb CQ_6$ will uniquely define the rectilinear generator $\mathbb C \mathbb P_1$ belonging to the quadric $\mathbb C \tilde Q_6$. In this case, the manifold of generators $\mathbb C \mathbb P_1 (\mathbb CQ_6)$ belonging to the same generator $\mathbb C\mathbb P_3(\mathbb CQ_6)$ defines the beam of generators $\mathbb C\mathbb P_1(\mathbb C\tilde Q_6)$ belonging to the quadric $\mathbb C \tilde Q_6$ (this beam is a cone). The center of the beam is determined by the system (\ref{e182}).
\item
    If from the system (\ref{e182}), we know only three equations
\begin{equation}
\label{e188}
    \left\{
    \begin{array}{lcl}
    ir^{ab}\dot Y_b & = & \dot X^a\sps\\
    ir^{ab}\dot T_b & = & \dot Z^a\sps\\
    ir^{ab}\dot N_b & = & \dot L^a
    \end{array}
    \right.
\end{equation}
    with the condition
\begin{equation}
\label{e189}
    \begin{array}{c}
    \dot X^a\dot Y_a=0\sps
    \dot Z^a\dot T_a=0\sps
    \dot L^a\dot N_a=0\sps\\
    \dot X^a\dot T_a=-\dot Z^a\dot Y_a\sps
    \dot X^a\dot N_a=-\dot L^a\dot Y_a\sps
    \dot Z^a\dot N_a=-\dot L^a\dot T_a
    \end{array}
\end{equation}
    then from the 12 equations, only 6 from them will be significant (the 6 unknowns and the 6 communication conditions (\ref{e189})). This means, that the 2-dimensional generator $\mathbb C\mathbb P_2$ of the quadric $\mathbb CQ_6$ will uniquely define the point of the quadric $\mathbb C \tilde Q_6$. At the same time, the manifold of generators $\mathbb C\mathbb P_2(\mathbb CQ_6)$ belonging to the same generator $\mathbb C\mathbb P_3(\mathbb CQ_6)$ uniquely determines the same point of the quadric $\mathbb C \tilde Q_6$. This point is determined by the system (\ref{e182}).
\end{enumerate}

\subsection{\texorpdfstring{The connection operators $\eta^A{}_{KL}$}{The connection operators for n=8}}$ $
    \indent Based on the foregoing, we consider the rectilinear generator of the quadric $\mathbb CQ_6$ defined by the bivector
\begin{equation}
\label{e190}
    \begin{array}{c}
    \hat R^{AB}=X_1^{\left[\right.A}X_2^{B\left.\right]}=
    \left(
    \begin{array}{cc}
    \dot X^a\dot Z^b-\dot X^b\dot Z^a & \dot X^a\dot T_d-\dot Y_d\dot Z^a \\
    \dot Y_c\dot Z^b-\dot T_c\dot X^b & \dot Y_c\dot T_d-\dot Y_d\dot T_c
    \end{array}
    \right)=\\ \\
    =\left(
    \begin{array}{cc}
    -2r^{\left[\right.a|k|}r^{b\left.\right]r}\dot Y_k\dot T_r &
    2ir^{ar}\dot Y_{\left[\right.r}\dot T_{b\left.\right]} \\
    i\varepsilon^{kbmn}r_{ck}\dot Y_{\left[\right.m}\dot T_{n\left.\right]}+
    \delta_c{}^b\dot X^k\dot T_k &
    2Y_{\left[\right.c}\dot T_{d\left.\right]} \\
    \end{array}
    \right)=\\ \\
    =\dot T_l\dot X^l\left(
    \begin{array}{cc}
    -\frac{1}{2}i\delta^a{}_kr^\gamma r_\gamma & r^{ar} \\
    r_{ck} & -i\delta_c{}^r
    \end{array}
     \right)\cdot\frac{1}{\dot T_l\dot X^l}\left(
    \begin{array}{cc}
    i\varepsilon^{kbmn}\dot Y_m\dot T_n  &
    0 \\
    -\delta_r{}^b r^{mn}\dot Y_m\dot T_n  &
    2iY_{\left[\right.r}T_{d\left.\right]}
    \end{array}
    \right):=R^A{}_KP^{KB}\spsd
    \end{array}
\end{equation}
    Put
\begin{equation}
\label{e191}
    R^{AB}:=\varepsilon^{BC}R^A{}_C=\dot T_l\dot X^l\left(
    \begin{array}{cc}
    r^{an} & -\frac{1}{2}i\delta^a{}_kr^\gamma r_\gamma \\
    -i\delta_c{}^n & r_{ck}
    \end{array}
    \right)\sps
\end{equation}
\begin{equation}
\label{e191a}
    R_{AB}:=\varepsilon_{AC}R^C{}_B=\dot T_l\dot X^l\left(
    \begin{array}{cc}
    r_{ck} & -i\delta_c{}^n \\
    -\frac{1}{2}i\delta^a{}_kr^\gamma r_\gamma & r^{an}
    \end{array}
    \right)\spsd
\end{equation}
    At the same time, the equation
\begin{equation}
\label{e192}
    R_A{}^C\hat R^{AB}=0
\end{equation}
    will be true that means that any spin-tensor $\hat R^{AB}$, representing the generator $\mathbb C \mathbb P_1 (\mathbb CQ_6)$, will contain the same spin-tensor $R^A{}_K$ in its expansion, wherein the second spin-tensor $P^{KB}$ of the decomposition will be responsible for the position of $\mathbb C \mathbb P_1$ in $\mathbb C \mathbb P_3$. Therefore, there is a reason to assign the bispinor $R^{AB}$ to the point of the quadric $\mathbb C \tilde Q_6$. This point is uniquely determined. In the transition to the space $\mathbb C\mathbb R^8$, we put
\begin{equation}
\label{e193a}
    \begin{array}{ccccccccc}
    r^{12} & = & \frac{1}{\sqrt 2}(v+iw)\sps &
    r^{13} & = & \frac{1}{\sqrt 2}(x+iy)\sps &
    r^{14} & = & \frac{i}{\sqrt 2}(t+iz)\sps \\
    r^{23} & = & \frac{i}{\sqrt 2}(iz-t)\sps &
    r^{24} & = & \frac{1}{\sqrt 2}(-x+iy)\sps&
    r^{34} & = & \frac{1}{\sqrt 2}(v-iw)\spsd
    \end{array}
\end{equation}
    Then we define \emph{homogeneous coordinates} of $\mathbb C\mathbb R^8$ as follows
\begin{equation}
\label{e193}
    \lambda=\left\{\begin{array}{ccccccccccccccl}
    R^{12} &:&  R^{13}  &:&  R^{14}  &:&  R^{23} &:& R^{24} &:& R^{34} &:&
    R^{15} &:& R^{51}\sps \\
    r^{12} &:& r^{13} &:& r^{14} &:& r^{23} &:& r^{24} &:& r^{34} &:&
    -\frac{1}{2}ir^\gamma r_\gamma &:& -i\sps
    \end{array}
    \right.
\end{equation}
\begin{equation}
\label{e194}
    \begin{array}{ccccccccc}
    R^{12} & = & \frac{1}{\sqrt 2}(V+iW)\sps &
    R^{13} & = & \frac{1}{\sqrt 2}(iY+X)\sps &
    R^{14} & = & \frac{1}{\sqrt 2}(iT-Z)\sps \\
    R^{23} & = & \frac{1}{\sqrt 2}(-Z-iT)\sps &
    R^{24} & = & \frac{1}{\sqrt 2}(-X+iY)\sps &
    R^{34} & = & \frac{1}{\sqrt 2}(V-iW)\sps \\
    R^{51} & = & S-iU\sps &
    R^{15} & = & \frac{1}{2}(iU+S)\sps &
    \end{array}
\end{equation}
    $$
    \begin{array}{cccccccc}
    R^{12} & = & -R^{21} & = & R^{78} & = & -R^{87}\sps \\
    R^{13} & = & -R^{31} & = & R^{86} & = & -R^{68}\sps \\
    R^{14} & = & -R^{41} & = & R^{67} & = & -R^{76}\sps \\
    R^{23} & = & -R^{32} & = & R^{58} & = & -R^{85}\sps \\
    R^{24} & = & -R^{42} & = & R^{75} & = & -R^{57}\sps \\
    R^{34} & = & -R^{43} & = & R^{56} & = & -R^{65}\sps \\
    R^{15} & = &  R^{26} & = & R^{37} & = &  R^{48}\sps \\
    R^{51} & = &  R^{62} & = & R^{73} & = &  R^{84}\sps
    \end{array}
    $$
\begin{equation}
\label{e195}
    \begin{array}{c}
    R^{AB}R_{AB}=8(R^{12}R^{34}-R^{13}R^{24}+R^{14}R^{23}+
    R^{15}R^{51})=\\ \\
    =4(T^2+V^2+U^2+S^2+W^2+X^2+Y^2+Z^2):=4nf(R)\\ \\
    R^{AB}R_{CB}=\frac{1}{2}nf(R)\delta_C{}^A
    \ \ \Rightarrow \ \ R^{AB}R_{AB}=4nf(R)\spsd
    \end{array}
\end{equation}
    For $(R_i)^\Lambda$, to define a generator of the quadric $\mathbb C\tilde Q_6$, it is necessary and sufficient to have the condition
\begin{equation}
\label{e197}
    G_{AB}R^A_iR^B_j=0\spsd
\end{equation}
    We define some connection operators $\eta_A{}^{BC}$ so that
\begin{equation}
\label{e198}
     R_A=\frac{1}{4}\eta_A{}^{BC}R_{BC}\sps
     R^A=\frac{1}{4}\eta^A{}_{BC}R^{BC}\spsd
\end{equation}
     Then these operators satisfy the reduced Clifford equation
\begin{equation}
\label{e199}
     G_{AB}\delta_K{}^L=\eta_{AK}{}^R\eta_B{}^L{}_R+
     \eta_{BK}{}^R\eta_A{}^L{}_R\spsd
\end{equation}
     Therefore, we can define the operators $\gamma_A$ as
\begin{equation}
\label{e2o}
     \gamma_A:=\sqrt{2}\left(
     \begin{array}{cc}
     0 & \sigma_A \\
     \eta_A & 0
     \end{array}
     \right)\sps
     \eta_A:=\eta_A{}^{KR}\sps
     \sigma_A:=(\eta_A)^T{}_{RL}\spsd
\end{equation}
    Then $\gamma_A$ will satisfy the Clifford equation
\begin{equation}
\label{e3o}
     \gamma_A\gamma_B+\gamma_B\gamma_A=2G_{AB} I\spsd
\end{equation}
    At the same time, the lowering and raising of single indices is done by using the metric spin-tensor $\varepsilon_{AB}$, determined above. Define
\begin{equation}
\label{e4o}
     \varepsilon_{PQRT}:=\eta^A{}_{PQ}\eta^B{}_{RT}G_{AB}\sps
     \varepsilon_{PQRT}=\varepsilon_{RTPQ}
\end{equation}
     that will give another metric spin-tensor $\varepsilon_{PQRT}$ with which the help we can raise and lower a pair of indices. Indeed, if the equation (\ref{e199}) contracts with $\delta_L{}^K$ then we obtain
\begin{equation}
\label{e5o}
     G_{AB}=\frac{1}{4}\eta_A{}^{PQ}\eta_{BPQ}\spsd
\end{equation}
    Contract (\ref{e4o}) with $\eta_C{}^{PQ}$ and obtain
\begin{equation}
\label{e6o}
     \eta_{CRT}=\frac{1}{4}\eta_C{}^{PQ}\varepsilon_{PQRT}\spsd
\end{equation}
    Now, the identity (\ref{e199}) can be rewritten as (contracting with  $\eta^A{}_{ST}\eta^B{}_{PQ}$)
\begin{equation}
\label{e7o}
     \varepsilon_{STPQ}\delta_K{}^L=
     \varepsilon_{STK}{}^R\varepsilon_{PQ}{}^L{}_R+
     \varepsilon_{PQK}{}^R\varepsilon_{ST}{}^L{}_R\spsd
\end{equation}
     Contracting (\ref{e7o}) with $\delta_L{}^K$, we obtain
\begin{equation}
\label{e8o}
     \varepsilon_{STPQ}=
     \frac{1}{4}\varepsilon_{ST}{}^{KR}\varepsilon_{PQKR}\spsd
\end{equation}
     The result of the applying for the two metric spin-tensors should be the same
\begin{equation}
\label{e9o}
     \varepsilon_{PQ}=\frac{1}{4}\varepsilon_{PQRT}\varepsilon^{RT}\spsd
\end{equation}
     Thus, in the presence of the 3 metric tensor $G_{AB},\varepsilon_{PQRT},\varepsilon_{RT}$, we obtain the equation (\ref{e199}) from the Clifford equation (\ref{e3o}). Then $\eta_A{}^{BC}$ will be the generators of the corresponding Clifford algebra. Now, we will lower the index L in (\ref{e7o}) and contract this equation with the $\varepsilon^{ST}$. Then we will obtain
\begin{equation}
\label{e10o}
     \varepsilon_{PQ(KL)}=\frac{1}{2}\varepsilon_{PQ}\varepsilon_{KL}\sps
     \varepsilon_{[PQ](KL)}=0
\end{equation}
     that will lead to the identity
\begin{equation}
\label{e11o}
    \eta_A{}^{(MN)}=\frac{1}{8}\eta_A{}^{KL}\varepsilon_{KL}\varepsilon^{MN}\spsd
\end{equation}
    If the matrixes of the tensors $g_{\Lambda \Psi}$ and  $\varepsilon_{KL}$ have the form
\begin{equation}
\label{e12o}
    \begin{array}{c}
    \parallel G_{AB} \parallel=\left(
    \begin{array}{cccccccc}
    1 & 0 & 0 & 0 & 0 & 0 & 0 & 0\\
    0 & 1 & 0 & 0 & 0 & 0 & 0 & 0\\
    0 & 0 & 1 & 0 & 0 & 0 & 0 & 0\\
    0 & 0 & 0 & 1 & 0 & 0 & 0 & 0\\
    0 & 0 & 0 & 0 & 1 & 0 & 0 & 0\\
    0 & 0 & 0 & 0 & 0 & 1 & 0 & 0\\
    0 & 0 & 0 & 0 & 0 & 0 & 1 & 0\\
    0 & 0 & 0 & 0 & 0 & 0 & 0 & 1\\
    \end{array}
    \right)\sps
    \parallel\varepsilon_{KL}\parallel=\left(
    \begin{array}{cccccccc}
    0 & 0 & 0 & 0 & 1 & 0 & 0 & 0\\
    0 & 0 & 0 & 0 & 0 & 1 & 0 & 0\\
    0 & 0 & 0 & 0 & 0 & 0 & 1 & 0\\
    0 & 0 & 0 & 0 & 0 & 0 & 0 & 1\\
    1 & 0 & 0 & 0 & 0 & 0 & 0 & 0\\
    0 & 1 & 0 & 0 & 0 & 0 & 0 & 0\\
    0 & 0 & 1 & 0 & 0 & 0 & 0 & 0\\
    0 & 0 & 0 & 1 & 0 & 0 & 0 & 0\\
    \end{array}
    \right)\sps
    \end{array}
\end{equation}
    then the significant coordinates of the operators $\eta^\Lambda{}_{KL}$ in some basis would be such
\begin{equation}
\label{e13o}
    \begin{array}{llll}
    \eta^2{}_{12}=\frac{1}{\sqrt{2}}\sps  &
    \eta^2{}_{34}=\frac{1}{\sqrt{2}}\sps  &
    \eta^5{}_{12}=-\frac{i}{\sqrt{2}}\sps &
    \eta^5{}_{34}=\frac{i}{\sqrt{2}}\sps  \\
    \eta^2{}_{78}=\frac{1}{\sqrt{2}}\sps  &
    \eta^2{}_{56}=\frac{1}{\sqrt{2}}\sps  &
    \eta^5{}_{78}=-\frac{i}{\sqrt{2}}\sps &
    \eta^5{}_{56}=\frac{i}{\sqrt{2}}\sps  \\
    \eta^1{}_{14}=-\frac{i}{\sqrt{2}}\sps &
    \eta^1{}_{23}=\frac{i}{\sqrt{2}}\sps  &
    \eta^8{}_{14}=-\frac{1}{\sqrt{2}}\sps &
    \eta^8{}_{23}=-\frac{1}{\sqrt{2}}\sps \\
    \eta^1{}_{67}=-\frac{i}{\sqrt{2}}\sps &
    \eta^1{}_{58}=\frac{i}{\sqrt{2}}\sps  &
    \eta^8{}_{67}=-\frac{1}{\sqrt{2}}\sps &
    \eta^8{}_{58}=-\frac{1}{\sqrt{2}}\sps\\
    \eta^7{}_{13}=-\frac{i}{\sqrt{2}}\sps &
    \eta^7{}_{24}=-\frac{i}{\sqrt{2}}\sps &
    \eta^6{}_{13}=\frac{1}{\sqrt{2}}\sps  &
    \eta^6{}_{24}=-\frac{1}{\sqrt{2}}\sps \\
    \eta^7{}_{68}=\frac{i}{\sqrt{2}}\sps  &
    \eta^7{}_{57}=\frac{i}{\sqrt{2}}\sps  &
    \eta^6{}_{68}=-\frac{1}{\sqrt{2}}\sps &
    \eta^6{}_{57}= \frac{1}{\sqrt{2}}\sps \\
    \eta^4{}_{15}=\frac{1}{\sqrt{2}}\sps  &
    \eta^4{}_{51}=\frac{1}{\sqrt{2}}\sps  &
    \eta^3{}_{15}=-\frac{i}{\sqrt{2}}\sps &
    \eta^3{}_{51}=\frac{i}{\sqrt{2}}\sps  \\
    \eta^4{}_{26}=\frac{1}{\sqrt{2}}\sps  &
    \eta^4{}_{62}=\frac{1}{\sqrt{2}}\sps  &
    \eta^3{}_{26}=-\frac{i}{\sqrt{2}}\sps &
    \eta^3{}_{62}=\frac{i}{\sqrt{2}}\sps  \\
    \eta^4{}_{37}=\frac{1}{\sqrt{2}}\sps  &
    \eta^4{}_{73}=\frac{1}{\sqrt{2}}\sps  &
    \eta^3{}_{37}=-\frac{i}{\sqrt{2}}\sps &
    \eta^3{}_{73}=\frac{i}{\sqrt{2}}\sps  \\
    \eta^4{}_{48}=\frac{1}{\sqrt{2}}\sps  &
    \eta^4{}_{84}=\frac{1}{\sqrt{2}}\sps  &
    \eta^3{}_{48}=-\frac{i}{\sqrt{2}}\sps &
    \eta^3{}_{84}=\frac{i}{\sqrt{2}}
    \end{array}
\end{equation}
    so that in an abbreviated form, we can rewrite \ref{e13o} as
\begin{equation}
\label{e14o}
    \eta_A{}^{MN}=\left(
    \begin{array}{cc}
    \eta_\alpha{}^{ab} & \lambda\delta^a{}_d \\
    \mu\delta_c{}^b & \eta_{\alpha cd}
    \end{array}
    \right)\sps
    \eta_{\alpha cd}=\frac{1}{2}\varepsilon_{abcd}\eta_\alpha{}^{ab}\sps
\end{equation}
    where $\varepsilon_{abcd}$ is the quadrivector (\ref{e2}). In fact, here we used the same basis as in the formula (\ref{e194}). \\
    \indent Next, we consider the bivector $\hat R^{AB}$ of the form (\ref{e190}) such that its vectors $X_1{}^A,X_2{}^A$, defined by the formula (\ref{e175}), satisfy the system (\ref{e182}). By this, the identities
\begin{equation}
\label{e15o}
    \begin{array}{c}
    ir^{ab}\dot Y_b=\dot X^a\sps
    ir_{ab}\dot Z^b=pf(r)\dot T_a\sps\\[2ex]
    \frac{1}{2}ir_{cd}\varepsilon^{abcd}\dot Y_b\varepsilon_{aklm}=
    \dot X^a\varepsilon_{aklm}\sps\\[2ex]
    3ir_{\left[\right. kl}\dot Y_{m\left.\right]}=
    \dot X^a\varepsilon_{aklm}\sps\\[2ex]
    ir_{kl}\dot Y_m+2ir_{m\left[\right. k}\dot Y_{l\left.\right]}=
    \dot X^a\varepsilon_{aklm}
    \end{array}
\end{equation}
    is determined. We contract the last identity with $\dot Z^m$ and obtain
\begin{equation}
\label{e16o}
    ir_{kl}\dot Y_m\dot Z^m=\dot X^a\dot Z^a\varepsilon_{klmn}+
    pf(r)\cdot 2\dot T_{\left[\right. k}\dot Y_{l\left.\right]}\spsd
\end{equation}
    Thus, the bispinor $\hat R^{AB}$ determines the rectilinear generator $\mathbb C\mathbb P_1 (\mathbb CQ_6)$ belonging to the planar generator $\mathbb C\mathbb P_3 (\mathbb CQ_6)$ which determines the point of the quadric $\mathbb C\tilde Q_6$. $R^{AB}$ will be the coordinates of this point. We define the spin-tensor $\hat{\hat R}^{AB}$
\begin{equation}
\label{e17o}
    \hat{\hat R}^{AB}:=\hat R^{AK} \hat P_K{}^B\sps
    \hat P_K{}^B:=\left(
    \begin{array}{cc}
    2\delta_m{}^k & 0\\
    0 & -2pf(r)i \delta^n{}_r
    \end{array}
    \right)\spsd
\end{equation}
    The spin-tensor $\hat{\hat R}^{AB}$ will continue to represent the rectilinear generator $\mathbb C\mathbb P_1 (\mathbb CQ_6)$ belonging to the planar generator $\mathbb C\mathbb P_3 (\mathbb CQ_6)$. We apply the operator $\eta^\Lambda{}_{KL}$ to the spin-tensor $\hat{\hat R}^{AB}$ and obtain
\begin{equation}
\label{e18o}
    \eta^A{}_{KL}\hat{\hat R}^{KL}=
    \eta^A{}_{KL}\dot Y_m\dot Z^m
    \left(
    \begin{array}{cc}
    r^{an} & -\frac{1}{2}i\delta^a{}_kr^\gamma r_\gamma \\
    -i\delta_c{}^n & r_{ck}
    \end{array}
    \right)=
    \eta^A{}_{KL}R^{KL}\spsd
\end{equation}
    Thus, the operators $\eta^\Lambda{}_{KL}$ take each rectilinear generator $\mathbb C \mathbb P_1 (\mathbb CQ_6)$, belonging to the planar generator $\mathbb C \mathbb P_3 (\mathbb CQ_6)$, to the same generator $\mathbb C \mathbb P_3 (\mathbb CQ_6)$, and this determines the point of the quadric $\mathbb C \tilde Q_6$. In homogeneous coordinates, the spin-tensor $R^{KL}$ determines the coordinates of the point R of the space $\mathbb C\mathbb R^8$.\\
    \indent We contract the identity (\ref{e199}) with the $\delta_L{}^B$
\begin{equation}
\label{e19o}
    G_{AK}=S_A{}^M\varepsilon_{KM}\sps
    S_A{}^M:=\eta_A{}^{MR}\eta_L{}^L{}_R+\eta_A{}^L{}_R\eta_L{}^{MR}\spsd
\end{equation}
    Since $G_{AB},\varepsilon_{KL}$ have the form (\ref{e12o}) then $\parallel S_A{}^M\parallel $ has the form
\begin{equation}
\label{e20o}
    \parallel S_A{}^M \parallel=\left(
    \begin{array}{cccccccc}
    0 & 0 & 0 & 0 & 1 & 0 & 0 & 0\\
    0 & 0 & 0 & 0 & 0 & 1 & 0 & 0\\
    0 & 0 & 0 & 0 & 0 & 0 & 1 & 0\\
    0 & 0 & 0 & 0 & 0 & 0 & 0 & 1\\
    1 & 0 & 0 & 0 & 0 & 0 & 0 & 0\\
    0 & 1 & 0 & 0 & 0 & 0 & 0 & 0\\
    0 & 0 & 1 & 0 & 0 & 0 & 0 & 0\\
    0 & 0 & 0 & 1 & 0 & 0 & 0 & 0\\
    \end{array}
    \right)\spsd
\end{equation}
    Therefore, $S_A{}^M$ is an involution tensor, and the quadrics are \emph{B-cylinders}. \\
    \indent Let us find out what is the number of a family which comprises this generator $\mathbb C\mathbb P_3 (\mathbb CQ_6)$. For this purpose, we consider the conditions
\begin{equation}
\label{e21o}
    \begin{array}{c}
    \varepsilon_{AB}X_i{}^AX_j{}^B=0\sps\\[2ex]
    X^{ABCD}:=\varepsilon^{ijkl}X_i{}^AX_j{}^BX_k{}^CX_l{}^D\sps
    \end{array}
\end{equation}
    where $\varepsilon^{IJKL}$ is a 4-vector antisymmetrical in all indices. Also we consider the 8-vector $e_{ABCDKLMN}$ antisymmetrical in all indices too. If in the condition
\begin{equation}
\label{e22o}
    \frac{1}{24}e_{ABCDKLMN}X^{ABCD}=\rho\varepsilon_{KR}\varepsilon_{LT}\varepsilon_{MU}
    \varepsilon_{NV}X^{RTUV}\sps \rho^2=1\sps
\end{equation}
    $\rho=1$, then we say that the planar generator $\mathbb C\mathbb P_3(\mathbb CQ_6)$ belongs to family I, and if $\rho=-1$, then the planar generator belongs to family II. In our case
\begin{equation}
\label{e23o}
    X_i{}^A=(\dot X_i{}^a,\dot Y_{ib})
\end{equation}
    and then
\begin{equation}
\label{e24o}
    \varepsilon^{ijkl}X_i{}^1X_j{}^2X_k{}^3X_l{}^4=\rho
    \varepsilon^{ijkl}X_i{}^1X_j{}^2X_k{}^3X_l{}^4\spsd
\end{equation}
    The spin-tensor $\varepsilon_{KL}$ has the form (\ref{e12o}). Whence, $\rho=1$. This means that our generators should belong to family I.\\
    \indent In addition, there is the tensor $\tilde S_K{}^L$
\begin{equation}
\label{e25o}
    \parallel \tilde S_K{}^L \parallel=\frac{1}{\sqrt{2}}\left(
    \begin{array}{cccccccc}
    i & 0 & 0 & 0 &-i & 0 & 0 & 0\\
    0 & i & 0 & 0 & 0 &-i & 0 & 0\\
    0 & 0 & i & 0 & 0 & 0 &-i & 0\\
    0 & 0 & 0 & i & 0 & 0 & 0 &-i\\
    1 & 0 & 0 & 0 & 1 & 0 & 0 & 0\\
    0 & 1 & 0 & 0 & 0 & 1 & 0 & 0\\
    0 & 0 & 1 & 0 & 0 & 0 & 1 & 0\\
    0 & 0 & 0 & 1 & 0 & 0 & 0 & 1\\
    \end{array}
    \right)\sps
     det \parallel\tilde S_K{}^L\parallel=1
\end{equation}
    such that
\begin{equation}
\label{e26o}
    \varepsilon_{AB}\tilde S_K{}^A\tilde S_L{}^B=G_{KL}\spsd
\end{equation}
    Therefor, we will obtain
\begin{equation}
\label{e27o}
    X_i{}^A=\tilde S_K{}^AR_i{}^K\sps
\end{equation}
    and (\ref{e21o}) can be rewritten as
\begin{equation}
\label{e28o}
    \begin{array}{c}
    G_{AB}R_i{}^AR_i{}^K=0\sps\\[2ex]
    X^{ABCD}=\tilde S_K{}^A\tilde S_L{}^B\tilde S_N{}^C\tilde S_M{}^DR^{KLNM}\spsd
    \end{array}
\end{equation}
    Since $\rho = 1$ then we have the identities
\begin{equation}
\label{e29o}
    \begin{array}{c}
    \frac{1}{24}e_{ABCDKLMN}X^{ABCD}=\varepsilon_{KR}\varepsilon_{LT}\varepsilon_{MU}
    \varepsilon_{NV}X^{RTUV}\sps\\[2ex]
    \frac{1}{24}e_{ABCDKLMN}\tilde S_U{}^A \tilde S_S{}^B\tilde S_T{}^C\tilde S_V{}^D
    R^{USTV}=\\
    =\varepsilon_{KP}\varepsilon_{LQ}\varepsilon_{MX}
    \varepsilon_{NY}\tilde S_G{}^P \tilde S_H{}^Q\tilde S_Z{}^X\tilde S_W{}^Y
    R^{GHZW}\\[2ex]
    \frac{1}{24}e_{ABCDKLMN}\tilde S_U{}^A \tilde S_S{}^B\tilde S_T{}^C\tilde S_V{}^D
    \tilde S_P{}^K \tilde S_Q{}^L\tilde S_X{}^M\tilde S_Y{}^NR^{USTV}=\\
    =G_{PG}G_{QM}G_{XZ}G_{YW}R^{GHZW}\sps\\[2ex]
    \frac{1}{24}det \parallel S_P{}^A \parallel e_{PQXYUSTV}R^{USTV}=G_{PG}G_{QM}G_{XZ}G_{YW}R^{GHZW}\sps\\[2ex]
    \frac{1}{24}e_{PQXYUSTV}R^{USTV}=G_{PG}G_{QM}G_{XZ}G_{YW}R^{GHZW}\spsd
    \end{array}
\end{equation}
    This shows that the generator $\mathbb C\mathbb P_3 (\mathbb C\tilde Q_6)$ must belong to family I. In order to achieve the same result for the generator of family II we should choose the spin-tensor
\begin{equation}
\label{e30o}
    \tilde\varepsilon_{KL}:=\sqrt{i}\varepsilon_{KL}
\end{equation}
    as the metric spin-tensor by means of which single indices lower and raise.

\subsection{\texorpdfstring{Correspondence $\mathbb C\tilde Q_6\longmapsto \mathbb CQ_6$}{Correspondence: C'Q to CQ}}$ $
    \indent Applying the operators $\eta_\Lambda{}^{KL}$ to (\ref{e197}), we obtain
\begin{equation}
\label{e199a}
      \begin{array}{c}
      R_i^{AB}R_{j\ AB}=0\psps\Leftrightarrow\psps
      ((R_i){}^{AB}-(R_j){}^{AB})((R_k){}_{AB}-(R_l){}_{AB})=0
      \Leftrightarrow\\[2ex]
      ((r_i){}^{ab}-(r_j){}^{ab})((r_k){}_{ab}-(r_l){}_{ab})=0\spsd
      \end{array}
\end{equation}
    Here as usual, i, j is the number of basis points. This defines the system
\begin{equation}
\label{e200}
    \left\{
    \begin{array}{lcl}
    i(r_1){}^{ab}\dot Y_b & = & \dot X^a\sps\\
    i(r_2){}^{ab}\dot Y_b & = & \dot X^a\sps\\
    i(r_3){}^{ab}\dot Y_b & = & \dot X^a\sps\\
    i(r_4){}^{ab}\dot Y_b & = & \dot X^a\sps
    \end{array}
    \right.
    \psps\Leftrightarrow\psps
    \left\{
    \begin{array}{lcl}
    i((r_1){}^{ab}-(r_2){}^{ab})\dot Y_b & = & 0\sps\\
    i((r_1){}^{ab}-(r_3){}^{ab})\dot Y_b & = & 0\sps\\
    i((r_3){}^{ab}-(r_4){}^{ab})\dot Y_b & = & 0\sps\\
    i(r_1){}^{ab}\dot Y_b & = & \dot X^a\spsd
    \end{array}
    \right.
\end{equation}
    Next we consider only the right system. It is constructed as follows. Always there is a covector $Y_a$ which resets 3 different simple bivectors. This statement is reduced to the existence of a covector orthogonal to three linearly independent vectors since any simple bivector is decomposed by the formula $r^{ab}=2P^{\left[\right.a} Q^{b\left.\right]}$. From the fourth equation, the spinor $X^a$ is determined. Therefore, such the system is always defined. On the other hand, on the basis of the fact that all $r_i{}^{ab}$ have the form (\ref{e169}) (for the fixed $\lambda_1, \lambda_2, \lambda_3$), the equality
\begin{equation}
\label{e201}
     ((r_i){}^{ab}-(r_j){}^{ab})((r_k){}_{ab}-(r_l){}_{ab})=0
\end{equation}
    is executed.

\begin{enumerate}
\item
    So, let us know the last equation of the system (\ref{e200})
\begin{equation}
\label{e202}
     i{r_1}^{ab}\dot Y_b  =  \dot X^a
\end{equation}
     then we have the 4 equations, all of which will be significant. Since, we have the eight unknowns $(X^a,Y_b)$ for fixed $(r_1){}^{ab}$ then the point of the quadric $\mathbb C \tilde Q_6$ uniquely defines the 3-dimensional planar generator $\mathbb C\mathbb P_3 (\mathbb CQ_6)$.
\item
     If we know all the equations of the system (\ref{e200}) with the conditions
\begin{equation}
\label{e204}
     \begin{array}{c}
     ((r_1){}^{ab}-(r_2){}^{ab})((r_1){}_{ab}-(r_2){}_{ab})=0\sps
     ((r_1){}^{ab}-(r_3){}^{ab})((r_1){}_{ab}-(r_3){}_{ab})=0\sps \\
     ((r_3){}^{ab}-(r_4){}^{ab})((r_3){}_{ab}-(r_4){}_{ab})=0\sps \\
      (r_1){}^{ab}(r_2){}_{ab}=0\sps
      (r_1){}^{ab}(r_3){}_{ab}=0\sps
      (r_1){}^{ab}(r_4){}_{ab}=0\sps \\
      (r_2){}^{ab}(r_3){}_{ab}=0\sps
      (r_2){}^{ab}(r_4){}_{ab}=0\sps
      (r_3){}^{ab}(r_4){}_{ab}=0
     \end{array}
\end{equation}
     then from the 16 equations, only 7 from them will be significant (the 8 unknowns and the 9 communication conditions (\ref{e204})). Therefore, the generator $\mathbb C \mathbb P_3 (\mathbb C \tilde Q_6)$ will uniquely define the point of the quadric $\mathbb CQ_6$.
\item
     If we know the three equations of system (\ref{e200})
\begin{equation}
\label{e205}
    \left\{
    \begin{array}{lcl}
    i((r_1){}^{ab}-(r_2){}^{ab})\dot Y_b & = & 0\sps\\
    i((r_1){}^{ab}-(r_3){}^{ab})\dot Y_b & = & 0\sps\\
    i(r_1){}^{ab}\dot Y_b & = & \dot X^a
    \end{array}
    \right.
\end{equation}
    with the conditions
\begin{equation}
\label{e206}
     \begin{array}{c}
     ((r_1){}^{ab}-(r_2){}^{ab})((r_1){}_{ab}-(r_2){}_{ab})=0\sps
     ((r_1){}^{ab}-(r_3){}^{ab})((r_1){}_{ab}-(r_3){}_{ab})=0\sps\\
      (r_1){}^{ab}(r_2){}_{ab}=0\sps
      (r_1){}^{ab}(r_3){}_{ab}=0\sps
      (r_2){}^{ab}(r_3){}_{ab}=0
     \end{array}
\end{equation}
     then from the 12 equations, only 7 from them will be significant (the 8 unknowns and the 5 communication conditions (\ref{e206})). This means that the generator $\mathbb C \mathbb P_2 (\mathbb C \tilde Q_6)$ will uniquely define the point of the quadric $\mathbb CQ_6$. In this case, the manifold of generators $\mathbb C \mathbb P_2 (\mathbb C \tilde Q_6)$, belonging to the generator $\mathbb C \mathbb P_3 (\mathbb C \tilde Q_6)$, uniquely defines the point of the quadric $\mathbb CQ_6$.
\item
     If we know the two equations of the system (\ref{e200})
\begin{equation}
\label{e207}
    \left\{
    \begin{array}{lcl}
    i((r_1){}^{ab}-(r_2){}^{ab})\dot Y_b & = & 0\sps\\
    i(r_1){}^{ab}\dot Y_b & = & \dot X^a\\
    \end{array}
    \right.
\end{equation}
     with the conditions
\begin{equation}
\label{e208}
     ((r_1){}^{ab}-(r_2^{ab})((r_1){}_{ab}-(r_2){}_{ab})=0\sps
      (r_1){}^{ab}(r_2){}_{ab}=0
\end{equation}
     then from the 8 equations, only 6 from them will be significant (the 8 unknowns and the 2 communication conditions \ref{e208})). This means that the rectilinear generator  $\mathbb C\mathbb P_1$ of the quadric $\mathbb C\tilde Q_6$ will uniquely define the rectilinear generator $\mathbb C \mathbb P_1$ belonging to the quadric $\mathbb CQ_6$. In this case, the manifold of generators $\mathbb C \mathbb P_1 (\mathbb C\tilde Q_6)$, belonging to the same generator $\mathbb C\mathbb P_3(\mathbb C\tilde Q_6)$, defines the beam of generators $\mathbb C\mathbb P_1(\mathbb CQ_6)$ belonging to the quadric $\mathbb CQ_6$ (this beam is a cone). The center of the beam is determined by the system (\ref{e200}).
\end{enumerate}

\subsection{Theorem on two quadrics}$ $
     \indent Thus, the theorem is proved
\begin{theorem}
\label{theorem9}(The triality principle for two B-cylinders).\\
     In the projective space $\mathbb C \mathbb P_7$, there are two quadrics (two B-cylinders) with the following main properties:
\begin{enumerate}
    \item The planar generator $\mathbb C \mathbb P_3$ of a one quadric will define one-to-one the point R on the other quadric.
    \item The planar generator $\mathbb C \mathbb P_2$ of a one quadric will uniquely define the point R on the other quadric. But the point R of the second quadric can be associated to the manifold of planar generators $\mathbb C \mathbb P_2$ belonging to the same planar generator $\mathbb C \mathbb P_3$ of the first quadric.
    \item The rectilinear generator $\mathbb C \mathbb P_1$ of a one quadric will define one-to-one the rectilinear generator $\mathbb C \mathbb P_1$ of the other quadric. And all the rectilinear generators belonging to the same planar generator $\mathbb C \mathbb P_3$ of the first quadric define the beam centered at R belonging to the second quadric.
\end{enumerate}
\end{theorem}

    This theorem is actually the generalization of the Klein correspondence. Prove this.
\begin{proof}$ $\\
    On the quadric $\mathbb CQ_6$, we consider only those generators which have the form
\begin{equation}
\label{e209}
     X^A=(0,Y_b)\spsd
\end{equation}
    The manifold of such the generators is diffeomorphic to $\mathbb C \mathbb P_3$. In this case, each generator can be associated to the point of the quadric $\mathbb CQ_4 \subset \mathbb C \tilde Q_6$. According to the system (\ref{e182}), the first equation of it takes the form
\begin{equation}
\label{e210}
      r^{ab}\dot Y_b=0\spsd
\end{equation}
    Until the end of the proof, we set $\bf{A,B,A',B',...}$$=\overline{1,2}$. In addition, we consider the spinor representation of a twistor according to \cite[v.2, p. 49, eq. (6.1.24) and p.65, (6.2.18)(eng)]{Penrose1}.
 \begin{equation}
\label{e211}
     \begin{array}{c}
     \dot Y_b=(\dot \pi_{\mbox{\scriptsize\bf B}},
     \dot{\bar\omega}^{{\mbox{\scriptsize\bf B}}'})\sps\\ \\
     r^{ab}=const\left(
     \begin{array}{cc}
     -\frac{1}{2}\varepsilon^{\mbox{\scriptsize\bf AB}}r_cr^c &
     ir^{\mbox{\scriptsize\bf A}}{}_{\mbox{\scriptsize\bf B}'} \\[1.2ex]
     -i\bar r_{{\mbox{\scriptsize\bf A}}'}{}^{\mbox{\scriptsize\bf B}} &
     \bar\varepsilon_{{\mbox{\scriptsize\bf A}}'{\mbox{\scriptsize\bf B}}'}
     \end{array}
     \right)\spsd
     \end{array}
\end{equation}
     Therefore, the equation $R^{ab}Y_b=0$ can be rewritten as a system of two equations
\begin{equation}
\label{e212}
     \left\{
     \begin{array}{ccc}
     -\frac{1}{2}\varepsilon^{\mbox{\scriptsize\bf AB}}r_cr^c
     \dot \pi_{\mbox{\scriptsize\bf B}}+
     ir^{\mbox{\scriptsize\bf A}}{}_{{\mbox{\scriptsize\bf B}}'}
     \dot{\bar\omega}^{{\mbox{\scriptsize\bf B}}'} & = & 0\sps \\[1.2ex]
     -i\bar r_{{\mbox{\scriptsize\bf A}}'}{}^{\mbox{\scriptsize\bf B}}
     \dot \pi_{\mbox{\scriptsize\bf B}}+
     \bar\varepsilon_{{\mbox{\scriptsize\bf A}}'{\mbox{\scriptsize\bf B}}'}
     \dot{\bar\omega}^{{\mbox{\scriptsize\bf B}}'} & = & 0\sps
     \end{array}
     \right.
\end{equation}
     and only one of them will be significant
\begin{equation}
\label{e213}
     ir^{{\mbox{\scriptsize\bf AA}}'}
     \dot{\bar\pi}_{{\mbox{\scriptsize\bf A}}'}=
     \dot \omega^{\mbox{\scriptsize\bf A}}\spsd
\end{equation}
    Here we use the metric spinors $\bar\varepsilon_{\mbox{\scriptsize\bf A'B'}}, \varepsilon_{\mbox{\scriptsize\bf AB}}$ with the help of which  spinor indices raise and lower. This spinors pass each other by the conjugation. This defines the system
\begin{equation}
\label{e215}
     \begin{array}{c}
     \left\{
     \begin{array}{ccc}
     ir^{\mbox{\scriptsize\bf AA}'}
     \dot{\bar\pi}_{\mbox{\scriptsize\bf A}'} & = &
     \dot \omega^{\mbox{\scriptsize\bf A}}\sps\\[1.2ex]
     ir^{\mbox{\scriptsize\bf AA}'}
     \dot{\bar\eta}_{\mbox{\scriptsize\bf A}'} & = &
     \dot \xi^{\mbox{\scriptsize\bf A}}\sps
     \end{array}
     \right.
     \dot Y_b=(\dot \pi_{\mbox{\scriptsize\bf B}},
     \dot{\bar\omega}^{\mbox{\scriptsize\bf B}'})\sps
     \dot T_b=(\dot \eta_{\mbox{\scriptsize\bf B}},
     \dot{\bar\xi}^{\mbox{\scriptsize\bf B}'})\spsd
     \end{array}
\end{equation}
     The system coincides with the one \cite[v. 2, p. 63, eq. (6.2.14)(eng)]{Penrose1} which, in turn, leads to the Klein correspondence.
\end{proof}

     It should be noted in conclusion that from this theorem, the Cartan triality principle implies:
     There are the 3 diffeomorphic manifolds:
\begin{enumerate}
      \item the manifold of all points of the quadric;
      \item the manifold of I-family maximal planar generators;
      \item the manifold of II-family maximal planar generators.
\end{enumerate}
      This is true, because the two constructed quadrics can be identified, for example, by means of the spin-tensor $\tilde S_K{}^L$. The manifold of all points of the quadric is diffeomorphic to the maximal plain generator manifold of one of the two families. In addition, since the Cartan triality principle is performed then the operators $\eta^i{}_{KL}$ for the inclusion $\mathbb R^8 \subset \mathbb C^8$ define the octave algebra\footnote{The algebraic definition of \emph{structural constant} has been published in \cite{Andreev_3}. This definition has the form $\eta_{ij}{}^k:=\sqrt{2}\eta_i{}^{AB}\eta_j{}_{YA}\eta^k{}_{ZB}X^YX^Z$. Here $X^AX_A=2$; $A,B,...=\overline{1,8}$; $i,j,...=\overline{1,8}$; \emph{algebra identity} is determined as $\frac{1}{4\sqrt{2}}\eta^i{}^{AB}\varepsilon_{AB}$; the metric spin-tensor $\varepsilon_{AB}$ is the same as in (\ref{e12o}). Later \cite{Andreev_4}, \cite{Andreev_5}, this will lead to the generalization of the definition on n-dimensional spaces with n mod 8 = 0 and hence to the \emph{group alternative-elastic algebra} definition \cite{Andreev_6}.
\foreignlanguage{russian}{
\References{
     \bibitem{Andreev_3}
      K. V. Andreev: Structure constants of the algebra of octaves. The Clifford equation. Izv. Vyssh. Uchebn. Zaved. Mat., 2001, no. 3, 3–6. \url{http://mi.mathnet.ru/eng/ivm856}
     \bibitem{Andreev_4}
      K.V. Andreev. On the spinor formalism for even n. [\href{http://arxiv.org/abs/1202.0941}{arXiv:1202.0941v2}].
      \bibitem{Andreev_5}
      К.В. Андреев [K.V. Andreev]: О спинорном формализме при четной размерности базового пространства [O spinornom formalizme pri chetno\u\i\ razmernosti bazovogo prostranstva]. ВИНИТИ - 298-B-11 [VINITI-298-V-11], июнь 2011 [iun' 2011]. [in Russian: On the spinor formalism for the base space of even dimension]
      \bibitem{Andreev_6}
      K.V. Andreev. On the metric hypercomplex group alternative-elastic algebras for n mod 8 = 0. (\href{http://arxiv.org/abs/1110.4737}{arXiv:1110.4737v1}).
}
}
      } since satisfy to the reduced Clifford equation. This assertion is based on the results given in the monography \cite[v.2, p. 461-464(eng)]{Penrose1} which deals with the structure constants of this algebra.

\newpage
\subsection{Summary}
     On the defence, the following major provisions are submitted:
\begin{enumerate}
     \item An inclusion $\mathbb R^6_{(p, q)}\subset \mathbb C \mathbb R^6$ is carried out by means of the operator $H_i{}^\alpha$. The operator $H_i{}^\alpha$ defines the involution, the spinor representation of which has the form
            $$
            S_\alpha{}^{\beta '}=\frac{1}{4}\eta_\alpha{}^{ab}
            \bar \eta^{\beta '}{}_{c'd'}\cdot 2s_a{}^{c'}s_b{}^{d '}\sps
            s_a{}^{c '}\bar s_{c'}{}^d=\pm\delta_a{}^d
            $$
         for odd q and
            $$
            S_\alpha{}^{\beta '}=\frac{1}{4}\eta_\alpha{}^{ab}
            \bar \eta^{\beta '}{}_{c'd'}\cdot s^{kc'}s^{nd'}\varepsilon_{knab}\sps
            s_{ac'}=\pm\bar s_{c'a}
            $$
         for even q.

        \item The operators $A_{\alpha\beta a}{}^b$ are founded in an explicit form. This operators are determined by using the correspondence between bivectors of the space $\mathbb C \mathbb R^6$ and traceless operators of the space $\mathbb C^4$. This allows us to study the algebraic structure of the curvature tensor $R_{\alpha\beta\gamma\delta}$ of the space $\mathbb C \mathbb R^6$ on its spinor image: the spin-tensor $R_a{}^b{}_c{}^d$.
        \item It is proved that a simple isotropic bivector of the space $\Lambda^2\mathbb C^6$ defines a degenerate Rosenfeld null-pair up to a complex factor: a vector and covector of the spaces $\mathbb C^4$, the contracting of which is zero.
        \item It is stated that a bivector of the space $\Lambda^2\mathbb R^6_{(p, q)}$ for even q can be reduced to the canonical form in some basis.
        \item The generalized triality principle for a pair of B-cylinders is proved.
        \item The operators $\eta_A{}^{KL}$ satisfying the Clifford equation and responsible for the correspondence between rectilinear generators of two B-cylinders are defined. In addition, these operators define the structure constants of the octave algebra.
\end{enumerate}

\newpage
\section{Appendix}
\subsection{Proof of the second chapter equations}
\subsubsection{\texorpdfstring{Proof of the equations containing the operator $A_{\alpha\beta a}{}^b$}{Proof of the equations containing the operator  A}}$ $

\indent Define $(\alpha,\beta,...=1,2,3,4,5,6;\ a_1,b_1,a,b,e,f,k,l,m,n,...=1,2,3,4)$. We have
\begin{equation}
\label{e1p}
       A_{\alpha\beta d}{}^c=
       \eta_{\left[\right.\alpha}{}^{ca}
       \eta_{\beta\left.\right]}{}_{da}
\end{equation}
\begin{equation}
\label{e2p}
\begin{array}{c}
       A_{\alpha\beta d}{}^cA_\gamma{}^\beta{}_r{}^s=
       \frac{1}{4}
       (\eta_\alpha{}^{ca}\eta_{\beta ad}-
       \eta_{\alpha ad}\eta_\beta{}^{ca})
       (\eta_\gamma{}^{sa_1}\eta^\beta{}_{a_1r}-
       \eta_{\gamma a_1r}\eta^{\beta sa_1})=\\
       =\frac{1}{4}
       (\eta_\alpha{}^{ca}\eta_\gamma{}^{sa_1}\varepsilon_{a_1rad}-
       \eta_{\alpha ad}\eta_\gamma{}^{sa_1}\varepsilon_{a_1r}{}^{ca}-
       \eta_\alpha{}^{ca}\eta_{\gamma a_1r}\varepsilon_{ad}{}^{sa_1}+
       \eta_{\alpha ad}\eta_{\gamma a_1r}\varepsilon^{casa_1})=\\
       =\frac{1}{4}
       (\frac{1}{2}\eta_\alpha{}^{ca}\eta_{\gamma kl}
       \varepsilon^{klsa_1}\varepsilon_{a_1rad}-
       \eta_{\alpha ad}\eta_\gamma{}^{sa_1}
       (\delta_{a_1}{}^c\delta_r{}^a-\delta_{a_1}{}^a\delta_r{}^c)-\\
       -\eta_\alpha{}^{ca}\eta_{\gamma a_1r}
       (\delta_a{}^s\delta_d{}^{a_1}-\delta_a{}^{a_1}\delta_d{}^s)+
       \frac{1}{2}\eta_\alpha{}^{kl}\eta_{\gamma a_1r}
       \varepsilon_{klad}\varepsilon^{casa_1})=\\
       =\frac{1}{4}
       (-3\eta_\alpha{}^{ca}\eta_{\gamma \left[\right.ra}\delta_{d\left.\right]}{}^s-
       \eta_{\alpha rd}\eta_\gamma{}^{sc}+
       \eta_{\alpha ad}\eta_\gamma{}^{sa}\delta_r{}^c-\\
       -\eta_\alpha{}^{cs}\eta_{\gamma dr}+
       \eta_\alpha{}^{ca}\eta_{\gamma ar}\delta_d{}^s-
       3\eta_\alpha{}^{\left[\right.cs}\delta_d{}^{a_1\left.\right]}\eta_{\gamma a_1r}=\\
       =\frac{1}{4}
       (-\eta_\alpha{}^{ck}\eta_{\gamma rk}\delta_d{}^s+
       \eta_\alpha{}^{cs}\eta_{\gamma rd}-
       \eta_\alpha{}^{ck}\eta_{\gamma kd}\delta_r{}^s-
       \eta_{\alpha rd}\eta_\gamma{}^{sc}+
       \eta_{\alpha ad}\eta_\gamma{}^{sa}\delta_r{}^c-\\
       -\eta_\alpha{}^{cs}\eta_{\gamma dr}+
       \eta_\alpha{}^{ca}\eta_{\gamma ar}\delta_d{}^s-
       \eta_\alpha{}^{cs}\eta_{\gamma dr}+
       \eta_\alpha{}^{ck}\eta_{\gamma kr}\delta_d{}^s-
       \eta_\alpha{}^{sk}\eta_{\gamma kr}\delta_d{}^c)=\\
\end{array}
\end{equation}
       $$
       \left|
       \begin{array}{c}
       \eta_{\alpha dr}\eta_\gamma{}^{sc}+
       \eta_{\alpha kd}\eta_\gamma{}^{sk}\delta_r{}^c=
       \frac{1}{4}
       (\eta_\alpha{}^{mm_1}\eta_{\gamma ll_1}
       \varepsilon_{mm_1dr}\varepsilon^{ll_1sc}+
       \eta_\alpha{}^{mm_1}\eta_{\gamma ll_1}
       \varepsilon_{mm_1kd}\varepsilon^{ll_1sk}\delta_r{}^c)=\\
       =6\eta_\alpha{}^{mm_1}\eta_{\gamma ll_1}
       \delta_{\left[\right. m}{}^l\delta_{m_1}{}^{l_1}
       \delta_d{}^s\delta_{r\left.\right]}{}^c-
       \frac{3}{2}\eta_\alpha{}^{mm_1}\eta_{\gamma ll_1}
       \delta_{\left[\right. m}{}^l\delta_{m_1}{}^{l_1}
       \delta_{d\left.\right]}{}^s\delta_r{}^c=\\ \\
       =-g_{\alpha\gamma}\delta_r{}^s\delta_d{}^c+
       \eta_\alpha{}^{kc}\eta_{\gamma kd}\delta_r{}^s+
       \eta_\alpha{}^{ck}\eta_{\gamma kr}\delta_d{}^s+
       \eta_\alpha{}^{ks}\eta_{\gamma kr}\delta_d{}^c+
       \eta_\alpha{}^{sc}\eta_{\gamma dr}\spsd\\ \\
       \end{array}
       \right|
       $$
       $$
       \begin{array}{c}
       =\eta_\alpha{}^{cs}\eta_{\gamma rd}+
       \eta_\alpha{}^{ck}\eta_{\gamma kr}\delta_d{}^s+
       \frac{1}{2}(
       \eta_\alpha{}^{ck}\eta_{\gamma dk}\delta_r{}^s+
       \eta_\alpha{}^{sk}\eta_{\gamma rk}\delta_d{}^c)-
       \frac{1}{4}g_{\alpha\gamma}\delta_r{}^s\delta_d{}^c\spsd
       \end{array}
       $$
Contract with $g^{\alpha\gamma}$
\begin{equation}
\label{e3p}
       \begin{array}{c}
       A_{\alpha\beta d}{}^cA^{\alpha\beta}{}_r{}^s=
       \varepsilon_{rd}{}^{cs}+\varepsilon_{kr}{}^{ck}\delta_d{}^s+
       \frac{1}{2}(\varepsilon^{ck}{}_{dk}\delta_r{}^s+
       \varepsilon^{sk}{}_{rk}\delta_d{}^c)-
       \frac{3}{2}\delta_r{}^s\delta_d{}^c=\\
       =\delta_r{}^c\delta_d{}^s-\delta_r{}^s\delta_d{}^c-
       3\delta_r{}^c\delta_d{}^s+\frac{3}{2}\delta_d{}^c\delta_r{}^s+
       \frac{3}{2}\delta_d{}^c\delta_r{}^s-
       \frac{3}{2}\delta_r{}^s\delta_d{}^c=
       \frac{1}{2}\delta_r{}^s\delta_d{}^c-
       2\delta_r{}^c\delta_d{}^s\spsd
       \end{array}
\end{equation}
In addition,
\begin{equation}
\label{e4p}
       \begin{array}{c}
       A_{\alpha\beta d}{}^cA^{\lambda\mu}{}_c{}^d=
       \frac{1}{2}\eta_{\left[\right.\alpha}{}^{ca}
       \eta_{\beta\left.\right]}{}^{nn_1}
       \varepsilon_{adnn_1}\cdot
       \frac{1}{2}
       \eta^{\left[\right.\lambda}{}_{ll_1}
       \eta^{\mu\left.\right]}{}_{a_1c}\varepsilon^{da_1ll_1}=\\
       =\frac{1}{4}\eta_{\left[\right.\alpha}{}^{ca}
       \eta_{\beta\left.\right]}{}^{nn_1}
       \eta^{\left[\right.\lambda}{}_{ll_1}
       \eta^{\mu\left.\right]}{}_{a_1c}\cdot
       (-6\delta_{\left[\right. a}{}^{a_1}\delta_n{}^l
       \delta_{n_1\left.\right]}{}^{l_1})=\\
       =-\frac{1}{2}(\eta_{\left[\right.\alpha}{}^{ca}
       \eta_{\beta\left.\right]}{}^{nn_1}
       \eta^{\left[\right.\lambda}{}_{nn_1}
       \eta^{\mu\left.\right]}{}_{ac}-
       2\eta_{\left[\right.\alpha}{}^{ca}
       \eta_{\beta\left.\right]}{}^{a_1k_1}
       \eta^{\left[\right.\lambda}{}_{ak_1}
       \eta^{\mu\left.\right]}{}_{a_1c})=\\
       =-2\delta_{\left[\right. \alpha}{}^{\lambda}
       \delta_{\beta\left.\right]}{}^\mu+
       \frac{1}{2}(\eta_\alpha{}^{ca}
       \eta_\beta{}^{a_1k_1}
       \eta^{\left[\right.\lambda}{}_{ak_1}
       \eta^{\mu\left.\right]}{}_{a_1c}-
       \eta_\beta{}^{ca}
       \eta_\alpha{}^{a_1k_1}
       \eta^{\left[\right.\lambda}{}_{ak_1}
       \eta^{\mu\left.\right]}{}_{a_1c})=\\
       =2\delta_{\left[\right. \alpha}{}^\mu
       \delta_{\beta\left.\right]}{}^\lambda+
       \frac{1}{2}(\eta_\alpha{}^{ca}
       \eta_\beta{}^{a_1k_1}
       \eta^{\left[\right.\lambda}{}_{ak_1}
       \eta^{\mu\left.\right]}{}_{a_1c}-
       \eta_\beta{}^{a_1k_1}
       \eta_\alpha{}^{ca}
       \eta^{\left[\right.\lambda}{}_{k_1a}
       \eta^{\mu\left.\right]}{}_{ca_1})=
       2\delta_{\left[\right. \alpha}{}^\mu
       \delta_{\beta\left.\right]}{}^\lambda\spsd
       \end{array}
\end{equation}
And finally
\begin{equation}
\label{e5p}
       \begin{array}{c}
       A^{\alpha\beta}{}_m{}^nT_{\alpha\beta}=
       A^{\alpha\beta}{}_m{}^nA_{\alpha\beta k}{}^lT_l{}^k=\\
       =(\frac{1}{2}\delta_m{}^n\delta_k{}^l-2\delta_k{}^n\delta_m{}^l)
       T_l{}^k=-2T_m{}^n\sps T_{\alpha\beta}=-T_{\beta\alpha}\spsd
       \end{array}
\end{equation}
\newpage

\subsubsection{\texorpdfstring{Proof of the equation of the 4-vector $e_{\alpha\beta\gamma\delta}$}{Proof of the equation of the 4-vector}}$ $

\indent Let
\begin{equation}
\label{e56p}
       \begin{array}{c}
       e_{\alpha\beta\gamma\delta}=e_{[\alpha\beta\gamma\delta]}=
       A_{\alpha\beta b}{}^aA_{\gamma\delta d}{}^ce_a{}^b{}_c{}^d\sps
       e_k{}^k{}_c{}^d=e_a{}^b{}_k{}^k=0\sps
       e_a{}^b{}_c{}^d= e_c{}^d{}_a{}^b\sps
       e_k{}^t{}_t{}^k= 0\spsd
       \end{array}
\end{equation}
The contraction with $A^{\alpha\beta}{}_ m{}^nA^{\gamma\delta}{}_ r{}^s$ gives
\begin{equation}
\label{e57p}
       \begin{array}{c}
       A^{\alpha\beta}{}_ m{}^nA^{\gamma\delta}{}_ r{}^s
       A_{\alpha\beta b}{}^aA_{\gamma\delta d}{}^ce_a{}^b{}_c{}^d=\\
       =(\frac{1}{2}\delta_r{}^s\delta_d{}^c-2\delta_r{}^c\delta_d{}^s)
       (\frac{1}{2}\delta_b{}^a\delta_m{}^n-2\delta_b{}^n\delta_m{}^a)
       e_a{}^b{}_c{}^d=4e_m{}^n{}_r{}^s=\\
       =-A^{\beta\alpha}{}_ m{}^nA^{\gamma\delta}{}_ r{}^s
       A_{\gamma\alpha b}{}^aA_{\beta\delta d}{}^ce_a{}^b{}_c{}^d=\\
       =-(\eta_\gamma{}^{an}\eta^\beta{}_{mb}+
       \eta_\gamma{}^{ak}\eta^\beta{}_{km}\delta_b{}^n+
       \frac{1}{2}
       \eta_\gamma{}^{ak}\eta^\beta{}_{bk}\delta_m{}^n)\cdot\\
       \cdot(\eta_\beta{}^{cs}\eta^\gamma{}_{rd}+
       \eta_\beta{}^{ct}\eta^\gamma{}_{tr}\delta_d{}^s+
       \frac{1}{2}
       \eta_\beta{}^{ct}\eta^\gamma{}_{dt}\delta_r{}^s)e_a{}^b{}_c{}^d=\\
       =-(\varepsilon_{rd}{}^{an}\varepsilon_{mb}{}^{cs}+
       \varepsilon_{mb}{}^{ck}\varepsilon_{kr}{}^{an}\delta_d{}^s+
       \frac{1}{2}
       \varepsilon_{mb}{}^{ck}\varepsilon_{dk}{}^{an}\delta_r{}^s+\\
       +\varepsilon_{rd}{}^{ak}\varepsilon_{km}{}^{cs}\delta_b{}^n+
       \varepsilon^{ak}{}_{tr}\varepsilon^{ct}{}_{km}\delta_d{}^s\delta_b{}^n+
       \frac{1}{2}
       \varepsilon^{ak}{}_{dt}\varepsilon^{ct}{}_{km}\delta_r{}^s\delta_b{}^n+\\
       +\frac{1}{2}
       \varepsilon_{rd}{}^{ak}\varepsilon_{bk}{}^{cs}\delta_m{}^n+
       \frac{1}{2}
       \varepsilon_{bk}{}^{ct}\varepsilon_{tr}{}^{ak}\delta_m{}^n\delta_d{}^s+
       \frac{1}{4}
       \varepsilon_{bk}{}^{ct}\varepsilon_{dt}{}^{ak}\delta_m{}^n\delta_r{}^s)
       e_a{}^b{}_c{}^d=\\
\end{array}
\end{equation}
Then
\begin{equation}
\label{e58p}
       \begin{array}{c}
       e_m{}^s:=\frac{1}{3}e_m{}^t{}_t{}^s\ \ \Rightarrow\ \ \
       2e_m{}^n{}_r{}^s=2(e_m{}^s\delta_r{}^n+e_r{}^n\delta_m{}^s)-
       (e_r{}^s\delta_m{}^n+e_m{}^n\delta_r{}^s)\spsd
       \end{array}
\end{equation}

\subsubsection{\texorpdfstring{Proof of the equation of the 6-vector $e_{\alpha\beta\gamma\delta\rho\sigma}$}{Proof of the equation of the 6-vector}}$ $

\indent Let
\begin{equation}
\label{e59p}
       \begin{array}{c}
       e_{\alpha\beta\gamma\delta\rho\sigma}=
       e_{[\alpha\beta\gamma\delta\rho\sigma]}=
       A_{\alpha\beta b}{}^aA_{\gamma\delta d}{}^cA_{\rho\sigma s}{}^r
       e_a{}^b{}_c{}^d{}_r{}^s\spsd
       \end{array}
\end{equation}
We use (\ref{e57p}) and obtain ($e_m{}^s{}_p{}^q:=e_m{}^k{}_k{}^s{}_p{}^q=e_k{}^s{}_p{}^q{}_m{}^k$)
\begin{equation}
\label{e60p}
       \begin{array}{c}
       6e_m{}^n{}_p{}^q{}_r{}^s=
       2(e_m{}^s{}_p{}^q\delta_r{}^n+e_r{}^n{}_p{}^q\delta_m{}^s)-
       (e_r{}^s{}_p{}^q\delta_m{}^n+e_m{}^n{}_p{}^q\delta_r{}^s)=\\
       =2(e_m{}^q{}_r{}^s\delta_p{}^n+e_p{}^n{}_r{}^s\delta_m{}^q)-
       (e_p{}^q{}_r{}^s\delta_m{}^n+e_m{}^n{}_r{}^s\delta_p{}^q)\spsd
       \end{array}
\end{equation}
The contraction of (\ref{e60p}) with $\delta_q{}^p$ gives
\begin{equation}
\label{e61p}
       \begin{array}{c}
       0=6e_m{}^n{}_r{}^s{}_p{}^p=2(e_m{}^n{}_r{}^s+e_m{}^n{}_r{}^s)-
       (e_p{}^p{}_r{}^s\delta_m{}^n+4e_m{}^n{}_r{}^s)=
       -(e_p{}^p{}_r{}^s)\cdot\delta_m{}^n=0\sps\\
       e_p{}^p{}_r{}^s=e_r{}^s{}_p{}^p=0\spsd
       \end{array}
\end{equation}
The contraction of (\ref{e60p}) with $\delta_n{}^r$ gives
\begin{equation}
\label{e62p}
       \begin{array}{c}
       8e_m{}^s{}_p{}^q+2e_t{}^t{}_p{}^q\delta_m{}^s-
       e_m{}^s{}_p{}^q-e_m{}^s{}_p{}^q=
       2(e_m{}^q{}_p{}^s+e_p{}^k{}_k{}^s\delta_m{}^q)-
       (e_p{}^q{}_m{}^s+e_m{}^k{}_k{}^s\delta_p{}^q)\sps\\
       6e_m{}^s{}_p{}^q+e_p{}^q{}_m{}^s-2e_m{}^q{}_p{}^s=
       2e_p{}^t{}_t{}^s\delta_m{}^q-e_m{}^k{}_k{}^s\delta_p{}^q\spsd
       \end{array}
\end{equation}
The contraction of (\ref{e62p}) with $\delta_q{}^m$ gives
\begin{equation}
\label{e63p}
       \begin{array}{c}
       6e_k{}^s{}_p{}^k+e_p{}^k{}_k{}^s=
       8e_p{}^t{}_t{}^s-e_p{}^k{}_k{}^s\sps
       e_p{}^t{}_t{}^q=e_t{}^q{}_p{}^t\spsd
       \end{array}
\end{equation}
The contraction of (\ref{e62p}) with $\delta_s{}^m$ gives
\begin{equation}
\label{e64p}
       \begin{array}{c}
       10\tilde e:=e_t{}^k{}_k{}^t\sps
       -2e_k{}^q{}_p{}^k=2e_p{}^t{}_t{}^q-e_t{}^k{}_k{}^t\delta_p{}^q\sps
       e_p{}^t{}_t{}^q=\frac{5\tilde e}{2}\delta_p{}^q\spsd
       \end{array}
\end{equation}
The contraction of (\ref{e60p}) with $\delta_s{}^p$ gives
\begin{equation}
\label{e65p}
       \begin{array}{c}
       2e_m{}^k{}_k{}^q\delta_r{}^n+2e_r{}^n{}_m{}^q-
       e_r{}^k{}_k{}^q\delta_m{}^n-e_m{}^n{}_r{}^q=
       2e_m{}^q{}_r{}^n+2e_k{}^n{}_r{}^k\delta_m{}^q-
       e_k{}^q{}_r{}^k\delta_m{}^n-e_m{}^n{}_r{}^q\sps\\
       e_r{}^n{}_m{}^q=e_m{}^q{}_r{}^n\spsd
       \end{array}
\end{equation}
Then from  (\ref{e62p}),(\ref{e60p}), it follows
\begin{equation}
\label{e66p}
       \begin{array}{c}
       7e_m{}^s{}_p{}^q-2e_m{}^q{}_p{}^s=
       5\tilde e\delta_m{}^q\delta_p{}^s-
       \frac{5\tilde e}{2}\delta_m{}^s\delta_p{}^q\sps
       9e_m{}^{\left[\right.s}{}_p{}^{q\left.\right]}=
       \frac{15\tilde e}{2}\delta_m{}^{\left[\right.q}\delta_p{}^{s\left.\right]}\sps\\
       5e_m{}^s{}_p{}^q+4e_m{}^{\left[\right.s}{}_p{}^{q\left.\right]}=
       5\tilde e\delta_m{}^q\delta_p{}^s-\frac{5\tilde e}{2}\delta_m{}^s\delta_p{}^q\sps
       6e_m{}^s{}_p{}^q=
       \tilde e(4\delta_m{}^q\delta_p{}^s-\delta_m{}^s\delta_p{}^q)\sps\\
       e_m{}^n{}_r{}^s{}_p{}^q=\tilde e(2((4\delta_p{}^n\delta_r{}^q-
       \delta_p{}^q\delta_r{}^n)\delta_m{}^s+
       (4\delta_p{}^s\delta_m{}^q-
       \delta_p{}^q\delta_m{}^s)\delta_r{}^n)-\\
       -((4\delta_p{}^n\delta_m{}^q-
       \delta_p{}^q\delta_m{}^n)\delta_r{}^s+
       (4\delta_p{}^s\delta_r{}^q-
       \delta_p{}^q\delta_r{}^s)\delta_m{}^n))\spsd
       \end{array}
\end{equation}
Calculate $\tilde e$
\begin{equation}
\label{e67p}
       \begin{array}{c}
       720=e^{\alpha\beta\gamma\delta\mu\nu}e_{\alpha\beta\gamma\delta\mu\nu}=
       A^{\alpha\beta}{}_{b_1}{}^{a_1}A_{\alpha\beta b}{}^a
       A^{\gamma\delta}{}_{d_1}{}^{c_1}A_{\gamma\delta d}{}^c
       A^{\mu\nu}{}_{n_1}{}^{m_1}A_{\mu\nu n}{}^m\sps\\
       e_a{}^b{}_c{}^d{}_m{}^ne_{a_1}{}^{b_1}{}_{c_1}{}^{d_1}{}_{m_1}{}^{n_1}=
       8(\frac{1}{4}\delta_{b_1}{}^{a_1}\delta_b{}^a-
       \delta_{b_1}{}^a\delta_b{}^{a_1})
       (\frac{1}{4}\delta_{d_1}{}^{c_1}\delta_d{}^c-
       \delta_{d_1}{}^c\delta_d{}^{c_1})\cdot\\
       \cdot(\frac{1}{4}\delta_{n_1}{}^{m_1}\delta_n{}^m-
       \delta_{n_1}{}^m\delta_n{}^{m_1})
       e_a{}^b{}_c{}^d{}_m{}^ne_{a_1}{}^{b_1}{}_{c_1}{}^{d_1}{}_{m_1}{}^{n_1}=
       -8e_a{}^b{}_c{}^d{}_m{}^ne_b{}^a{}_d{}^c{}_n{}^m=\\
       =-8{\tilde e}^2(8\delta_m{}^b\delta_c{}^n\delta_a{}^d-
       4\delta_m{}^n\delta_c{}^b\delta_a{}^d+
       8\delta_m{}^d\delta_a{}^n\delta_c{}^b-
       4\delta_m{}^b\delta_a{}^n\delta_c{}^d+
       2\delta_m{}^n\delta_a{}^b\delta_c{}^d-
       4\delta_m{}^d\delta_c{}^n\delta_a{}^b)\\
       (8\delta_n{}^a\delta_d{}^m\delta_b{}^c-
       4\delta_n{}^m\delta_d{}^a\delta_b{}^c+
       8\delta_n{}^c\delta_b{}^m\delta_d{}^a-
       4\delta_n{}^a\delta_b{}^m\delta_d{}^c+
       2\delta_n{}^m\delta_b{}^a\delta_d{}^c-
       4\delta_n{}^c\delta_d{}^m\delta_b{}^a)=\\
       =-720\cdot 64{\tilde e}^2\spsd
       \end{array}
\end{equation}
Whence,
\begin{equation}
\label{e68p}
       \tilde e=\frac{1}{8}i\spsd
\end{equation}

\subsubsection{\texorpdfstring{Proof of the equation containing the operator $N_a{}^f$}{Proof of the equation containing the operator N}}$ $

\indent Let
\begin{equation}
\label{e69p}
       \begin{array}{c}
       N_{\left[\right.a}{}^bN_{c\left.\right]}{}^d=
       \delta_{\left[\right.a}{}^b\delta_{c\left.\right]}{}^d
       \end{array}
\end{equation}
 then the contraction with $\delta_b{}^a$ will give
\begin{equation}
\label{e70p}
       \begin{array}{c}
       N_a{}^aN_c{}^d-N_a{}^dN_c{}^a=3\delta_c{}^d\sps\\[2ex]
       N_a{}^dN_c{}^a=N_a{}^aN_c{}^d-3\delta_c{}^d\spsd
       \end{array}
\end{equation}
Multiply (\ref{e69p}) by $N_d{}^f$
\begin{equation}
\label{e71p}
       \begin{array}{c}
       N_{\left[\right.a}{}^bN_{c\left.\right]}{}^dN_d{}^f=
       \delta_{\left[\right.a}{}^bN_{c\left.\right]}{}^f\spsd
       \end{array}
\end{equation}
Substitute (\ref{e70p}) into (\ref{e71p})
\begin{equation}
\label{e72p}
       \begin{array}{c}
       N_{\left[\right.a}{}^bN_{c\left.\right]}{}^fN_r{}^r-
       3N_{\left[\right.a}{}^b\delta_{c\left.\right]}{}^f=
       \delta_{\left[\right.a}{}^bN_{c\left.\right]}{}^f\spsd
       \end{array}
\end{equation}
The equation (\ref{e72p}) contracts  with $\delta_f{}^c$ substituting (\ref{e70p}) then
\begin{equation}
\label{e73p}
       \begin{array}{c}
       3\delta_a{}^bN_r{}^r=N_r{}^r\delta_a{}^b+8N_a{}^b\sps\\[2ex]
       N_a{}^b=\frac{1}{4}N_r{}^r\delta_a{}^b\spsd
       \end{array}
\end{equation}
Therefor,
\begin{equation}
\label{e74p}
       \begin{array}{c}
       N_a{}^b=n\delta_a{}^b\sps n^2=1\spsd
       \end{array}
\end{equation}

\subsection{Proof of the forth chapter equations}
\subsubsection{Proof the Bianchi identity}$ $

\indent The Bianchi identity have the form
\begin{equation}
\label{e16p}
       \begin{array}{c}
       R_{\alpha\beta\gamma\delta}+R_{\alpha\gamma\delta\beta}+
       R_{\alpha\delta\beta\gamma}=0\sps\\
       (A_{\alpha\beta d}{}^cA_{\gamma\delta r}{}^s+
       A_{\alpha\gamma d}{}^cA_{\delta\beta r}{}^s+
       A_{\alpha\delta d}{}^cA_{\beta\gamma r}{}^s)
       R_c{}^d{}_s{}^r=0\spsd
       \end{array}
\end{equation}
Let's contract it with $A^{\alpha\beta}{}_t{}^lA^{\gamma\delta}{}_m{}^n$
\begin{equation}
\label{e17p}
       %
\end{equation}
The contraction with $\delta_n{}^t$ gives
\begin{equation}
\label{e18p}
       \begin{array}{c}
       16R_k{}^l{}_m{}^k+4R_t{}^k{}_k{}^t\delta_m{}^l-
       2R_k{}^l{}_m{}^k-2R_m{}^k{}_k{}^l-
       8R_k{}^r{}_r{}^k\delta_m{}^l+R_k{}^r{}_r{}^k\delta_m{}^l=0\sps\\
       R_k{}^l{}_m{}^k=\frac{1}{4}R_k{}^r{}_r{}^k\delta_m{}^l=
       -\frac{1}{8}R\delta_m{}^l
       \end{array}
\end{equation}
that finishes the proof of (\ref{e18}).

\subsubsection {Proof of the identities related to the Weyl tensor}$ $

\indent The following relation
\begin{equation}
\label{e12p}
       \begin{array}{c}
       R_{\left[\right.\alpha}{}^{\left[\right.\gamma}
       g_{\beta\left.\right]}{}^{\delta\left.\right]}=
       A_{\alpha\beta d}{}^cA^{\gamma\delta}{}_r{}^s
       \frac{1}{16}
       (P^{kdnr}\varepsilon_{kcns}-P_{kc}{}^{nr}\varepsilon^{kd}{}_{ns}+
       P_{kcns}\varepsilon^{kdnr}-P^{kd}{}_{ns}\varepsilon_{kc}{}^{nr})=\\
       =A_{\alpha\beta d}{}^cA^{\gamma\delta}{}_r{}^s
       \frac{1}{16}
       (P^{kd}{}_{mm_1}\frac{1}{2}\varepsilon^{mm_1nr}\varepsilon_{kcns}+
       P_{kc}{}^{mm_1}\frac{1}{2}\varepsilon_{mm_1ns}\varepsilon^{kdnr}-\\
       -P_{kc}{}^{nr}(\delta_n{}^k\delta_s{}^d-\delta_n{}^d\delta_s{}^k)-
       P^{kd}{}_{ns}(\delta_k{}^n\delta_c{}^r-\delta_k{}^r\delta_c{}^n))=\\
       =A_{\alpha\beta d}{}^cA^{\gamma\delta}{}_r{}^s
       \frac{1}{16}
       (P^{kd}{}_{kc}\delta_s{}^r-P^{kd}{}_{ks}\delta_c{}^r+P^{rd}{}_{cs}+\\
       +P_{kc}{}^{kd}\delta_s{}^r-P_{kc}{}^{kr}\delta_s{}^d+P_{sc}{}^{dr}-
       P_{kc}{}^{kr}\delta_s{}^d+P_{sc}{}^{dr}-
       P^{kd}{}_{ks}\delta_c{}^r+P^{rd}{}_{cs})=\\
       =A_{\alpha\beta d}{}^cA^{\gamma\delta}{}_r{}^s
       \frac{1}{4}
       (P_{sc}{}^{dr}-\frac{1}{2}P^{kd}{}_{ks}\delta_c{}^r+
       \frac{1}{2}P_{kc}{}^{kd}\delta_s{}^r-
       \frac{1}{2}P_{kc}{}^{kr}\delta_s{}^d)=\\
       =A_{\alpha\beta d}{}^cA^{\gamma\delta}{}_r{}^s
       \frac{1}{4}
       (P_{sc}{}^{dr}-\frac{1}{2}R\delta_s{}^d\delta_c{}^r+
       \frac{1}{4}R\delta_s{}^r\delta_c{}^d)\spsd
       \end{array}
\end{equation}
will be true. Similarly,
\begin{equation}
\label{e13p}
       \begin{array}{c}
       g_{\left[\right.\alpha}{}^{\left[\right.\gamma}
       g_{\beta\left.\right]}{}^{\delta\left.\right]}=
       A_{\alpha\beta d}{}^cA^{\gamma\delta}{}_r{}^s
       \frac{1}{4}
       (\varepsilon_{sc}{}^{dr}-
       \frac{1}{2}\varepsilon^{kd}{}_{ks}\delta_c{}^r+
       \frac{1}{2}\varepsilon^{kd}{}_{kc}\delta_s{}^r-
       \frac{1}{2}\varepsilon_{kc}{}^{kr}\delta_s{}^d)=\\
       =A_{\alpha\beta d}{}^cA^{\gamma\delta}{}_r{}^s
       \frac{1}{4}
       (\frac{1}{2}\delta_s{}^r\delta_c{}^d-2\delta_s{}^d\delta_c{}^r)\spsd
       \end{array}
\end{equation}
Thus, from
\begin{equation}
\label{e14p}
       \begin{array}{c}
       C_{\alpha\beta}{}^{\gamma\delta}:=
       R_{\alpha\beta}{}^{\gamma\delta}-
       R_{\left[\alpha \right.}{}^{\left[ \gamma \right.}
       g_{\left.\beta \right]}{}^{\left. \delta \right]}+
       1/10Rg_{\left[\alpha \right.}{}^{\left[ \gamma \right.}
       g_{\left.\beta \right]}{}^{\left. \delta \right]}\sps
       \end{array}
\end{equation}
it follows
\begin{equation}
\label{e15p}
       \begin{array}{c}
       C_{\alpha\beta}{}^{\gamma\delta}=
       A_{\alpha\beta d}{}^cA^{\gamma\delta}{}_r{}^s
       (R_c{}^d{}_s{}^r-\frac{1}{4}P_{sc}{}^{dr}+
       \frac{1}{8}R\delta_s{}^d\delta_c{}^r-
       \frac{1}{16}R\delta_s{}^r\delta_c{}^d+
       \frac{1}{40}R(\frac{1}{2}\delta_s{}^r\delta_c{}^d-\\
       -2\delta_s{}^d\delta_c{}^r))=
       A_{\alpha\beta d}{}^cA^{\gamma\delta}{}_r{}^s
       (R_{\left(\right.c}{}^d{}_{s\left.\right)}{}^r+
       R_{\left[\right.c}{}^k{}_{|k|}{}^{\left[\right.r}
       \delta_{s\left.\right]}{}^{d\left.\right]}+
       \frac{1}{40}R(3\delta_s{}^d\delta_c{}^r-
       2\delta_s{}^r\delta_c{}^d))\spsd
       \end{array}
\end{equation}
Then
\begin{equation}
\label{e19p}
       \begin{array}{c}
       C_{\alpha\beta}{}^{\gamma\delta}=
       A_{\alpha\beta d}{}^cA^{\gamma\delta}{}_r{}^s
       (R_{\left(\right.c}{}^d{}_{s\left.\right)}{}^r+
       \frac{1}{80}R\delta_s{}^d\delta_c{}^r)\spsd
       \end{array}
\end{equation}
Contract it with $A^{\alpha\beta}{}_k{}^lA_{\gamma\delta t}{}^n$
\begin{equation}
\label{e20p}
       \begin{array}{c}
       4C_k{}^l{}_t{}^n:=A^{\alpha\beta}{}_k{}^lA_{\gamma\delta t}{}^n
       C_{\alpha\beta}{}^{\gamma\delta}=
       A^{\alpha\beta}{}_k{}^lA_{\alpha\beta d}{}^c
       A^{\gamma\delta}{}_r{}^sA_{\gamma\delta t}{}^n
       (R_{\left(\right.c}{}^d{}_{s\left.\right)}{}^r+
       \frac{1}{80}R\delta_s{}^d\delta_c{}^r)=\\
       =(\frac{1}{2}\delta_k{}^l\delta_d{}^c-2\delta_k{}^c\delta_d{}^l)
        (\frac{1}{2}\delta_r{}^s\delta_t{}^n-2\delta_r{}^n\delta_t{}^s)
        (\frac{1}{2}R_c{}^d{}_s{}^r+\frac{1}{2}R_s{}^d{}_c{}^r+
         \frac{1}{80}R\delta_s{}^d\delta_c{}^r)=\\
        =(2R_k{}^l{}_t{}^n+
        \frac{1}{8}R_{k_1}{}^{r_1}{}_{r_1}{}^{k_1}\delta_k{}^l\delta_t{}^n-
        \frac{1}{2}R_{k_1}{}^l{}_k{}^{k_1}\delta_t{}^n-
        \frac{1}{2}R_t{}^{k_1}{}_{k_1}{}^n\delta_k{}^l+\\
        +2R_t{}^l{}_k{}^n+\frac{1}{80}R
        (\delta_t{}^n\delta_k{}^l-\delta_t{}^n\delta_k{}^l-
        \delta_t{}^n\delta_k{}^l+4\delta_t{}^l\delta_k{}^n)=\\
        =(4R_{\left(\right.k}{}^l{}_{t\left.\right)}{}^n-
        \frac{1}{16}R\delta_k{}^l\delta_t{}^n+
        \frac{1}{16}R\delta_k{}^l\delta_t{}^n+
        \frac{1}{16}R\delta_t{}^n\delta_k{}^l+
        \frac{1}{80}R(4\delta_t{}^l\delta_k{}^n-\delta_t{}^n\delta_k{}^l))=\\
        =4(R_{\left(\right.k}{}^l{}_{t\left.\right)}{}^n+
        \frac{1}{40}R\delta_{\left(\right.k}{}^l\delta_{t\left.\right)}{}^n)\sps
        \end{array}
\end{equation}
as finishes the proof of (\ref{e16}).

\subsubsection{The proof of the Ricci identity}$ $

\indent By definition,
\begin{equation}
\label{e2100p}
       \begin{array}{c}
       \Box_a{}^d:=\frac{1}{2}(\nabla_{ak}\nabla^{dk}-\nabla^{dk}\nabla_{ak})\sps
       \Box_{\alpha\beta}:=2\nabla_{\left[\right.\alpha}
       \nabla_{\beta\left.\right]}\sps
       \Box_{\alpha\beta}=A_{\alpha\beta d}{}^a\Box_a{}^d\spsd
       \end{array}
\end{equation}
In addition, by the natural way, the covariant constancy of the following values
\begin{equation}
\label{e2101p}
       \begin{array}{c}
       \nabla_\alpha \eta_\gamma{}^{ab}=0\sps
       \nabla_\alpha A_{\beta\gamma d}{}^a=0\sps
       \nabla_\alpha g_{\alpha\beta}=0\sps
       \nabla_\alpha \varepsilon_{abcd}=0
       \end{array}
\end{equation}
is assumed.  The Ricci identity has the form
\begin{equation}
\label{e25p}
       \begin{array}{c}
       \Box_{\alpha\beta}k^{\gamma\delta}=
       R_{\alpha\beta\lambda}{}^\gamma k^{\lambda\delta}+
       R_{\alpha\beta\lambda}{}^\delta k^{\gamma\lambda}\spsd
       \end{array}
\end{equation}
Then
\begin{equation}
\label{e26p}
       \begin{array}{c}
       A_{\alpha\beta b}{}^a A^{\gamma\delta}{}_c{}^d\Box_a{}^bk_d{}^c=
       A_{\alpha\beta b}{}^a 2A_\lambda{}^{\left[\right.\delta}{}_c{}^{|l|}
       A^{\gamma\left.\right]\lambda}{}_r{}^dR_a{}^b{}_l{}^ck_d{}^r\spsd
       \end{array}
\end{equation}
Whence, taking into account (\ref{e2p}), we obtain
\begin{equation}
\label{e27p}
       \begin{array}{c}
       A^{\gamma\delta}{}_c{}^d\Box_a{}^bk_d{}^c=
       -2R_a{}^b{}_l{}^ck_d{}^r\cdot
       (\eta^{\left[\right.\delta |ld|}\eta^{\gamma\left.\right]}{}_{rc}-
       \eta^{\left[\right.\delta |lk|}\eta^{\gamma\left.\right]}{}_{rk}\delta_c{}^d)=\\
       =(\eta^{\gamma ld}\eta^\delta{}_{rc}-
       \eta^\gamma{}_{rc}\eta^{\delta ld})
       (R_a{}^b{}_l{}^ck_d{}^r-R_a{}^b{}_l{}^kk_k{}^r\delta_d{}^c)=
       \eta^{\gamma ld}\eta^\delta{}_{rc}(R_a{}^b{}_l{}^ck_d{}^r-\\
       -R_a{}^b{}_l{}^kk_k{}^r\delta_d{}^c-
       6\delta_{\left[\right.l}{}^r\delta_d{}^c\delta_{r_1}{}^{l_1}
       \delta_{c_1\left.\right]}{}^{d_1}R_a{}^b{}_{l_1}{}^{c_1}k_{d_1}{}^{r_1}+
       \frac{3}{2}
       \delta_{\left[\right.l}{}^r\delta_d{}^c\delta_{r_1\left.\right]}{}^{l_1}
       \delta_{c_1}{}^{d_1}R_a{}^b{}_{l_1}{}^{c_1}k_{d_1}{}^{r_1})=\\
       =\eta^{\gamma ld}\eta^\delta{}_{rc}
       (-R_a{}^b{}_l{}^kk_k{}^r\delta_d{}^c-
       R_a{}^b{}_k{}^rk_d{}^k\delta_l{}^c)
       =A^{\gamma\delta}{}_r{}^l\cdot
       (-R_a{}^b{}_l{}^kk_k{}^r+R_a{}^b{}_k{}^rk_l{}^k)\sps\\
       \Box_a{}^b k_l{}^r=
       -R_a{}^b{}_l{}^kk_k{}^r+R_a{}^b{}_k{}^rk_l{}^k\spsd
       \end{array}
\end{equation}
On the other hand,
\begin{equation}
\label{e28p}
       \begin{array}{c}
       \Box_{\alpha\beta}r^\gamma=R_{\alpha\beta\delta}{}^\gamma r^\delta\spsd
       \end{array}
\end{equation}
Multiply the both sides by $\frac{1}{2}\eta_\gamma{}^{mm_1}A^{\beta\alpha}{}_a{}^b$
\begin{equation}
\label{e29p}
       \begin{array}{c}
       \Box_a{}^br^{mm_1}
       =\eta_\gamma{}^{mm_1} A_\delta{}^\gamma{}_c{}^l R_a{}^b{}_l{}^c r^\delta=
       \eta_\gamma{}^{mm_1}\eta^\gamma{}_{ck}
       \eta_\delta{}^{lk}R_a{}^b{}_l{}^cr^\delta=\\
       =2\delta_{\left[\right.c}{}^m\delta_{k\left.\right]}{}^{m_1}
       R_a{}^b{}_l{}^cr^{lk}=
       R_a{}^b{}_l{}^mr^{lm_1}+R_a{}^b{}_l{}^{m_1}r^{ml}\sps\\
       \Box_a{}^br^{mm_1}=
       R_a{}^b{}_l{}^mr^{lm_1}+R_a{}^b{}_l{}^{m_1}r^{ml}\spsd
       \end{array}
\end{equation}
So, let $k^{\alpha\beta}$ be a simple isotropic bivector that, according to Corollary \ref{corollary113} of the third chapter, we have the traceless image $k_a{}^b$ of the form
\begin{equation}
\label{e32p}
       \begin{array}{c}
       k_a{}^b=P^bQ_a\sps P^a Q_a=0\spsd
       \end{array}
\end{equation}
Let's substitute (\ref{e32p}) in (\ref{e27p}) then
\begin{equation}
\label{e33p}
       \begin{array}{c}
       P^c\Box_a{}^b Q_d+Q_d\Box_a{}^bP^c=
       -R_a{}^b{}_d{}^kQ_kP^c+R_a{}^b{}_k{}^cQ_dP^k\spsd
       \end{array}
\end{equation}
Contract (\ref{e33p}) with any covector $S_c$, that $S_cP^c=0$, then 
\begin{equation}
\label{e34p}
       \begin{array}{c}
       S_c\Box_a{}^bP^c=R_a{}^b{}_k{}^c P^kS_c\spsd
       \end{array}
\end{equation}
Let's consider an isotropic vector $r^\gamma$ such that
\begin{equation}
\label{e35p}
       \begin{array}{c}
       r^\gamma=\frac{1}{2}\eta^\gamma{}_{ab}r^{ab}:=
       \eta^\gamma{}_{ab}X^aY^b\spsd
       \end{array}
\end{equation}
Then from (\ref{e29p}), the equation
\begin{equation}
\label{e36p}
       \begin{array}{c}
       X^m\Box_a{}^bY^{m_1}+
       Y^{m_1}\Box_a{}^bX^m-
       X^{m_1}\Box_a{}^bY^m-
       Y^m\Box_a{}^bX^{m_1}=\\
       =X^mR_a{}^b{}_k{}^{m_1}Y^k+
       Y^{m_1}R_a{}^b{}_k{}^mX^k-
       X^{m_1}R_a{}^b{}_k{}^mX^k-
       Y^mR_a{}^b{}_k{}^{m_1}X^k
       \end{array}
\end{equation}
will follow. Contract this equation with $Z_m$, that $X^mZ_m=0$ and $Y^mZ_m \ne 0$, then, taking into account (\ref{e34p}) (this means that $Y^{m_1}Z_m\Box_a{}^bX^m=R_a{}^b{}_k{}^m Z_mX^kY^{m_1}$),  we obtain the equality
\begin{equation}
\label{e37p}
       \begin{array}{c}
       -Z_mX^{m_1}\Box_a{}^b Y^m-Z_mY^m\Box_a{}^bX^{m_1}=
       -X^{m_1}R_a{}^b{}_k{}^m Y^kZ_m-Z_mY^mR_a{}^b{}_k{}^{m_1}X^k\spsd
       \end{array}
\end{equation}
For any $T_d$, the spinors $X^m$, $Y^m$ and $Z_m$ can always be chosen, that $X^mT_m=0$, $Y^mT_m=0$, $X^mZ_m=0$, $Y^mZ_m \ne 0$. Multiply (\ref{e37p}) by $T_d$ and obtain
\begin{equation}
\label{e38p}
       \begin{array}{c}
       -Z_mX^{m_1}T_d\Box_a{}^b Y^m-Z_mY^mT_d\Box_a{}^bX^{m_1}=\\
       =-Z_mX^{m_1}\Box_a{}^b(T_dY^m)+Z_mX^{m_1}Y^m\Box_a{}^bT_d
       -Z_mY^m\Box_a{}^b(X^{m_1}T_d)+Z_mY^mX^{m_1}\Box_a{}^bT_d=\\
       =-X^{m_1}R_a{}^b{}_k{}^m Y^kZ_mT_d-Z_mY^mR_a{}^b{}_k{}^{m_1}X^kT_d\spsd
       \end{array}
\end{equation}
Let's subtract from (\ref{e38p}) the following identities obtained from (\ref{e33p})
\begin{equation}
\label{e39p}
       \begin{array}{c}
       -Z_mX^{m_1}\Box_a{}^b(T_dY^m)=Z_mX^{m_1}Y^mR_a{}^b{}_d{}^kT_k-
       Z_mX^{m_1}T_dR_a{}^b{}_k{}^mY^k\sps\\
       -Z_mY^m\Box_a{}^b(X^{m_1}T_d)=Z_mX^{m_1}Y^mR_a{}^b{}_d{}^kT_k-
       Z_mY^mT_dR_a{}^b{}_k{}^{m_1}X^k\spsd\\
       \end{array}
\end{equation}
Whence,
\begin{equation}
\label{e40p}
       \begin{array}{c}
       Z_mY^mX^{m_1}\Box_a{}^bT_d=-X^{m_1}Z_mY^mR_a{}^b{}_d{}^kT_k\sps\\
       \Box_a{}^bT_d=-R_a{}^b{}_d{}^kT_k\spsd
       \end{array}
\end{equation}
Thus,
\begin{equation}
\label{e42p}
       \begin{array}{c}
       \Box_a{}^bT_d=-R_a{}^b{}_d{}^rT_r\sps
       \Box_a{}^bT^d=R_a{}^b{}_r{}^dT^r\spsd
       \end{array}
\end{equation}

\subsubsection{The proof of the differential Bianchi identity}$ $

\indent The differential Bianchi identity has the form
\begin{equation}
\label{e43p}
       \begin{array}{c}
       \nabla_{\left[\right.\alpha}R_{\beta\gamma\left.\right]\delta\lambda}=0\spsd
       \end{array}
\end{equation}
From this, the equation
\begin{equation}
\label{e44p}
       \begin{array}{c}
       A_{\left[\right.\beta\gamma |c|}{}^b\eta_{\alpha\left.\right]}{}^{aa_1}
       \nabla_{aa_1}R_b{}^c{}_r{}^s=0
       \end{array}
\end{equation}
will follow. Contract it with $A^{\beta\gamma}{}_m{}^n$ and obtain
\begin{equation}
\label{e45p}
       %
       $$
Contract this equation with $\eta^\alpha{}_{tp}$ then
\begin{equation}
\label{e46p}
       \begin{array}{c}
       -2\varepsilon_{tp}{}^{aa_1}\nabla_{aa_1}R_m{}^n{}_r{}^s+
       4\varepsilon_{tp}{}^{bn}\nabla_{cm}R_b{}^c{}_r{}^s
       -4\varepsilon_{tp}{}^{ba}\nabla_{am}R_b{}^n{}_r{}^s-
       2\varepsilon_{tp}{}^{ba}\delta_m{}^n\nabla_{ca}R_b{}^c{}_r{}^s=\\
       =-4\nabla_{tp}R_m{}^n{}_r{}^s+
       4\delta_p{}^n\nabla_{cm}R_t{}^c{}_r{}^s-
       4\delta_t{}^n\nabla_{cm}R_p{}^c{}_r{}^s- \\
       -4\nabla_{pm}R_t{}^n{}_r{}^s+
       4\nabla_{tm}R_p{}^n{}_r{}^s-
       2\delta_m{}^n\nabla_{cp}R_t{}^c{}_r{}^s+
       2\delta_m{}^n\nabla_{ct}R_p{}^c{}_r{}^s=\\
       =-4\nabla_{tp}R_m{}^n{}_r{}^s+4\nabla_{tm}R_p{}^n{}_r{}^s-
       4\nabla_{pm}R_t{}^n{}_r{}^s+4\delta_p{}^n\nabla_{cm}R_t{}^c{}_r{}^s-\\
       -4\delta_t{}^n\nabla_{cm}R_p{}^c{}_r{}^s-
       2\delta_m{}^n\nabla_{cp}R_t{}^c{}_r{}^s+
       2\delta_m{}^n\nabla_{ct}R_p{}^c{}_r{}^s=\\
       =12\nabla_{\left[\right.tm}R_{p\left.\right]}{}^n{}_r{}^s-
       8\delta_{\left[\right.t}{}^n\nabla_{|cm|}R_{p\left.\right]}{}^c{}_r{}^s-
       4\delta_m{}^n\nabla_{c\left[\right.p}R_{t\left.\right]}{}^c{}_r{}^s=0\spsd
       \end{array}
\end{equation}
Contract this equation with $\delta_n{}^m$ then obtain the identity
\begin{equation}
\label{e47p}
       \begin{array}{c}
       4\nabla_{tk}R_p{}^k{}_r{}^s-
       4\nabla_{pk}R_t{}^k{}_r{}^s+
       4\nabla_{cp}R_t{}^c{}_r{}^s-\\ \\
       -4\nabla_{ct}R_p{}^c{}_r{}^s-
       8\nabla_{cp}R_t{}^c{}_r{}^s+
       8\nabla_{ct}R_p{}^c{}_r{}^s=0\spsd
       \end{array}
\end{equation}
Contract (\ref{e46p}) with $\delta_n{}^t$ then
\begin{equation}
\label{e48p}
       \begin{array}{c}
       -4\nabla_{np}R_m{}^n{}_r{}^s+
       4\nabla_{tm}R_p{}^t{}_r{}^s+
       4\nabla_{cm}R_p{}^c{}_r{}^s
       -16\nabla_{cm}R_p{}^c{}_r{}^s-
       2\nabla_{cp}R_m{}^c{}_r{}^s+
       2\nabla_{cm}R_p{}^c{}_r{}^s=0\sps\\
       \nabla_{c\left(\right.p}R_{m\left.\right)}{}^c{}_r{}^s=0\spsd
       \end{array}
\end{equation}
Then the Bianchi identity takes the form
\begin{equation}
\label{e49p}
       \begin{array}{c}
       \nabla_{\left[\right.tm}R_{p\left.\right]}{}^n{}_r{}^s=
       \delta_{\left[\right.t}{}^n\nabla_{|c|m}R_{p\left.\right]}{}^c{}_r{}^s\spsd
       \end{array}
\end{equation}
The contraction of (\ref{e48p}) with $\delta_s{}^m$ gives the equation
\begin{equation}
\label{e50p}
       \begin{array}{c}
       \nabla_{cp}R_m{}^c{}_r{}^m+
       \nabla_{cm}R_p{}^c{}_r{}^m=0\sps
       \nabla_{cm}R_p{}^c{}_r{}^m=\frac{1}{8}\nabla_{rp}R\spsd
       \end{array}
\end{equation}

\subsubsection{\texorpdfstring{Proof of the formulas associated with the metric induced in the cross-section of the cone $K_6$}{Proof of the formulas associated with the metric induced in the cross-section of the cone K}}$ $

\indent Let the section of the cone $K_6$
\begin{equation}
\label{e81p}
     T^2+V^2-W^2-X^2-Y^2-Z^2=0
\end{equation}
by the plane V+W=1 be set. We make the stereographic projection of the resulting hyperboloid onto the plane (V = 0, W = 1)
\begin{equation}
\label{e82p}
\begin{array}{c}
     \frac{t}{T}=\frac{x}{X}=\frac{y}{Y}=\frac{z}{Z}=-\frac{\frac{1}{2}}{V-\frac{1}{2}}\sps
     x^2+y^2+z^2-t^2=\frac{1}{2V-1}\sps\\
     T=-t(2V-1)\sps X=-x(2V-1)\sps Y=-y(2V-1)\sps Z=-z(2V-1)\sps\\
     dT=-(dt\ (2V-1) +2tdV)\sps dX=-(dx\ (2V-1) +2xdV)\sps\\ dY=-(dy\ (2V-1) +2ydV)\sps dZ=-(dz\ (2V-1) +2zdV)\spsd\\
\end{array}
\end{equation}
Then
\begin{equation}
\label{e83p}
\begin{array}{c}
     dS^2=dT^2+dV^2-dW^2-dX^2-dY^2-dZ^2=|dV=-dW|=\\=dT^2-dX^2-dY^2-dZ^2=
     (dt^2-dx^2-dy^2-dz^2)(2V-1)^2+\\+4(tdt-xdx-ydy-zdz)dV(2V-1)+4(t^2-x^2-y^2-z^2)dV^2=\\
     =|tdt-xdx-ydy-zdz=-\frac{1}{2}d(\frac{1}{2V-1})|=\Large\frac{dt^2-dx^2-dy^2-dz^2}{(t^2-x^2-y^2-z^2)^2}\spsd
\end{array}
\end{equation}
We make the substitution
\begin{equation}
\label{e84p}
\begin{array}{c}
     \varsigma=\frac{-iX+Y}{2V-1}=-y+ix\sps \omega=-\frac{i(T+Z)}{2V-1}=i(t+z)\sps \eta=\frac{i(Z-T)}{2V-1}=i(t-z)\sps\\
     \varsigma\bar\varsigma+\eta\omega=\frac{1}{2V-1}\sps\\
     d\varsigma     = (\frac{dY}{2V-1}-\frac{2YdV}{(2V-1)^2})-
                i(\frac{dX}{2V-1}-\frac{2XdV}{(2V-1)^2})\sps
     d\bar\varsigma = (\frac{dY}{2V-1}-\frac{2YdV}{(2V-1)^2})+
                i(\frac{dX}{2V-1}-\frac{2XdV}{(2V-1)^2})\sps\\
     d\omega    =-i(\frac{dT}{2V-1}-\frac{2TdV}{(2V-1)^2}+
                  \frac{dZ}{2V-1}-\frac{2ZdV}{(2V-1)^2})\sps
     d\eta  =-i(\frac{dT}{2V-1}-\frac{2TdV}{(2V-1)^2}-
                  \frac{dZ}{2V-1}+\frac{2ZdV}{(2V-1)^2})\sps\\
     d\varsigma d\bar\varsigma + d\omega d\eta=\frac{(dX^2+dY^2+dZ^2-dT^2)}{(2V-1)^2}-4\frac{dV}{(2V-1)^3}(XdX+YdY+ZdZ-TdT)+\\
     +4\frac{dV^2}{(2V-1)^4}(X^2+Y^2+Z^2-T^2)=-\frac{dT^2-dX^2-dY^2-dZ^2}{(T^2-X^2-Y^2-Z^2)^2}\sps\\
     ds^2:=dT^2-dX^2-dY^2-dZ^2=-\frac{d\varsigma d\bar\varsigma+d\eta d\omega}{(\varsigma\bar\varsigma+\eta\omega)^2}\sps
     x^2+y^2+z^2-t^2=\frac{1}{2V-1}\spsd
\end{array}
\end{equation}
Therefore, there is a reason to put
\begin{equation}
\label{e85p}
    X:=
    \left(
    \begin{array}{ccc}
    \omega & \varsigma \\
    -\bar\varsigma & \eta
    \end{array}
    \right) \sps
    dX:=
    \left(
    \begin{array}{ccc}
    d\omega & d\varsigma \\
    -d\bar\varsigma & d\eta
    \end{array}
    \right) \sps
    \frac{\partial}{\partial X}:=
    \left(
    \begin{array}{ccc}
    \frac{\partial}{\partial\omega}   &
    \frac{\partial}{\partial\varsigma} \\[1.2ex]
    -\frac{\partial}{\partial\bar\varsigma} &
    \frac{\partial}{\partial\eta}
    \end{array}
    \right)\spsd
\end{equation}
Then
\begin{equation}
\label{e86p}
    d s^2=-\frac{det(dX)}{(det(X))^2}\sps \bar X^T+X=0\spsd
\end{equation}

\subsubsection{Proof of the first invariant formulas}$ $

\indent Consider the fractional linear transformation group L
\begin{equation}
\label{e87p}
    {\tilde X}=(AX+B)(CX+D)^{-1} \sps
    S:=
    \left(
    \begin{array}{ccc}
     A & B \\
     C & D
    \end{array}
    \right) \sps
    detS=1\spsd
\end{equation}
Then there are two consecutive transformations
\begin{equation}
\label{e88p}
    \tilde X=(AX+B)(CX+D)^{-1} \sps
    \tilde{\tilde X}=(\tilde A\tilde X+\tilde B)(\tilde C\tilde X+\tilde D)^{-1}
\end{equation}
\begin{equation}
\label{e89p}
%
\end{equation}
For unitary fractional-linear transformations, we have
\begin{equation}
\label{e90p}
\begin{array}{c}
    X^*+X\equiv\tilde X^*+\tilde X=0\sps\\
    0=(AX+B)(CX+D)^{-1}+(CX+D)^{*-1}(AX+B)^*\sps\\
    0=(CX+D)^*(AX+B)+(AX+B)^*(CX+D)=\\
    =X^*(A^*C+C^*A)X+X^*(A^*D+C^*B)+
    (B^*C+D^*A)X+D^*B+B^*D\equiv \\
    \equiv X^*+X\spsd
\end{array}%
\end{equation}
Therefore, the identities
\begin{equation}
\label{e95p}
\begin{array}{lcllcl}
    \hat{\hat X}=AX+B   & \Rightarrow & d\hat{\hat X}=AdX\sps &
    \hat{\hat X}=X^{-1} & \Rightarrow & d\hat{\hat X}=-X^{-1}\ dX\ X^{-1}
\end{array}
\end{equation}
are true. The proof of the second is that
\begin{equation}
\label{e96p}
\begin{array}{c}
    \hat{\hat X}X=1\sps
    (d\hat{\hat X})\ X+\hat{\hat X}\ (dX)=0\sps\\
    (d\hat{\hat X})\ X+X^{-1}\ (dX)=0\sps
    d\hat{\hat X}=-X^{-1}\ (dX)\ X^{-1}\spsd
\end{array}
\end{equation}
Then we will obtain the identities
\begin{equation}
\label{e97p}
\begin{array}{c}
    \tilde X^*+\tilde X=0\sps
    -\tilde X^*=\tilde X\sps\\
    -(CX+D)^{*-1}(AX+B)^*=(AX+B)(CX+D)^{-1}\sps\\
    -(AX+B)^*=(CX+D)^*(AX+B)(CX+D)^{-1}\spsd
\end{array}
\end{equation}
Multiply the both parts on $CdX$
\begin{equation}
\label{e98p}
\begin{array}{c}
    (-X^*A^*-B^*)CdX=(CX+D)^*(AX+B)(CX+D)^{-1}CdX\sps\\
    (-X^*A^*C-B^*C)dX=(CX+D)^*(AX+B)(CX+D)^{-1}CdX\spsd
\end{array}
\end{equation}
We use (\ref{e91p})
\begin{equation}
\label{e99p}
\begin{array}{c}
    (X^*C^*A+D^*A-E)dX=(CX+D)^*(AX+B)(CX+D)^{-1}CdX\sps\\
    (CX+D)^*AdX-(CX+D)^*(AX+B)(CX+D)^{-1}CdX=dX\sps\\
    A(dX)(CX+D)^{-1}-(AX+B)(CX+D)^{-1}C(dX)(CX+D)^{-1}=\\=(CX+D)^{*-1}(dX)(CX+D)^{-1}\spsd
\end{array}
\end{equation}
We use (\ref{e92p})
\begin{equation}
\label{e100p}
\begin{array}{c}
    (d(AX+B))\ (CX+D)^{-1}+(AX+B)d((CX+D)^{-1})=\\=(CY+DZ)^{*-1}Z^*(dX)Z(CY+DZ)^{-1}\sps\\
    \tilde Z^*\ d\tilde X \tilde Z=Z^*\ dX Z\spsd
\end{array}
\end{equation}
Thus, the first invariant is obtained.

\subsubsection{Proof of the second invariant formulas} $ $

\indent Let
\begin{equation}
\label{e101p}
    \hat X=AX+B\sps
\end{equation}
or in the componentwise record,
\begin{equation}
\label{e102p}
    \begin{array}{cc}
     \hat x_1=a_1x_1+a_2x_3+b_1\sps & \hat x_2=a_1x_2+a_2x_4+b_2\sps \\
     \hat x_3=a_3x_1+a_4x_3+b_3\sps & \hat x_4=a_3x_2+a_4x_4+b_4\spsd
    \end{array}
\end{equation}
Then
\begin{equation}
\label{e103p}
    \begin{array}{lcr}
     \frac{\partial}{\partial x_1}=
     a_1\frac{\partial}{\partial \hat x_1}+a_3\frac{\partial}{\partial \hat x_3}\sps & &
     \frac{\partial}{\partial x_2}=
     a_1\frac{\partial}{\partial \hat x_2}+a_3\frac{\partial}{\partial \hat x_4}\sps  \\ \\
     \frac{\partial}{\partial x_3}=
     a_2\frac{\partial}{\partial \hat x_1}+a_4\frac{\partial}{\partial \hat x_3}\sps &  &
     \frac{\partial}{\partial x_4}=
     a_2\frac{\partial}{\partial \hat x_2}+a_4\frac{\partial}{\partial \hat x_4}\sps
    \end{array}
\end{equation}
or in the abbreviated form,
\begin{equation}
\label{e104p}
\begin{array}{ccc}
     \frac{\partial}{\partial X}=A^T\frac{\partial}{\partial \hat X} &
     \Rightarrow &
     \frac{\partial}{\partial \hat X^T}=\frac{\partial}{\partial X^T}A^{-1}
     \ (det(A)\ne 0)\spsd
\end{array}
\end{equation}
Now let
\begin{equation}
\label{e105p}
     \hat X=X^{-1}\sps
\end{equation}
or in the componentwise record,
\begin{equation}
\label{e106p}
\begin{array}{cccccc}
     \hat x_1= \frac{x_4}{x_1x_4-x_2x_3} & \hat x_2=-\frac{x_2}{x_1x_4-x_2x_3}\sps
     \hat x_3=-\frac{x_3}{x_1x_4-x_2x_3} & \hat x_4=\frac{x_1}{x_1x_4-x_2x_3}\spsd
\end{array}
\end{equation}
Then
\begin{equation}
\label{e107p}
\begin{array}{c}
      \frac{\partial}{\partial x_1}=
      -\frac{x_4x_4}{(x_1x_4-x_2x_3)^2}\frac{\partial}{\partial \hat x_1}+
       \frac{x_2x_4}{(x_1x_4-x_2x_3)^2}\frac{\partial}{\partial \hat x_2}+
       \frac{x_3x_4}{(x_1x_4-x_2x_3)^2}\frac{\partial}{\partial \hat x_3}-
       \frac{x_2x_3}{(x_1x_4-x_2x_3)^2}\frac{\partial}{\partial \hat x_4}\sps\\ \\
      \frac{\partial}{\partial x_2}=
       \frac{x_4x_3}{(x_1x_4-x_2x_3)^2}\frac{\partial}{\partial \hat x_1}-
       \frac{x_1x_4}{(x_1x_4-x_2x_3)^2}\frac{\partial}{\partial \hat x_2}-
       \frac{x_3x_3}{(x_1x_4-x_2x_3)^2}\frac{\partial}{\partial \hat x_3}+
       \frac{x_1x_3}{(x_1x_4-x_2x_3)^2}\frac{\partial}{\partial \hat x_4}\sps\\ \\
      \frac{\partial}{\partial x_3}=
       \frac{x_4x_2}{(x_1x_4-x_2x_3)^2}\frac{\partial}{\partial \hat x_1}-
       \frac{x_2x_2}{(x_1x_4-x_2x_3)^2}\frac{\partial}{\partial \hat x_2}-
       \frac{x_1x_4}{(x_1x_4-x_2x_3)^2}\frac{\partial}{\partial \hat x_3}+
       \frac{x_1x_2}{(x_1x_4-x_2x_3)^2}\frac{\partial}{\partial \hat x_4}\sps\\ \\
      \frac{\partial}{\partial x_4}=
      -\frac{x_2x_3}{(x_1x_4-x_2x_3)^2}\frac{\partial}{\partial \hat x_1}+
       \frac{x_1x_2}{(x_1x_4-x_2x_3)^2}\frac{\partial}{\partial \hat x_2}+
       \frac{x_1x_3}{(x_1x_4-x_2x_3)^2}\frac{\partial}{\partial \hat x_3}-
       \frac{x_1x_1}{(x_1x_4-x_2x_3)^2}\frac{\partial}{\partial \hat x_4}\sps
\end{array}
\end{equation}
or in the abbreviated form,
\begin{equation}
\label{e108p}
\begin{array}{ccc}
       \frac{\partial}{\partial X}=-X^{-1T}
       \frac{\partial}{\partial \hat X}X^{-1T} &
       \Leftrightarrow &
       \frac{\partial}{\partial \hat X^T}=-X
       \frac{\partial}{\partial X^T}X\spsd
\end{array}
\end{equation}
Assume that $det(C)\ne 0$ then
\begin{equation}
\label{e108pa}
\begin{array}{c}
        \tilde X=(AX+B)(CX+D)^{-1}=AC^{-1}+(B-AC^{-1}D)(CX+D)^{-1}\spsd
\end{array}
\end{equation}
Therefor,
\begin{equation}
\label{e109p}
\begin{array}{c}
       \frac{\partial}{\partial \tilde X^T}=\frac{\partial}{\partial (CX+D)^{T-1}}(B-AC^{-1}D)^{-1}=\\
       =-(CX+D)\frac{\partial}{\partial (CX+D)^T}(CX+D)(B-AC^{-1}D)^{-1}=\\
       =-(CX+D)\frac{\partial}{\partial X^T}C^{-1}(CX+D)(B-AC^{-1}D)^{-1}=\\
      =-(CX+D)\frac{\partial}{\partial X^T}(XC^*+C^{-1}DC^*)(BC^*-AC^{-1}DC^*)^{-1}=\\
      =(CX+D)\frac{\partial}{\partial X^T}(CX+D)^*(BC^*+AD^*)^{-1}=
      (CX+D)\frac{\partial}{\partial X^T}(CX+D)^*\spsd
\end{array}
\end{equation}
Whence,
\begin{equation}
\label{e110p}
       \tilde Z^{-1}\frac{\partial}{\partial \tilde X^T}\tilde Z^{*-1}=
       Z^{-1}\frac{\partial}{\partial X^T} Z^{*-1}\spsd
\end{equation}
If now $det(C)=0$ then put
\begin{equation}
\label{e111p}
\begin{array}{c}
       \tilde{\tilde X}=\tilde X^{-1}=(CX+D)(AX+B)^{-1}\sps
       \tilde{\tilde Z}=AY+BZ\sps \tilde{\tilde C}=A\sps
       det(\tilde{\tilde C})\ne 0\spsd
\end{array}
\end{equation}
If $det(A)\ne 0$ then
\begin{equation}
\label{e112p}
\begin{array}{c}
       Z^{-1}\frac{\partial}{\partial \tilde X^T} Z^{*-1}=
       \tilde{\tilde Z}^{-1}\frac{\partial}{\partial \tilde{\tilde X}^T}\tilde{\tilde Z}^{*-1}=\\
       =-(AY+BZ)^{-1}\tilde X\frac{\partial}{\partial \tilde X^T}
       \tilde X(AY+BZ)^{*-1}=|\tilde X^*=-\tilde X|=\\
       =(AY+BZ)^{-1}(AY+BZ)(CY+DZ)^{-1}\frac{\partial}{\partial \tilde X^T}\times\\
        \times(CY+DZ)^{*-1}(AY+BZ)^*(AY+BZ)^{*-1}=
       \tilde Z^{-1}\frac{\partial}{\partial \tilde X^T}\tilde Z^{*-1}\spsd
\end{array}
\end{equation}
If $det(A)=0$ then  put
\begin{equation}
\label{e113p}
\begin{array}{c}
       \tilde{\tilde X}=\tilde X+\tilde D\sps
       \tilde{\tilde{\tilde X}}=\tilde{\tilde X}^{-1}\sps
       \tilde{\tilde{\tilde X}}=
       (\tilde{\tilde{\tilde A}}X+\tilde{\tilde{\tilde B}})
       (\tilde{\tilde{\tilde C}}X+\tilde{\tilde{\tilde D}})^{-1}\sps
        \tilde{\tilde{\tilde C}}=A+\tilde DC\sps
        det(\tilde{\tilde{\tilde C}})\ne 0\spsd
\end{array}
\end{equation}
It is obvious that $\tilde D $ can always be chosen, that $det(\tilde{\tilde {\tilde C}}) \ne 0 $. So
\begin{equation}
\label{e114p}
\begin{array}{c}
         Z^{-1}\frac{\partial}{\partial X^T}Z^{*-1}=
         \tilde{\tilde{\tilde Z}}^{-1}
         \frac{\partial}{\partial \tilde{\tilde{\tilde X}}^T}
         \tilde{\tilde{\tilde Z}}^{*-1}
         =\tilde{\tilde Z}^{-1}
         \frac{\partial}{\partial \tilde{\tilde X}^T}
         \tilde{\tilde Z}^{*-1}=
         \tilde Z^{-1}
         \frac{\partial}{\partial \tilde X^T}
         \tilde Z^{*-1}
\end{array}
\end{equation}
that will give the second invariant.

\subsection{Proof of the fifth chapter equations}
\subsubsection{The proof of the integrability conditions of the bitwistor equation}$ $

\indent By the definition,
\begin{equation}
\label{e2500}
\begin{array}{c}
     \Box_a{}^d=\frac{1}{2}(\nabla_{ak}\nabla^{dk}-\nabla^{dk}\nabla_{ak})=
     \frac{1}{4}\varepsilon_{akmn}(\nabla^{mn}\nabla^{dk}+\nabla^{md}\nabla^{nk})
     =\frac{1}{2}\varepsilon_{akmn}\nabla^{m\left(\right.n}\nabla^{d\left)\right.k}\spsd
\end{array}
\end{equation}
Therefor,
\begin{equation}
\label{e2501}
      \Box_a{}^dX^c=R_a{}^d{}_r{}^cX^r=
      \frac{1}{2}\varepsilon_{akmn}\nabla^{m\left(\right.n}\nabla^{d\left)\right.k}X^c\spsd
\begin{array}{c}
\end{array}
\end{equation}
If $\nabla^{a\left(\right.b}X^{c\left)\right.}=0$ then the integrability conditions of this equation have the form
\begin{equation}
\label{e2502}
\begin{array}{c}
       \frac{1}{2}\varepsilon_{akmn}\nabla^{m\left(\right.n}\nabla^{d|k|}
       X^{c\left)\right.}=\\
       =\frac{1}{6}\varepsilon_{akmn}
       (\nabla^{m\left(\right.n}\nabla^{d\left)\right.k}X^c+
        \nabla^{m\left(\right.n}\nabla^{c\left)\right.k}X^d+
        \nabla^{m\left(\right.c}\nabla^{d\left)\right.k}X^n)=
       \\
       =\frac{2}{3}R_a{}^{\left(\right.d}{}_l{}^{c\left.\right)}X^l+
        \frac{1}{12}\varepsilon_{akmn}
        (\nabla^{mc}\nabla^{dk}+\nabla^{md}\nabla^{ck})X^n=
\end{array}
\end{equation}
      $$
      \left|
\begin{array}{c}
        \varepsilon_{akmn}
        (\nabla^{mc}\nabla^{dk}+\nabla^{md}\nabla^{ck})X^n=\\
        =\varepsilon_{akmn}(\varepsilon^{cdkr}
        \frac{1}{2}(\nabla^{ms}\nabla_{sr}-\nabla_{sr}\nabla^{ms})-
        \varepsilon^{mrdk}
        \frac{1}{2}(\nabla_{sr}\nabla^{sc}-\nabla^{sc}\nabla_{sr}))X^n=\\
        =6\delta_a{}^{\left[\right. c}\delta_m{}^r\delta_n{}^{d\left.\right]}
        \frac{1}{2}(\nabla_{sr}\nabla^{sm}-\nabla^{sm}\nabla_{sr})X^n
        +4\delta_{\left[\right. a}{}^r\delta_{n\left.\right]}{}^d
        \frac{1}{2}(\nabla_{sr}\nabla^{sc}-\nabla^{sc}\nabla_{sr})X^n=\\
        =6\delta_a{}^{\left[\right. c}\delta_m{}^r\delta_n{}^{d\left.\right]}
        R_r{}^m{}_l{}^nX^l+
        4\delta_{\left[\right. a}{}^r\delta_{n\left.\right]}{}^d
        R_r{}^c{}_l{}^nX^l
        =2R_a{}^{\left(\right. c}{}_l{}^{d\left.\right)}X^l+
        \frac{1}{4}\delta_a{}^{\left(\right.c}\delta_l{}^{d\left.\right)} RX^l
        \end{array}
        \right|
        $$
\begin{equation}
\label{e2503}
\begin{array}{c}
        =\frac{5}{6}R_a{}^{\left(\right. c}{}_l{}^{d\left.\right)}X^l+
        \frac{1}{48}\delta_a{}^{\left(\right.c}\delta_l{}^{d\left.\right)} RX^l=
        \frac{5}{6}C_a{}^c{}_l{}^dX^l\sps
\end{array}
\end{equation}
where $C_a{}^c{}_l{}^d$ is the Weyl tensor analogue.

\newpage
}
\References{
\bibitem{Besse1}
A.L. Besse: Einstein Manifolds. Ergeb. Math. Grenzgeb. (3), Vol. 10. Springer-Verlag, Berlin, 1987. Reprinted: Classics in Mathematics. Springer-Verlag, Berlin, 2008. Russian translation: А. Бессе [A. Besse]: Многообразия Эйнштейна [Mnogoobraziya \`E\u\i nshte\u\i na]. т. I, II [Tom I, II]. Мир [Mir], Москва [Moscow], 1990. Russian translation by Д. В. Алексеевский [D.V. Alekseevski\u\i]. The list of isomorphisms between classical Lie algebras for $n\le 6$ is given on the pages 200-201, item 7.101.
\bibitem{Dubrovin1}
В.А. Дубровин [V.A. Dubrovin], С.П. Новиков [S.P. Novikov], А.Т. Фоменко [A.T. Fomenko]: Современная геометрия [Sovremennaya geometriya]. Наука [Nauka], Москва [Moskva], 1986. English translation: V.A. Dubrovin, S.P. Novikov, A.T. Fomenko: Modern Geometry. Grad. Texts Math., Part 1: Vol. 93, Part 2: Vol. 104. Springer, New York, Part 1: 1984, Part 2: 1985.
\bibitem{Cartan1}
\'{E}. Cartan: Le\c{c}ons sur la Th\'{e}orie des Spineurs, 2 Vols.. Vol.\ I: Les Spineurs de l'Espace a Trois Dimensions. Actual. Sci. Ind., Vol.\ 643, Expos\'{e}s G\'{e}om., Vol.\ 9. Vol.\ II: Les Spineurs de l'Espace a n $>$ 3 dimensions. Les Spineurs en G\'{e}om\'{e}trie Riemanienne. Actual. Sci. Ind., Vol.\ 701, Expos\'{e}s G\'{e}om., Vol.\ 11. Hermann, Paris, 1938. English translation: The Theory of Spinors. Hermann, Paris, 1966. Reprinted: Dover Publications, Inc., New York, 1981. Russian translation: Э. Картан [\`E Kartan]: Теория спиноров [Teoriya spinorov]. Платон [Platon], Москва [Moskva], 1997. Russian translation by П. А. Широков [P.A. Shirokov].
\bibitem{Koboyasi1}
S. Kobayashi, K. Nomizu: Foundations of Differential Geometry. Volume 2. Interscience Tracts Pure Appl. Math., Vol. 15.2. Interscience, New York, 1969. Reprinted: Wiley Classics Library. Wiley-Interscience. New York, 1996. Russian translation: Ш. Кобаяси [Sh. Kobayasi],  К. Номидзу [K. Nomidzu]: Основы дифференциальной геометрии [Osnovy differentsial'no\u\i\ geometrii]. т. 2 [Tom 2]. Наука [Nauka], Москва [Moskva], 1981. Russian translation by Л.В. Сабинин [L.V. Sabinin]. Almost complex manifolds are studied on the pages 114-141. However, as real manifolds are considered Hermitian manifolds unlike this thesis.
\bibitem{Koboyasi2}
S. Kobayashi: Transformation Groups in Differential Geometry. Ergeb. Math. Grenzgeb. (2), Vol. 70. Springer-Verlag, Berlin, 1972. Russian translation: Ш. Кобаяси [Sh. Kobayasi]: Группы преобразований в дифференциальной геометрии [Gruppy preobrazovani\u\i\ v differentsial'no\u\i\ geometrii]. Наука [Nauka], Москва [Moskva], 1986. Russian translation by Л.В. Сабинин [L.V. Sabinin]. A piece of the information on the Riemannian space and complex manifolds is given on pages 1-119.
\bibitem{Kotelnikov1}
А.П. Котельников [A.P. Kotel'nikov]: Винтовое счисление и некоторые приложения его к геометрии и механнике [Vintovoe schislenie i nekotorye prilozheniya ego k geometrii i mekhanike]. S.n., Казань [Kazan'], 1895. [in Russian: Screw Calculus and Some of Its Applications to Geometry and Mechanics]
\bibitem{Lichnerowicz1}
A. Lichnerowicz: Th\'eorie globale des connexions et des groupes d'holonomie. Travaux Rech. Math.. Dunod, Paris, 1955. Consiglio Naz. Rich., Monogr. Mat., Vol. 2. Edizioni Cremonese, Rome, 1955, 1962. English translation: Global Theory of Connections and Holonomy Groups. Noordhoff, Leyden, 1976. Russian translation: А. Лихнерович [A. Likhnerovich]: Теория связностей в целом и группы голономий [Teoriya svyaznoste\u\i\ v tselom i gruppy golonomi\u\i]. ИЛ [IL], Москва [Moskva], 1960. Russian translation by С.П. Фиников [S.P. Finikov] under edition В.В. Рыжкова [V.V. Ryzhkova]. Almost complex manifolds and connections on they are studied. However, as real manifolds are considered Hermitian manifolds unlike this thesis.
\bibitem{Manin1}
Ю.И. Манин [Yu.I. Manin]: Калибровочные поля и комплексная геометрия [Kalibrovochnye polya i kompleksnaya geometriya]. Москва [Moskva], Наука [Nauka], 1996. English translation: Yu.I. Manin: Gauge Field Theory and Complex Geometry. Grundlehren Math. Wiss., Vol. 289. Springer-Verlag, Berlin, 1988. The Minkowski space is studied as the manifold of real points of the big cell of the Grassmannian of complex planes in the twistor space on the pages 15-72.
\bibitem{Landau1}
Л.Д. Ландау [L.D. Landau], Е.М. Лифшиц [E.M. Livshits]: Теория поля [Teoriya polya]. Наука [Nauka], Москва [Moskva], 1988. English translation: L.D. Landau, E.M. Lifshitz: The Classical Theory of Fields. Pergamon Press, Oxford, 1961.
\bibitem{Neifeld1}
Э.Г. Нейфельд [\`E.G. Ne\u\i fel'd]: Об инволюциях в комплексных пространствах [Ob involyutsiyakh v kompleksnykh prostranstvakh]. Тр. Геом. Семин. [Tr. Geom. Semin.], Казанский университет [Kazanski\u\i\ universitet] (Выпуск [Vypusk]) 19(1989)71-82 (Mathnet URL: \url{http://mi.mathnet.ru/eng/kutgs98}). [in Russian: Involutions in complex spaces]
\bibitem{Neifeld2}
Э.Г. Нейфельд [\`E.G. Ne\u\i fel'd]: Геометрия поверхности в проективном пространстве над алгеброй [Geometriya poverkhnosti v proektivnom prostranstve nad algebro\u\i]. In: Ю.А. Яфаров [Yu.A. Yafarov] (Ed.): Геометрия обобщенных пространств [Geometriya obobshchennykh prostranstv]. Башкирский Государственный Педагогический Институт [Bashkirski\u\i\ Gosudarstvenny\u\i\ Pedagogicheski\u\i\ Institut], Уфа [Ufa], 1982, pp. 32-51. [in Russian: Geometry of a surface in the projective space over an algebra]
\bibitem{Neifeld3}
Э.Г. Нейфельд [\`E.G. Ne\u\i fel'd]: О внутренних геометриях поляризованных комплексных грассманианов [O vnutrennikh geometriyakh polyarizovannykh kompleksnykh grassmanianov]. Изв. Высш. Учебн. Завед., Матем. [Izv. Vyssh. Uchebn. Zaved., Mat.] (1995) No. 5 (396), 51-54 (Mathnet URL: \url{http://mi.mathnet.ru/eng/ivm1741}). English translation: \`E.G. Ne\u\i fel'd: On the intrinsic geometries of polarized complex Grassmannians. Russian Math. (Iz. VUZ) 39(1995) No. 5, 46-49.
\bibitem{Neifeld4}
Э.Г. Нейфельд [\`E.G. Ne\u\i fel'd]: Нормализация комплексных грассманианов и квадрик [Normalizatsiya kompleksnykh grassmanianov i kvadrik]. Тр. Геом. Семин. [Tr. Geom. Semin.], Казанский университет [Kazanski\u\i\ universitet] (Выпуск [Vypusk]) 20(1990)58-69 (Mathnet URL: \url{http://mi.mathnet.ru/kutgs82}). [in Russian: Normalization of complex Grassmannians and quadrics]
\bibitem{Neifeld5}
Э.Г. Нейфельд [\`E.G. Ne\u\i fel'd]: О внутренних геометриях нормализованного пенроузиана [O vnutrennikh  geometriyakh normalizovannogo penrouziana]. Тр. Геом. Семин. [Tr. Geom. Semin.], Казанский университет [Kazanski\u\i\ universitet] (Выпуск [Vypusk]) 20(1990)70-73 (Mathnet URL: \url{http://mi.mathnet.ru/eng/kutgs83}). [in Russian: Intrinsic geometries of a normalized Penrosian]
\bibitem{Neifeld6}
Э.Г. Нейфельд [\`E.G. Ne\u\i fel'd]: Аффинные связности на нормализованном многообразии плоскостей проективного пространства [Affinnye svyaznosti na normalizovannom mnogoobrazii ploskoste\u\i\ proektivnogo prostranstva].  Изв. Высш. Учебн. Завед., Матем. [Izv. Vyssh. Uchebn. Zaved., Mat.] (1976) No. 11 (174) 48-55 (Mathnet URL: \url{http://mi.mathnet.ru/eng/ivm8577}). English translation: E.G. Ne\u\i fel'd: Affine connections on the normalized manifold of planes in a projective space. Sov. Math. (Iz. VUZ) 20(1978) No. 11, 42–48.
\bibitem{Neifeld7}
Э.Г. Нейфельд [\`E.G. Ne\u\i fel'd]: О внутренних геометриях нуль-плоскостей максимальной размерности поляритетов второго порядка [O vnutrennikh  geometriyakh nul'-ploskoste\u\i\ maksimal'no\u\i\ razmernosti polyaritetov vtorogo poryadka]. Тр. Геом. Семин. [Tr. Geom. Semin.], Казанский университет [Kazanski\u\i\ universitet] (Выпуск [Vypusk]) 14(1982)50-55 (Mathnet URL: \url{http://mi.mathnet.ru/eng/kutgs178}). [in Russian: Intrinsic geometries of manifolds of zero-planes of maximal dimension of second-order polarities]
\bibitem{Norden1}
А.П. Норден [A.P. Norden]: О комплексном представлении тензоров пространства Лоренца [O kompleksnom predstavlenii tenzorov prostranstva Lorentsa]. Изв. Высш. Учебн. Завед., Матем. [Izv. Vyssh. Uchebn. Zaved., Mat.] (1959) No. 1 (8), 156-164 (Mathnet URL: \url{http://mi.mathnet.ru/eng/ivm2415}). [in Russian: On a complex representation of the tensors of Lorentz space]
\bibitem{Norden2}
А.П. Норден [A.P. Norden]: Обобщение основной теоремы теории нормализации. [Obobshenie osnovno\u\i\ teoremy normalizatsii]. Изв. Высш. Учебн. Завед., Матем. [Izv. Vyssh. Uchebn. Zaved., Mat.] (1966) No. 2 (51), 78-82 (Mathnet URL: \url{http://mi.mathnet.ru/eng/ivm2690}). [in Russian: A generalization of the fundamental theorem of the theory of normalization]
\bibitem{Norden3}
А.П. Норден [A.P. Norden]: О структуре связности на многообразии прямых неевклидового пространства. [O structure svyaznosti na mnogoobrazii pryamykh neevklidivogo prostranstva]. Изв. Высш. Учебн. Завед., Матем. [Izv. Vyssh. Uchebn. Zaved. Mat.] (1972) No. 12 (127), 84-94 (Mathnet URL: \url{http://mi.mathnet.ru/eng/ivm4158}). [in Russian: The structure of the connection on a manifold of lines in a non-Euclidean space]
\bibitem{Norden4}
А.П. Норден [A.P. Norden]: Аффинная связность на поверхностях проективного пространства. [Affinnaya svyaznost' na poverkhnostyakh proektivnogo prostranstva]. Мат. Сборник [Mat. Sbornik], N.S. 20(62)(1947)263-281 (Mathnet URL: \url{http://mi.mathnet.ru/eng/msb6216}). [in Russian: Affine connection on the surfaces of a projective space]
\bibitem{Norden5}
А.П. Норден [A.P. Norden]: Теория нормализации и векторные расслоения [Teoriya normalizatsii i vektornye rassloeniya]. Тр. Геом. Семин. [Tr. Geom. Semin.], Казанский университет [Kazanski\u\i\ universitet] (Выпуск [Vypusk]) 9(1976)68-77 (Mathnet URL: \url{http://mi.mathnet.ru/eng/kutgs256}). [in Russian: Normalization theory and vector bundles]
\bibitem{Norden6}
А.П. Норден [A.P. Norden]: Пространства аффинной связности [Prostranstva affinno\u\i\ svyaznosti]. Наука [Nauka], Москва [Moskva], 1976 (EqWorld URL: \url{http://eqworld.ipmnet.ru/ru/library/books/Norden1976ru.djvu}). [in Russian: Affinely Connected Spaces]
\bibitem{Penrose1}
R. Penrose, W. Rindler: Spinors and Space-Time. Vol. 1: Two-Spinor Calculus and Relativistic Fields. Vol. 2: Spinor and Twistor Methods in Space-Time Geometry. Cambridge Monogr. Math. Phys.. Cambridge University Press, Cambridge, Vol. 1: 1984, Vol. 2: 1986. Russian translation: Р. Пенроуз [R. Penrouz], В. Риндлер [V. Rindler]: Спиноры и пространство-время [Spinory i prostranstvo-vremya]. Мир [Mir], Москва [Moskva], т. 1 [Tom 1]: 1987, т. 2 [Tom 2]: 1988. Russian translation by Е.М. Серебрянный [E.M. Serebryanny\u\i] and З.А. Штейнгард [Z.A. Shеу\u\i ngard] under edition Д.М. Гальцов [D.M. Gal'tsov].
\bibitem{Penrose2}
R. Penrose: The twistor programme. Rep. Math. Phys. 12(1977)65-76 (\href{http://dx.doi.org/10.1016/0034-4877(77)90047-7}{DOI: 10.1016/0034-4877(77)90047-7}).
\bibitem{Penrose3}
R. Penrose: Structure of space-time. In: C.M. DeWitt, J.A. Wheeler (Eds.): Battelles Rencontres. 1967 Lectures in Mathematics and Physics. W.A. Benjamin Inc., New York, 1968, Chap. VII, pp. 121-235. Russian translation: Р. Пенроуз [R. Penrouz]: Структура пространства-времени [Struktura prostranstva-vremeni]. Мир [Mir], Москва [Moskva], 1972 (EqWorld URL: \url{http://eqworld.ipmnet.ru/ru/library/books/Penrouz1972ru.djvu}).
\bibitem{Penrose4}
R. Penrose: Spinor classification of energy tensors. In:
В.П. Шелест [V.P. Shelest], А.Е. Левашов [A.E. Levashov], М.Ф. Широков [M.F. Shirokov], К.А. Пирагас [K.A. Piragas] (Eds.):
Гравитация: проблемы, перспективы. Памяти Алексея Зиновьевича Петрова посвящается. [Gravitatsiya: problemy, perspektivy. Pamyati Alekseya Zinov'evicha Petrova posvyashchaetsya] - Gravitation: Problems, Perspectives. The book is dedicated to the memory of Aleksey Zinovievich Petrov. Наукова Думка [Naukova Dumka], Киев [Kiev], 1972, pp. 203-215.
\bibitem{Petrov1}
А.З. Петров [A.Z. Petrov]: Классификация пространств, определяющих поля тяготения [Klassifikatsiya prostranstv, opredelyayushchikh polya tyagoteniya]. Уч. Зап. Казан. Гос. Унив. [Uch. Zap. Kazan. Gos. Univ.] 114(1954) No. 8, 55-69 (Mathnet URL: \url{http://mi.mathnet.ru/uzku344}). English translation: A.Z. Petrov: The classification of spaces defined by gravitational fields. Gen. Rel. Grav. 22(2000)1665-1685 (\href{http://dx.doi.org/10.1023/A:1001910908054}{DOI: 10.1023/A:1001910908054}).
\bibitem{Petrov2}
А.З. Петров [A.Z. Petrov]: Пространства Эйнштейна [Prostranstva \`E\u\i nshte\u\i na]. Физматгиз [Fizmatgiz], Москва [Moskva], 1961 (EqWorld URL: \url{http://eqworld.ipmnet.ru/ru/library/books/Petrov1961ru.djvu}). English translation: A.Z. Petrov: Einstein Spaces. Pergamon Press, Oxford, 1969.
\bibitem{Shevale1}
C. Chevalley: Theory of Lie Groups. I. Princeton Math. Ser., Vol. 8. Princeton University Press, Princeton, 1946. C. Chevalley: Th\'eorie des groupes de Lie.
Tome II: Groupes alg\'ebriques. Actual. Sci. Ind., Vol. 1152. Hermann, Paris, 1951. [in French: Theory of Lie Groups. Volume II: Algebraic Groups] C. Chevalley:
Th\'eorie des groupes de Lie. Tome III: Th\'eor\`emes g\'en\'eraux sur les alg\`ebres de Lie. Actual. Sci. Ind., Vol. 1226. Hermann, Paris, 1955.
[in French: Theory of Lie Groups. 
\nopagebreak[4]Russian edition: К. Шевалле [K. Shevalle]: Теория групп Ли [Teoriya grupp Li]. ИЛ [IL], Москва [Moskva], T. 1 [Tom 1]: 1948
\nopagebreak[4](EqWorld URL: \url{http://eqworld.ipmnet.ru/ru/library/books/Shevalle_t1_1948ru.djvu}), T. 2 [Tom 2]: 1958 (EqWorld URL's:
\nopagebreak[4]\url{http://eqworld.ipmnet.ru/ru/library/books/Shevalle_t2_1958ru.djvu}), T. 3 [Tom 3]: 1958 (EqWorld URL: 
\nopagebreak[4]\url{http://eqworld.ipmnet.ru/ru/library/books/Shevalle_t3_1958ru.djvu}). Russian translation by Л.А. Калужин [L.A. Kaluzhin].
\bibitem{Postnikov1}
М.М. Постников. [M.M. Postnikov] Группы и алгебры Ли [Gruppy i algebry Li]. Наука [Nauka], Москва [Moskva], 1986. English translation: M. Postnikov: Lie Groups and Lie Algebras. Lectures in Geometry, Semester 5. Mir, Moscow, 1986; URSS Publishing, Moscow, 1994. The main ideas of the hypercomplex number constraction on the base of the Bott periodicity are given in the lectures 13-16.
\bibitem{Rosenfeld2}
Б.А. Розенфельд [B.A. Rozenfel'd]: Неевклидовы геометрии [Neevklidovy geometrii]. ГИТТО [GITTO], Москва [Moskva], 1955. [in Russian: Non-Euclidean Geometries]. The Cartan triality principle is given on the page 534.
\bibitem{Rosenfeld1}
Б.А. Розенфельд [B.A. Rozenfel'd]: Многомерные пространства [Mnogomernye prostranstva]. Наука [Nauka], Москва [Moskva], 1966. [in Russian: Multidimensional Spaces]. The m-pair definition is given on the page 384. For this thesis, m=0.
\bibitem{Rosenfeld3}
Хуа Ло-гэн [Khua Lo-g\`en], Б.А. Розенфельд [B.A. Rozenfel'd]: Геометрия прямоугольных матриц и ее применение к вещественной проективной и неевклидовой геометрии.
[Geometriya pryamougol'nykh matrits i ee primenenie k veshchestvenno\u\i\ proektivno\u\i\ i neevklidovo\u\i\ geometrii]. Изв. Высш. Учебн. Завед., Матем. [Izv. Vyssh. Uchebn. Zaved., Matem.] (1957) No. 1, 233-247 (Mathnet URL: \url{http://mi.mathnet.ru/ivm3038}). English translation: Hua Loo-geng (Hua Loo-keng), B.A. Rozenfel'd: The geometry of rectangular matrices and its application to real-projective and non-euclidean geometry. Chin. Math. 8(1966)726-737.
\bibitem{Sintsov1}
Д.М.Синцов [D.M. Sintsov]: Теория коннексов в пространстве в связи с теорией дифференциальных уравнений в частных производных первого порядка. [Teoriya konneksov v prostranstve v svyazi s teorie\u\i\ differentsial'nykh uravneni\u\i\ v chastnykh proizvodnykh pervogo poryadka]. S.n., Казань [Kazan'], 1894. [in Russian: Theory of connexes in space in relation to the theory of first order partial differential equations]
\bibitem{Fadeev1}
Д.К. Фаддеев [D.K. Faddeev]: Лекции по алгебре [Lektsii po algebre]. Наука [Nauka], Москва [Moskva], 1984. [in Russian: Lectures on Algebra]
\bibitem{Hirtseburh1}
F. Hirzebruch: Neue topologische Methoden in der algebraischen Geometrie. Ergeb. Math. Grenzgeb. (2), Vol. 9. Springer-Verlag, Berlin, 1. ed.: 1956, 2. ext. ed.: 1962. English translation of the 2. ext. ed.: Topological Methods in Algebraic Geometry. Grundlehren Math. Wiss., Vol. 131. Springer-Verlag, Berlin, 1966. Russian translation: Ф. Хирцебрух [F. Khirtsebrukh]: Топологические методы в алгебраической геометрии [Topologicheskie metody v algebraichesko\u\i\ geometrii]. Мир [Mir], Москва [Moskva], 1973. Russian translation by Б.Б. Венков [B.B. Venkov].
\bibitem{Hodzh1}
W.V.D. Hodge, D. Pedoe: Methods of Algebraic Geometry, Vol. 2. Cambridge University Press, Cambridge, 1952. Russian edition: В.Д. Ходж [V. D. Khodzh], Д. Пидо [D. Pido]: Методы алгебраической геометрии [Metody algebraichesko\u\i\ geometrii]. Т. 2 [Tom 2]. ИЛ [IL], Москва [Moskva], 1954. Russian translation by А.И. Узков [A.I. Uzkov].
\bibitem{Einshtein1}
J. Stachel (Ed.): The Collected Papers of Albert Einstein. Volume 2: The Swiss Years: Writings, 1900-1909. Princeton University Press, Princeton, 1987. Russian equivalent: А. Эйнштейн [A. \`E\u\i nshte\u\i n]: Сборник научных трудов [Sbornik nauchnykh trudov]. Т.1 [Tom 1]. Наука [Nauka], Москва [Moskva], 1966 (EqWorld URL: \url{http://eqworld.ipmnet.ru/ru/library/books/Einstein_t1_1965ru.djvu}).
\bibitem{Adams1}
J.F. Adams: Spin(8), triality, $F_4$ and all that. In: S.W. Hawking, M. Ro\v{c}ek (Eds.): Superspace and Supergravity. Cambridge University Press, Cambridge, 1981, pp. 435-445.
\bibitem{Brauer1}
R. Brauer  H. Weyl: Spinors in n dimensions. Am. J. Math. 57(1935)425-449 (stable JSTOR URL: \url{http://www.jstor.org/stable/2371218}).
\bibitem{Chevalley1}
C. Chevalley: The Algebraic Theory of Spinors. Columbia University Press, New York, 1954.
\bibitem{Frank1}
F.W. Warner: Foundation of Differentiable Manifolds and Lie Groups. Grad. Texts Math., Vol. 94. Springer, New York, 1983.
\bibitem{Hughston1}
L.P. Hughston: Applications of SO(8) spinors. In: W. Rindler, A. Trautman (Eds.): Gravitation and Geometry, a Volume in Honour of Ivor Robinson. Monogr. Textbook Phys. Sci., Vol. 4. Bibliopolis, Naples, 1987, pp. 253-287.
\bibitem{LeBrun1}
C.R. LeBrun: Ambi-twistors and Einstein's equations. Class. Quantum Grav. 2(1985)555-563 (\href{http://dx.doi.org/10.1088/0264-9381/2/4/020}{DOI: 10.1088/0264-9381/2/4/020}).
\bibitem{Penrose_1}
R. Penrose: Twistor algebra. J. Math. Phys. 8(1967)345-366 (\href{http://dx.doi.org/10.1063/1.1705200}{DOI: 10.1063/1.1705200}).
\bibitem{Penrose_2}
R. Penrose: Twistor theory: its aims and achievements. In: C.J. Isham, R. Penrose, D.W. Sciama (Eds.): Quantum Gravity: An Oxford Symposium held at the Rutherford Laboratory, Chilton, February 15-16, 1974. Clarendon Press, Oxford, 1975, pp. 268-407.
\bibitem{Penrose_3}
R. Penrose: On the  origins of twistor theory. In: W. Rindler, A. Trautman (Eds.): Gravitation and Geometry, a Volume in Honour of Ivor Robinson. Monogr. Textbook Phys. Sci., Vol. 4. Bibliopolis, Naples, 1987, pp. 341-361.
\bibitem{Penrose_4}
R. Penrose: Relativistic symmetry groups. In: A.O. Barut (Ed.): Group Theory in Non-Linear Problems: Lectures presented at the NATO Advanced Study Institute on Mathematical Physics, held in Istanbul, Turkey, August 7-18, 1972. NATO Adv. Study Inst. Ser., Ser. C, Math. Phys. Sci., Vol. 7. Reidel, Dordrecht, 1974, pp. 1-58.
\bibitem{Hochschild1}
G. Hochschild: The Structure of Lie Groups. Holden-Day Ser. Math.. Holden-Day, San Francisco, 1965.
\bibitem{Dirac1}
P.A. Dirac: Wave equations in conformal space. Ann. of Math. (2) 37(1936)429-442 (stable JSTOR URL: \url{http://www.jstor.org/stable/1968455}).
\bibitem{Dirac2}
P.A. Dirac: Relativistic wave equations. Proc. Roy. Soc. London A 155(1936)447-459 (\href{http://dx.doi.org/10.1098/rspa.1936.0111}{DOI: 10.1098/rspa.1936.0111}; stable JSTOR URL: \url{http://www.jstor.org/stable/96758}).
\bibitem{Klein1}
F. Klein: Zur Theorie der Liniencomplexe des ersten und zweiten Grades. Math. Ann. 2(1870)198-226 (\href{http://dx.doi.org/10.1007/BF01444020}{DOI: 10.1007/BF01444020}; Digizeitschriften URL: \url{http://resolver.sub.uni-goettingen.de/purl?GDZPPN002240505}). [in German: On the theory of first and second degree line complexes]
\bibitem{Klein2}
F. Klein: Vorlesungen \"uber h\"ohere Geometrie. Grundlehren Math. Wiss., Vol. 22. Springer-Verlag, Berlin, 1926, pp. 80, 262. Reprinted: Chelsea, New York, 1949, 1957. [in German: Lectures on Higher Geometry]
}

\newpage
\renewcommand{\proofname}{Доказательство.}
\renewcommand{\refname}{Литература}
\def\ptctitle{Содержание}
\def\plttitle{Список таблиц}
\renewcommand{\figurename}{Рис.}
\renewcommand{\tablename}{Таблица}
\newtheorem{theoremr}{Теорема}[section]
\newtheorem{lemmar}{Лемма}[section]
\newtheorem{noter}{Замечание}[section]
\newtheorem{corollaryr}{Следствие}[section]

\def\thetheoremr{\arabic{theoremr}}
\def\thelemmar{\arabic{lemmar}}
\def\thenoter{\arabic{noter}}
\def\thecorollaryr{\arabic{corollaryr}}
\let\Contentsline\contentsline
\def\References#1{{\normalsize\baselineskip=12pt\begin{thebibliography}{99}{#1}\end{thebibliography}}}
\newcounter{dividepage}
\setcounter{dividepage}{-\thepage}
\addtocounter{dividepage}{1}
\newcounter{myfootpage}
\newcounter{mypage}
\makeatletter
\def\addcontentsline#1#2#3{
  \begingroup
    \let\label\@gobble
    \ifx\@currentHref\@empty
      \Hy@Warning{%
        No destination for bookmark of \string\addcontentsline,%
        \MessageBreak destination is added%
      }%
      \phantomsection
    \fi
    \expandafter\ifx\csname toclevel@#2\endcsname\relax
      \begingroup
        \def\Hy@tempa{#1}%
        \ifx\Hy@tempa\Hy@bookmarkstype
          \Hy@WarningNoLine{%
            bookmark level for unknown #2 defaults to 0%
          }%
        \else
          \Hy@Info{bookmark level for unknown #2 defaults to 0}%
        \fi
      \endgroup
      \expandafter\gdef\csname toclevel@#2\endcsname{0}%
    \fi
    \edef\Hy@toclevel{\csname toclevel@#2\endcsname}%
    \Hy@writebookmark{\csname the#2\endcsname}%
      {#3}%
      {\@currentHref}%
      {\Hy@toclevel}%
      {#1}%
    \ifHy@verbose
      \begingroup
        \def\Hy@tempa{#3}%
        \@onelevel@sanitize\Hy@tempa
        \let\temp@online\on@line
        \let\on@line\@empty
        \Hy@Info{%
          bookmark\temp@online:\MessageBreak
          thecounter {\csname the#2\endcsname}\MessageBreak
          text {\Hy@tempa}\MessageBreak
          reference {\@currentHref}\MessageBreak
          toclevel {\Hy@toclevel}\MessageBreak
          type {#1}%
        }%
      \endgroup
    \fi
    \setcounter{mypage}{\thepage}\addtocounter{mypage}{\thedividepage}
    \addtocontents{#1}{%
      \protect\contentsline{#2}{#3}{\thepage(\themypage)}{\@currentHref}%
    }%
  \endgroup
}
\renewcommand{\@oddhead}{\setcounter{myfootpage}{\thepage}\addtocounter{myfootpage}{\thedividepage}}
\renewcommand{\@evenhead}{\setcounter{myfootpage}{\thepage}\addtocounter{myfootpage}{\thedividepage}}
\renewcommand{\@oddfoot}{\thepage\hbox to 170mm{\quad\hrulefill\quad \themyfootpage}}
\renewcommand{\@evenfoot}{\thepage\hbox to 170mm{\quad\hrulefill\quad \themyfootpage}}
\makeatother

\large
\label{originb}
\thispagestyle{empty}
\begin{center}
БАШКИРСКИЙ ГОСУДАРСТВЕННЫЙ УНИВЕРСИТЕТ\\
\rule{12cm}{.4pt} \\
Математический факультет
\end{center}
\vspace{5cm}
\begin{center}
АНДРЕЕВ Константин Васильевич\\
\vspace{1cm}

{\sf СПИНОРНЫЙ ФОРМАЛИЗМ И ГЕОМЕТРИЯ\\
     ШЕСТИМЕРНЫХ РИМАНОВЫХ ПРОСТРАНСТВ}\\

\vspace{1cm}
01.01.04 - дифференциальная геометрия и топология\\
\vspace{1cm}
Диссертация на соискание ученой степени\\
кандидата физико-математических наук
\end{center}
\vspace{1cm}
\begin{flushleft}
\hspace*{5cm}
Научный руководитель: \\
\hspace*{5cm}
кандидат физико - математических наук,\\
\hspace*{5cm}
доцент Э.Г.Нейфельд.
\vspace{5cm}
\end{flushleft}
\begin{center}
УФА - 1997 г.
\end{center}
\pagebreak

\newpage
\part{Русская редакция}
\parttoc
\newpage
\partlot
\newpage
\section{Введение}$ $
    \indent Предлагаемая диссертационная работа является теоретическим
    исследованием по геометрии 6-мерных римановых пространств
    и посвящена вопросам, связанным с этой геометрией.\\
    \indent Изучение 6-мерных римановых пространств производится с
    помощью соответствующего 6-мерного спинорного формализма
    \cite{Cartan1r}, \cite{Brauer_1r}, \cite{Penrose1r}   и
    теории нормализации Нордена-Нейфельда
    \cite{Neifeld1r}-\cite{Neifeld7r},
    \cite{Norden1r}-\cite{Norden6r},  что
    позволяет упростить важные соотношения, записанные в тензорном виде,
    и приводит к оригинальным результатам.\\
    \indent Выбор темы обусловлен возросшим в последнее время интересом к
    таким пространствам. Они  естественным образом появляются,
    например, в спинорно-твисторном формализме Пенроуза
    \cite{Penrose1r}-\cite{Penrose4r},
    \cite{Penrose_1r}-\cite{Penrose_4r}. Здесь
    важную роль играет псевдоевклидово пространство $\mathbb R^6_{(2,4)}$,
    изотропный конус которого позволяет определить
    конформно-псевдоевклидово пространство Минковского. Более того,
    твисторы в теории Пенроуза будут представлять спиноры, согласованные с
    пространством $\mathbb R^6_{(2,4)}$. Однако, если в монографии \cite{Penrose1r} твисторы
    образуют 4-мерное комплексное векторное расслоение
    с базой - 4-мерным действительным многообразием, то в данной
    работе базой служит 6-мерное аналитическое комплексное
    риманово пространство $\mathbb CV^6$. Это приводит к новым
    результатам в твисторной теории. Конформно-(псевдо-)евклидово
    (псевдо-)риманово пространство в этом случае связывается с комплексной
    аналитической квадрикой $\mathbb CQ_6$
    \cite{Hirtseburh1r}, \cite{Hodzh1r}, что приводит к изучению
    свойств группы $SO(8,\mathbb C)$ \cite{Adams_1r}, \cite{Hughston_1r},
    а следовательно и к принципу тройственности
    Э. Картана \cite{Cartan1r}. Указанные комплексно-евклидовы геометрии
    в данном случае появляются как внутренние геометрии
    этой нормализованной квадрики. Выписывая деривационные уравнения
    для такой квадрики \cite{Neifeld4r},
    можно определить инвариантное при
    конформных преобразованиях, а следовательно и при
    замене нормализации, уравнение, которое по аналогии
    с твисторным уравнением Пенроуза назовем
    битвисторным уравнением. Решения этого уравнения образуют
    пары, которые можно интерпретировать как нуль-пары
    Розенфельда \cite{Rosenfeld1r},
    что приводит к 6-мерной квадрике и принципу
    тройственности Э. Картана. Целесообразно рассматривать
    следующие три диффеоморфных между собой многообразия:
\begin{enumerate}
    \item многообразие точек квадрики $\mathbb CQ_6$;
    \item многообразие плоских образующих $\mathbb C\mathbb P_3$ квадрики
    $\mathbb CQ_6$ максимальной размерности I семейства;
    \item многообразие плоских образующих $\mathbb C\mathbb P_3$ квадрики
    $\mathbb CQ_6$ максимальной размерности II семейства.
\end{enumerate}
    Нормализация этих многообразий позволяет рассматривать
    конформно-(псевдо-)евклидовы связности на этих многообразиях,
    которые будут вейлевыми. Это приводит к обобщению
    принципа тройственности на B-пространства в терминологии
    Нордена.\\
    \indent Итак, 6-мерный спинорный формализм основан
    на работах Э. Картана \cite{Cartan1r} и Брауера
    \cite{Brauer_1r}. 4-мерный спинорный формализм
    и твисторная алгебра описаны в работах Пенроуза
    \cite{Penrose1r}-\cite{Penrose4r},
    \cite{Penrose_1r}-\cite{Penrose_4r}. Связанные
    с этими формализмами изоморфизмы групп и алгебр Ли рассмотрены
    в работах \cite{Besse1r}, \cite{Shevale1r}, \cite{Postnikov1r},
    \cite{Hochschild_1r}. Кроме того, сведения по
    клиффордовым алгебрам и октавам взяты из
    \cite{Postnikov1r} и \cite{Penrose1r}. Сведения о квадриках
    и их плоских образующих приведены в работах
    \cite{Hirtseburh1r} и \cite{Hodzh1r}. Нормализация многообразия
    плоских образующих происходит также, как описано в
    работах \cite{Neifeld4r}-\cite{Neifeld6r},
    \cite{Norden2r}, \cite{Norden4r}, \cite{Norden6r}.
    Связности в расслоениях вводятся согласно
    \cite{Neifeld6r}, \cite{Norden3r}, \cite{Norden5r}, \cite{Norden6r}.
    Вложения действительных пространств в комплексное
    рассмотрены в работе \cite{Neifeld1r}. Действительное и комплексное представления
    римановых пространств проводится согласно \cite{Koboyasi1r},
    \cite{Koboyasi2r} и \cite{Lichnerowicz1r}. Нуль-пары
    Розенфельда взяты из работы \cite{Rosenfeld1r}.
    Соответствие Кляйна приведено в \cite{Klein_1r} и \cite{Klein_2r}.
    О физических приложениях твисторов можно посмотреть
    в работах \cite{Penrose1r}, \cite{Dirac_1r} и \cite{Dirac_2r}.\\
    \indent Рассмотрим основное содержание по главам.
    Для этого необходимо предварительно сделать некоторые определения.

\subsection{Основные определения}$ $
    \indent Эти определения введены согласно работам \cite{Neifeld1r}-\cite{Neifeld7r}.
    Отметим, что все функции, участвующие в построениях
    предполагаются достаточно гладкими. Все определения, утверждения
    и построения носят локальный характер.
    Под комплексным аналитическим римановым пространством
    $\mathbb CV^n$ в дальнейшем будем понимать аналитическое комплексное
    многообразие, снабженное аналитической квадратичной метрикой,
    т.е. метрикой, определенной с помощью симметрического
    невырожденного тензора $g_{\alpha\beta}$, координаты
    которого - аналитические функции координат точки.
    Этому тензору соответствует комплексная риманова связность
    без кручения, коэффициенты которой определяются символами Кристофеля
    и поэтому являются аналитическими функциями.\\
    \indent Касательное расслоение этого многообразия $\tau^\mathbb C(\mathbb CV^n)$ имеет
    слои $\tau_x^\mathbb C\cong \mathbb C\mathbb R^n$, то есть слои, изоморфные n-мерному комплексному
    евклидовому пространству, метрика которого определяется
    значением евклидового метрического тензора в данной точке $x$. Пусть n=6.
    Обозначим через $\Lambda^2\mathbb C^4$ пространство бивекторов
    пространства $\mathbb C^4$, а через $\Lambda$ - соответствующее расслоение
    с базой $\mathbb CV^6$ и слоями, которые изоморфны $\mathbb C\mathbb R^6$ ($\mathbb C\mathbb R^6\cong \Lambda^2\mathbb C^4$).
    Отсюда следует, что в 6-мерном случае комплексное риманово пространство
    $\mathbb CV^6$ будет базой расслоения $A^\mathbb C=\mathbb C^4(\mathbb CV^6)$. При
    этом каноническая проекция $p:\mathbb C^4_x\longmapsto x\in \mathbb CV^6$
    отображает слой $\mathbb C^4_x$ в точку $x$ базы.\\
    \indent Вещественное (псевдо-)риманово пространство $V^n_{(p,q)}$ будем
    рассматривать как поверхность вещественной размерности
    n  в пространстве $\mathbb CV^n$, т.е. локально определять
    с помощью параметрического уравнения
\begin{equation}
\label{re1v}
    z^\alpha=z^\alpha(u^i)\ (\alpha,\beta ,... ,i,j,g,h=\overline{1,n})\sps
\end{equation}
    где $z^\alpha$ - комплексные координаты точки $x$ базы,
    а $u^i$ - параметры: локальные координаты точки
    пространства $V^n_{(p,q)}$. Частные производные $(\partial_i z^\alpha =:
    H_i{}^\alpha)$ определяют вложение вещественного
    касательного пространства $\tau_x^\mathbb R(V^n_{(p,q)})\cong \mathbb R^n_{(p,q)}$ поверхности (\ref{re1v}) в
    комплексное касательное пространство $\tau_x^\mathbb C$
    следующим образом
\begin{equation}
\label{re2v}
    H:\tau_x^\mathbb R\longmapsto\tau_x^\mathbb C\sps
\end{equation}
\begin{equation}
\label{re3v}
\begin{array}{c}
    z^\alpha=z^\alpha(u^i(t))\sps V^\alpha:=\frac{dz^\alpha}{dt}=
    H_i{}^\alpha\frac{du^i}{dt}=:H_i{}^\alpha v^i\sps\\[2ex]
    \frac{du^i}{dt}\in\tau_x^\mathbb R\longmapsto \frac{dz^\alpha}{dt}\in\tau_x^\mathbb C\sps
\end{array}
\end{equation}
    где дифференцирование ведется вдоль вещественной кривой $\gamma (t)$
    поверхности (\ref{re1v}).
    Так как матрица $\parallel H_i{}^\alpha \parallel$ есть невырожденная якобиева
    матрица, то существует оператор $H^i{}_\alpha$ такой, что
\begin{equation}
\label{re4v}
    \left\{
\begin{array}{l}
    H^i{}_\alpha H_i{}^\beta=\delta_\alpha{}^\beta\sps\\
    H^i{}_\alpha H_j{}^\alpha=\delta_j{}^i\spsd
\end{array}
    \right.
\end{equation}
    Отсюда следует, что оператор $H_i{}^\alpha$ определяет
    в комплексном пространстве инволюцию
\begin{equation}
\label{re5v}
    S_\alpha {}^{\beta '}=H^i{}_\alpha\bar  H_i{}^{\beta '}\sps
\end{equation}
    где координаты $\bar H_i{}^{\beta '}$ комплексно
    сопряжены координатам $H_i{}^\beta$ \cite{Neifeld1r}.
    Поэтому
\begin{equation}
\label{re6v}
    v^i=H^i{}_\alpha V^\alpha=\overline{H^i{}_\alpha V^\alpha}
    \ \ \Rightarrow \ \ S_\alpha{}^{\beta '}V^\alpha=\bar V^{\beta '}\spsd
\end{equation}
    Это есть необходимое и достаточное условие того, что
    вектор $V^\alpha\in \tau^\mathbb C_x$ будет вещественным.
    При этом
\begin{equation}
\label{re7v}
    S_\alpha{}^{\beta '}\bar S_{\beta '}{}^\gamma=\delta_\gamma{}^\alpha\spsd
\end{equation}
    Метрику в $\tau_x^\mathbb R(V^n_{(p,q)})$ определим
    условием
\begin{equation}
\label{re8v}
    g_{\alpha\beta}V^\alpha V^\beta=\overline{g_{\alpha\beta}V^\alpha V^\beta}\sps
    \forall \bar V^{\beta '}=S_\alpha{}^{\beta '} V^\alpha\spsd
\end{equation}
    Это означает, что вещественный тензор пространства $\tau_x^\mathbb R(V^n_{(p,q)})$ определяется
    как тензор, самосопряженный относительно указанной эрмитовой инволюции
\begin{equation}
\label{re9v}
    g_{\alpha\beta}=S_\alpha{}^{\gamma '}S_\beta{}^{\delta '}
    \bar g_{\gamma '\delta '}\spsd
\end{equation}
    Поэтому, тензор
\begin{equation}
\label{re10v}
    g_{ij}:=H_i{}^\alpha H_j{}^\beta g_{\alpha\beta}=
    \overline{H_i{}^\alpha H_j{}^\beta g_{\alpha\beta}}
\end{equation}
    будет метрическим тензором $\tau_x^\mathbb R(V^n_{(p,q)})\subset \tau_x^\mathbb C(\mathbb CV^n)$. Вид метрики
    $g_{ij}$  существенно зависит от структуры оператора
    $H_i{}^\alpha$ и следовательно от тензора инволюции
    $S_\alpha{}^{\beta '}$.
    Комплексная риманова связность пространства $\mathbb CV^n$
    будет индуцировать на вещественной поверхности связность вида
\begin{equation}
\label{re11v}
    \nabla_i:=H_i{}^\alpha\nabla_\alpha
\end{equation}
    такую, что
\begin{equation}
\label{re12v}
     \nabla_i H_j{}^\alpha=i\ b_{ij}{}^gH_g{}^\alpha\sps
     \nabla_i g_{jg}=2ib_{i(jg)}\sps b_{ijg}:=b_{ij}{}^hg_{hg}\spsd
\end{equation}
     Потребуем, чтобы индуцированная связность была римановой, тогда
\begin{equation}
\label{re13v}
     b_{ijh}=b_{jih}\sps b_{ijh}=-b_{ihj}\ \ \Rightarrow\ \ \ b_{ijh}=0
\end{equation}
     и следовательно
\begin{equation}
\label{re14v}
     \nabla_iS_\alpha{}^{\beta '}=0\ \ \Rightarrow\ \ \
     \nabla_\gamma S_\alpha{}^{\beta '}=0\spsd
\end{equation}\\

     При n=6 сужение базы $\mathbb CV^6\longmapsto V^6_{(p,q)}$ позволяет
     определить расслоение $A^\mathbb C(S)=\mathbb C^4(S)(V^6_{(p,q)})$, слои которого
     должны быть снабжены некоторой дополнительной структурой $s$.
     Ниже будет показано, что эта структура в случае
     метрики четного индекса (т.е. количество минусов
     четно) определяется эрмитово-симметричным тензором,
     а в случае метрики нечетного индекса структура
     будет определена эрмитовой инволюцией.
     Для псевдориманова пространства четного индекса 4
     расслоение $A^\mathbb C(S)$ назовем твисторным, так как
     его слои будут изоморфны векторному пространству
     $\mathbb T$ \cite{Penrose1r}.

\subsection{Вторая глава}$ $
     \indent Эта часть основана на работах \cite{Penrose1r}-\cite{Penrose4r},
     где развит 4-мерный спинорный формализм. В \cite{Norden1r}
     введены связующие операторы $\eta_\alpha{}^{ab}$.
     Данные по группам Ли взяты из работ \cite{Besse1r},
     \cite{Shevale1r}, \cite{Postnikov1r}.
     6-мерный спинорный формализм, построенный в этой главе,
     основан на следующих 3 изоморфизмах:
\begin{enumerate}
     \item изоморфизм пространств $\mathbb C\mathbb R^6\cong\Lambda^2\mathbb C^4$;
     \item изоморфизм групп $SO(6,\mathbb C)\cong SL(4,\mathbb C)/\{\pm 1\}$;
     \item изоморфизм алгебр Ли $so(6,\mathbb C)\cong sl(4,\mathbb C)$.
\end{enumerate}
     В явном виде эти изоморфизмы описываются так:
\begin{enumerate}
     \item $r^\alpha=\frac{1}{2}\eta^\alpha{}_{ab}R^{ab}$,
     где $r^\alpha$ - координаты вектора в $\mathbb C\mathbb R^6$,
     $R^{ab}$ - координаты бивектора в $\Lambda^2\mathbb C^4$,
     а $\eta^\alpha{}_{ab}$ - координаты связующих операторов Нордена;
     \item $K_\alpha{}^\beta =\frac{1}{4} \eta^\beta{}_{ab}
     \eta_\alpha{}^{cd}\cdot 2 S_c{}^aS_d{}^b$, где
     $K_\alpha{}^\beta$ - координаты преобразования из группы
     $SO(6,\mathbb C)$, а $S_a{}^b$ - координаты преобразования из группы $SL(4,\mathbb C)$;
     \item $T^{\alpha\beta}=A^{\alpha\beta}{}_a{}^bT_b{}^a$,
     где $T^{\alpha\beta}$ - координаты бивектора пространства $\Lambda^2 \mathbb C\mathbb R^6$,
     а $T_a{}^b$ - координаты бесследного оператора пространства $\mathbb C^4$.
\end{enumerate}
     При этом связующие операторы удовлетворяют соотношениям
\begin{equation}
\label{re1r}
    g^{\alpha\beta}=1/4\cdot \eta^{\alpha}{}_{aa_1} \eta^{\beta}{}_{bb_1}
    \varepsilon^{aa_1bb_1}\sps
    \varepsilon^{aa_1bb_1}=
    \eta_{\alpha}{}^{aa_1} \eta_{\beta}{}^{bb_1}g^{\alpha\beta}\sps
\end{equation}
    где $(\alpha,\beta,...=1,2,3,4,5,6;$
    $a_1,b_1,a,b,e,f,k,l,m,n,...=1,2,3,4;$ $i,j,g,h=1,2,3,4,5,6)$.
    Из этого следует, что операторы Нордена удовлетворяют уравнению
    Клиффорда и определяют некоторую полную
    клиффордову алгебру, которая реализуется с помощью алгебры
    матриц размерности $8\times 8$.\\
    \indent Рассматривая далее вложение $\mathbb R^6_{(p,q)}\subset \mathbb C\mathbb R^6$,
    можно получить разложения и для тензора эрмитовой инволюции
    $S_\alpha{}^{\beta '}(=\frac{1}{4}\eta_\alpha{}^{aa_1}\bar\eta^{\beta '}{}_{b'b_1 '}
    s_{aa_1}{}^{b'b_1 '})$
\begin{equation}
\label{re2r}
    s_{aa_1}{}^{b'b'{}_1}=s^{cb'}s^{c_1b'{}_1} \varepsilon_{cc_1aa_1}
    \sps \bar s_{a'b}=\pm s_{ba'}
\end{equation}
    в случае метрики четного индекса инерции q и
\begin{equation}
\label{re3r}
    s_{aa_1}{}^{b'b'{}_1}=2s_{\left[ a \right.}{}^{b'}s_{\left. a_1\right]}
    {}^{b'{}_1}\sps s_a{}_{b '}\bar s_{b '}{}^c=\pm \delta_a{}^c
\end{equation}
    в случае метрики нечетного индекса.
    В качестве показательного примера в специальном базисе рассматривается
    вложение $\mathbb R^6_{(2,4)}\subset \mathbb C\mathbb R^6$.\\
    В последнем параграфе данной главы вводятся обобщенные операторы Нордена,
    как аналитические функции координат точки $z^\gamma$ так, что
    выполнено
\begin{equation}
\label{re4r}
    g^{\alpha\beta}(z^\gamma)=1/4\cdot \eta^{\alpha}{}_{aa_1}(z^\gamma)
    \eta^{\beta}{}_{bb_1}(z^\gamma)\varepsilon^{aa_1bb_1}(z^\gamma)\spsd
\end{equation}
    На этом заканчивается построение необходимого спинорного
    формализма для пространства $\mathbb CV^6$.\\
    \indent Следует отметить, что указанный спинорный формализм
    во многом сходен с 4-мерным спинорным формализмом Пенроуза,
    при котором риманово пространство $V^4$ есть база расслоений
    $\mathbb C^2(V^4_{(1,3)})$ и $\mathbb C^4(V^4_{(1,3)})$. Векторы пространства $\mathbb C^2(V^4_{(1,3)})$
    называются у Пенроуза спинорами, а векторы $\mathbb C^4(V^4_{(1,3)})$ -
    твисторами. Основной отличительной чертой твисторов данной
    диссертации и является то, что твистором называется вектор расслоения
    $\mathbb C^4(V^6_{(2,4)})$, и такое истолкование помогает получить новые результаты. Это может
    привести к новой трактовке физической интерпретации твисторов,
    изложенной в монографии \cite{Penrose1r}.

\subsection{Третья глава}$ $
    \indent Третья глава посвящена введению связности в расслоениях
    с комплексной базой $\mathbb CV^6$.
    Введение связности осуществляется согласно \cite{Neifeld6r},
    \cite{Norden3r}, \cite{Penrose1r}. Комплексное и действительное представления взяты
    из \cite{Koboyasi1r}, \cite{Lichnerowicz1r}.
    В качестве базы задается комплексно-аналитическое
    риманово пространство $\mathbb CV^6$. Для этого рассматривается
    комплексная аналитическая квадрика $\mathbb CQ_6$, вложенная
    в проективное пространство $\mathbb CP_7$
\begin{equation}
\label{re5r}
    G_{AB}X^AX^B=0\sps
\end{equation}
     где $(A,B,...=\overline{1,8})$. Многообразие плоских
     образующих максимального ранга $\mathbb C\mathbb P_3$ какого-нибудь
     семейства (их как известно два) - комплексно шестимерно. Далее рассматривается
     гармоническая нормализация такого семейства, которая в локальных
     координатах имеет вид
\begin{equation}
\label{re6r}
     X^a=X^a(u^\Lambda)\sps Y_b=Y_b(u^\Lambda)\sps
\end{equation}
     где $u^\Lambda$ - двенадцать вещественных параметров
     $(\Lambda ,\Psi ,...=\overline{1,12})$. Первые деривационные
     уравнения этого нормализованного семейства имеют вид
\begin{equation}
\label{re7r}
     \nabla_\Lambda X_a=Y^bM_{\Lambda ab}\sps
     M_{\Lambda ab}=-M_{\Lambda ba}\spsd
\end{equation}
     Для переброски парных индексов используется квадривектор
     $\varepsilon_{abcd}$, кососимметричный по всем
     своим индексам
\begin{equation}
\label{re8r}
     M_\Lambda{}^{ab}=\frac{1}{2}M_{\Lambda cd}\varepsilon^{abcd}\spsd
\end{equation}
     Кроме того, операторы $M_\Lambda{}^{ab}$ каждому вектору
     базы ставит в соответствие бивектор из $\Lambda^2\mathbb C^4$
\begin{equation}
\label{re9r}
     V^{ab}=M_\Lambda{}^{ab}V^\Lambda\spsd
\end{equation}
     Это определит метрический тензор в касательном расслоении
\begin{equation}
\label{re10r}
       G_{\Lambda\Psi}=\frac{1}{4}
       (M_\Lambda{}^{ab}M_\Psi{}^{cd}\varepsilon_{abcd}+
       \bar M_\Lambda{}^{a'b'}\bar M_\Psi{}^{c'd'}\varepsilon_{a'b'c'd'})\spsd
\end{equation}
     Таким образом, база - многообразие плоских образующих - превратится
     в 12-мерное риманово пространство $V^{12}_{(6,6)}$ с метрическим тензором
     $G_{\Lambda\Psi}$ с заданной на нем комплексной структурой
     $f_\Lambda{}^\Psi$, удовлетворяющей следующему соотношению
\begin{equation}
\label{re11r}
       \bigtriangleup_\Lambda{}^\Psi=\frac{1}{2}
       (\delta_\Lambda{}^\Psi+if_\Lambda{}^\Psi)=\frac{1}{2}
       M_{\Lambda ab} M^{\Psi ab}\spsd
\end{equation}
     В качестве слоев расслоения $A^\mathbb C$ рассматривается пространство
     $\mathbb C^4$, определенное 4 базисными точками $X_a$ плоской образующей.
     Комплексной реализацией  пространства $V^{12}_{(6,6)}$  является пространство
     $\mathbb CV^6$ так, что отображение касательных пространств  происходит
     с помощью операторов Нейфельда $m_\alpha{}^\Lambda$.
     Коэффициенты связности
     определятся в этом случае через уравнения
\begin{equation}
\label{re12r}
     \nabla_\alpha m_\alpha{}^\Lambda=0\sps
     \bar\nabla_{\alpha '}\bar m_{\alpha '}{}^\Lambda=0\spsd
\end{equation}
     При этом за комплексную ковариантную производную
     можно принять производную такого вида
\begin{equation}
\label{re13r}
     \nabla_\alpha=m_\alpha{}^\Lambda\nabla_\Lambda\sps
     \bar\nabla_{\alpha '}=\bar m_{\alpha '}{}^\Lambda\nabla_\Lambda\spsd
\end{equation}
     Затем в этой главе устанавливаются свойства римановой связности без кручения,
     продолженной на слои расслоения $A^\mathbb C$. Оказывается, это продолжение
     единственно и задается требованием ковариантного постоянства
     квадривектора $\varepsilon_{abcd}$. После этого с помощью
     операторов вложения удается перейти к вещественной связности,
     но при этом необходимо потребовать ковариантное постоянство эрмитовой
     инволюции.\\
     \indent Все это позволяет рассмотреть конформно-инвариантное
     битвисторное уравнение
\begin{equation}
\label{re14r}
       \nabla^{c\left(\right.d}X^{a\left.\right)}=0\spsd
\end{equation}
     Его решения связываются с нуль парами Розенфельда,
     которые играют важную роль в дальнейших исследованиях.

\subsection{Четвертая глава}$ $
     \indent Четвертая глава посвящена проблеме классификации
     тензора кривизны 6-мерных (псевдо-)римановых пространств и
     свойствам его бивекторов.\\
     \indent Показывается, как, пользуясь определенным в первой главе
     спинорным формализмом, упростить тензорную запись
     основных тождеств для тензора кривизны. Утверждается,
     что классификацию такого тензора можно
     свести к классификации тензора пространства $\mathbb  C^4$ такого, что
\begin{equation}
\label{re15r}
     R_{\alpha\beta\gamma\delta}=
     A_{\alpha\beta a}{}^b A_{\gamma\delta c}{}^d R_b{}^a{}_d{}^c\spsd
\end{equation}
     При этом тождество Бианки, которым удовлетворяет тензор кривизны
\begin{equation}
\label{re16r}
    R_{\alpha \beta \gamma \delta}+R_{\alpha \delta \beta \gamma}+
    R_{\alpha \gamma \delta \beta}=0\sps
\end{equation}
    будет иметь вид
\begin{equation}
\label{re17r}
     R_l{}^d{}_s{}^l=-\frac{1}{8}R\delta_s{}^d\spsd
\end{equation}
    Как видно, вместо 105 уравнений из (\ref{re16r}) можно
    рассматривать всего 16 уравнений из (\ref{re17r}), из
    которых 15 будут существенными. Можно построить таким
    же образом спинорный аналог тензора Вейля
\begin{equation}
\label{re18r}
     C_{\alpha\beta\gamma\delta}=
     A_{\alpha\beta a}{}^b A_{\gamma\delta c}{}^d C_b{}^a{}_d{}^c\spsd
\end{equation}

    Интересен, в качестве следствия из данной теоремы, следующий факт.
    Произвольный простой изотропный бивектор пространства $\Lambda^2\mathbb C^6$
    определит с точностью до множителя вектор пространства $\mathbb C^4$.
    Это даст возможность в пространстве $\mathbb R^6_{(2,4)}$ построить
    геометрическую интерпретацию твистора, во многом схожую с интерпретацией
    спинора Пенроуза в пространстве $\mathbb R^4_{(1,3)}$ - флаг, составленный
    из флагштока и полотнища флага.\\
    \indent И, наконец, утверждается, что в пространстве $\mathbb R^6_{(p,q)}$ четного
    индекса q  любой бивектор может быть приведен к каноническому виду
    в некоторым базисе
\begin{equation}
\label{re19r}
    \frac{1}{2}R_{\alpha\beta}X^\alpha Y^\beta=
    R_{16}X^{\left[\right.1} Y^{6\left.\right]}+
    R_{23}X^{\left[\right.2} Y^{3\left.\right]}+
    R_{45}X^{\left[\right.4} Y^{5\left.\right]}\spsd
\end{equation}

\subsection{Пятая глава}$ $
   \indent В последней главе, используя вектор $X^a$ и ковектор $Y_b$ из
   слоев расслоения $A^\mathbb C$ и ему дуального $A^{\mathbb C*}$, строится 8-мерное комплексное пространство
   $\mathbb T^2$ как прямая сумма $\mathbb C^4\oplus \mathbb C^{*4}$
\begin{equation}
\label{re20r}
    X^A:=(X^a,Y_b)\sps
\end{equation}
    и $X^A\in \mathbb T^2$. При этом $X^a$ и $Y_b$ удовлетворяют
    следующей системе
\begin{equation}
\label{re21r}
    \left\{
    \begin{array}{lcl}
    X^a & = & \dot X^a -ir^{ab}Y_b\sps \\
    Y_b & = & \dot Y_b\sps
    \end{array}
    \right.
\end{equation}
    где $r^{ab}$ - координаты бивектора пространства $\mathbb C\mathbb R^6$,
    а $\dot X^a$, $\dot Y_b$ - значения $X^a$, $Y_b$ в некоторой
    точке O. На самом деле, систему (\ref{re21r}) можно рассматривать
    как решения битвисторного уравнения, а $\dot X^a$, $\dot Y_b$
    будут его частными решениями. Рассматривая г.м.т., для которых
    $X^a=0$, можно прийти к нуль-парам Розенфельда и сформулировать
    следующее утверждение.
\begin{theoremr}
     (Принцип тройственности для двух B- цилиндров).\\
     В проективном пространстве $\mathbb C\mathbb P_7$ существуют
     две квадрики (два B - цилиндра), обладающие
     следующими общими свойствами:
     \begin{enumerate}
     \item Плоская образующая $\mathbb C\mathbb P_3$ одной квадрики
           взаимооднозначно определит точку R другой.
     \item Плоская образующая $\mathbb C\mathbb P_2$ одной квадрики
           однозначно определит точку R другой. Но точке
           R можно сопоставить многообразие плоских образующих
           $\mathbb C\mathbb P_2$, принадлежащих одной плоской образующей $\mathbb C\mathbb P_3$
           второй квадрики.
     \item Прямолинейная образующая $\mathbb C\mathbb P_1$ одной квадрики
           взаимооднозначно определит прямолинейную образующую
           $\mathbb C\mathbb P_1$ из другой. Причем все прямолинейные образующие,
           принадлежащие одной плоской образующей $\mathbb C\mathbb P_3$ первой квадрики,
           определят пучок с центром в точке R, принадлежащий
           второй квадрике.
     \end{enumerate}
\end{theoremr}
      Это позволяет ввести операторы $\eta^A{}_{KL}$ такие, что
\begin{equation}
\label{re22r}
      r^A=\frac{1}{4}\eta^A{}_{KL}R^{KL}\sps
\end{equation}
      где $r^A$ - координаты вектора пространства $\mathbb C\mathbb R^8$, а
      $R^{KL}$ - координаты некоторого тензора пространства $\mathbb C\mathbb R^8$.
      При этом связующие операторы $\eta^A{}_{KL}$ определят
      некоторую полную алгебру Клиффорда, поскольку будут
      удовлетворять клиффордову уравнению
\begin{equation}
\label{re23r}
      G_{AB}\delta_K{}^L=\eta_{AK}{}^R\eta_B{}^L{}_R+\eta_{BK}{}^R\eta_A{}^L{}_R\spsd
\end{equation}
      В этом случае у нас будет пара метрических тензоров
\begin{equation}
\label{re24r}
      \begin{array}{c}
      \varepsilon_{KLMN}=\eta^A{}_{KL}\eta_{AMN}\sps\\[2ex]
      G_{AB}=\frac{1}{4}\eta_A{}^{KL}\eta_B{}^{MN}\varepsilon_{KM}\varepsilon_{LN}\sps
      \varepsilon_{(KL)(MN)}=\frac{1}{2}\varepsilon_{KL}\varepsilon_{MN}\spsd
      \end{array}
\end{equation}
      С помощью первого тензора можно поднимать  и опускать
      парные индексы, а с помощью второго проделывать указанную
      операцию с одиночными индексами. Это накладывает
      жесткие условия на связующие операторы, например
\begin{equation}
\label{re25r}
      \eta^A{}_{(MN)}=\frac{1}{8}\eta^A{}_{KL}\varepsilon^{KL}\varepsilon_{MN}\spsd
\end{equation}
      Такие связующие операторы будут определять структурные константы
      алгебры октав и приведут к двулистному накрытию \linebreak[4]
      $Spin(8,\mathbb C)/\{\pm 1\}\cong SO(8,\mathbb C)$. Поэтому операторы
      $\eta^A{}_{KL}$ во многом схожи с операторами Нордена
      $\eta^\alpha{}_{kl}$ по своим свойствам.\\

\subsection{Заключение}$ $
      \indent Следует отметить, что в конце диссертации имеется Приложение,
      в котором приведены все необходимые алгебраические выкладки.\\
      \indent Основные результаты диссертации опубликованы в открытой печати:
\begin{enumerate}
      \item "О бивекторах 6-мерных римановых пространств".
            УТИС, Уфа, 1996, с. 59-61;
      \item "О структуре тензора кривизны 6-мерных римановых пространств".
             Вестник БГУ, Уфа, N2(I), 1996, с. 44-47;
      \item "О твисторных расслоениях с 6-мерной базой".
            МГС, Казань, 1997, c. 13;
      \item "О геометрии битвисторов".
            РКСА, Уфа, 1997, с. 85-87.
\end{enumerate}
      и докладывались на конференциях:
\begin{enumerate}
      \item "Ленинские горы - 95", г. Москва;
      \item "Чебышевские чтения - 96", г. Москва;
      \item "Лобачевские чтения - 97", г. Казань;
      \item многочисленные конференции в г. Уфе и
      семинары, проходившие в г. Казани (кафедра геометрии КГУ).
\end{enumerate}

      \indent Автор выражает благодарность за помощь в подготовке диссертации
      своему научному руководителю доц. Э.Г. Нейфельду и кафедре
      геометрии КГУ (зав.кафедрой проф. Б.Н. Шапуков).

\newpage
\section{Основные тождества и формулы}$ $
    \indent Эта глава посвящена изучению алгебраических свойств накрытия
    $$
    SO(6,\mathbb C)\cong SL(4,\mathbb C)/{\{\pm 1\}}\spsd
    $$
    На основе этого изоморфизма строится элементарная
    алгебраическая база, необходимая для дальнейших построений.
    Для этого рассматриваются различные векторные расслоения
    с базой $\mathbb CV^6(\mathbb C\mathbb R^6)$. Касательное расслоение $\tau^\mathbb C(\mathbb CV^6)$, содержащее слои,
    изоморфные $\mathbb C\mathbb R^6$, будет изоморфно расслоению $\Lambda$
    со слоями, изоморфными $\Lambda^2\mathbb C^4$, что следует из
    существования связующих операторов Нордена
    $$
    r^\alpha=\frac{1}{2}\eta^\alpha{}_{aa_1}R^{aa_1}\sps
    $$
    где $(\alpha,\beta,...=1,2,3,4,5,6;a_1,b_1,a,b,...=1,2,3,4)$.
    Кроме того, рассматривается расслоение $A^\mathbb C$ со слоями, изоморфными
    $\mathbb C^4$, и базой $\mathbb CV^6(\mathbb C\mathbb R^6)$. Отсюда будет следовать существование
    таких операторов $A_{\alpha\beta a}{}^b$, что
    $$
    T_{\alpha \beta}= A_{\alpha \beta d}{}^cT_c{}^d \sps
    T_a{}^a=0 \sps T_{\alpha \beta}=-T_{\alpha \beta}\spsd
    $$
    Как следует из результатов предпоследнего пункта этой главы,
    рассматривающего инфинитизимальные преобразования, построенные
    операторы являются алгебраической реализацией изоморфизма
    алгебр Ли
    $$
    so(6,\mathbb C)\cong sl(4,\mathbb C)\spsd
    $$
    Затем исследуются вещественные вложения для указанных изоморфизмов
    с помощью оператора вложения $H_i{}^\alpha$ и инволюции $S_\alpha{}^{\beta '}$. При этом операция
    сопряжения, индуцируемая в расслоении $A^\mathbb C$, разбивается
    на два класса. В первом случае (пространство $V^6_{(p,q)}(\mathbb R^6_{(p,q)})$
    имеет метрику четного индекса q) сопряжение осуществляется с помощью
    тензора эрмитового поляритета $s^{aa'}$
    $$
    \bar X^{a'}:=s^{aa'}X_a\spsd
    $$
    Во втором случае, когда q - нечетно, - с помощью тензора эрмитовой
    инволюции $s_a{}^{a'}$
    $$
    \bar X^{a'}:=s_a{}^{a'}X^a\spsd
    $$
    Второй пункт как раз и посвящен выяснению этого факта,
    который доказывается с использованием теорем из монографии
    \cite{Penrose1r}.

\subsection{\texorpdfstring{Бивектора пространства $\Lambda^2\mathbb C^4(\Lambda^2\mathbb R^4)}{О бивекторах}$}
\subsubsection{Операторы Нордена}$ $

     Известно, что можно установить изоморфизм между комплексным
    евклидовым пространством $\mathbb C\mathbb R^6$ $(\mathbb R^6_{(3,3)})$
    и пространством бивекторов  $\Lambda^2\mathbb C^4(\Lambda^2\mathbb R^4)$.
    Этот изоморфизм определяется связующими операторами
    Нордена \cite{Norden1r}, удовлетворяющими следующим условиям
\begin{equation}
\label{re0}
    \frac{1}{2}\eta^\alpha{}_{aa_1}\eta_\beta{}^{aa_1}=\delta_\alpha{}^\beta\sps
    \eta^\alpha{}_{aa_1}\eta_\alpha{}^{bb_1}=\delta_{aa_1}{}^{bb_1}:=
    2\delta_{\left[\right.a}{}^{\left[\right.b}
    \delta_{a_1\left.\right]}{}^{b_1\left.\right]}
\end{equation}
    так, что выполнено
\begin{equation}
\label{re1}
    r^{\alpha}=1/2\cdot \eta^{\alpha}{}_{aa_1}R^{aa_1}\sps
    R^{aa_1}=\eta_{\alpha}{}^{aa_1}r^{\alpha}\sps
\end{equation}
    где $(\alpha,\beta,...=1,2,3,4,5,6;$
    $a_1,b_1,a,b,e,f,k,l,m,n,...=1,2,3,4;$ $i,j,g,h=1,2,3,4,5,6)$,
    причем
\begin{equation}
\label{re2}
    \begin{array}{c}
    g^{\alpha\beta}=1/4\cdot \eta^{\alpha}{}_{aa_1} \eta^{\beta}{}_{bb_1}
    \varepsilon^{aa_1bb_1}\sps
    \varepsilon^{aa_1bb_1}=
    \eta_{\alpha}{}^{aa_1} \eta_{\beta}{}^{bb_1}g^{\alpha\beta}\sps\\[2ex]
    g_{\alpha\beta}=1/4\cdot \eta_{\alpha}{}^{aa_1} \eta_{\beta}{}^{bb_1}
    \varepsilon_{aa_1bb_1}\sps
    \varepsilon_{aa_1bb_1}=
    \eta^{\alpha}{}_{aa_1} \eta^{\beta}{}_{bb_1}g_{\alpha\beta}\spsd
    \end{array}
\end{equation}
    При этом $R^{aa_1}$ - координаты бивектора из пространства
    $\Lambda^2\mathbb C^4$, а $r^{\alpha}$ - координаты его образа в $\mathbb C\mathbb R^6$;
    $g^{\alpha\beta}$ - метрический тензор пространства
    $\mathbb C\mathbb R^6$,  а его образ - тензор $\varepsilon_{aa_1bb_1}$,
    кососимметричный по всем индексам.

\begin{noter}
\label{note10r} Отметим, что с помощью метрического тензора $g_{\alpha\beta}$,
    заданного на пространстве $\mathbb C\mathbb R^6$,
    мы можем поднимать и опускать одиночные индексы, в то время
    как с помощью метрического 4-вектора $\varepsilon_{aa_1bb_1}$,
    заданного в расслоении $\Lambda$,
    мы можем поднимать и опускать только парные кососимметричные индексы, и
    нет такого метрического тензора, с помощью которого
    можно было бы проделать подобную операцию с одиночными
    индексами.
\end{noter}

    Отсюда следует, что существуют операторы $A_{{\alpha}{\beta}d}{}^c$
    такие,  что
\begin{equation}
\label{re4}
    T_{\alpha \beta}= A_{\alpha \beta d}{}^cT_c{}^d \sps
    T_k{}^k=0 \sps T_{\alpha \beta}=-T_{\alpha \beta}\spsd
\end{equation}
\newpage
    \noindent Приведем доказательство этого факта.
\begin{proof} Для символов Леви-Чивита имеем
\begin{equation}
\label{re7d}
    \begin{array}{c}
    \varepsilon^{abcd}\varepsilon_{klmn}=
    24\delta_{\left[\right.k}{}^a\delta_l{}^b\delta_m{}^c\delta_{n\left.\right]}{}^d\sps\\[2ex]
    \varepsilon^{abcd}\varepsilon_{klmd}=
    6\delta_{\left[\right.k}{}^a\delta_l{}^b\delta_{m\left.\right]}{}^c\sps\\[2ex]
    \varepsilon^{abcd}\varepsilon_{klcd}=
    4\delta_{\left[\right.k}{}^a\delta_{l\left.\right]}{}^b\sps\\[2ex]
    \varepsilon^{abcd}\varepsilon_{kbcd}=
    6\delta_k{}^a\sps\\[2ex]
    \varepsilon^{abcd}\varepsilon_{abcd}=
    24\sps\\[2ex]
    \delta_{kl}^{ab}:=2\delta_{\left[\right.k}{}^a\delta_{l\left.\right]}{}^b\spsd
    \end{array}
\end{equation}
    Кроме того, поскольку $\varepsilon^{abcd}$ - метрический, следует,
    что (\cite[т. 2, стр. 83, (6.2.19)]{Penrose1r})
\begin{equation}
\label{re8d}
    R_{ab}=\frac{1}{2}\varepsilon_{abcd}R^{cd}\spsd
\end{equation}
    Тогда с учетом формул (\ref{re1}) и (\ref{re2}) для
    некоторого тензора $T_{\alpha\beta}$
    имеем
    $$
    T_{\left[\alpha \beta \right]}=
    1/4\cdot\eta_\alpha{}^{aa_1}\eta_\beta{}^{bb_1}\cdot
    3/2(T_{a\left[ \right.a_1bb_1\left. \right]}-T_{\left[ \right. bb_1a
    \left.  \right]a_1})=
    $$
    $$
    =1/4\eta_\alpha{}^{aa_1}\eta_\beta{}^{bb_1}\cdot
    1/2(T_{ak}{}^{kd}\varepsilon_{a_1dbb_1}-T^{kd}{}_{ka_1}\varepsilon_
    {adbb_1})=
    1/2\eta_\alpha{}^{ca}\eta_\beta{}^{bk_1}\varepsilon_{adbk_1}\cdot
    $$
\begin{equation}
\label{re20}
    1/4(T^{kd}{}_{kc}-T_{kc}{}^{kd})
    =-1/2\eta_\alpha{}^{bk_1}\eta_\beta{}^{ca}\varepsilon_{adbk_1}\cdot
    1/4(T^{kd}{}_{kc}-T_{kc}{}^{kd})\spsd
\end{equation}
    Поэтому, чтобы получить формулу (\ref{re4}), следует положить
    $$
    A_{\alpha \beta d}{}^c:=\frac{1}{2}\eta_{\left[\alpha \right.}{}^{ca}\eta_{\left.
    \beta \right]}{}^{bk_1}\varepsilon_{dabk_1}
    \sps T_c{}^d:=1/4(T_{kc}{}^{kd}-T^{kd}{}_{kc})\sps
    $$
    $$
    T_{aa_1bb_1}-T_{bb_1aa_1}=4 \varepsilon_{bb_1c\left[\right.a_1}
    T_{a\left.\right]}{}^c\sps
    $$
\end{proof}

    \noindent Выпишем наиболее важные соотношения для операторов
    $A_{\alpha\beta b}{}^a$
\begin{equation}
\label{re21}
    \begin{array}{c}
    A_{\alpha\beta d}{}^cA^{\alpha\beta}{}_r{}^s=\frac{1}{2}\delta_r{}^s\delta_d{}^c-
    2\delta_ r{}^c\delta_d{}^s\sps
    A_{\alpha\beta d}{}^cA^{\lambda\mu}{}_c{}^d=
    2\delta_{\left[\alpha \right.}{}^\mu\delta_{\left.\beta\right]}{}^\lambda\sps
    \\[3ex]
    A_{\alpha\beta d}{}^cA_\gamma{}^\beta{}_r{}^s=
    (\eta_\alpha{}^{cs}\eta_\gamma{}_{rd}+
    \eta_\alpha{}^{ck}\eta_\gamma{}_{kr}\delta_d{}^s)+
    1/2(\eta_\alpha{}^{sn}\eta_\gamma{}_{rn}\delta_d{}^c+
    \\
    +\eta_\alpha{}^{ck}\eta_\gamma{}_{dk}\delta_r{}^s)-
    1/4g_{\alpha\gamma}\delta_r{}^s\delta_d{}^c\sps
    \\[3ex]
    T_m{}^n=\frac{1}{2}A^{\alpha\beta}{}_m{}^nT_{\beta\alpha}
    \sps T_{\beta\alpha}=-T_{\alpha\beta}\spsd
    \end{array}
\end{equation}
    Доказательство этих формул довольно громоздко и, поэтому, вынесено
    в приложение  (\ref{re1p})-(\ref{re5p}).\\

    Следует отметить, что операторы Нордена определяют некоторую
    Клиффордову алгебру.
\begin{proof} Рассмотрим тождество
\begin{equation}
\label{re15v}
    \eta_{\left(\right.\alpha}{}^{a\left[\right.b}
    \eta_{\beta\left.\right)}{}^{cd\left.\right]}=
    -\eta_{\left(\right.\alpha}{}^{\left[\right.cb}
    \eta_{\beta\left.\right)}{}^{|a|d\left.\right]}=
    \eta_{\left(\right.\alpha}{}^{\left[\right.ab}
    \eta_{\beta\left.\right)}{}^{cd\left.\right]}\sps
\end{equation}
    свертка которого с $\varepsilon_{bcdn}$ даст цепочку тождеств
\begin{equation}
\label{re16v}
    \begin{array}{c}
    \eta_{\left(\right.\alpha}{}^{ab}\eta_{\beta\left.\right)}{}^{cd}
    \varepsilon_{bcdn}=\frac{1}{24}
    \eta_{\left(\right.\alpha}{}^{kl}\eta_{\beta\left.\right)}{}^{mn}
    \varepsilon_{klmn}\varepsilon^{abcd}\varepsilon_{bcdn}\sps\\[3ex]
    2\eta_{\left(\right.\alpha}{}^{ab}\eta_{\beta\left.\right)}{}_{bn}=
    -\frac{1}{4}
    \eta_{\left(\right.\alpha}{}^{kl}\eta_{\beta\left.\right)}{}^{mn}
    \varepsilon_{klmn}\delta_n{}^a\sps\\[3ex]
    \eta_{\left(\right.\alpha}{}^{ab}\eta_{\beta\left.\right)}{}_{nb}=
    \frac{1}{2}g_{\alpha\beta}\delta_n{}^a\sps
    \end{array}
\end{equation}
    где $g_{\alpha\beta}$ тот же, что и в формуле (\ref{re2}).
    Определим
\begin{equation}
\label{re17v}
    \eta_\alpha:=\parallel \eta_\alpha{}^{aa_1} \parallel \sps
    \sigma_\alpha:=\parallel -\eta_\alpha{}_{aa_1} \parallel\sps
    \gamma_\alpha:=\sqrt{2}\left(
    \begin{array}{cc}
    0           & \sigma_\alpha \\
    \eta_\alpha & 0
    \end{array}
    \right)\sps
\end{equation}
    где $\lambda ,\psi,...=\overline{1,6}$. Тогда операторы $\gamma_\lambda$
    удовлетворяют следующему тождеству
\begin{equation}
\label{re18v}
    \gamma_\lambda\gamma_\psi+\gamma_\psi\gamma_\lambda=2g_{\lambda\psi}I\sps
\end{equation}
    что следует из уравнения (\ref{re16v}). Это уравнение является
    уравнением Клиффорда (\cite[т. 2, стр. 519, (Б.1)]{Penrose1r})
    так, что $\gamma_1, \gamma_2, \gamma_3, \gamma_4, \gamma_5, \gamma_6$ - операторы с
    комплексными матрицами размерности $8\times 8,\ g_{\alpha\beta}$ -
    метрический тензор (\ref{re2}), а I - единичный оператор.\\

    Верно обратное. Пусть мы имеем уравнение (\ref{re18v}). Тогда
    мы можем построить элемент $\gamma_7$
\begin{equation}
\label{re18va}
    \gamma_7:=\gamma_1\gamma_2\gamma_3\gamma_4\gamma_5\gamma_6\sps
    (\gamma_7){}^2=I\spsd
\end{equation}
    В этом случае, поскольку n=6 (четно), $\gamma_7$ антикоммутирует
    с любым элементом $\gamma_\alpha,(\alpha=\overline{1,6})$.
    Это означает, что для $\gamma_\alpha$ возможно представление
    (\ref{re17v}), и следовательно верны тождества (\ref{re15v}).\\

    Отсюда следует, что операторы Нордена определяют полную
    Клиффордову алгебру, которая образована конечными суммами
\begin{equation}
\label{re18vab}
     AI+B^\lambda\gamma_\lambda+C^{\lambda\mu}\gamma_\lambda\gamma_\mu+...
\end{equation}
     Размерность этой алгебры равна $2^6=64$. Такая алгебра
     может быть представлена полной матричной алгеброй,
     элементы которой имеют размерность $8\times 8$
     (\cite[т. 2, стр. 518-546]{Penrose1r}).
\end{proof}

\subsubsection[Сопряжение в расслоении A]{Сопряжение в расслоении $A^\mathbb C(S)$}$ $

    \indent В этом параграфе формулируется некоторое утверждение,
    касающееся вложения вещественных пространств в комплексные,
    доказательство приводится же в следующем пункте.

    Рассмотрим 6-мерное (псевдо-)евклидово пространство $\mathbb R^6_{(p,q)}$,
    вложенное в $\mathbb C\mathbb R^6$, с метрикой произвольного индекса q,
    касательное пространство $\tau^\mathbb R(\mathbb R^6_{(p, q)})$ которого будем
    рассматривать как вещественное
    подпространство пространства $\mathbb C^4$. Это приведет к
    векторному расслоению $A^\mathbb C(S)$ со слоями, изоморфными $\mathbb C^4$, и
    некоторой структурой s. Нам необходимо выяснить природу этой структуры.
    Для этого рассмотрим простой бивектор из $\tau^\mathbb C(\tau^\mathbb R)$.
    Необходимое и достаточное условие простоты выражается формулой
\begin{equation}
\label{re5}
    R^{ab}=X^aY^b-X^bY^a\sps X^a\, Y^a\in \mathbb C^4\sps
\end{equation}
    где $i,j=\overline{1,6};a,b,c,d,k,l,m,n,...=\overline{1,4}$. \\

\begin{noter}
\label{note11r}
    На основании формулы (\ref{re1}) простому бивектору из
    пространства $\Lambda^2\mathbb C^4$ ставится в соответствие
    изотропный вектор пространства $\mathbb C\mathbb R^6$. Это следует из
    соотношения
\begin{equation}
\label{re19v}
    \begin{array}{c}
    0=\varepsilon_{abcd}X^aY^bX^cY^d=
    \frac{1}{4}\varepsilon_{abcd}R^{ab}R^{cd}
    =\frac{1}{2}\eta_\alpha{}^{km} r^\alpha \eta^\beta{}_{km}r_\beta=
    r^\alpha r_\alpha=0\spsd
    \end{array}
\end{equation}
\end{noter}

    Далее, любой бивектор должен быть самосопряжен относительно
    тензора $s_{aa_1b'b'{}_1}$
    (обозначения введены согласно \cite{Penrose1r})
\begin{equation}
\label{re6}
    s_{aa_1b'b'{}_1}=g_{ij}\eta^i{}_{b'b'{}_1}
    \eta^j{}_{aa_1}\sps
    g_{ij}=1/4\cdot\eta_i{}^{b'b'{}_1}\eta_j{}^{aa_1}
    s_{aa_1b'b'{}_1}\sps
\end{equation}
    $$
    \bar s_{b'b'{}_1aa_1}=s_{aa_1b'b'{}_1}\sps
    $$
    где $g_{ij}$ - метрический тензор (\ref{re10v}).
    Последнее уравнение выражает эрмитову симметрию тензора s.
    Такой тензор был введен в работе \cite{Neifeld1r}. В случае
    метрики четного индекса тензор $s_{aa_1}{}^{b'b'{}_1}$ (поднятие и
    спуск двойных индексов осуществляется с помощью метрического
    4-вектора $\varepsilon_{aa_1bb_1}$) имеет вид
\begin{equation}
\label{re7}
    s_{aa_1}{}^{b'b'{}_1}=s^{cb'}s^{c_1b'{}_1} \varepsilon_{cc_1aa_1}
    \sps \bar s_{a'b}=\pm s_{ba'}\sps
\end{equation}
    а в случае нечетного, получим
\begin{equation}
\label{re8}
    s_{aa_1}{}^{b'b'{}_1}=2s_{\left[ a \right.}{}^{b'}s_{\left. a_1\right]}
    {}^{b'{}_1}\sps s_a{}^{b '}\bar s_{b '}{}^c=\pm\delta_a{}^c\spsd
\end{equation}
    Если $R^{ab}$ прост и принадлежит касательному пространству $\tau^\mathbb R$,
    то для составляющих его векторов
    выполнено (сравн. \cite[т. 2, стр. 80, (6.2.13)]{Penrose1r})
    в случае метрики четного индекса
\begin{equation}
\label{re9}
    X_aX_{a'}s^{aa'}=0 \sps Y_aY_{a'}s^{aa'}=0 \sps X_aY_{a'}s^{aa'}=0\sps
\end{equation}
    а для нечетного
\begin{equation}
\label{re10}
    X_{a'}X^as_a{}^{a'}=0 \sps
    X_{a'}Y^as_a{}^{a'}=0 \sps
    Y_{a'}Y^as_a{}^{a'}=0\spsd
\end{equation}
    Таким образом определяется структура s расслоения $A^\mathbb C(S)$.
    При этом тензор $s_{kk'}$ в случае метрики четного индекса выполняет роль
    метрического тензора, с помощью которого можно поднимать и опускать одиночный
    индекс; а в случае метрики нечетного индекса с помощью тензора $s_k{}^{k'}$
    происходит отождествление штрихованных (комплексно-сопряженных)
    и нештрихованных пространств. Доказательству этого утверждения как раз и посвящен следующий параграф.

\subsection{Спинорное представление тензоров специального вида.
            Накрытия, соответствующие этому разложению}
\subsubsection{\texorpdfstring{Теорема о двулистности накрытия группы $SO(6,\mathbb C)$
            группой $SL(4,\mathbb C)$}{Теорема о двулистном накрытии}}$ $
    \indent Прежде чем приступить к указанному доказательству,
    нам необходимо более тщательно разобраться с
    двулистным накрытием \linebreak[4] $SL(4,\mathbb C)/\{\pm 1\}\cong SO(6,\mathbb C)$.
    Ниже будет получено явное представление этого накрытия
    с помощью связующих операторов Нордена $\eta_\alpha{}^{ab}$.
    Используя это представление, легче разобраться в том,
    как происходит вещественное вложение $\mathbb R^6_{(2,4)}\subset \mathbb C\mathbb R^6$,
    и, соответственно, как построить явное представление
    оператора инволюции в спинорном виде. Кроме того, результаты
    этого пункта пригодятся в исследовании структуры бивекторов
    пространства $\mathbb C\mathbb R^6$.

    Обозначим через $K_\alpha{}^\beta$ преобразования из
    группы $SO(6,\mathbb C)$, а через $S_a{}^b$ - преобразования из
    группы $SL(4,\mathbb C)$. Тогда будет верна следующая теорема.
\vspace{1cm}
\begin{theoremr}
\label{theorem5r} Всякому преобразованию $K_\alpha{}^\beta$ соответствует
    два и только два преобразования $\pm S_a{}^b$($\pm S_{ab}$) таких, что
    $det\parallel S_c{}^d \parallel=1$ $(det\parallel S_{cd} \parallel=1)$.
    И наоборот, любым преобразованиям $\pm S_a{}^b$($\pm S_{ab}$)
    соответствует одно и только одно преобразование $K_\alpha{}^\beta$.
\end{theoremr}
\vspace{1cm}

\begin{proof} Пусть имеется некоторое преобразование $\pm S_a{}^b(\pm S_{ab})$
    такое, что
     $$
     \varepsilon:=\varepsilon_{1234}\sps
     $$
\begin{equation}
\label{re1b}
    S_a{}^b S_{a_1}{}^{b_1} S_c{}^d S_{c_1}{}^{d_1}\varepsilon_{bb_1dd_1}=
    \varepsilon_{aa_1cc_1}\
    (S_{ab} S_{a_1b_1} S_{cd} S_{c_1d_1}\varepsilon^{bb_1dd_1}=
    \varepsilon^{-2}\varepsilon_{aa_1cc_1})\spsd
\end{equation}
    Последнее означает, что $det\parallel S_c{}^d \parallel=1\ (det\parallel S_{cd} \parallel=1)$.
    Положим
\begin{equation}
\label{re2b}
    \begin{array}{c}
    K_\alpha{}^\beta:=\frac{1}{4}\eta_\alpha{}^{ab}\eta^\beta{}_{cd}
    \cdot 2\beta S_{\left[ a \right.}{}^c S_{\left. b \right]}{}^d\
    (K_\alpha{}^\beta:=\frac{1}{4}\eta_\alpha{}^{ab}\eta^\beta{}_{cd}
    \cdot\varepsilon  \beta S_{ak} S_{bl}\varepsilon^{klcd})\sps\\[4ex]
    \beta:=\pm 1\spsd
    \end{array}
\end{equation}
    Тогда на основании (\ref{re1b}) и (\ref{re2b}) получим
\begin{equation}
\label{re3b}
    K_\alpha{}^\beta K_\gamma{}^\delta g_{\beta\delta}=g_{\alpha\gamma}\spsd
\end{equation}
    Таким образом из (\ref{re1b}) следует (\ref{re3b}).\\ \\

    Если теперь, наоборот, задано $K_\alpha{}^\beta$ вида (\ref{re3b}).
    Положим
\begin{equation}
\label{re4b}
    K_{aa_1}{}^{bb_1}:=\eta^\alpha{}_{aa_1}\eta_\beta{}^{bb_1}K_\alpha{}^\beta\spsd
\end{equation}
    При этом (\ref{re3b}) перепишется как
\begin{equation}
\label{re46ab}
    \frac{1}{4}K_{aa_1}{}^{bb_1}K_{cc_1}{}^{dd_1}
    \varepsilon_{bb_1dd_1}=\varepsilon_{aa_1cc_1}\spsd
\end{equation}
    Формула  (\ref{re46ab}) означает, что преобразование
    $K_{aa_1}{}^{bb_1}$ должно быть регулярным, т.е.
    $\forall r^{aa_1}\ne 0 \ \Rightarrow\ K_{aa_1}{}^{bb_1}r^{aa_1}\ne 0,\
    K_{aa_1}{}^{bb_1}r_{bb_1}\ne 0$.\\
\vspace{1cm}
\begin{proof} Действительно,
    предположим обратное: $\exists r^{aa_1}\ne 0$ так, что
    $K_{aa_1}{}^{bb_1}r_{bb_1}=0$. Это будет означать,
    что преобразование $K_{aa_1}{}^{bb_1}$ сингулярно
\begin{equation}
\label{re46abc}
     0=r_{bb_1}\cdot\frac{1}{4} K_{aa_1}{}^{bb_1} K_{cc_1}{}^{dd_1}
     \varepsilon_{bb_1dd_1}=\varepsilon_{aa_1cc_1}r^{aa_1}\spsd
\end{equation}
    Из этого следует, что $r^{aa_1}=0$. Противоречие.
\end{proof}
\vspace{1cm}

    Для дальнейших выкладок нам потребуется следующая лемма.

\vspace{1cm}
\begin{lemmar}
\label{lemma1r}
    Выберем два ненулевых вектора $r_1{}^\alpha ,r_2{}^\beta \in
    \tau^\mathbb C_x$, где $x$ - произвольная точка базы:
    комплексного евклидового пространства $\mathbb C\mathbb R^6$, снабженного
    метрическим тензором $g_{\alpha\beta}$. Тогда следующие три
    условия эквивалентны:
    \begin{enumerate}
     \item  $r_1{}^\alpha(r_1)_\alpha=0$ ,
            $r_2{}^\alpha(r_2)_\alpha=0$ ,
            $r_1{}^\alpha(r_2)_\alpha=0$ ;
     \item  $r_1{}^\alpha=\frac{1}{2}\eta^\alpha{}_{aa_1} X^aY^{a_1}$,
            $r_2{}^\alpha=\frac{1}{2}\eta^\alpha{}_{aa_1} X^aZ^{a_1}$;
     \item  $(r_1)_\alpha=\frac{1}{2}\eta_\alpha{}^{aa_1} \tilde X_a\tilde Y_{a_1}$,
            $(r_2)_\alpha=\frac{1}{2}\eta_\alpha{}^{aa_1} \tilde X_a\tilde Z_{a_1}$;
    \end{enumerate}
    где векторы $X^a,Y^a,Z^a$ принадлежат слою $\mathbb C^4_x$ расслоения $A^\mathbb C$, а
    $\tilde X_a,\tilde Y_a,\tilde Z_a$ - ковекторы двойственного слоя.
\end{lemmar}
\vspace{1cm}

\begin{proof} $1).\Rightarrow 2).$\\
    Рассмотрим первое уравнение условия 1). и положим
\begin{equation}
\label{re1vv}
    r_1{}^{aa_1}:=\eta_\alpha{}^{aa_1}r_1{}^\alpha\sps
\end{equation}
\newpage
    \noindent тогда на основании (\ref{re2}) получим
\begin{equation}
\label{re2vv}
    \begin{array}{c}
    r_1{}^\alpha (r_1)_\alpha=g_{\alpha\beta}r_1{}^\alpha r_1{}^\beta=
    \frac{1}{4}\eta_\alpha{}^{aa_1}\eta_\beta{}^{cc_1}
    \varepsilon_{aa_1cc_1}\frac{1}{2}\eta^\alpha{}_{dd_1}r_1{}^{dd_1}
    \frac{1}{2}\eta^\beta{}_{kk_1}r_1{}^{kk_1}=\\[2ex]
    =\frac{1}{4}r_1{}^{dd_1}
    \delta_{\left[\right.d}{}^a\delta_{d_1\left.\right]}{}^{a_1}
    r_1{}^{kk_1}
    \delta_{\left[\right.k}{}^c\delta_{k_1\left.\right]}{}^{c_1}
    \varepsilon_{aa_1cc_1}=
    \frac{1}{2}
    r_1{}^{aa_1}(r_1)_{aa_1}=0\spsd
    \end{array}
\end{equation}
    Определим
\begin{equation}
\label{re3vv}
    pf(r):=r^\alpha r_\alpha=\frac{1}{2}r^{aa_1}r_{aa_1}\spsd
\end{equation}
    Тогда из (\ref{re16v}) следует
\begin{equation}
\label{re4vv}
    r^{ab}r_{bc}=-pf(r)\delta_c{}^d\ \ \Leftrightarrow\ \ \
    3r^{a\left[\right.b}r^{cd\left.\right]}=pf(r)\varepsilon^{abcd}\spsd
\end{equation}
    Отсюда для вектора $r_1{}^\alpha$ получаем
\begin{equation}
\label{re5vv}
    r_1{}^{ab}r_1{}^{cd}=r_1{}^{ac}r_1{}^{bd}-r_1{}^{bc}r_1{}^{ad}\spsd
\end{equation}
    Поскольку $r_1{}^{cd}$ ненулевой бивектор, то существуют
    такие ковекторы $A_c,B_d$, что $r_1{}^{cd}A_cB_d\ne 0$ и вещественно.
    Положим
\begin{equation}
\label{re6vv}
    P^a:=\sqrt{2}r_1{}^{ak}A_k/\sqrt{(r_1{}^{cd}A_cB_d)}\sps
    Q^a:=\sqrt{2}r_1{}^{bk}B_k/\sqrt{(r_1{}^{cd}A_cB_d)}\sps
\end{equation}
    тогда из (\ref{re5vv}) следует
\begin{equation}
\label{re7vv}
     r_1{}^{ab}=P^{\left[\right.a}Q^{b\left.\right]}\spsd
\end{equation}
     При этом $P^a,Q^a$ - линейно независимы. Таким же образом
     можно получить разложение и для $r_2{}^{ab}$
\begin{equation}
\label{re8vv}
     r_2{}^{ab}=R^{\left[\right.a}S^{b\left.\right]}
\end{equation}
     из второго условия 1). так, что вектора $R^a,S^a$ будут
     также линейно независимы. Из третьего уравнения условия 1).
     вытекает следующее соотношение
\begin{equation}
\label{re9vv}
    \begin{array}{c}
    0=r_1{}^\alpha(r_2)_\alpha=\frac{1}{4}
    \varepsilon_{abcd}r_1{}^{ab}r_2{}^{cd}=
    \frac{1}{4}\varepsilon_{abcd}P^aQ^bR^cS^d=0\spsd
    \end{array}
\end{equation}
    Это означает,  что вектора $P^a,Q^b,R^c,S^d$ - линейно завиcимы
\begin{equation}
\label{re10vv}
    \alpha P^a+\beta Q^a+\gamma R^a+\delta S^a=0\sps
    |\alpha|+|\beta|+|\gamma|+|\delta|\ne 0\spsd
\end{equation}
    При этом  либо $\alpha \ne 0$, либо $\beta \ne 0$.
    В ином случае $(\alpha=\beta=0)$ векторы $R^c,S^c$
    были бы линейно зависимы. Пусть для определенности
    $\alpha \ne 0$. Тогда опять либо $\gamma \ne 0$,
    либо $\delta \ne 0$. Положим
\begin{equation}
\label{re11vv}
    \begin{array}{c}
    X^a:=P^a+(\beta/\alpha) Q^a=-((\gamma/\alpha) R^a+(\delta/\alpha) S^a)\sps
    Y^a:=Q^a\sps\\[2ex]
    Z^a:=\left\{
    \begin{array}{l}
    \ (\alpha/\delta) R^a,\ \delta\ne 0,\gamma =0\sps \\
    -(\alpha/\gamma) S^a,\ \gamma \ne 0\spsd
    \end{array}
    \right.
    \end{array}
\end{equation}
    Таким образом из (\ref{re11vv}) следует условие 2). леммы.\\ \\
    $2).\Rightarrow 1).$\\
    Проверяется непосредственно, например
\begin{equation}
\label{re12vv}
    r_1{}^\alpha(r_2)_\alpha=\frac{1}{4}\varepsilon_{aa_1bb_1}
    X^aY^{a_1}X^bZ^{b_1}=0\spsd
\end{equation}
    \indent Таким же образом доказывается и эквивалентность
    $1).\Rightarrow 3).$ и $3).\Rightarrow 1).$ Эти
    импликации возможны из-за наличия метрического
    тензора в касательном расслоении и метрического
    4-вектора в расслоении $A^\mathbb C$.
\end{proof}
\vspace{1cm}
    Возьмем два изотропных ненулевых вектора
\begin{equation}
\label{re13vv}
    \begin{array}{c}
    r_1{}^\alpha=\frac{1}{2}\eta^\alpha{}_{aa_1} M^aN^{a_1}\sps
    r_2{}^\alpha=\frac{1}{2}\eta^\alpha{}_{aa_1} M^aL^{a_1}
    \end{array}
\end{equation}
    и два ненулевых изотропных ковектора
\begin{equation}
\label{re14vv}
    (\tilde r_1)_\alpha=\frac{1}{2}\eta_\alpha{}^{aa_1} \tilde M_a\tilde N_{a_1}\sps
    (\tilde r_2)_\alpha=\frac{1}{2}\eta_\alpha{}^{aa_1} \tilde M_a\tilde L_{a_1}\sps
\end{equation}
    удовлетворяющих соответственно условиям 2). и 3). леммы \ref{lemma1r}.
    Подействуем на (\ref{re13vv}), (\ref{re14vv}) ортогональным
    преобразованием $K_\alpha{}^\beta$ и получим
\begin{equation}
\label{re15vv}
    \begin{array}{c}
    r_3{}^\alpha:=K_\beta{}^\alpha r_1{}^\beta\sps
    r_4{}^\alpha:=K_\beta{}^\alpha r_2{}^\beta\sps\\[2ex]
    (\tilde r_3)_\alpha:=K_\alpha{}^\beta (\tilde r_1)_\beta\sps
    (\tilde r_4)_\alpha:=K_\alpha{}^\beta (\tilde r_2)_\beta\spsd
    \end{array}
\end{equation}
    Тогда из условия 1). леммы \ref{lemma1r} следует с учетом
    (\ref{re3b}) и (\ref{re46ab})
\begin{equation}
\label{re16vv}
    \begin{array}{c}
    r_3{}^\alpha (r_3)_\alpha=K_\alpha{}^\beta K_\gamma{}^\delta
    g_{\beta\delta} r_1{}^\alpha r_1{}^\gamma =
    r_1{}^\alpha (r_1)_\alpha=0\sps\\[2ex]
    r_4{}^\alpha (r_3)_\alpha=r_2{}^\alpha (r_1)_\alpha=0\sps
    r_4{}^\alpha (r_4)_\alpha=r_2{}^\alpha (r_2)_\alpha=0\sps\\[2ex]
    \tilde r_4{}^\alpha (\tilde r_3)_\alpha=\tilde r_2{}^\alpha (\tilde r_1)_\alpha=0\sps
    \tilde r_4{}^\alpha (\tilde r_4)_\alpha=\tilde r_2{}^\alpha (\tilde r_2)_\alpha=0\sps\\[2ex]
    \tilde r_3{}^\alpha (\tilde r_3)_\alpha=\tilde r_1{}^\alpha \tilde r_1{}_\alpha=0\spsd
    \end{array}
\end{equation}
    Поскольку преобразование $K_{aa_1}{}^{bb_1}$ регулярно, то векторы
    и ковекторы (\ref{re15vv}) будут ненулевыми, следовательно, из
    условий 2). и 3).  леммы \ref{lemma1r} получим
\begin{equation}
\label{re17vv}
    \begin{array}{c}
    r_3{}^\alpha=\frac{1}{2}\eta^\alpha{}_{aa_1} X^a Y^{a_1}\sps
    r_4{}^\alpha=\frac{1}{2}\eta^\alpha{}_{aa_1} X^a Z^{a_1}\\[2ex]
    (\tilde r_3)_\alpha=\frac{1}{2}\eta_\alpha{}^{aa_1}\tilde X_a\tilde  Y_{a_1}\sps
    (\tilde r_4)_\alpha=\frac{1}{2}\eta_\alpha{}^{aa_1}\tilde X_a\tilde  Z_{a_1}\spsd
    \end{array}
\end{equation}
    Рассмотрим теперь тождество
\begin{equation}
\label{re18vv}
    r_3{}^{\left[\right.\alpha}r_4{}^{\beta \left.\right]}=
    K_{\left[\right.\gamma}{}^{\left[\right.\alpha}
    K_{\delta \left.\right]}{}^{\beta \left.\right]}
    r_1{}^{\left[\right.\gamma}r_2{}^{\delta \left.\right]}\spsd
\end{equation}
    Распишем его с помощью формул (\ref{re4}) и (\ref{re21})
\begin{equation}
\label{re19vv}
    \begin{array}{c}
    A^{\alpha\beta}{}_a{}^b\cdot\frac{1}{4}X^aY^{b_1}X^cZ^{c_1}
    \varepsilon_{cc_1bb_1}=\\
    =\frac{1}{4}A^{\gamma\delta}{}_r{}^s M^rN^{k_1}M^l L^{l_1}\varepsilon_{ll_1sk_1}
    A_{\gamma\delta c}{}^d A^{\alpha\beta}{}_a{}^b\cdot
    \frac{1}{8}(K_{dm}{}^{ak}K^{cm}{}_{bk}-K^{cmak}K_{dmbk})\sps\\[2ex]
    X^aY^{b_1}X^cZ^{c_1}\varepsilon_{cc_1bb_1}=\\
    =2 \delta_r{}^d
    \delta_c{}^s M^rN^{k_1}M^lL^{l_1}\varepsilon_{k_1ll_1s}
    \cdot\frac{1}{8}(K_{dm}{}^{ak}K^{cm}{}_{bk}-K^{cmak}K_{dmbk})\sps\\[2ex]
    X^a(Y^{b_1}X^cZ^{c_1}\varepsilon_{cc_1bb_1})=\\
    =M^d(N^{k_1}M^lL^{l_1}\varepsilon_{k_1ll_1c})
    \cdot\frac{1}{4}(K_{dm}{}^{ak}K^{cm}{}_{bk}-K^{cmak}K_{dmbk})\spsd
    \end{array}
\end{equation}
    Определим
\begin{equation}
\label{re20vv}
    \begin{array}{c}
    T_b:=Y^{b_1}X^cZ^{c_1}\varepsilon_{cc_1bb_1}\sps
    P_c:=N^{k_1}M^lL^{l_1}\varepsilon_{k_1ll_1c}\sps\\[2ex]
    \tilde K_d{}^c{}_b{}^a:=
    \frac{1}{8}(K^{cmak}K_{dmbk}-K_{dm}{}^{ak}K^{cm}{}_{bk})
    \end{array}
\end{equation}
    так, что выполнено
\begin{equation}
\label{re21vv}
    X^cT_c=0\sps
    M^cP_c=0\sps
    \tilde K_c{}^c{}_b{}^a=0\sps\tilde K_d{}^c{}_b{}^b=0\sps
\end{equation}
\begin{equation}
\label{re22vv}
    K_{aa_1}{}^{bb_1}K_{cc_1}{}^{dd_1}-
    K_{aa_1}{}^{dd_1}K_{cc_1}{}^{bb_1}=
    8\varepsilon_{cc_1k\left[\right.a_1}
    \tilde K_{a\left.\right]}{}^k{}_r{}^{\left[\right.b}
    \varepsilon^{b_1\left.\right]rdd_1}\spsd
\end{equation}
    Откуда
\begin{equation}
\label{re23vv}
    X^aT_b=-2 M^dP_c\tilde K_d{}^c{}_b{}^a\spsd
\end{equation}
    Таким же образом из тождества
\begin{equation}
\label{re24vv}
    (\tilde r_3)_{\left[\right.\gamma}
    (\tilde r_4)_{\delta\left.\right]}=
    K_{\left[\right.\gamma}{}^{\left[\right.\alpha}
    K_{\delta\left.\right]}{}^{\beta\left.\right]}
    (\tilde r_1)_{\left[\right.\alpha}
    (\tilde r_2)_{\beta\left.\right]}\sps
\end{equation}
    определяя
\begin{equation}
\label{re25vv}
    \tilde T^b:=\tilde Y_{b_1}\tilde X_c\tilde Z_{c_1}\varepsilon^{cc_1bb_1}\sps
    \tilde P^b:=\tilde N_{k_1}\tilde M_l\tilde Z_{l_1}\varepsilon^{k_1ll_1b}\sps
\end{equation}
    можно получить
\begin{equation}
\label{re26vv}
    \tilde X_d\tilde T^c=-2 \tilde M_a\tilde P^b\tilde K_d{}^c{}_b{}^a\spsd
\end{equation}
    Найдем теперь однородное решение, удовлетворяющее уравнениям
    (\ref{re23vv}) и (\ref{re26vv})
\begin{equation}
\label{re27vv}
    \left\{
    \begin{array}{c}
    (\tilde K_{\mbox{однородное}})_d{}^c{}_b{}^aM^dP_c=0\sps \\
    (\tilde K_{\mbox{однородное}})_d{}^c{}_b{}^a\tilde M_a\tilde P^b=0\sps
    \end{array}
    \right.
    \ \ \Leftrightarrow \ \ \
    \left\{
    \begin{array}{c}
    M^dP_d=0\sps  \\
    \tilde M_a\tilde P^a=0\spsd
    \end{array}
    \right.
\end{equation}
    Эти две системы должны совпадать тождественно, поскольку
    левая система верна для любых $M^a,\tilde M_a,P_a,\tilde P^a$,
    удовлетворяющих правой системе. Это возможно только
    при
\begin{equation}
\label{re28vv}
    (\tilde K_{\mbox{однородное}})_d{}^c{}_b{}^a=\alpha \delta_d{}^c
    \delta_b{}^a \sps \alpha \in \mathbb C\spsd
\end{equation}
    Рассмотрим далее частное решение, например, уравнения
    (\ref{re23vv}). Это решение должно быть регулярным, что
    означает невозможность выполнения условия
\begin{equation}
\label{re29vv}
    \exists\ M^d\ne 0 , P_c\ne 0\sps{\mbox{   что   }}\
    (\tilde K_{\mbox{частное}})_d{}^c{}_b{}^aM^dP_c=0
\end{equation}
    (выполнение условия (\ref{re29vv}) равносильно (формула
    (\ref{re22vv})) сингулярности преобразования $K_{aa_1}{}^{bb_1}$).
    Для решения (\ref{re23vv}) нам потребуется следующая лемма.
\vspace{1cm}
\begin{lemmar}
\label{lemma2r}
    Пусть {\bf {A,B,C,...}} - собирательные индексы. Тогда
    следующие 3 условия на
    $\lambda_{\mbox{\scriptsize\bf AB}}{}^{\mbox{\scriptsize\bf Q}}$
    эквивалентны:
\begin{enumerate}
    \item $\lambda_{\mbox{\scriptsize\bf AB}}{}^{\mbox{\scriptsize\bf Q}}
          \xi_{\mbox{\scriptsize\bf Q}}$ имеет вид $\rho_{\mbox{\scriptsize\bf A}}
          \xi_{\mbox{\scriptsize\bf B}}$ для всякого $\xi_{\mbox{\scriptsize\bf Q}}$;
    \item $\lambda_{{\mbox{\scriptsize\bf A}_1}\left[\right. {\mbox{\scriptsize\bf B}_1}}{}^
                {\left(\right.{\mbox{\scriptsize\bf  Q}_1}}
           \lambda_{|{\mbox{\scriptsize\bf A}_2}|{\mbox{\scriptsize\bf B}_2}\left.\right]}{}^
                {{\mbox{\scriptsize\bf  Q}_1}\left.\right)}$=0;
    \item $\lambda_{\mbox{\scriptsize\bf AB}}{}^{\mbox{\scriptsize\bf Q}}$
          можно представить как $\alpha_{\mbox{\scriptsize\bf A}}
          \varphi_{\mbox{\scriptsize\bf B}}{}^{\mbox{\scriptsize\bf Q}}$,
          либо как $\theta_{\mbox{\scriptsize\bf A}}{}^
          {\mbox{\scriptsize\bf Q}}\beta{\mbox{\scriptsize\bf B}}$.\\
\end{enumerate}
\end{lemmar}
\vspace{1cm}

\begin{proof} Оно приведено на стр. 205 монографии \cite[т. 1]{Penrose1r}.
    Поскольку в ее доказательстве метрический тензор
    участия не принимал, то эта лемма справедлива для
    любого расположения индексов: сверху или снизу.\\
\end{proof}
\vspace{1cm}

    Применим лемму \ref{lemma2r} к уравнению (\ref{re23vv}),
    тогда получим 2 варианта:
\begin{equation}
\label{re30vv}
    \begin{array}{ccccc}
    a). & (\tilde K_{\mbox{частное}})_d{}^c{}_b{}^aP_c=A^aB_{bd}\sps &
    b). & (\tilde K_{\mbox{частное}})_d{}^c{}_b{}^aP_c=A_d{}^aB_b
    \end{array}
\end{equation}
     \indent Во-первых. Предположим, что условия a). и  b).  выполняются
     одновременно. Воспользуемся еще одной леммой.
\vspace{1cm}
\begin{lemmar}
\label{lemma3r}
     Из $\psi_{\mbox{\scriptsize\bf AB}}\varphi_{\mbox{\scriptsize\bf C}}=
     \chi_{\mbox{\scriptsize\bf A}}\theta_{\mbox{\scriptsize\bf BC}}$
     следует выполнение $\psi_{\mbox{\scriptsize\bf AB}}=
     \chi_{\mbox{\scriptsize\bf A}}\xi_{\mbox{\scriptsize\bf B}}$,
     $\theta_{\mbox{\scriptsize\bf BC}}=\xi_{\mbox{\scriptsize\bf B}}
     \varphi_{\mbox{\scriptsize\bf C}}$ для некоторого
     $\xi_{\mbox{\scriptsize\bf B}}$.
\end{lemmar}
\vspace{1cm}

\begin{proof} Оно приведено на стр. 205 монографии \cite[т. 1]{Penrose1r}.
    И так же, как и в предыдущей лемме, расположение индексов
    не существенно.
\end{proof}
\vspace{1cm}

    Применим эту лемму к уравнению (\ref{re30vv}), что даст
\begin{equation}
\label{re31vv}
    (\tilde K_{\mbox{частное}})_d{}^c{}_b{}^aP_c=A_dB^aC_b\spsd
\end{equation}
    Но существует такой вектор $M^d\ne 0$, что $M^dP_d=0$ и
    $M^dA_d=0$, поэтому  из (\ref{re31vv}) следует
    (\ref{re29vv}), что невозможно. Из этого
    заключаем, что a).  и  b). из (\ref{re30vv}) одновременно
    выполняться не могут.\\ \\

    Во-вторых. Применим лемму \ref{lemma2r} теперь уже к уравнению (\ref{re30vv}).
    Это даст 4 варианта:
\begin{equation}
\label{re32vv}
    \begin{array}{ccccc}
    I). & a). & (\tilde K_{\mbox{частное}})_d{}^c{}_b{}^a=A^{ac}B_{db}\sps &
          b). & (\tilde K_{\mbox{частное}})_d{}^c{}_b{}^a=C^aD^c{}_{db}\sps\\[2ex]
    II).& a). & (\tilde K_{\mbox{частное}})_d{}^c{}_b{}^a=S_d{}^aE_b{}^c\sps &
          b). & (\tilde K_{\mbox{частное}})_d{}^c{}_b{}^a=U_d{}^{ac}V_b\spsd
    \end{array}
\end{equation}
    Варианты b). в обоих случаях отпадают, поскольку приводят
    к сингулярным преобразованиям (смотри пояснения после
    формулы (\ref{re31vv})).\\
\newpage
    Для определенности рассмотрим случай II).a). Свернем
    общее решение  уравнения (\ref{re23vv})
\begin{equation}
\label{re33vv}
    \tilde K_d{}^c{}_b{}^a=S_d{}^aE_b{}^c+\alpha\delta_d{}^c
    \delta_b{}^a
\end{equation}
    с $\delta_c{}^d$ и, используя (\ref{re21vv}), получим
\begin{equation}
\label{re34vv}
    0=S_k{}^aE_b{}^k+4\alpha\delta_b{}^a\ \ \Rightarrow\ \ \
    E_b{}^k=-4\alpha(S^{-1})_b{}^k
\end{equation}
    (преобразование $S_k{}^a$  невырождено, т.к. в ином случае
    преобразование $\tilde K_d{}^c{}_b{}^a$ будет сингулярным,
    что повлечет за собой сингулярность преобразования
    $K_{aa_1}{}^{bb_1})$.
    Поэтому
\begin{equation}
\label{re35vv}
    \tilde K_d{}^c{}_b{}^a=(-\alpha)(4S_d{}^a(S^{-1})_b{}^c-
    \delta_d{}^c\delta_b{}^a)\spsd
\end{equation}
    Свернем (\ref{re22vv}) с $\varepsilon_{dd_1pp_1}K_{ss_1}{}^{pp_1}$,
    что даст с учетом (\ref{re46ab})
\begin{equation}
\label{re36vv}
    \begin{array}{c}
    K_{aa_1}{}^{bb_1}\varepsilon_{ss_1cc_1}-
    K_{cc_1}{}^{bb_1}\varepsilon_{ss_1aa_1}=
    8\varepsilon_{cc_1k\left[\right.a_1}
    \tilde K_{a\left.\right]}{}^k{}_r{}^{\left[\right.b}
    K_{ss_1}{}^{b_1\left.\right]r}\spsd
    \end{array}
\end{equation}
    Свернем (\ref{re36vv}) с $\varepsilon^{ss_1cc_1}$,
    используя формулы (\ref{re7d}),
\begin{equation}
\label{re37vv}
    5K_{aa_1}{}^{bb_1}=8\tilde
    K_{\left[\right.a}{}^k{}_{\left|\right.r}{}^{\left[\right.b}
    K_{k\left.\right|a_1\left.\right]}{}^{b_1\left.\right]r}
\end{equation}
    и подставим (\ref{re35vv}) в (\ref{re37vv})
\begin{equation}
\label{re38vv}
    \begin{array}{c}
    5K_{aa_1}{}^{bb_1}=(-8\alpha)
    K_{k\left[\right.a_1}{}^{\left[\right.b_1|r|}
    (4S_{a\left.\right]}{}^{b\left.\right]}(S^{-1})_r{}^k-
    \delta_{a\left.\right]}{}^{|k|}\delta_r{}^{b\left.\right]})\sps\\[2ex]
    K_{aa_1}{}^{bb_1}=\frac{32\alpha}{5+8\alpha}
    K_{k\left[\right.a_1}{}^{r\left[\right.b_1}
    S_{a\left.\right]}{}^{b\left.\right]}(S^{-1})_r{}^k
    \end{array}
\end{equation}
     ($\alpha\ne 0,\alpha\ne \pm 5/8$,  в ином случае преобразование
     $\tilde K_a{}^k{}_r{}^b$  будет сингулярным). Положим
\begin{equation}
\label{re39vv}
     K_{aa_1}{}^{bb_1}:=
     2M_{\left[\right.a_1}{}^{\left[\right.b_1}
     S_{a\left.\right]}{}^{b\left.\right]}
\end{equation}
     и получим
\begin{equation}
\label{re40vv}
    \begin{array}{c}
     K_{aa_1}{}^{bb_1}:=
    \frac{32\alpha}{5+8\alpha}
    S_{\left[\right.a}{}^{\left[\right.b}
    (M_{a_1\left.\right]}{}^{b_1\left.\right]}+
    \frac{1}{2}
    S_{a_1\left.\right]}{}^{b_1\left.\right]}M_k{}^r(S^{-1})_r{}^k)=
    2M_{\left[\right.a_1}{}^{\left[\right.b_1}
    S_{a\left.\right]}{}^{b\left.\right]}\spsd
    \end{array}
\end{equation}
     Положим
\begin{equation}
\label{re41vv}
     M_k{}^r:=\beta S_k{}^r\Rightarrow
     \beta=\frac{8\alpha}{5-8\alpha}M_k{}^r(S^{-1})_r{}^k\Rightarrow \alpha=\frac{1}{8}\spsd
\end{equation}
     Тогда из (\ref{re40vv}) следует
\begin{equation}
\label{re42vv}
    \begin{array}{c}
     K_{aa_1}{}^{bb_1}:=
     2\beta S_{\left[\right.a_1}{}^{\left[\right.b_1}
     S_{a\left.\right]}{}^{b\left.\right]}\spsd
    \end{array}
\end{equation}
    Подстановкой (\ref{re42vv}) в (\ref{re46ab}) находим,
    что $\beta=\pm 1$.\\
    \indent Подобным же образом рассматривается и пункт I).a).
    В этом случае преобразование $K_{aa_1}{}^{bb_1}$
    имеет вид
\begin{equation}
\label{re42vva}
    K_{aa_1}{}^{bb_1}:=
    \beta S_{ac}S_{a_1c_1}\varepsilon\varepsilon^{cc_1bb_1}\spsd
\end{equation}
     Заметим, что  множитель $\varepsilon$ можно включить в определение $S_{ac}$.\\
     \indent Таким образом от (\ref{re3b}) можно действительно прийти к (\ref{re1b}), чем
     и закончено доказательство обратной части теоремы.
     Поэтому преобразованию $S_a{}^b$($S_a{}_b$) будет соответствовать
     одно и только одно преобразование $K_\alpha{}^\beta$ и, наоборот,
     любому преобразованию $K_\alpha{}^\beta$ будет соответствовать
     два и только два преобразования $\pm S_a{}^b$($\pm S_a{}_b$)таких, что
     $det\parallel S_c{}^d \parallel=1\ (det\parallel S_{cd} \parallel=1)$.     \\

    Выясним, какое из преобразований соответствует собственным преобразованиям
    $K_\alpha{}^\beta$. Для этого рассмотрим следующее тождество
    $$
    K_\alpha{}^\beta K_\gamma{}^\delta K_\lambda{}^\mu
    K_\nu{}^\chi K_\pi{}^\omega K_\sigma{}^\xi
    e_{\beta\delta\mu\chi\omega\xi}=\pm
    e_{\alpha\gamma\lambda\nu\pi\sigma}\sps
    $$
\begin{equation}
\label{re5b}
    e_{\beta\delta\mu\chi\omega\xi}=
    e_{[\beta\delta\mu\chi\omega\xi]}\sps
    \hat e:=e_{123456}\spsd
\end{equation}
    При этом под $e_{\beta\delta\mu\chi\omega\xi}$ понимается
    6 - вектор, кососимметричный по всем индексам. Следовательно,
    мы можем получить эквивалентную (\ref{re5b}) запись
\begin{equation}
\label{re20v}
    \begin{array}{c}
    K_1{}^\beta K_2{}^\delta K_3{}^\mu
    K_4{}^\chi K_5{}^\omega K_6{}^\xi
    e_{\beta\delta\mu\chi\omega\xi}=\pm e_{123456} \ \ \Leftrightarrow \ \ \
    det\parallel K_\alpha{}^\beta \parallel=\pm 1\spsd
    \end{array}
\end{equation}
    Если $K_\alpha{}^\beta$ - собственное преобразование, то в
    (\ref{re5b}) выбирается знак ''+''.
    Это означает, что $det\parallel K_\alpha{}^\beta \parallel=1$, в ином
    случае (несобственные преобразования) выбирается знак ''-''. Поскольку
    для 4-вектора имеются тождества, следующие из формулы (\ref{re4})
    $$
    e_{\alpha\beta\gamma\delta}=
    e_{[\alpha\beta\gamma\delta]}\sps
    $$
\begin{equation}
\label{re6b}
    e_{\alpha\beta\gamma\delta}=A_{\alpha\beta b}{}^{a}
    A_{\gamma\delta d}{}^{c} e_a{}^b{}_c{}^d\sps
\end{equation}
    $$
    e_a{}^b:=\frac{1}{3}e_a{}^k{}_k{}^b\sps
    $$
    то, воспользовавшись его симметриями,
    можно получить разложение
    $$
    B_{\alpha\beta\gamma\delta r}{}^k:=
    A_{\alpha\beta r}{}^d A_{\gamma\delta d}{}^k+
    A_{\alpha\beta c}{}^k A_{\gamma\delta r}{}^c\sps
    $$
\begin{equation}
\label{re7b}
    e_{\alpha\beta\gamma\delta}:=B_{\alpha\beta\gamma\delta r}{}^k e_k{}^r\sps
\end{equation}
    $$
    e_k{}^k=0
    $$
    (доказательство в приложении (\ref{re56p}) - (\ref{re58p})).
    В свою очередь с помощью этих формул можно получить разложение и
    для 6-вектора
    $$
    e_{\alpha\beta\gamma\delta\lambda\mu}=
    A_{\alpha\beta b}{}^aA_{\gamma\delta d}{}^c
    A_{\lambda\mu l}{}^k\ e_a{}^b{}_c{}^d{}_k{}^l\sps
    $$
    $$
    e_a{}^b{}_c{}^d{}_k{}^l=\frac{i}{8}
    (2((4\delta_k{}^b\delta_c{}^l-\delta_k{}^l\delta_c{}^b)\delta_a{}^d
        +(4\delta_k{}^d\delta_a{}^l-\delta_k{}^l\delta_a{}^d)\delta_c{}^b)-
    $$
\begin{equation}
\label{re8b}
        -(4\delta_k{}^b\delta_a{}^l-\delta_k{}^l\delta_a{}^b)\delta_c{}^d-
        (4\delta_k{}^d\delta_c{}^l-\delta_k{}^l\delta_c{}^d)\delta_a{}^b)
\end{equation}
    (доказательство в приложении (\ref{re59p}) - (\ref{re68p})).
     Из (\ref{re8b}) вытекает тождество
\begin{equation}
\label{re78p}
       \begin{array}{c}
       e_{\alpha\gamma\lambda\nu\pi\sigma}=2
       \eta_{\left[\right.\alpha}{}^{bb_1}\eta_{\gamma |dd_1|}
       \eta_{\lambda}{}^{mm_1}\eta_{\nu |xx_1|}
       \eta_{\pi}{}^{rr_1}\eta_{\sigma\left.\right]ss_1}\cdot
       i \delta_r{}^d\delta_m{}^s\delta_b{}^x
       \delta_{b_1}{}^{d_1}\delta_{m_1}{}^{x_1}\delta_{r_1}{}^{s_1}=\\
       =\frac{1}{4}\eta_{\left[\right.\alpha}{}^{bb_1}\eta_{\gamma}{}^{dd_1}
       \eta_{\lambda}{}^{mm_1}\eta_{\nu}{}^{xx_1}
       \eta_{\pi}{}^{rr_1}\eta_{\sigma\left.\right]}{}^{ss_1}\cdot
       i\ \varepsilon_{rb_1dd_1}\varepsilon_{mr_1ss_1}\varepsilon_{bm_1xx_1}=\\
       =i(A_{\alpha\gamma b}{}^aA_{\pi\sigma a}{}^cA_{\lambda\nu c}{}^b+A_{\alpha\gamma b}{}^aA_{\lambda\nu c}{}^bA_{\pi\sigma a}{}^c)
       \end{array}
\end{equation}
       Применяя определение (\ref{re5b}) к собственным (несобственным)
       преобразованиям, получим из (\ref{re8b}) и (\ref{re78p})
\begin{equation}
\label{re79p}
       \begin{array}{c}
       S_a{}^b S_{a_1}{}^{b_1} S_c{}^d S_{c_1}{}^{d_1}\varepsilon_{bb_1dd_1}=
       \varepsilon_{aa_1cc_1}\ (
       S_{ab} S_{a_1b_1} S_{cd} S_{c_1d_1}\varepsilon^{bb_1dd_1}=
       \varepsilon^{-2}\varepsilon_{aa_1cc_1})\sps
       \end{array}
\end{equation}
     что даст тождество (\ref{re1b}).
     Отсюда следует, что (\ref{re42vv}) соответствует собственным, а (\ref{re42vva})
     соответствует несобственным преобразованиям $K_\alpha{}^\beta$.\\
     \indent И наконец, преобразования $S_a{}^b$ и $iS_a{}^b$ принадлежат
     одной и той же группе $SL(4,\mathbb C)$. Это означает, что можно в формуле
     (\ref{re42vv}) рассматривать только случай, когда $\beta=+1$.
     Поэтому группа $SL(4,\mathbb C)$ двулистно накрывает связную компоненту единицы
     группы $SO(6,\mathbb C)$ (ее мы обозначим через $SO^e(6,\mathbb C)$).\\
\end{proof}

\newpage
\subsubsection{\texorpdfstring{Вещественная реализация двулистного накрытия группы $SO(6,\mathbb C)$
               группой $SL(4,\mathbb C)$ в присутствии инволюции $S_\alpha{}^{\beta '}}{Вещественная реализация двулистного накрытия}$}
\begin{theoremr}
\label{theorem6r}
     Пусть в 6-мерном комплексном евклидовом пространстве $\mathbb C\mathbb R^6$
     задана инволюция вида
\begin{equation}
\label{re28b}
     S_\alpha{}^{\beta '}\bar S_{\beta '}{}^\gamma =
     \delta_\alpha{}^\gamma\sps
     S_\alpha{}^{\beta '}S_\gamma {}^{\delta '}\bar g_{\beta '\delta '}=
     g_{\alpha\gamma}\spsd
\end{equation}
     Определим
\begin{equation}
\label{re29b}
     \begin{array}{ccccc}
     & s_{aba'b'} & = & \bar\eta_{\beta 'a'b'}\eta^{\alpha}{}_{ab}
     S_\alpha{}^{\beta '}\sps &
     \end{array}
\end{equation}
     тогда будут выполнены следующие соотношения
\begin{equation}
\label{re29ba}
     s_{aba'b'}=\bar s_{a'b'ab}\sps
     s_{ab}{}^{a'b'}\bar s_{a'b'}{}^{cd}=2\delta_{ab}^{cd}\sps
\end{equation}
     и будет существовать два и только два разложения
\begin{equation}
\label{re30b}
     \begin{array}{ccccccc}
     I). &  s_{ab}{}^{a'b'}   =
     2s_{\left[\right.a}{}^{a'}  s_{b\left.\right]}{}^{b'}\sps
       & s_a{}^{b'}\bar s_{b'}{}^c=\pm\delta_a{}^c\sps\\[2ex]
     II). &    s_{ab}{}_{a'b'}    =
     2s_{\left[\right.a}{}_{|a'|}s_{b\left.\right]}{}_{b'}\sps
      &  s_{ab'}=\pm \bar s_{b'a}\spsd &
     \end{array}
\end{equation}
     Кроме того, для вещественного случая будут
     верны следующие тождества
\begin{equation}
\label{re22v}
     \begin{array}{cccc}
     I).  &\bar \eta_i{}^{a'b'}=\eta_j{}^{cd}s_c{}^{a'}s_d{}^{b'} &
     II).&\bar \eta_i{}^{a'b'}=\eta_{jcd}s^{ca'}s^{db'}\sps \\
     & \bar A_{ij a'}{}^{b'}=A_{ijc}{}^d\bar s_{a'}{}^cs_d{}^{b'} &
     & \bar A_{ij a'}{}^{b'}=-A_{ijc}{}^d s_{da'}s^{cb'}\sps\\
     \end{array}
\end{equation}
\end{theoremr}
\vspace{1cm}
\begin{proof}
      Доказательство разложения (\ref{re30b}) проводится также
      как и в предыдущей теореме. Все изменения сводятся только
      к замене преобразования $K_\alpha{}^\beta$ на преобразование
      $S_\alpha{}^{\beta '}$ так, что аналогом (\ref{re3b}) служит
\begin{equation}
\label{re23v}
       S_\alpha{}^{\beta '}S_\gamma{}^{\delta '} \bar g_{\beta '\delta '}=
       g_{\alpha\gamma}\sps
\end{equation}
      что даст уравнение аналогичное (\ref{re1b}) (соответствующий множитель включен в определение спин-тензора s)
\begin{equation}
\label{re2vvv}
    \begin{array}{c}
    s_a{}^{b'} s_{a_1}{}^{{b '}_1} s_c{}^{d'} s_{c_1}{}^{{d'}_1}
    \bar \varepsilon_{b'{b'}_1d'{d'}_1}=
    \varepsilon_{aa_1cc_1}\sps\\[2ex]
    (s_{ab'} s_{a_1{b'}_1} s_{cd'} s_{c_1{d'}_1}
    \bar\varepsilon^{b'{b'}_1d'{d'}_1}=
    \varepsilon_{aa_1cc_1})\spsd
    \end{array}
\end{equation}
      Из (\ref{re28b}) и (\ref{re23v}) можно получить
\begin{equation}
\label{re3vvv}
     S_\alpha{}^{\beta '}=\bar S^{\beta '}{}_\alpha\spsd
\end{equation}
     Из этого следует
\begin{equation}
\label{re4vvv}
     \bar s_{a'b'ab}=\bar \eta{}^{\alpha '}{}_{a'b'}
     \eta_{\beta ab} \bar S_{\alpha '}{}^\beta=
     \bar \eta{}_{\beta 'a'b'}
     \eta^\alpha{}_{ab} \bar S^{\beta '}{}_\alpha=
     s_{aba'b'}\spsd
\end{equation}
     \indent Отметим, что в касательном расслоении $\tau^\mathbb C(\tau^\mathbb R)$
     существует метрический тензор $g_{\alpha\beta}(g_{ij})$,
     с помощью которого осуществляется спуск и подъем
     одиночных индексов. В расслоении $A^\mathbb C(S)$ аналогичную
     роль выполняет тензор $\varepsilon_{abcd}$. За тензор
     $\bar g_{\alpha '\beta '}$($\bar\varepsilon_{a'b'c'd'}$)
     принимается тензор, координаты которого сопряжены координатам
     тензора $g_{\alpha\beta}(\varepsilon_{abcd})$.\\
     \indent Рассмотрим цепочку тождеств, следующих из (\ref{re28b})
\begin{equation}
\label{re5vvv}
     \begin{array}{c}
     S_\alpha{}^{\beta '}\bar S_{\beta '}{}^\gamma =
     \delta_\alpha{}^\gamma\sps\\[2ex]
     \frac{1}{4}
     \eta_\alpha{}^{aa_1}\bar\eta^{\beta '}{}_{b'{b'}_1}
     s_{aa_1}{}^{b'{b'}_1}
     \frac{1}{4}
     \bar\eta_{\beta '}{}^{d'{d'}_1}\eta^\gamma{}_{cc_1}
     \bar s_{d'{d'}_1}{}^{cc_1}=
     \frac{1}{4}
     \eta_\alpha{}^{aa_1}\eta^\gamma{}_{cc_1}
     2\delta_{\left[\right.a}{}^c\delta_{a_1\left.\right]}{}^{c_1}\sps\\[2ex]
     \frac{1}{2}s_{aa_1}{}^{b'{b'}_1}\bar s_{b'{b'}_1}{}^{cc_1}=
     \delta_{aa_1}^{cc_1}\spsd
\end{array}
\end{equation}
     Исследуем теперь  случай II). Из последнего тождества
     (\ref{re5vvv}) получим
\begin{equation}
\label{re6vvv}
     \begin{array}{c}
     s_{ac'}s_{a_1{c'}_1}\bar \varepsilon^{c'{c'}_1b'{b'}_1}
     \cdot
     \bar s_{b'd}\bar s_{{b'}_1d_1} \varepsilon^{dd_1ff_1}=
     2\delta_{ff_1}^{cc_1}\sps\\[2ex]
     s_{ac'}s_{a_1{c'}_1}
     \bar s_{b'd}\bar s_{{b'}_1d_1} \bar\varepsilon^{c'{c'}_1b'{b'}_1}=
     \varepsilon_{aa_1dd_1}\spsd
\end{array}
\end{equation}
     Определим $s^{kl'}$ следующим образом
\begin{equation}
\label{re7vvv}
      s^{kl'}s_{km'}=\delta_{m'}{}^{l'}
\end{equation}
      так, что
\begin{equation}
\label{re8vvv}
     s^{kl'}s^{k_1{l'}_1}s^{mn'}s^{m_1{n'}_1}
     \varepsilon_{kk_1mm_1}=
     \bar \varepsilon^{l'{l'}_1n'{n'}_1}\spsd
\end{equation}
     Домножим (\ref{re6vvv}) на  $s^{ak'} s^{a_1{k'}_1}s^{dn'}s^{d_1{n'}_1}$
     и с учетом (\ref{re8vvv}) получим
\begin{equation}
\label{re9vvv}
     s^{dn'}\bar s_{b'd}s^{d_1{n'}_1}\bar s_{{b'}_1d_1}
     \bar \varepsilon^{k'{k'}_1b'{b'}_1}=
     \bar \varepsilon^{k'{k'}_1n'{n'}_1}\spsd
\end{equation}
     Положим
\begin{equation}
\label{re10vvv}
     \bar N_{b'}{}^{n'}:=s^{dn'}\bar s_{b'd}\sps
\end{equation}
     тогда (\ref{re9vvv}) перепишется так
\begin{equation}
\label{re11vvv}
     \bar N_{\left[\right.b'}{}^{n'}
     \bar N_{{b'}_1\left.\right]}{}^{{n'}_1}=
     \delta_{\left[\right.b'}{}^{n'}
     \delta_{{b'}_1\left.\right]}{}^{{n'}_1}\spsd
\end{equation}
     Отсюда следует (доказательство в приложении
     (\ref{re69p}) - (\ref{re74p}))
\begin{equation}
\label{re12vvv}
     \bar N_{b'}{}^{n'}=s^{dn'}\bar s_{b'd}=n\delta_{b'}{}^{n'}=
     ns^{dn'}s_{db'}\sps n^2=1\spsd
\end{equation}
     Поэтому из (\ref{re7vvv}) вытекает
\begin{equation}
\label{re13vvv}
     \bar s_{b'd}=\pm s_{db'}\spsd
\end{equation}
     Подобным же образом разбирается случай I). Из
     тождества (\ref{re5vvv}) следует
\begin{equation}
\label{re14vvv}
     s_{\left[\right.a}{}^{b'}
     \delta_{a_1\left.\right]}{}^{{b'}_1}
     \bar s_{b'}{}^c\bar s_{{b_1}'}{}^{c_1}=
     \delta_{\left[\right.a}{}^c
     \delta_{a_1\left.\right]}{}^{c_1}\spsd
\end{equation}
     Положим
\begin{equation}
\label{re15vvv}
     N_a{}^c:=s_a{}^{b'}\bar s_{b'}{}^c\sps
\end{equation}
     то получим
\begin{equation}
\label{re16vvv}
     \bar N_{\left[\right.a}{}^c
     \bar N_{a_1\left.\right]}{}^{c_1}=
     \delta_{\left[\right.a}{}^c
     \delta_{a_1\left.\right]}{}^{c_1}\spsd
\end{equation}
     Откуда и следует
\begin{equation}
\label{re17vvv}
     N_a{}^c=n\delta_a{}^c=s_a{}^{b'}\bar s_{b'}{}^c \sps n^2=1\sps
\end{equation}
     определяя окончательно следующее соотношение
\begin{equation}
\label{re17vvva}
     s_a{}^{b'}\bar s_{b'}{}^c=\pm\delta_a{}^c\spsd
\end{equation}

      И нам останется доказать только (\ref{re22v}). Воспользуемся
      оператором вложения $H_i{}^\alpha$ и инволюцией
      $S_\alpha{}^{\beta '}$, определенными по формуле (\ref{re5v}).
      Для случая II). имеем
\begin{equation}
\label{re24v}
      \begin{array}{c}
      \bar \eta_i{}^{a'b'}=\bar H_i{}^{\alpha '}
      \bar \eta_{\alpha '}{}^{a'b'}=
      \bar H_i{}^{\alpha '}\bar S_{\alpha '}{}^\beta
      \eta_{\beta cd} s^{ca'}s^{db'}=\\[2ex]
      = H_i{}^\beta \eta_{\beta cd} s^{ca'}s^{db'}=
      \eta_{i cd} s^{ca'}s^{db'}\sps
      \end{array}
\end{equation}
\begin{equation}
\label{re25v}
      \begin{array}{c}
      \bar A_{ij a'}{}^{b'}=\bar \eta_{\left[\right. i}{}^{b'k'}
      \bar \eta_{j\left.\right] a'k'}=
      H_i{}^{\gamma}H_j{}^\delta \eta_{\left[\right. \gamma |bk|}
      \eta_{\delta\left.\right]}{}^{ak} s_{aa'}s^{bb'}=
      -A_{ij b}{}^a s_{aa'}s^{bb'}\spsd
      \end{array}
\end{equation}
      В случае I). доказательство такое
\begin{equation}
\label{re24va}
      \begin{array}{c}
      \bar \eta_i{}^{a'b'}=\bar H_i{}^{\alpha '}
      \bar \eta_{\alpha '}{}^{a'b'}=
      \bar H_i{}^{\alpha '}\bar S_{\alpha '}{}^\beta
      \eta_\beta{}^{cd} s_c{}^{a'}s_d{}^{b'}=\\[2ex]
      = H_i{}^\beta \eta_\beta{}^{cd} s_c{}^{a'}s_d{}^{b'}=
      \eta_i{}^{cd} s_c{}^{a'}s_d{}^{b'}\sps
      \end{array}
\end{equation}
\begin{equation}
\label{re25va}
      \begin{array}{c}
      \bar A_{ij a'}{}^{b'}=\bar \eta_{\left[\right. i}{}^{b'k'}
      \bar \eta_{j\left.\right] a'k'}=
      H_i{}^{\gamma}H_j{}^\delta \eta_{\left[\right. \gamma}{}^{ck}
      \eta_{\delta\left.\right] dk} s_c{}^{b'}\bar s_{a'}{}^d=\\[2ex]
      =\eta_{\left[\right. i}{}^{ck} \eta_{j\left.\right]}{}_{dk}
      s_c{}^{b'}\bar s_{a'}{}^d=
      A_{ij d}{}^c s_c{}^{b'}\bar s_{a'}{}^d\spsd
      \end{array}
\end{equation}
\end{proof}

\subsubsection{\texorpdfstring{Вложение $\mathbb R^6_{(2,4)}\subset \mathbb C\mathbb R^6$  в специальном базисе}{Действительное вложение}}$ $
    \indent Рассмотрим теперь в качестве примера вложение вещественного пространства
    $\mathbb R^6_{(2,4)}$ в комплексное пространство $\mathbb C\mathbb R^6$.
    В этом случае у нас появится возможность
    с помощью тензора $s_{aa'}$ осуществить отождествление
    верхних штрихованных с нижними нештрихованными индексами.
    Рассмотрим цепочку тождеств
\begin{equation}
\label{re31b}
    \begin{array}{c}
    K_i{}^j=\bar K_i{}^j\sps
    K_i{}^j:=H_i{}^\alpha H^j{}_\beta K_\alpha{}^\beta\sps\\[2ex]
    \eta_j{}^{ab}K_i{}^j\eta^i{}_{cd}=
    \eta_j{}^{ab}\bar K_i{}^j\eta^i{}_{cd}\sps\\[2ex]
    2S_{\left[\right.c}{}^a S_{d\left.\right]}{}^b=\frac{1}{4}
    \eta_j{}^{ab}\bar\eta^{jm'n'}2\bar S_{m'}{}^{k'}\bar S_{n'}{}^{l'}
    \bar\eta_{ik'l'}\eta^{i}{}_{cd}\sps
    \end{array}
\end{equation}
\begin{equation}
\label{re31bb}
    \begin{array}{c}
    S_{\left[\right.c}{}^aS_{d\left.\right]}{}^b=\frac{1}{4}
    s^{abm'n'}\bar S_{\left[\right.m'}{}^{k'}\bar S_{n'\left.\right]}{}^{l'}
    s_{cdk'l'}\sps\\[2ex]
    S_{\left[\right.c}{}^aS_{d\left.\right]}{}^b=
    s^{am'}s^{bn'}\bar S_{\left[\right.m'}{}^{k'}\bar S_{n'\left.\right]}{}^{l'}
    s_{ck'}s_{dl'}\sps\\[2ex]
    s^{lk'}\bar S_{k'}{}^{m'}s_{am'}
    S_{\left[\right.c}{}^aS_{d\left.\right]}{}^b
    s_{bn'}\bar S_{r'}{}^{n'}s^{sr'}=
    \delta_{\left[\right.c}{}^a\delta_{d\left.\right]}{}^b\spsd
    \end{array}
\end{equation}
    Определим
\begin{equation}
\label{re31bc}
    N_c{}^l:=s^{lk'}\bar S_{k'}{}^{m'}s_{am'}S_c{}^a
\end{equation}
    и получим
\begin{equation}
\label{re31bd}
    N_{\left[\right.c}{}^aN_{d\left.\right]}{}^b=
    \delta_{\left[\right.c}{}^a\delta_{d\left.\right]}{}^b\spsd
\end{equation}
    Откуда будет следовать выражение (доказательство в приложении
     (\ref{re69p}) - (\ref{re74p}))
\begin{equation}
\label{re32b}
    N_c{}^l=s^{lk'}\bar S_{k'}{}^{m'}s_{am'}S_c{}^a=n\delta_c{}^l\sps
    n=\pm 1\spsd
\end{equation}
    Выбирая знак ''+'' в (\ref{re32b}), мы получаем
    преобразования из  группы,  изоморфной группе $SU(2,2)$, которая
    будет, как видно из вышесказанного, двулистно накрывать
    связную компоненту единицы группы $SO^e(2,4)$. Эта
    компонента определится следующими условиями
\begin{equation}
\label{re33b}
    1).\ \ det\| K_\alpha{}^\beta\|=1\ \ \alpha ,\beta=\overline{1,6}\sps
    2).\ \ det\| K_\alpha{}^\beta\|>0\ \ \alpha ,\beta=\overline{1,2}\spsd
\end{equation}
    Если в (\ref{re32b}) выбрать ''-'', то знак в 2). из (\ref{re33b})
    изменится на противоположный. Далее, чтобы лучше уяснить как это
    происходит на практике,
    воспользуемся представлением полученных результатов
    в специальном базисе. Для этого определим
    базис пространства $\mathbb C\mathbb R^6$  следующим образом
\begin{equation}
\label{re26v}
    \begin{array}{lccc}
    t^\alpha & =(1,0,0,0,0,0)\sps &
    v^\alpha & =(0,1,0,0,0,0)\sps \\
    w^\alpha & =(0,0,i,0,0,0)\sps &
    x^\alpha & =(0,0,0,i,0,0)\sps \\
    y^\alpha & =(0,0,0,0,i,0)\sps &
    z^\alpha & =(0,0,0,0,0,i)\spsd
    \end{array}
\end{equation}
     Пусть в этом базисе матрица метрического тензора $g_{\alpha\beta}$
     имеет вид
\begin{equation}
\label{re27v}
     \parallel g_{\alpha\beta} \parallel=
     \left(
       \begin{array}{cccccc}
       1 & 0 & 0 & 0 & 0 & 0\\
       0 & 1 & 0 & 0 & 0 & 0\\
       0 & 0 & 1 & 0 & 0 & 0\\
       0 & 0 & 0 & 1 & 0 & 0\\
       0 & 0 & 0 & 0 & 1 & 0\\
       0 & 0 & 0 & 0 & 0 & 1\\
       \end{array}
    \right)\spsd
\end{equation}
    Определим вещественную реализацию вложения $\mathbb R^6_{(2,4)}\subset \mathbb C\mathbb R^6$
    оператором $H_i{}^\alpha$
\begin{equation}
\label{re28v}
     \parallel H_i{}^\alpha \parallel=
     \left(
       \begin{array}{cccccc}
       1 & 0 & 0 & 0 & 0 & 0\\
       0 & 1 & 0 & 0 & 0 & 0\\
       0 & 0 & i & 0 & 0 & 0\\
       0 & 0 & 0 & i & 0 & 0\\
       0 & 0 & 0 & 0 & i & 0\\
       0 & 0 & 0 & 0 & 0 & i\\
       \end{array}
    \right)\sps
     \parallel H^i{}_\alpha \parallel=
     \left(
       \begin{array}{cccccc}
       1 & 0 &  0 &  0 &  0 &  0\\
       0 & 1 &  0 &  0 &  0 &  0\\
       0 & 0 & -i &  0 &  0 &  0\\
       0 & 0 &  0 & -i &  0 &  0\\
       0 & 0 &  0 &  0 & -i &  0\\
       0 & 0 &  0 &  0 &  0 & -i\\
       \end{array}
    \right)\spsd
\end{equation}
    Тогда базис (\ref{re26v}) будет самосопряжен относительно инволюции
    $S_\alpha{}^{\beta '}$ вида
\begin{equation}
\label{re29v}
     \parallel S_\alpha{}^{\beta '} \parallel=
     \left(
       \begin{array}{cccccc}
       1 & 0 &  0 &  0 &  0 &  0\\
       0 & 1 &  0 &  0 &  0 &  0\\
       0 & 0 & -1 &  0 &  0 &  0\\
       0 & 0 &  0 & -1 &  0 &  0\\
       0 & 0 &  0 &  0 & -1 &  0\\
       0 & 0 &  0 &  0 &  0 & -1\\
       \end{array}
    \right)\spsd
\end{equation}
\newpage
    \noindent Поэтому в пространстве $\mathbb R^6_{(2,4)}$ индуцируемый метрический тензор
    $g_{ij}$ будет иметь матрицу
\begin{equation}
\label{re30v}
     \parallel g_{ij} \parallel=
     \left(
       \begin{array}{cccccc}
       1 & 0 &  0 &  0 &  0 &  0\\
       0 & 1 &  0 &  0 &  0 &  0\\
       0 & 0 & -1 &  0 &  0 &  0\\
       0 & 0 &  0 & -1 &  0 &  0\\
       0 & 0 &  0 &  0 & -1 &  0\\
       0 & 0 &  0 &  0 &  0 & -1\\
       \end{array}
    \right)=\parallel H_i{}^\alpha H_j{}^\beta g_{\alpha\beta} \parallel
\end{equation}
\vspace{1cm}
    в базисе
\vspace{1cm}
\begin{equation}
\label{re31v}
    \begin{array}{lcccccc}
    t^i  &= H^i{}_\alpha  t^\alpha  &=(1,0,0,0,0,0) & \sps &
    v^i  &= H^i{}_\alpha  v^\alpha  &=(0,1,0,0,0,0) \sps\\[3ex]
    w^i  &= H^i{}_\alpha  w^\alpha  & =(0,0,1,0,0,0) & \sps &
    x^i  &= H^i{}_\alpha  x^\alpha  & =(0,0,0,1,0,0) \sps\\[3ex]
    y^i  &= H^i{}_\alpha  y^\alpha  & =(0,0,0,0,1,0) & \sps &
    z^i  &= H^i{}_\alpha  z^\alpha  & =(0,0,0,0,0,1)\spsd
    \end{array}
\end{equation}
\vspace{1cm}
    Определим векторный базис в расслоении $A^\mathbb C(S)$ так
    $$
    X^a=(1,0,0,0) \sps
    Y^a=(0,1,0,0) \sps
    $$
\vspace{1cm}
\begin{equation}
\label{re51}
    Z^a=(0,0,1,0) \sps
    T^a=(0,0,0,1)\sps
\end{equation}
\vspace{1cm}
    $$
    \varepsilon_{abcd}X^aY^bZ^cT^d=1\sps \varepsilon=1\spsd
    $$
\vspace{1cm}
    Тогда в базисах (\ref{re30v}) и (\ref{re51}) имеет
    место разложение
    $$
    R^{ab}=2(R^{12}X^{\left[ \right.a}Y^{b\left. \right]}+
             R^{13}X^{\left[ \right.a}Z^{b\left. \right]}+
             R^{14}X^{\left[ \right.a}T^{b\left. \right]}+
    $$
\newpage
\begin{equation}
\label{re52}
             +R^{23}Y^{\left[ \right.a}Z^{b\left. \right]}+
             R^{24}Y^{\left[ \right.a}T^{b\left. \right]}+
             R^{34}Z^{\left[ \right.a}T^{b\left. \right]})=
\end{equation}
    $$
    \begin{array}{c}
    =\frac{1}{\sqrt{2}}
    (R^{12}+R^{34})\cdot {\sqrt{2}}(X^{\left[ \right.a}Y^{b\left. \right]}+
                    Z^{\left[ \right.a}T^{b\left. \right]})+\\ \\
    +\frac{1}{\sqrt{2}}
    (R^{12}-R^{34})\cdot {\sqrt{2}}(X^{\left[ \right.a}Y^{b\left. \right]}-
                    Z^{\left[ \right.a}T^{b\left. \right]})+\\ \\
    +\frac{1}{\sqrt{2}}
    (R^{13}+R^{24})\cdot {\sqrt{2}}(X^{\left[ \right.a}Z^{b\left. \right]}+
                    Y^{\left[ \right.a}T^{b\left. \right]})+\\ \\
    +\frac{i}{\sqrt{2}}
    (R^{13}-R^{24})\cdot (-i{\sqrt{2}})(X^{\left[ \right.a}Z^{b\left. \right]}-
                    Y^{\left[ \right.a}T^{b\left. \right]})+\\ \\
    +\frac{-i}{\sqrt{2}}
    (R^{14}+R^{23})\cdot i{\sqrt{2}}(X^{\left[ \right.a}T^{b\left. \right]}+
                    Y^{\left[ \right.a}Z^{b\left. \right]})+\\ \\
    +\frac{-i}{\sqrt{2}}
    (R^{14}-R^{23})\cdot i{\sqrt{2}}(X^{\left[ \right.a}T^{b\left. \right]}-
                    Y^{\left[ \right.a}Z^{b\left. \right]})=\\ \\
    =(Tt^i+Vv^i+Ww^i+Xx^i+Yy^i+Zz^i)\cdot
    \eta_i{}^{ab}=r^i\eta_i{}^{ab}\spsd
    \end{array}
    $$
    Поэтому мы можем положить
\begin{equation}
\label{re32v}
    \begin{array}{cccccc}
    v^i\eta_i{}^{ab} & := & {\sqrt{2}}(X^{\left[ \right.a}Y^{b\left. \right]}+
                    Z^{\left[ \right.a}T^{b\left. \right]})\sps &
    w^i\eta_i{}^{ab} & := & {\sqrt{2}}(X^{\left[ \right.a}Y^{b\left. \right]}-
                    Z^{\left[ \right.a}T^{b\left. \right]})\sps \\[2ex]
    y^i\eta_i{}^{ab} & := & {\sqrt{2}}(X^{\left[ \right.a}Z^{b\left. \right]}+
                    Y^{\left[ \right.a}T^{b\left. \right]})\sps &
    x^i\eta_i{}^{ab} & := & -{\sqrt{2}i}(X^{\left[ \right.a}Z^{b\left. \right]}-
                    Y^{\left[ \right.a}T^{b\left. \right]})\sps \\[2ex]
    z^i\eta_i{}^{ab} & := & {\sqrt{2}i}(X^{\left[ \right.a}T^{b\left. \right]}+
                    Y^{\left[ \right.a}Z^{b\left. \right]})\sps &
    t^i\eta_i{}^{ab} & := & {\sqrt{2}i}(X^{\left[ \right.a}T^{b\left. \right]}-
                    Y^{\left[ \right.a}Z^{b\left. \right]})\sps
    \end{array}
\end{equation}
    что определит операторы Нордена $\eta_i{}^{aa_1}$ в этих базисах как
\begin{equation}
\label{re56}
    \begin{array}{crcrcrcr}
    \eta_2{}^{12}= & \frac{1}{\sqrt{2}}\sps &
    \eta_2{}^{34}= & \frac{1}{\sqrt{2}}\sps &
    \eta_3{}^{12}= & \frac{1}{\sqrt{2}}\sps &
    \eta_3{}^{34}= & -\frac{1}{\sqrt{2}}\sps \\
    \eta_1{}^{14}= & \frac{i}{\sqrt{2}}\sps &
    \eta_1{}^{23}= & -\frac{i}{\sqrt{2}}\sps  &
    \eta_6{}^{14}= & \frac{i}{\sqrt{2}}\sps &
    \eta_6{}^{23}= & \frac{i}{\sqrt{2}}\sps \\
    \eta_5{}^{13}= & \frac{1}{\sqrt{2}}\sps  &
    \eta_5{}^{24}= & \frac{1}{\sqrt{2}}\sps  &
    \eta_4{}^{13}= & -\frac{i}{\sqrt{2}}\sps  &
    \eta_4{}^{24}= & \frac{i}{\sqrt{2}}\spsd
    \end{array}
\end{equation}
    Из (\ref{re32v}) вытекают следующие тождества
\begin{equation}
\label{re33v}
    \begin{array}{cc}
    T=\frac{i}{\sqrt{2}}(R^{23}-R^{14})\sps &
    V=\frac{1}{\sqrt{2}}(R^{12}+R^{34})\sps \\[2ex]
    W=\frac{1}{\sqrt{2}}(R^{12}-R^{34})\sps &
    X=\frac{i}{\sqrt{2}}(R^{13}-R^{24})\sps \\[2ex]
    Y=\frac{1}{\sqrt{2}}(R^{13}+R^{24})\sps &
    Z=\frac{-i}{\sqrt{2}}(R^{14}+R^{23})\sps\\[4ex]
    R^{12}=\frac{1}{\sqrt{2}}(V+W)\sps  &
    R^{13}=\frac{1}{\sqrt{2}}(Y-iX)\sps \\[2ex]
    R^{14}=\frac{i}{\sqrt{2}}(T+Z)\sps  &
    R^{23}=\frac{i}{\sqrt{2}}(Z-T)\sps  \\[2ex]
    R^{24}=\frac{1}{\sqrt{2}}(Y+iX)\sps &
    R^{34}=\frac{1}{\sqrt{2}}(V-W)\sps
    \end{array}
\end{equation}
    так, что обратные величины $\eta^i{}_{aa_1}$ имеют вид
\begin{equation}
\label{re57}
    \begin{array}{crcrcrcr}
    \eta^2{}_{12}= & \frac{1}{\sqrt{2}} \sps &
    \eta^2{}_{34}= & \frac{1}{\sqrt{2}} \sps &
    \eta^3{}_{12}= & \frac{1}{\sqrt{2}} \sps &
    \eta^3{}_{34}= & -\frac{1}{\sqrt{2}}\sps \\
    \eta^1{}_{14}= & -\frac{i}{\sqrt{2}}\sps &
    \eta^1{}_{23}= & \frac{i}{\sqrt{2}} \sps &
    \eta^6{}_{14}= & -\frac{i}{\sqrt{2}}\sps &
    \eta^6{}_{23}= & -\frac{i}{\sqrt{2}}\sps \\
    \eta^5{}_{13}= & \frac{1}{\sqrt{2}} \sps &
    \eta^5{}_{24}= & \frac{1}{\sqrt{2}} \sps &
    \eta^4{}_{13}= & \frac{i}{\sqrt{2}} \sps &
    \eta^4{}_{24}= & -\frac{i}{\sqrt{2}}\spsd
    \end{array}
\end{equation}
    И, кроме того, будут верны равенства
\begin{equation}
\label{re35v}
    \begin{array}{ccccccccccc}
    \overline{R^{23}} & = & -R^{23} & = &  R_{41} & \sps &
    \overline{R^{34}} & = &  R^{34} & = &  R_{12}   \sps\\[2ex]
    \overline{R^{12}} & = &  R^{12} & = &  R_{34} & \sps &
    \overline{R^{13}} & = &  R^{24} & = &  R_{31}   \spsd
    \end{array}
\end{equation}
    Выберем ковекторный базис, согласуя его с базисом (\ref{re51})
    при выполнении (\ref{re35v}), следующим образом
\begin{equation}
\label{re34v}
\begin{array}{l}
    X_a=s_{aa'}\bar X^{a'}=(0,0,1,0)\sps
    Y_a=s_{aa'}\bar Y^{a'}=(0,0,0,1)\sps\\
    Z_a=s_{aa'}\bar Z^{a'}=(1,0,0,0)\sps
    T_a=s_{aa'}\bar T^{a'}=(0,1,0,0)\spsd
\end{array}
\end{equation}
    Этим определится эрмитовый поляритет, которым
    наделено расслоение $A^\mathbb C(S)$ (его база - $\mathbb R^6_{(2,4)}$), с матрицей
\begin{equation}
\label{re36v}
       \parallel s_{aa'} \parallel=
       \left(
       \begin{array}{cccc}
       0 & 0 & 1 & 0\\
       0 & 0 & 0 & 1\\
       1 & 0 & 0 & 0\\
       0 & 1 & 0 & 0
       \end{array}
       \right)\spsd
\end{equation}
\newpage
     \noindent Отсюда следует, что пффафиан бивектора $R^{ab}$ имеет вид
\begin{equation}
\label{re53}
    \begin{array}{c}
    pf(R):=\frac{1}{2}R^{ab}R_{ab}=\\
    =2(R^{12}R^{34}-R^{13}R^{24}+R^{14}R^{23})=\\
    =T^2+V^2-W^2-X^2-Y^2-Z^2\spsd
    \end{array}
\end{equation}
      Вид матрицы тензора s в некотором специальном базисе
      для остальных случаев вложения приведен в таблице \ref{table1r}.\\

\begin{table}
\caption{Вид матрицы тензора s для действительных  вложений.}
\label{table1r}
\begin{center}
\begin{tabular}{|c|c|c|c|c|}
      \hline
      П-п & Пространство & s &
      s в спец. базисе & Изоморфизм\\
      \hline
      1  &
      $
      \begin{array}{c}
      R^6   \\
      \ll ++++++ \gg
      \end{array}
      $  &
      $s_{kk'}$ &
      $\left(
      \begin{array}{cccc}
      1 & 0 & 0 & 0\\
      0 & 1 & 0 & 0\\
      0 & 0 & 1 & 0\\
      0 & 0 & 0 & 1
      \end{array}
      \right)
      $&
      $\small{SU(4)/\{\pm 1\} \cong SO^e(6)}$\\
      \hline
      2  &
      $
      \begin{array}{c}
      R^6_{(1,5)}   \\
      \ll +----- \gg
      \end{array}
      $  &
      $s_k{}^{k'}$ &
      $\left(
      \begin{array}{cccc}
      0 & 1 & 0 & 0\\
     -1 & 0 & 0 & 0\\
      0 & 0 & 0 & 1\\
      0 & 0 &-1 & 0
      \end{array}
      \right)
      $&
      $\small{SL(2,H)/\{\pm 1\} \cong SO^e(1,5)}$\\
      \hline
      3  &
      $
      \begin{array}{c}
      R^6_{(2,4)}   \\
      \ll ++---- \gg
      \end{array}
       $ &
       $s_{kk'}$ &
       $\left(
       \begin{array}{cccc}
       0 & 0 & 1 & 0\\
       0 & 0 & 0 & 1\\
       1 & 0 & 0 & 0\\
       0 & 1 & 0 & 0
       \end{array}
       \right)
       $&
      $\small{SU(2,2)/\{\pm 1\} \cong SO^e(2,4)}$\\
      \hline
       4  &
      $
      \begin{array}{c}
      R^6_{(3,3)}   \\
      \ll +++--- \gg
      \end{array}
       $ &
       $s_k{}^{k'}$ &
       $\left(
       \begin{array}{cccc}
       1 & 0 & 0 & 0\\
       0 & 1 & 0 & 0\\
       0 & 0 & 1 & 0\\
       0 & 0 & 0 & 1
       \end{array}
       \right)
       $&
      $\small{SL(4,R)/\{\pm 1\} \cong SO^e(3,3)}$\\
      \hline
\end{tabular}
\end{center}
\end{table}

\subsubsection{Инфинитезимальные преобразования}$ $
    \indent Пусть имеется $K_\alpha{}^\beta(\lambda)$ -
    однопараметрическое семейство, удовлетворяющее условию
\begin{equation}
\label{re34b}
    g_{\alpha\delta}=K_\alpha{}^\beta(\lambda)
    K_\delta{}^\gamma(\lambda)g_{\beta\gamma}\sps
    K_\alpha{}^\beta(0)=\delta_\alpha{}^\beta\spsd
\end{equation}
\newpage
    \noindent Инфинитезимальные преобразования, ему соответствующие, определим так
\begin{equation}
\label{re35b}
     T_\delta{}^\gamma=\left.\left[\frac{d}{d\lambda}
     K_\delta{}^\gamma(\lambda)\right]\right|_{\lambda =0}\spsd
\end{equation}
     Тогда из (\ref{re34b}) следует
\begin{equation}
\label{re36b}
     T_{\alpha\beta}=-T_{\beta\alpha}\spsd
\end{equation}
     Согласно \cite[т. 1, стр. 224]{Penrose1r} из (\ref{re34b})
     всегда следует (\ref{re36b}), а из
     (\ref{re36b}) $\ll$экспонентцированием$\gg$
\begin{equation}
\label{re37b}
     K_\delta{}^\gamma(\lambda ):=exp(\lambda T_\delta{}^\gamma)
\end{equation}
     можно получить (\ref{re34b}).\\
     \indent Пусть, кроме того задано однопараметрическое семейство
     $S_a{}^b(\lambda )$ такое, что
\begin{equation}
\label{re38b}
      S_a{}^b(\lambda )S_c{}^d(\lambda )S_{a_1}{}^{b_1}(\lambda )
      S_{c_1}{}^{d_1}(\lambda )\varepsilon_{bb_1dd_1}=
      \varepsilon_{aa_1cc_1}\sps
      S_a{}^b(0)=\delta_a{}^b\spsd
\end{equation}
      Продифференцируем его, предварительно полагая
\begin{equation}
\label{re39b}
     T_a{}^b:=\left.\left[\frac{d}{d\lambda}
     S_a{}^b(\lambda)\right]\right|_{\lambda =0}\sps
\end{equation}
     и получим
\begin{equation}
\label{re40b}
     \varepsilon_{b\left[\right.a_1cc_1}T_{a\left.\right]}{}^b=0\ \ \
     \Leftrightarrow\ \ T_a{}^a=0\spsd
\end{equation}
     Верно обратное. Пусть
\begin{equation}
\label{re41b}
     S_a{}^b(\lambda ):=exp(\lambda T_a{}^b)\sps
\end{equation}
     тогда будет выполнено следующее тождество
\begin{equation}
\label{re42b}
      \begin{array}{c}
      S_a{}^b(\lambda )S_c{}^d(\lambda )S_{a_1}{}^{b_1}(\lambda )
      S_{c_1}{}^{d_1}(\lambda )\varepsilon_{bb_1dd_1}=\\ \\
      =exp(\lambda T_a{}^b)exp(\lambda T_c{}^d)exp(\lambda T_{a_1}{}^{b_1})
      exp(\lambda T_{c_1}{}^{d_1})\varepsilon_{bb_1dd_1}=\\ \\
      =det(exp(\lambda T_a{}^b))\varepsilon_{aa_1cc_1}=
      exp(\lambda tr(T_a{}^b))\varepsilon_{aa_1cc_1}=\varepsilon_{aa_1cc_1}\spsd
      \end{array}
\end{equation}
      Поскольку
\begin{equation}
\label{re43b}
       K_\alpha{}^\beta(\lambda)=\frac{1}{4}\eta_\alpha{}^{aa_1}
       \eta^\beta{}_{bb_1}2S_{\left[\right.a}{}^b(\lambda)
       S_{a_1\left.\right]}{}^{b_1}(\lambda)\sps
\end{equation}
      то дифференцируя по $\lambda$, полагая $\lambda =0$ и
      опуская верхний индекс с помощью метрического тензора
      $g_{\alpha\beta}$, получим
\begin{equation}
\label{re44b}
      T_{\alpha\beta}=\frac{1}{2}\eta_\alpha{}^{aa_1}
      \eta_{\beta bb_1}(T_a{}^b\delta_{a_1}{}^{b_1}+
      T_{a_1}{}^{b_1}\delta_a{}^b)=A_{\alpha\beta b}{}^aT_a{}^b\spsd
\end{equation}
      Теперь видна и цель этого пункта. На самом деле (\ref{re44b})
      есть алгебраическая интерпретация изоморфизма алгебр Ли
\begin{equation}
\label{re45b}
      so(6,\mathbb C)\cong sl(4,\mathbb C)\sps
\end{equation}
      и определение  (\ref{re4}) в начале этой главы вполне
      оправдано.

\subsection{Обобщенные операторы Нордена}$ $
      \indent Если задано аналитическое комплексное риманово пространство $\mathbb CV^6$,
      которое будет является базой касательного расслоения
      $\tau^\mathbb C$ и расслоения $\Lambda$, то существует тензор $g_{\alpha\beta}(z^\gamma)$,
      который на этом пространстве является метрическим и
      аналитичен по $z^\gamma$ ($z^\gamma$ -координаты точки базы).
      Обозначим через $\tilde{\dot g}_{\alpha\beta}$ значение этого
      тензора в некоторой точке $O(z_o^\gamma)$
\begin{equation}
\label{re42v}
      \tilde{\dot g}_{\alpha\beta}:=
      g_{\alpha\beta}(z_o^\gamma)\spsd
\end{equation}
      Поскольку тензор $\tilde{\dot g}_{\alpha\beta}$ имеет
      симметрическую матрицу, то она может быть
      приведена к диагональному виду в некотором базисе с
      помощью невырожденного преобразования $\dot P_\alpha{}^\gamma$
\begin{equation}
\label{re43v}
      \dot g_{\alpha\beta}=
      \dot P_\alpha{}^\gamma\dot P_\beta{}^\delta
      \tilde{\dot g}_{\gamma\delta}\sps
      \dot P_\alpha{}^\gamma:=
      P_\alpha{}^\gamma(z_o^\delta)\sps
\end{equation}
      где $P_\alpha{}^\gamma(z^\delta)$ - аналитические функции
      координат точки. Но для тензора $\dot g_{\alpha\beta}$
      будут справедливы следующие соотношения
\begin{equation}
\label{re44v}
    \dot g^{\alpha\beta}=1/4\cdot \dot\eta^{\alpha}{}_{aa_1} \dot\eta^{\beta}{}_{bb_1}
    \dot\varepsilon^{aa_1bb_1}\sps
    \dot\varepsilon^{aa_1bb_1}=
    \dot\eta_{\alpha}{}^{aa_1} \dot\eta_{\beta}{}^{bb_1}\dot g^{\alpha\beta}\sps
\end{equation}
    где $\dot\eta^{\alpha}{}_{aa_1}$ - связующие операторы Нордена,
    удовлетворяющие соотношению (\ref{re0}). Тогда из (\ref{re43v})
    следует
\begin{equation}
\label{re45v}
      g_{\alpha\beta}(z_o^\gamma):=
      \dot g_{\alpha\beta}
      (\dot P^{-1}){}_\gamma{}^\alpha
      (\dot P^{-1}){}_\delta{}^\beta=
      (\dot P^{-1}){}_\gamma{}^\alpha
      (\dot P^{-1}){}_\delta{}^\beta
      \dot\eta_{\alpha}{}^{aa_1} \dot\eta_{\beta}{}^{bb_1}
      \dot\varepsilon_{aa_1bb_1}\spsd
\end{equation}
\newpage
    \noindent Определим обобщенные операторы Нордена как
\begin{equation}
\label{re46v}
      \begin{array}{c}
      \eta_{\alpha}{}^{aa_1}(z_o^\delta):=
      (\dot P^{-1}){}_\gamma{}^\alpha\dot\eta_{\alpha}{}^{aa_1}\sqrt{\varepsilon^{-1}(z_o^\delta)}\sps\\[2ex]
      \varepsilon_{abcd}(z_o^\delta)=\varepsilon(z_o^\delta)\dot\varepsilon_{abcd}\sps \varepsilon_{1234}(z_o^\delta)=\varepsilon(z_o^\delta)\spsd
      \end{array}
\end{equation}
      В качестве корня можно взять любой из 2 вариантов.
      Вообще говоря, далее $\ll$нолик$\gg$ можно будет опустить,
      поскольку все выкладки справедливы для произвольной точки
      O, и при этом функции $P_\alpha{}^\gamma (z^\delta),\ \varepsilon(z_o^\delta)$ -
      аналитичны. Тогда из (\ref{re44v})  будет следовать
\begin{equation}
\label{re47v}
    \begin{array}{c}
    g^{\alpha\beta}(z^\delta)=1/4\cdot \eta^{\alpha}{}_{aa_1}(z^\delta)
    \eta^{\beta}{}_{bb_1}(z^\delta)
    \varepsilon^{aa_1bb_1}(z^\delta)\sps\\[2ex]
    \varepsilon^{aa_1bb_1}(z^\delta)
    =\eta_{\alpha}{}^{aa_1}(z^\delta)
    \eta_{\beta}{}^{bb_1}(z^\delta)g^{\alpha\beta}(z^\delta)\spsd
    \end{array}
\end{equation}
    В дальнейшем мы будем пользоваться обобщенными
    операторами Нордена.

\newpage
\section{\texorpdfstring{Связности в расслоении $A^\mathbb C$ и базой $\mathbb CV^6$}{Связности в расслоении}}$ $

    \indent В этой главе рассматривается два подхода к введению
    связности в расслоении $A^\mathbb C$. Первый описан в монографии
    \cite{Penrose1r}, а второй следует из теории нормализации
    Нордена-Нейфельда. В первом пункте как раз и рассматривается
    определение связности в расслоениях согласно этим
    теориям.\\
    \indent Во втором пункте рассматривается нормализация многообразия
    одного из двух семейств
    плоских образующих  квадрики $\mathbb CQ_6$, вложенной в проективное
    пространство $\mathbb C\mathbb P_7$. Это многообразие диффеоморфно многообразию
    точек самой квадрики. Рассматривая деривационные
    уравнения нормализованного семейства плоских образующих мы
    приходим к определению операторов Нордена через
    операторы Нейфельда. Если в расслоении $A^\mathbb C$ рассматривать в
    качестве метрического тензора квадривектор $\varepsilon_{abcd}$,
    кососимметричный по всем индексам, то на базе индуцируется
    метрический тензор $G_{\Lambda\Psi}$, что превращает многообразие
    плоских образующих в вещественное риманово пространство $V^{12}_{(6,6)}$
    с комплексной структурой $f_\Lambda{}^\Psi$.
    Можно перейти к комплексной реализации нашего многообразия с базой $\mathbb CV^6$.
    Поставив каждой 4-мерной образующей конуса
    8-мерного пространства $\mathbb C\mathbb R^8$ (т.е в проективной геометрии это
    как раз и будет 3-мерная образующая квадрики $\mathbb CQ_6\subset \mathbb C\mathbb P_7$)
    соответствующий слой из расслоения $A^\mathbb C$ с базой
    $\mathbb CV^6$, получим, что риманова связность, введенная по формулам
    $$
    \nabla_\alpha g_{\beta\gamma}=0\sps
    \bar\nabla_{\alpha '} \bar g_{\beta '\gamma '}=0\sps
    $$
    где $(\alpha,\beta,...=1,2,3,4,5,6$),
    единственным образом продолжается до эквиаффинной связности
    в расслоении $A^\mathbb C(\mathbb CV^6)$ вида
    $$
    \nabla_\alpha\varepsilon_{abcd}=0\sps
    \bar\nabla_{\alpha '}\bar \varepsilon_{a'b'c'd'}=0\spsd
    $$
    Существование и единственность такой связности и доказывается
    в данной главе.\\
    \indent Далее рассматривается вещественная связность, индуцируемая
    вложением $V^6_{(p,q)}\subset \mathbb CV^6$. Такая связность
    должна быть согласована с инволюцией, т.е. должно быть
    выполнено следующее соотношение
    $$
    \nabla_\alpha S_\beta{}^{\gamma '}=0\sps
    \bar\nabla_{\alpha '} S_\beta{}^{\gamma '}=0\spsd
    $$
\newpage
    \noindent Затем, используя результаты первой главы,
    вводится либо эрмитов поляритет, либо эрмитова инволюция в
    расслоении $A^\mathbb C$. При этом указанная структура должна быть
    ковариантно постоянна.
    Из результатов этих пунктов получается битвисторное
    уравнение
    $$
    \nabla^{a\left(\right.b}X^{c\left.\right)}=0\sps
    $$
    которое конформно-инвариантно и инвариантно при преобразованиях
    нормализации. Решения этого уравнения будут рассмотрены
    в следующей главе.

\subsection{\texorpdfstring{Связность в расслоении}{Связность в расслоении}}$ $
     \indent Пусть задано расслоение R с базой $V^{2n}_{(n,n)}$ и
     слоями, изоморфными $\mathbb C^k$.
     Определим оператор ковариантной производной,
     действующий в расслоении R вдоль векторного поля
     X как отображение двух гладких сечений слоя
     $\mathbb C^k_x$
\begin{equation}
\label{re1w}
     \begin{array}{c}
     \nabla_X s:x\longmapsto \nabla_X s(x)\sps
     \end{array}
\end{equation}
     где $s(x)$ - сечение. При $X=\frac{\partial}{\partial x^i}$
     это даст разложение
\begin{equation}
\label{re2w}
     \nabla_\frac{\partial}{\partial x^i} s=\nabla_i s
\end{equation}
     $(i,j,k,...=\overline{1,2n})$. При этом оператор $\nabla_i$
     должен удовлетворять следующим соотношениям (которые, кстати говоря,
     можно положить в его определение)
\begin{equation}
\label{re3w}
     \begin{array}{c}
     \nabla_i(X^a+Y^a)=\nabla_iX^a+\nabla_iY^a\sps\\
     \nabla_i(fX^a)=f\nabla_iX^a+X^a\nabla_if\sps\\
     \nabla_i(X_aY^a)=Y^a\nabla_iX_a+X_a\nabla_iY^a\sps\\
     \nabla_i \bar X^{a'}=\overline{\nabla_i  X^a}\sps
     \nabla_i \bar X_{a'}=\overline{\nabla_i  X_a}\sps\\
     \nabla_ik=0\sps\\
     \nabla_i(g+h)=\nabla_ig+\nabla_ih\sps\\
     \nabla_i(gh)+g\nabla_ih+h\nabla_ig
     \end{array}
\end{equation}
     $(a,b,c,...,f=\overline{1,n})$. При этом k,g,h - аналитические функции, $k=const$; $X^a,\ Y^a$ - векторы
     слоя $\mathbb C^k_x$, а  $X_a,\ Y_a$ - ковекторы двойственного
     пространства ${\mathbb C^*}^k_x$. Сечение $s(x)$  может быть разложено
     по базису $s_a(x)$ слоя  $\mathbb C^k_x$
\begin{equation}
\label{re3wa}
     s=s^as_a
\end{equation}
     так, что коэффициенты связности определятся из следующего
     уравнения
\begin{equation}
\label{re3wb}
     \nabla_is_a=\Gamma_{ia}{}^cs_c\spsd
\end{equation}
     Тогда дифференцирование можно осуществить следующим
     образом
\begin{equation}
\label{re3wc}
     \nabla_i X^a=\partial_i X^a+\Gamma_{ic}{}^aX^c\sps
\end{equation}
     при этом повторная ковариантная производная запишется
     в следующем виде
\begin{equation}
\label{re3wd}
     \nabla_i\nabla_jX^a=\partial_i\nabla_jX{}^a-
     \Gamma_{ij}{}^k\nabla_kX^a+\Gamma_{ic}{}^aX_j{}^c\sps
\end{equation}
     где $\Gamma_{ij}{}^k$ определяют связность в касательном
     расслоении.\\

     Тензором кручения назовем тензор
     $T_{ij}{}^k$, удовлетворяющий соотношению
\begin{equation}
\label{re4w}
     2\nabla_{\left[\right.i}\nabla_{j\left.\right]}f=
     T_{ij}{}^k\nabla_kf\spsd
\end{equation}
     Тензором кривизны назовем тензор $R_{ijk}{}^l$, удовлетворяющий
     следующему соотношению
\begin{equation}
\label{re4wa}
     (2\nabla_{\left[\right.i}\nabla_{j\left.\right]}-T_{ij}{}^k\nabla_k)
     X^i=R_{ijk}{}^lX^k\spsd
\end{equation}
     Если кручение нулевое, то соответствующий оператор $\nabla_i$
     назовем симметричным. Пусть $\nabla_i$ - симметричный оператор,
     а $\tilde\nabla_i$ - произвольный. Тогда
\begin{equation}
\label{re5w}
     (\tilde\nabla_i-\nabla_i)f=0\sps
\end{equation}
     и можно определить такой тензор $Q_{ib}{}^a$, называемый тензором
     деформации, что
\begin{equation}
\label{re6w}
     \begin{array}{cc}
     (\tilde\nabla_i-\nabla_i)X^a=Q_{ib}{}^a X^b\sps &
     (\tilde\nabla_i-\nabla_i)X_a=-Q_{ia}{}^b X_b\sps\\
     (\tilde\nabla_i-\nabla_i)X^{a'}=Q_{ib'}{}^{a'} X^{b'}\sps &
     (\tilde\nabla_i-\nabla_i)X_{a'}=-Q_{ia}{}^{b'} X_{b'}\spsd\\
     \end{array}
\end{equation}
     Если теперь $R=\tau^\mathbb R(V^{2n}_{(n,n)})$ есть касательное расслоение,
     то кручение оператора $\tilde\nabla_i$ будет иметь вид
\begin{equation}
\label{re7w}
     \tilde T_{ij}{}^k=2Q_{[ij]}{}^k\sps
\end{equation}
     где $Q_{ij}{}^k$ - тензор деформации в касательном расслоении.

\subsubsection{\texorpdfstring{Нормализация (спинорная) квадрики $\mathbb CQ_6$ в $\mathbb CP_7$}{Спинорная нормализация}}$ $
     \indent Рассмотрим невырожденную квадрику $\mathbb CQ_6$, вложенную в
     проективное пространство $\mathbb C\mathbb P_7$.
     Она может быть описана уравнением
\begin{equation}
\label{re1aa}
     G_{AB}X^AX^B=0\ \ \ \Leftrightarrow\ \ \  (X,X)=0\
     (A,B,...=\overline{1,8})\spsd
\end{equation}
     На основании принципа тройственности Картана \cite[стр. 175]{Cartan1r}
     многообразие точек квадрики диффеоморфно многообразию трехмерныx
     плоских образующих, составляющих 2 семейства (таким образом
     мы имеем 3 изоморфных друг другу многообразия). Базисные точки
     этих образующих
\begin{equation}
\label{re2aa}
     X_a=(X_a{}^A)\psps (a,b,...,i,j,...,p,q,...=\overline{1,4})
\end{equation}
     определят уравнения
\begin{equation}
\label{re3aa}
      (X_a,X_b)=0\spsd
\end{equation}
     Определим плоскую образующую ее матричной координатой
     $Z=(Z_a^p)$ \cite{Rosenfeld3r}
\begin{equation}
\label{re4aa}
     X_a:=A_a+B_pZ^p_a\sps (A_a,B_p):=d_{ap}\sps B^a:=d^{ap}B_p\sps
\end{equation}
     тогда из (\ref{re3aa}) следует
\begin{equation}
\label{re5aa}
     Z_{ab}=-Z_{ba}\sps Z_{ab}:=d_{ap}Z_a^p\spsd
\end{equation}
     Это означает, что $X_a$ зависят от 6 комплексных параметров.
     Как известно \cite{Neifeld2r}, нормализация многообразия
     плоских образующих квадрики определяется заданием такого
     вещественного дифференциального соответствия между ее плоскими
     образующими максимальной размерности
\begin{equation}
\label{re6aa}
     f:\ \mathbb C\mathbb P_3(X_a)\rightarrow \mathbb C\mathbb P_3(Y_p)\sps
\end{equation}
     что образующей $\mathbb C\mathbb P_3(X_a)$ соответствует
     плоскость $\mathbb C\mathbb P_3(Y_p)$, не пересекающая первую. Для 6-мерной
     квадрики эти плоские образующие необходимо принадлежат
     одному семейству.
     Мы будем требовать, чтобы нормализация была гармонической
     \cite[стр. 209]{Norden6r} .
     В локальных координатах нормализация
     определяется параметрическими уравнениями
\begin{equation}
\label{re7aa}
      X_a=X_a(u^\Lambda )\sps Y_a=Y_a(u^\Lambda )\ \
      (\Lambda ,\Psi ,...=\overline{1,12})\spsd
\end{equation}
     При этом выполнены соотношения
\begin{equation}
\label{re8aa}
     (X_a,X_b)=0\sps (Y_p,Y_q)=0\sps (X_a,Y_p)=c_{ap}\spsd
\end{equation}
     Ввиду невырожденности $c_{ap}$ мы можем определить
\begin{equation}
\label{re9aa}
     Y^a:=c^{ap}Y_p\sps c^{ap}c_{pb}=\delta_b^a\sps (X_a,Y^b)=\delta_b^a\spsd
\end{equation}

\subsubsection{Операторы Нейфельда}$ $
     \indent Деривационные уравнения нормализованного семейства плоских
     образующих имеют вид \cite{Neifeld2r}
\begin{equation}
\label{re10aa}
     \left\{
     \begin{array}{lcl}
     \nabla_\Lambda X_a & = & Y^b M_{\Lambda ab}\sps\\
     \tilde \nabla_\Lambda Y^b & = & X_a N_\Lambda{}^{ab}\spsd
     \end{array}
     \right.
\end{equation}
     Далее, из (\ref{re8aa}) вытекает
\begin{equation}
\label{re11aa}
     M_{\Lambda (ab)}=0\sps N_\Lambda{}^{(ab)}=0\sps
     \Gamma_{\Lambda a}{}^c=\tilde \Gamma_{\Lambda a}{}^c\sps
\end{equation}
     где $\Gamma_{\Lambda a}{}^c$ - коэффициенты конформно-псевдоевклидовой
     связности в комплексном векторном расслоении, база
     которого есть многообразие плоских образующих.
     Отметим, что комплексное векторное расслоение
     метризуемо в том смысле, что в нем можно задать
     поле метрического 4-вектора $\varepsilon_{abcd}$, и
     поскольку нормализация гармоническая,
     то определенная выше связность - эквиаффинна, а 4-вектор
     $\varepsilon_{abcd}$ ковариантно постоянен. Это позволяет
     использовать его для переброски индексов.
     Операторы $M_\Lambda{}^{ab}$ есть связующие операторы,
     которые каждому бивектору слоя ставят в соответствие
     вещественный вектор касательного расслоения
\begin{equation}
\label{re12aa}
     V^{ab}:=M_\Lambda{}^{ab}V^\Lambda\spsd
\end{equation}
     Это соответствие будет взаимнооднозначно. Отсюда следует, что
     можно определить
\begin{equation}
\label{re13aa}
       \left\{
       \begin{array}{ccc}
       M^{\Lambda ab}M_{\Lambda cd} & = & \delta_{cd}^{ab}\sps\\
       \bar M^{\Lambda a'b'}M_{\Lambda cd} & = & 0\sps
       \end{array}
       \right.
       det
       \left\|
       \begin{array}{c}
       M_{\Lambda ab} \\
       \bar M_{\Lambda a'b'}
       \end{array}
       \right\|
       \ne 0\spsd
\end{equation}
       Тогда оператор
\begin{equation}
\label{re14aa}
       \bigtriangleup_\Lambda{}^\Psi=\frac{1}{2}
       (\delta_\Lambda{}^\Psi+if_\Lambda{}^\Psi)=\frac{1}{2}
       M_{\Lambda ab} M^{\Psi ab}
\end{equation}
       есть единичный аффинор Нордена \cite{Norden1r} такой, что
\begin{equation}
\label{re15aa}
       f_\Lambda{}^\Psi M^{\Lambda cd}=-iM^{\Psi cd}\sps
\end{equation}
       где $f_\Lambda{}^\Psi$  есть оператор комплексной
       структуры
\begin{equation}
\label{re16aa}
      f^2=-E\spsd
\end{equation}
      Определим согласно работе \cite{Neifeld1r} операторы
      $m_\alpha{}^\Lambda$ таким образом, что
\begin{equation}
\label{re17aa}
       \left\{
       \begin{array}{ccc}
       m_\alpha{}^\Lambda m^\beta{}_\Lambda & = & \delta_\alpha{}^\beta\sps\\
       m_\alpha{}^\Lambda \bar m^{\beta '}{}_\Lambda & = & 0\sps
       \end{array}
       \right.
       det
       \left\|
       \begin{array}{c}
       m^\alpha{}_\Lambda \\
       \bar m^{\alpha '}{}_\Lambda
       \end{array}
       \right\|
       \ne 0\sps
\end{equation}
       и тогда
\begin{equation}
\label{re18aa}
       \bigtriangleup_\Lambda{}^\Psi=\frac{1}{2}
       (\delta_\Lambda{}^\Psi+if_\Lambda{}^\Psi)=
       m^\alpha{}_\Lambda m_\alpha{}^\Psi\ \
       (\alpha ,\beta ,...=\overline{1,6})
\end{equation}
       есть все тот же  единичный аффинор Нордена \cite{Norden1r}.
       При этом
\begin{equation}
\label{re19aa}
       f_\Lambda{}^\Psi m_\alpha{}^\Lambda=-im_\alpha{}^\Psi\spsd
\end{equation}
       Это означает, что верно следующее разложение
\begin{equation}
\label{re20aa}
       m_\alpha{}^\Lambda=\frac{1}{2}\eta_\alpha{}^{ab}
       M^\Lambda{}_{ab}\spsd
\end{equation}
       для некоторых $\eta_\alpha{}^{ab}=-\eta_\alpha{}^{ba}$.
       Для произвольного тензора $A_{\Lambda\Psi}$
       будет иметь место следующее разложение
\begin{equation}
\label{re21aa}
       \left\{
       \begin{array}{ccc}
       a_{\alpha\beta} & = & m_\alpha{}^\Lambda m_\beta{}^\Psi
       A_{\Lambda\Psi}\sps\\
       a_{\alpha '\beta} & = & \bar m_{\alpha '}{}^\Lambda m_\beta{}^\Psi
       A_{\Lambda\Psi}\sps
       \end{array}
       \right.
       \left\{
       \begin{array}{ccc}
       a_{abcd} & = & M^\Lambda{}_{ab} M^\Psi{}_{cd}
       A_{\Lambda\Psi}\sps\\
       a_{a'b'cd} & = & \bar M^\Lambda{}_{a'b'} M^\Psi{}_{cd}
       A_{\Lambda\Psi}\spsd\\
       \end{array}
       \right.
\end{equation}
       При этом метрическому 4-вектору будет соответствовать
       метрический тензор $G_{\Lambda\Psi}$ так, что
\begin{equation}
\label{re22aa}
       \left\{
       \begin{array}{ccl}
       g_{\alpha \beta} & = & m_\alpha{}^\Lambda m_\beta{}^\Psi
       G_{\Lambda\Psi}\sps\\
       g_{\alpha '\beta} & =  & 0\sps
       \end{array}
       \right.
       \left\{
       \begin{array}{ccl}
       \varepsilon_{abcd} & = & M^\Lambda{}_{ab} M^\Psi{}_{cd}
       G_{\Lambda\Psi}\sps\\
       \varepsilon_{a'b'cd} & =  & 0\spsd
       \end{array}
       \right.
\end{equation}
       Обратные соотношения имеют вид
       $$
       G_{\Lambda\Psi}=\frac{1}{4}
       (M_\Lambda{}^{ab}M_\Psi{}^{cd}\varepsilon_{abcd}+
       \bar M_\Lambda{}^{a'b'}\bar M_\Psi{}^{c'd'}\bar\varepsilon_{a'b'c'd'})\sps
       $$
\begin{equation}
\label{re23aa}
       \eta^\alpha{}_{ab}=m^\alpha{}_\Lambda M^\Lambda{}_{ab}
       \sps
       \bar \eta^{\alpha '}{}_{a'b'}=
       \bar m^{\alpha '}{}_\Lambda \bar M^\Lambda{}_{a'b'}\sps
\end{equation}
       $$
       \eta^{\alpha '}{}_{ab}=\bar m^{\alpha '}{}_\Lambda M^\Lambda{}_{ab}\equiv 0
       \sps
       \bar \eta^{\alpha }{}_{a'b'}=
       m^\alpha {}_\Lambda \bar M^\Lambda{}_{a'b'}\equiv 0\spsd
       $$
       Последняя пара уравнений появляется ввиду аналитичности $M^\Lambda{}_{ab}$.
       Отсюда с учетом (\ref{re13aa}),(\ref{re17aa}),(\ref{re20aa})
       будет следовать
\begin{equation}
\label{re24aa}
       \begin{array}{c}
       \eta^\alpha{}_{ab}\eta_\alpha{}^{cd}=
       M^\Lambda{}_{ab}M_\Lambda{}^{cd}=\delta_{ab}^{cd}\sps
       \frac{1}{4}\eta_\alpha{}^{ab}\eta^\beta{}_{cd}\delta_{ab}^{cd}=
       \delta_\alpha{}^\beta\spsd
       \end{array}
\end{equation}
       Таким образом многообразие плоских образующих
       квадрики $\mathbb CQ_6$ снабжено метрическим тензором $G_{\Lambda\Psi}$
       и поэтому диффеоморфно псевдориманову вещественному пространству
       $V^{12}_{(6,6)}$  с комплексной структурой $f_\Lambda{}^\Psi$.

\subsubsection{Вещественная и комплексная реализации связности}$ $
       \indent Перейдем к построению более общей связности.
       Назовем две связности эквивалентными, если они определяют
       один и тот же параллельный перенос вдоль любой
       кривой базы. Вещественная и комплексная реализации даны
       согласно \cite[с. 169-178]{Lichnerowicz1r}.
\vspace{1cm}
\begin{theoremr}
\label{theorem111r}
       Пусть $V^{2n}_{(n,n)}$ - вещественное псевдориманово пространство
       с комплексной структурой, а $\mathbb CV^n$ - комплексное аналитическое
       риманово пространство: комплексная реализация $V^{2n}_{(n,n)}$. Тогда следующие
       два определения эквивалентны (определяют одну и ту же
       связность)
       \begin{enumerate}
       \item  В касательном расслоении $\tau^\mathbb R(V^{2n}_{(n,n)})$ существует
              риманова связность без кручения такая, что тензор $m_\alpha{}^\Lambda$
              ковариантно постоянен
\begin{equation}
\label{re20w}
        \nabla_\Lambda G_{\Theta\Psi}=0\sps
\end{equation}
\begin{equation}
\label{re21w}
       \nabla_\Lambda m_\alpha{}^\Psi=0\sps
       \nabla_\Lambda \bar m_{\alpha '}{}^\Psi=0\spsd
\end{equation}
       \item  В касательном расслоении $\tau^\mathbb C(\mathbb CV^n)$ существует
              риманова связность без кручения такая, что тензор $m_\alpha{}^\Psi$
              ковариантно постоянен
\begin{equation}
\label{re22w}
       \left\{
       \begin{array}{l}
       \nabla_\alpha g_{\beta\gamma}=0\sps\\
       \bar\nabla_{\alpha '} g_{\beta\gamma}=0\sps
       \end{array}
       \right.
       \left\{
       \begin{array}{l}
       \nabla_\alpha \bar g_{\beta '\gamma '}=0\sps\\
       \bar\nabla_{\alpha '} \bar g_{\beta '\gamma '}=0\sps
       \end{array}
       \right.
\end{equation}
\begin{equation}
\label{re23w}
       \left\{
       \begin{array}{l}
       \nabla_\beta m_\alpha{}^\Psi=0\sps\\
       \bar\nabla_{\beta '} m_\alpha{}^\Psi=0
       \end{array}
       \right.
\end{equation}
       \end{enumerate}
       и сделано определение
\begin{equation}
\label{re24w}
       \nabla_\alpha:=m_\alpha{}^\Lambda{}\nabla_\Lambda\sps
       \nabla_{\alpha '}:=\bar m_{\alpha '}{}^\Lambda\nabla_\Lambda\spsd
\end{equation}
\end{theoremr}
\vspace{1cm}

\begin{proof}$ $\\
        \indent Во-первых.
        Пусть связность 1). существует, тогда домножим
        (\ref{re21w}) на $m_\beta{}^\Lambda$, то с учетом определения
        (\ref{re23w}) получим
\begin{equation}
\label{re25w}
        \begin{array}{c}
        \nabla_\Lambda m_\alpha{}^\Psi=0\sps\\[2ex]
        0=m_\beta{}^\Lambda\nabla_\Lambda m_\alpha{}^\Psi=\nabla_\beta m_\alpha{}^\Psi\sps\\[2ex]
        0=\bar m_{\beta '}{}^\Lambda\nabla_\Lambda m_\alpha{}^\Psi=\nabla_{\beta '} m_\alpha{}^\Psi\spsd
        \end{array}
\end{equation}
        Обратно. Пусть выполнено (\ref{re23w}), тогда c учетом
        определений (\ref{re17aa}) и (\ref{re18aa})
\begin{equation}
\label{re26w}
        \begin{array}{c}
        \nabla_\alpha:=m_\alpha{}^\Lambda\nabla_\Lambda\sps\\[2ex]
        m^\alpha{}_\Psi\nabla_\alpha=
        \bigtriangleup_\Psi{}^\Lambda\nabla_\Lambda
        \ \ \Leftrightarrow\ \ \
        \bar m^{\alpha '}{}_\Psi\bar\nabla_{\alpha '}
        =\bar\bigtriangleup_\Psi{}^\Lambda\nabla_\Lambda\spsd
        \end{array}
\end{equation}
        Сложим два последних уравнения и получим
\begin{equation}
\label{re27w}
        \begin{array}{c}
        \nabla_\Lambda=
        (\bigtriangleup_\Psi{}^\Lambda+
        \bar\bigtriangleup_\Psi{}^\Lambda)\nabla_\Lambda=
        m^\alpha{}_\Lambda\nabla_\alpha+
        \bar m^{\alpha '}{}_\Lambda\bar \nabla_{\alpha '}\spsd
        \end{array}
\end{equation}
        Тогда из условий (\ref{re23w}) следует
\begin{equation}
\label{re28w}
        \nabla_\Lambda m^\beta{}_\Psi{}=
        m^\alpha{}_\Lambda{}\nabla_\alpha m^\beta{}_\Psi+
        \bar m^{\alpha '}{}_\Lambda\bar \nabla_{\alpha '} m^\beta{}_\Psi{}=0\spsd
\end{equation}
        \indent Во-вторых.
        Поскольку из (\ref{re21w}) или из  (\ref{re23w}) следует ковариантное постоянство
        оператора комплексной структуры из-за выполнения
        (\ref{re19aa}), то согласно \cite[т. 2, с. 135-139]{Koboyasi1r} из этого вытекает 
        существование аффинной связности в касательном расслоении
        $\tau^\mathbb C(\mathbb CV_n)$. Рассматривая риманову связность без кручения получим,
        что если известна связность вида 1)., то можно определить
        символы связности вида 2)., расписав условие (\ref{re21w})
\begin{equation}
\label{re29w}
        \begin{array}{c}
        \Gamma_{\Lambda\alpha}{}^\beta:=
        \Gamma_{\Lambda\Theta}{}^\Psi m_\alpha{}^\Theta m^\beta{}_\Psi+
        m^\beta{}_\Psi\partial_\Lambda m_\alpha{}^\Psi\sps\\[2ex]
        \bar \Gamma_{\Lambda\alpha '}{}^{\beta '}:=
        \Gamma_{\Lambda\Theta}{}^\Psi \bar m_{\alpha '}{}^\Theta \bar m^{\beta '}{}_\Psi+
        \bar m^{\beta '}{}_\Psi\partial_\Lambda \bar m_{\alpha '}{}^\Psi\spsd
        \end{array}
\end{equation}
        А если известна связность вида 2)., то можно определить
        символы связности вида 1)., расписав условие (\ref{re23w})
\begin{equation}
\label{re30w}
\begin{array}{l}
        \Gamma_{\beta\phantom{'}\Theta}{}^\Psi:=
        \Gamma_{\beta\phantom{'}\alpha}{}^\gamma m^\alpha{}_\Theta m_\gamma{}^\Psi+
        \bar \Gamma_{\beta\phantom{'} \alpha '}{}^{\gamma '}
        \bar m^{\alpha '}{}_\Theta \bar m_{\gamma '}{}^\Psi-
        m^\alpha{}_\Theta\partial_{\beta\phantom{'}} m_\alpha{}^\Psi-
        \bar m^{\alpha '}{}_\Theta\partial_{\beta\phantom{'}} \bar m_{\alpha '}{}^\Psi\sps\\[2ex]
        \Gamma_{\beta'\Theta}{}^\Psi:=
        \Gamma_{\beta'\alpha}{}^\gamma m^\alpha{}_\Theta m_\gamma{}^\Psi+
        \bar \Gamma_{\beta' \alpha '}{}^{\gamma '}
        \bar m^{\alpha '}{}_\Theta \bar m_{\gamma '}{}^\Psi-
        m^\alpha{}_\Theta\partial_{\beta '} m_\alpha{}^\Psi-
        \bar m^{\alpha '}{}_\Theta\partial_{\beta '} \bar m_{\alpha '}{}^\Psi\spsd
\end{array}
\end{equation}
        При этом выполнено
\begin{equation}
\label{re31w}
        \begin{array}{c}
        \Gamma_{\Lambda\Theta}{}^\Psi:=
        \Gamma_{\beta\Theta}{}^\Psi m^\beta{}_\Lambda+\Gamma_{\beta '\Theta}{}^\Psi \bar m^{\beta'}{}_\Lambda\sps
        \Gamma_{\beta\Theta}{}^\Psi=\Gamma_{\Lambda\Theta}{}^\Psi m_\beta{}^\Lambda\sps
        \Gamma_{\beta '\Theta}{}^\Psi=\Gamma_{\Lambda\Theta}{}^\Psi \bar m_{\beta '}{}^\Lambda\sps\\[2ex]
        \partial_\beta = m_\beta{}^\Psi\partial_\Psi\sps
        \bar \partial_{\beta '}= \bar m_{\beta '}{}^\Psi\partial_\Psi\sps
        \partial_\Lambda = m^\beta{}_\Lambda\partial_\beta +\bar m^{\beta '}{}_\Lambda\bar \partial_{\beta '}\spsd
        \end{array}
\end{equation}
        \indent В-третьих.
        Из (\ref{re17aa}) и (\ref{re18aa}) следует
\begin{equation}
\label{re32w}
        \begin{array}{c}
        g_{\alpha\beta}=G_{\Psi\Lambda} m_\alpha{}^\Psi m_\beta{}^\Lambda\sps\\
        G_{\Theta\Upsilon}\bigtriangleup_\Lambda{}^\Theta
        \bigtriangleup_\Psi{}^\Upsilon=m^\alpha{}_\Lambda
        m^\beta{}_\Psi g_{\alpha\beta}\sps\\
        \frac{1}{2}(G_{\Lambda\Psi}+i
        G_{\Theta\left(\right.\Lambda} f_{\Psi\left.\right)}{}^\Theta)=
        m^\alpha{}_\Lambda m^\beta{}_\Psi g_{\alpha\beta}\sps\\
        \frac{1}{2}(G_{\Lambda\Psi}-i
        G_{\Theta\left(\right.\Lambda} f_{\Psi\left.\right)}{}^\Theta)=
        \bar m^{\alpha '}{}_\Lambda \bar m^{\beta '}{}_\Psi{}
        \bar g_{\alpha '\beta '}\sps\\
        G_{\Lambda\Psi}=m^\alpha{}_\Lambda m^\beta{}_\Psi g_{\alpha\beta}+
        \bar m^{\alpha '}{}_\Lambda \bar m^{\beta '}{}_\Psi{}
        g_{\alpha '\beta '}\spsd
        \end{array}
\end{equation}
        Поэтому из условий (\ref{re22w}) следуют
        условия (\ref{re20w}). Обратно, при выполнении (\ref{re20w})
        имеем
\begin{equation}
\label{re33w}
        \begin{array}{c}
        m_\alpha{}^\Lambda m_\beta{}^\Psi m_\gamma{}^\Theta{}
        \nabla_\Lambda G_{\Psi\Theta}=0\ \ \Leftrightarrow\ \ \
        \nabla_\alpha g_{\beta\gamma}=0\sps\\
        \bar m_{\alpha '}{}^\Lambda m^\Psi{}_\beta m_\gamma{}^\Theta
        \nabla_\Lambda G_{\Psi\Theta}=0\ \ \Leftrightarrow\ \ \
        \nabla_{\alpha '} g_{\beta\gamma}=0\spsd
        \end{array}
\end{equation}
        \indent В-четвертых.
        Поскольку риманова связность без кручения условия 1). единственна,
        то и единственна связность условия 2).\\
\end{proof}

\begin{theoremr}
\label{theorem112r}
      Пусть в качестве базы расслоения задано вещественное псевдориманово
      пространство $V^{12}_{(6,6)}$. Тогда две связности без кручения, заданные
      в расслоениях $\tau^\mathbb R(V^{12}_{(6,6)})$ и $A^\mathbb C$ эквивалентны:
      \begin{enumerate}
      \item Риманова связность, заданная в расслоении $\tau^\mathbb R(V^{12}_{(6,6)})$
            условием
\begin{equation}
\label{re34w}
            \nabla_\Lambda G_{\Psi\Theta}=0\spsd
\end{equation}
      \item Риманова связность, заданная в расслоении $A^\mathbb C$
            условиями
\begin{equation}
\label{re35w}
            \nabla_\Lambda \varepsilon_{abcd}=0\sps
            \nabla_\Lambda \bar\varepsilon_{a'b'c'd'}=0\spsd
\end{equation}
      При этом коэффициенты связности 2). однозначно определятся
      из условия
\begin{equation}
\label{re36w}
            \nabla_\Lambda M_\Psi{}^{ab}=0\sps
            \nabla_\Lambda \bar M_\Psi{}^{a'b'}=0\spsd
\end{equation}
      \end{enumerate}
\end{theoremr}

\begin{proof}
      Риманова связность без кручения, заданная условием (\ref{re34w})
      в касательном расслоении всегда существует и единственна.
      Распишем первое условие (\ref{re36w})
\begin{equation}
\label{re37w}
      \nabla_\Lambda M_\Psi{}^{aa_1}=
      \partial_\Lambda M_\Psi{}^{aa_1} -
      \Gamma_{\Lambda\Psi}{}^{\Theta}M_\Theta{}^{aa_1}+
      \Gamma_{\Lambda c}{}^aM_\Psi{}^{ca_1}+
      \Gamma_{\Lambda c}{}^{a_1}M_\Psi{}^{ac}=0\spsd
\end{equation}
      Домножим это уравнение на $M^\Psi{}_{ac_1}$
\begin{equation}
\label{re38w}
      \Gamma_{\Lambda c_1}{}^{a_1}=-\frac{1}{2}
      (M^\Psi{}_{ac_1}\partial_\Lambda M_\Psi{}^{aa_1}
      -\Gamma_{\Lambda\Psi}{}^\Theta M_\Theta{}^{aa_1} M^\Psi{}_{ac_1}+
      \Gamma_{\Lambda a}{}^a\delta_{c_1}{}^{a_1})\spsd
\end{equation}
      Кроме того, из условий (\ref{re35w}) и (\ref{re13aa}) следует
\begin{equation}
\label{re39w}
      \begin{array}{c}
      \frac{1}{24} M_\Psi{}^{ab} M^\Psi{}^{cd}
      \partial_\Lambda(M^\Theta{}_{ab} M_\Theta{}_{cd})=
      \frac{1}{24}\varepsilon^{abcd}\partial_\Lambda\varepsilon_{abcd}=\\[2ex]
      =\frac{1}{24}\varepsilon^{abcd}(\nabla_\Lambda\varepsilon_{abcd}+
      4\Gamma_{\Lambda\left[\right. a}{}^k\varepsilon_{|k|bcd\left.\right]})=
      \frac{1}{6}\varepsilon^{abcd}\Gamma_{\Lambda a}{}^k
      \varepsilon_{kbcd}=\Gamma_{\Lambda k}{}^k
      \end{array}
\end{equation}
      Исходя из этого, можно положить уравнение (\ref{re38w})
      в определение символов связности 2).

      Пусть в расслоении $A^\mathbb C$ существует еще один
      оператор симметричной ковариантной производной
      $\tilde\nabla_\Lambda$ такой, что
\begin{equation}
\label{re40w}
       \tilde\nabla_\Lambda\varepsilon_{abcd}=0\ \ \Rightarrow\ \ \
       (\tilde\nabla_\Lambda-\nabla_\Lambda)\varepsilon_{abcd}=0
       \ \ \Leftrightarrow\ \ \ Q_{\Lambda k}{}^k=0\sps
\end{equation}
       где тензор $Q_{\Lambda a}{}^b$ - тензор деформации,
       определенный в расслоении $A^\mathbb C$. Пусть тензор
       $Q_{\Lambda \Psi}{}^\Theta$ - тензор деформации в
       касательном расслоении $\tau^\mathbb R(V^{12}_{(6,6)})$. Рассмотрим
       действие таких операторов на бивекторах
       $R^{ab}=M_\Psi{}^{ab}r^\Psi$
\begin{equation}
\label{re41w}
      \begin{array}{c}
      (\tilde\nabla_\Lambda-\nabla_\Lambda)R^{ab}=
      (Q_{\Lambda k}{}^a\delta_t{}^b-
       Q_{\Lambda k}{}^b\delta_t{}^a)R^{kt}=\\
      =M_\Psi{}^{ab}(\tilde\nabla_\Lambda-\nabla_\Lambda)r^\Psi=
       M_\Psi{}^{ab}Q_{\Lambda \Theta}{}^\Psi r^\Theta\sps\\
      (\tilde\nabla_\Lambda-\nabla_\Lambda)R_{ab}=-M^\Theta{}_{ab}Q_{\Lambda \Theta}{}^\Psi r_\Psi\spsd
      \end{array}
\end{equation}
      Отсюда следует цепочка тождеств
\begin{equation}
\label{re42w}
      \begin{array}{c}
      M_\Psi{}^{ab}Q_{\Lambda \Theta}{}^\Psi r^\Theta=
      2Q_{\Lambda\left[\right. k}{}^a\delta_{t\left.\right]}{}^b R^{kt}\sps\\
      M_\Psi{}^{ab}Q_{\Lambda \Theta}{}^\Psi r^\Theta=
      2Q_{\Lambda\left[\right. k}{}^a\delta_{t\left.\right]}{}^b M_\Theta{}^{kt} r^\Theta\sps\\
      M_\Psi{}^{ab}Q_{\Lambda \Theta}{}^\Psi =M_\Theta{}^{kt}
      2Q_{\Lambda\left[\right. k}{}^a\delta_{t\left.\right]}{}^b=
      2M_\Theta{}^{k\left[\right.b}Q_{\Lambda k}{}^{a\left.\right]}\sps\\
      Q_{\Lambda \Theta\Psi} =M_\Theta{}^{kb}M_\Psi{}_{ab}
      Q_{\Lambda k}{}^a+\bar M_\Theta{}^{k'b'}\bar M_\Psi{}_{a'b'}
      \bar Q_{\Lambda k'}{}^{a'}=\\
      =-M_\Psi{}^{kb}M_\Theta{}_{ab}
      Q_{\Lambda k}{}^a-\bar M_\Psi{}^{k'b'}\bar M_\Theta{}_{a'b'}
      \bar Q_{\Lambda k'}{}^{a'}\spsd
      \end{array}
\end{equation}
\newpage
      \noindent Откуда получаем
\begin{equation}
\label{re43w}
      Q_{\Lambda\Theta\Psi}=-Q_{\Lambda\Psi\Theta}\spsd
\end{equation}
      В отсутствии кручения имеем
\begin{equation}
\label{re44w}
      Q_{\Lambda\Theta\Psi}=Q_{\Theta\Lambda\Psi}\sps
\end{equation}
      поэтому
\begin{equation}
\label{re45w}
      Q_{\Lambda\Theta\Psi}=0\sps
\end{equation}
      и это означает единственность связности 2).
\end{proof}

\vspace{2cm}
\begin{corollaryr}
\label{corollary113r}
      Пусть в качестве базы расслоения задано комплексное аналитическое риманово
      пространство $\mathbb CV^6$. Тогда две связности без кручения, заданные
      в расслоениях $\tau^\mathbb C(\mathbb CV^6)$ и $A^\mathbb C$ эквивалентны:
      \begin{enumerate}
      \item Риманова аналитическая связность, заданная в расслоении \linebreak[4] $\tau^\mathbb C(\mathbb CV^6)$
            условиями
\begin{equation}
\label{re46w}
            \nabla_\alpha g_{\beta\gamma}=0\sps
            \bar\nabla_{\alpha '} g_{\beta '\gamma '}=0\spsd
\end{equation}
      \item Риманова аналитическая связность, заданная в расслоении $A^\mathbb C$
            условиями
\begin{equation}
\label{re47w}
            \nabla_\alpha\varepsilon_{abcd}=0\sps
            \bar\nabla_{\alpha '}\bar \varepsilon_{a'b'c'd'}=0\spsd
\end{equation}
      При этом коэффициенты связности 2). однозначно определятся
      из условия
\begin{equation}
\label{re49w}
            \nabla_\alpha \eta_\beta{}^{ab}=0\sps
            \bar\nabla_{\alpha '} \bar\eta_{\beta '}{}^{a'b'}=0\spsd
\end{equation}
      \end{enumerate}
\end{corollaryr}
\vspace{1cm}

\begin{proof}
      Доказательство следует из теорем \ref{theorem111r},\ref{theorem112r}, аналитичности (\ref{re23aa}) и уравнения (\ref{re19aa}).
      В частности, аналитичность $\eta_\beta{}^{ab}$ означает  $\partial_{\alpha '}\eta_\beta{}^{ab}\equiv 0$, а из уравнения (\ref{re19aa})
      следует $\Gamma_{\alpha '\beta}{}^\gamma\equiv 0$
\end{proof}

\newpage
\subsubsection{\texorpdfstring{Инволюция в $\mathbb C\mathbb P_7$}{Инволюция в проективном пространстве}}$ $
       \indent Пусть теперь в $\mathbb C\mathbb P_7$ нам задана инволюция в смысле
       \cite{Neifeld1r}
\begin{equation}
\label{re25aa}
       \bar S_{A'}{}^B S_B{}^{D'}=\delta_{A'}{}^{D'}\sps
\end{equation}
       тогда условие действительности точки $X^A$ примет вид
\begin{equation}
\label{re26aa}
       S_A{}^{B'}\bar X^A=X^{B'}\spsd
\end{equation}
       Потребуем, чтобы эта инволюция определяла вложение
       действительной квадрики в комплексную, что равносильно
       тому, что определяющий ее тензор также будет самосопряжен
       относительно этой инволюции. Тогда плоские образующие
       максимальной размерности вещественной квадрики
       должны удовлетворять условиям
\begin{equation}
\label{re27aa}
       1).\ \ \bar S_{A'}{}^B\bar X_{a'}^{A'}s_a{}^{a'}=X_a^B\sps
       2).\ \ \bar S_{A'}{}^B\bar X_{a'}^{A'}s^{aa'}=X^{aB}\spsd
\end{equation}
       Здесь тензоры $s_a{}^{a'}$ и $s^{aa'}$ определяют
       в комплексном расслоении соответственно эрмитову
       инволюцию и эрмитов поляритет соответственно.
       Эти два случая возникают из-за того, что у нас
       в расслоении со слоями, изоморфными $\mathbb C^4$, нет тензора,
       с помощью которого можно поднимать и опускать одиночные индексы.
       Первый случай означает, что сама образующая и
       сопряженная ей принадлежат одному семейству; во-втором
       же случае указанные образующие представляют два
       различных семейства. На основании результатов
       второй главы этим исчерпываются все возможные случаи
       вещественного вложения.
       Из (\ref{re30b}), (\ref{re25aa}) - (\ref{re27aa})
       следует
\begin{equation}
\label{re28aa}
       1).\ \ s_a{}^{a'}\bar s_{a'}{}^b=\pm\delta_a{}^b\sps
       2).\ \ s_{aa'}\bar s^{a'b}=\pm\delta_a{}^b\spsd
\end{equation}
     Далее, будем рассматривать только случай 2). как
     наиболее интересный с точки зрения физики
     \cite[т. 2, с. 86]{Penrose1r}. Случай 1). рассматривается
     аналогично, и его мы опустим. Тогда
\begin{equation}
\label{re29aa}
       \bar X^{b'}=\bar s^{b'a}X_a\spsd
\end{equation}
       Поэтому мы можем написать эквивалентные (\ref{re8aa}), (\ref{re9aa})
       выражения
\begin{equation}
\label{re30aa}
     (X_a,\bar X^{b'})=0\sps (Y_p,\bar Y^{q'})=0\sps
     (X_a,\bar Y_{b'})=s_{ab'}\spsd
\end{equation}
\newpage
     \noindent Положим
\begin{equation}
\label{re31aa}
     \left\{
     \begin{array}{lcl}
     \nabla_\Lambda \bar X^{a'} & = & \bar Y_{b'} \tilde M_\Lambda{}^{a'b'}\sps\\
     \nabla_\Lambda \bar Y_{b'} & = & \bar X^{a'} \tilde N_{\Lambda a'b'}\sps
     \end{array}
     \right.
\end{equation}
     тогда из (\ref{re30aa}) будут следовать тождества
\begin{equation}
\label{re32aa}
     \tilde M_{\Lambda a'b'}=-\frac{1}{2}s_{ca'}s_{db'}
     \varepsilon^{cdab}M_{\Lambda ab}\sps
     \nabla_\Lambda s_{ab'}=0\spsd
\end{equation}
     Поэтому равенством
\begin{equation}
\label{re33aa}
     S_\Lambda{}^\Theta=\frac{1}{2}
     (M_{\Lambda ab} \bar M^\Theta{}_{c'd'}
     \bar s^{c'a}\bar s^{d'b}+
     \bar M_{\Lambda a' b'} M^\Theta{}_{cd}
     s^{ca'}s^{db'})
\end{equation}
     мы определим вещественную инволюцию вида
\begin{equation}
\label{re34aa}
     S_\Lambda{}^\Theta S_\Theta{}^\Psi=\delta_\Lambda{}^\Psi\sps
     \bar M_{\Lambda a'b'}=-S_\Lambda{}^\Psi\tilde M_{\Psi a'b'}\sps
     S_\Lambda{}^\Theta f_\Theta{}^\Lambda=
     -f_\Lambda{}^\Theta S_\Theta{}^\Lambda\spsd
\end{equation}
     Кроме того можно определить
     еще одну инволюцию (согласно \cite{Neifeld1r})
\begin{equation}
\label{re38aa}
       \left\{
       \begin{array}{ccl}
       S_\alpha{}^\beta & =  & 0 \sps\\
       S_\alpha{}^{\beta '} & = & m_\alpha{}^\Lambda \bar m^{\beta '}{}_\Psi
       S_\Lambda{}^\Psi\sps
       \end{array}
       \right.
       S_\alpha{}^{\beta '}\bar S_{\beta '}{}^\gamma=\delta_\alpha{}^\gamma\spsd
\end{equation}

\subsubsection{Риманова связность, согласованная с инволюцией}
\begin{corollaryr}
\label{corollary114r}
      Пусть в качестве базы расслоения задано комплексное аналитическое риманово
      пространство $\mathbb CV^6$. Тогда две вещественные связности без кручения, заданные
      в расслоениях $\tau^\mathbb C(\mathbb CV^6)$ и $A^\mathbb C(S)$ эквивалентны
      \begin{enumerate}
      \item Риманова вещественная связность, заданная в расслоении $\tau^\mathbb C(\mathbb CV^6)$
            условиями
\begin{equation}
\label{re50w}
            \nabla_\alpha g_{\beta\gamma}=0\sps
            \nabla_\alpha S_\gamma{}^{\beta'}=0
\end{equation}
      (такую риманову связность назовем согласованной с инволюцией).
      \item Риманова вещественная связность, заданная в расслоении $A^\mathbb C(S)$
            условиями
\begin{equation}
\label{re51w}
            \nabla_\alpha\varepsilon_{abcd}=0\sps
            \nabla_\alpha s_{aba'b'}=0\spsd
\end{equation}
      При этом коэффициенты связности 2). однозначно определятся
      из условия
\begin{equation}
\label{re52w}
            \nabla_\alpha \eta_\beta{}^{ab}=0\spsd
\end{equation}
      \end{enumerate}
\end{corollaryr}

\begin{proof}
      В условиях следствия \ref{corollary113r} рассмотрим связность 1).,
      заданную условиями (\ref{re49w}), тогда из условий вещественности следует
\begin{equation}
\label{re54w}
      S_\beta{}^{\gamma '}\bar\partial_{\gamma '}=\partial_\beta\sps
\end{equation}
      поэтому из ковариантного постоянства тензора инволюции
      получим
\begin{equation}
\label{re55w}
      \nabla_\gamma=S_\gamma{}^{\beta '}\nabla_{\beta '}\sps
\end{equation}
      что определит вещественную связность.
      Если положить
\begin{equation}
\label{re56w}
      s_{aba'b'}:=\eta^\alpha{}_{ab}\bar \eta_{\beta 'a'b'}S_\alpha{}^{\beta '}\sps
\end{equation}
      то из (\ref{re52w}) и (\ref{re50w}) вытекает
\begin{equation}
\label{re57w}
       \nabla_\alpha s_{aba'b'}=0\spsd
\end{equation}
\end{proof}
\vspace{2cm}

\begin{corollaryr}
\label{corollary115r}
      Пусть в качестве базы расслоения задано вещественное риманово
      пространство $V^6_{(2,4)}$. Тогда две вещественные связности без кручения, заданные
      в расслоениях $\tau^\mathbb R(V^6_{(2,4)})$ и $A^\mathbb C(S)$ эквивалентны
      \begin{enumerate}
      \item Риманова вещественная связность, заданная в расслоении $\tau^\mathbb R(V^6_{(2,4)})$
            условиями
\begin{equation}
\label{re58w}
            \nabla_i g_{jk}=0\spsd
\end{equation}
      \item Риманова вещественная связность, заданная в расслоении $A^\mathbb C(S)$
            условиями
\begin{equation}
\label{re59w}
            \nabla_i\varepsilon_{abcd}=0\sps
            \nabla_is_{ab'}=0\spsd
\end{equation}
      При этом коэффициенты связности 2). однозначно определятся
      из условия
\begin{equation}
\label{re60w}
            \nabla_i \eta_j{}^{ab}=0\spsd
\end{equation}
      \end{enumerate}
\end{corollaryr}
\newpage

\begin{proof}
      Это следствие вытекает из предыдущего следствия \ref{corollary114r}
      при условии ковариантного постоянства оператора вложения $H_i{}^\alpha$,
      что определит соответствующие коэффициенты связности.
      Нам останется только доказать ковариантное постоянство
      тензора эрмитового поляритета. Поскольку
\begin{equation}
\label{re61w}
    \nabla_\alpha s_{abc'd'}=
    \nabla_\alpha s_{\left[ \right. a
    \left| c' \right|} s_{b \left. \right] d'}=0\sps
\end{equation}
    развертывая его по правилу Лейбница и свертывая с $s^{ac'}$, получим
\begin{equation}
\label{re62w}
    \nabla_\alpha s_{bd'}=-1/2s_{bd'}s^{ac'}\nabla_\alpha s_{ac'}\spsd
\end{equation}
    После свертки с $s^{bd'}$ этого уравнения окончательно имеем
\begin{equation}
\label{re63w}
    s^{ac'}\nabla_\alpha s_{ac'}=0\sps \nabla_\alpha s_{ac'}=0\spsd
\end{equation}
\end{proof}
\vspace{1cm}

\subsubsection{Битвисторное уравнение}
       Из выполнения (\ref{re10aa}),(\ref{re13aa}),(\ref{re20aa}),
       полагая
\begin{equation}
\label{re51aa}
       \nabla_{ab}:=\eta^\alpha{}_{ab}\nabla_\alpha\sps
\end{equation}
       получим
\begin{equation}
\label{re52aa}
       \nabla_\alpha X_a=Y^b\eta_{\alpha ab}\ \ \Leftrightarrow\ \
       \nabla_{cd} X_a=Y^b\varepsilon_{cdab}
\end{equation}
        так, что будут выполнены уравнения
\begin{equation}
\label{re53aa}
       \nabla_{c\left(\right.d}X_{a\left.\right)}=0\sps
       \nabla^{c\left(\right.d}X^{a\left.\right)}=0\sps
\end{equation}
        последнее из которых мы назовем битвисторным уравнением.
        С помощью этого уравнения можно исследовать конформную
        структуру пространств вида $\mathbb C\mathbb R^6$. Следует
        отметить, что битвисторное уравнение не меняется
        при конформных преобразованиях метрики и инвариантно
        при преобразованиях нормализации в смысле \cite{Neifeld2r}, \cite{Neifeld4r}.
\begin{proof}
        Действительно, положим, что конформное преобразование
        метрики имеет вид
\begin{equation}
\label{re144}
        g_{\alpha\beta}\longmapsto\hat g_{\alpha\beta}=
        \Omega^2g_{\alpha\beta}\spsd
\end{equation}
        Тогда из формулы
\begin{equation}
\label{re145}
        \hat\nabla_\alpha\hat\varepsilon_{abcd}=
        \nabla_\alpha\varepsilon_{abcd}=0
\end{equation}
        следует
\begin{equation}
\label{re146}
        0=\hat\nabla_\alpha(\Omega^2\varepsilon_{abcd})=
        \varepsilon_{abcd}(2\Omega\nabla_\alpha\Omega-
        \Omega^2\Theta_{\alpha k}{}^k)=0\spsd
\end{equation}
        Положим
\begin{equation}
\label{re147}
        B_\alpha:=\frac{1}{2}\Theta_{\alpha k}{}^k\spsd
\end{equation}
        Поскольку $\nabla_\alpha$ и $\hat\nabla_\alpha$
        симметричны, то выполнено
\begin{equation}
\label{re148}
        Q_{\alpha\beta\gamma}=Q_{\beta\alpha\gamma}
\end{equation}
        в касательном расслоении $\tau^\mathbb C(\mathbb CV^6)$. Тогда в расслоении $A^\mathbb C$ выполнено
\begin{equation}
\label{re149}
        \Theta_{cc_1a}{}^b=B_{ca}\delta_{c_1}{}^b-B_{c_1a}\delta_c{}^b
        \sps
        B_{ab}=-B_{ba}\spsd
\end{equation}
        Откуда
\begin{equation}
\label{re150}
        B_\alpha=\frac{1}{2}\eta_\alpha{}^{ab}B_{ab}=
        \Omega^{-1}\nabla_\alpha\Omega\spsd
\end{equation}
        Положим
\begin{equation}
\label{re151}
        \hat X^c=X^c\spsd
\end{equation}
        Тогда с учетом
\begin{equation}
\label{re152}
        \hat\nabla_{ab}X^c=\nabla_{ab}X^c+
        2B_{\left[\right. a|k|}\delta_{\left.b\right]}{}^cX^k
\end{equation}
        получим
\begin{equation}
\label{re153}
        \hat\nabla^{a\left(b\right.}\hat X^{c\left.\right)}=
        \Omega^{-2}\nabla^{a\left(\right.b}X^{c\left.\right)}\spsd
\end{equation}
        Это значит,  что битвисторное уравнение
        конформно-инвариантно.
\end{proof}

\newpage
\section{\texorpdfstring{Теоремы о тензоре кривизны. Каноническая форма бивекторов
         6-мерных (псевдо-) евклидовых пространств $\mathbb R^6_{(p,q)}$
         с метрикой четного индекса~q}{Теоремы о тензоре кривизны. Каноническая форма бивекторов}}$ $

    \indent Так как введенная в касательном расслоении к $\mathbb CV_6$
    связность удовлетворяет условию
    $$
    \nabla_\alpha g_{\gamma\delta}=0\sps
    \bar\nabla_{\alpha '} \bar g_{\gamma '\delta '}=0\sps
    $$
    а связность в расслоении $A^\mathbb C$ определяется из уравнений
    $$
    \nabla_\alpha\eta_\beta{}^{ab}=0\sps
    \bar \nabla_{\alpha '}\bar\eta_{\beta '}{}^{a'b'}=0\sps
    $$
    то можно выбрать некоторый неголономный
    специальный базис такой, что метрика
    $g_{\gamma\delta}$ будет иметь в нем диагональный вид с $\ll$+1$\gg$
    на главной диагонали, а обобщенные операторы Нордена будут
    иметь постоянные существенные координаты наподобие
    формул (\ref{re56}) - (\ref{re57}). Из этого
    следует, что операторы $A_{\alpha\beta a}{}^b$ в
    этом базисе тоже имеют в качестве координат константы.
    Тогда тензор кривизны с помощью операторов $A_{\alpha\beta a}{}^b$
    можно представить в следующем виде
    $$
    R_{\alpha\beta\gamma\delta}=
    A_{\alpha\beta a}{}^bA_{\gamma\delta c}{}^d R_b{}^a{}_d{}^c\spsd
    $$
    При этом, зная структуру тензора $R_b{}^a{}_d{}^c$, можно
    восстановить структуру тензора кривизны. Но
    исследование структуры тензора  $R_b{}^a{}_d{}^c$
    облегчается тем, что он почти не содержит несущественных
    компонент. В 4-мерном случае подобные $R_b{}^a{}_d{}^c$
    тензоры, названные спинорами кривизны \cite{Penrose1r},
    сильно упрощают классификацию тензора кривизны 4-мерного
    пространства, впервые осуществленную Петровым прямыми
    тензорными методами. Поэтому
    следует ожидать, что легче классифицировать будет
    тензор $R_b{}^a{}_d{}^c$ нежели заниматься
    классификацией тензора $R_{\alpha\beta\gamma\delta}$ 6-мерного пространства.
    Первая часть этой главы и посвящена связи таких тензоров.\\
    \indent В третьем пункте рассматривается вопрос о каноническом
    виде кососимметричной билинейной форме для
    метрики четного индекса q в пространстве $\mathbb R^6_{(p,q)}$.
    Указанная форма в некотором базисе имеет вид
    $$
    \frac{1}{2}R_{\alpha\beta}X^\alpha Y^\beta=
    R_{16}X^{\left[\right.1} Y^{6\left.\right]}+
    R_{23}X^{\left[\right.2} Y^{3\left.\right]}+
    R_{45}X^{\left[\right.4} Y^{5\left.\right]}\spsd
    $$
    Кроме того, устанавливается такой факт, как соответствие
    вектора из слоя расслоения $A^\mathbb C$ с базой $\mathbb C\mathbb V^6$ и
    изотропного простого бивектора, принадлежащего
    изотропному конусу $K_6$ слоя касательного расслоения над одной и той же точкой x. Это соответствие
    с точностью до множителя $re^{i\Theta} \in \mathbb C$ определит указанный вектор слоя.
    На основании этого соответствия мы можем говорить
    о геометрической интерпретации изотропного
    (в смысле $s_{aa'}X^aX^{a'}=0$) твистора  из $\mathbb C^4$ в
    пространстве $R^6_{(2,4)}$. Для ее осуществления нам необходимо
    научиться сравнивать изотропные векторы, принадлежащие
    конусу $K_6$. Поэтому с помощью стереографической
    проекции мы инвариантным (координатно-независимым) образом определяем
    некоторый касательный к $K_6$ вектор, приложенный к точке P.
    Его норма, взятая со знаком ''-'', сопоставляется изотропному
    вектору K с началом в вершине конуса, а концом в точке P и
    называется протяженностью вектора K. Тогда можно
    выбрать вектор k единичной протяженности и все изотропные вектора
    сравнивать с этим вектором. При этом неоднозначность соответствия
    устраняется так: r - есть протяженность любого изотропного вектора,
    определенного указанным изотропным простым бивектором, принадлежащим
    конусу $K_6$ (флагшток), а $\Theta$ есть угол поворота 3-полуплоскости П
    (полотнище флага), натянутой на бивектор и некоторый вектор,
    ортогональный плоскости П${}_1$, определяемой бивектором, вокруг
    этой плоскости П${}_1$. Полученная интерпретация аналогична
    соответствию спиноров и изотропных векторов пространства
    Минковского, рассмотренного в монографии \cite{Penrose1r}.
\vspace{1cm}

\subsection{Теорема о  битензорах 6-мерных пространств}$ $
    \indent Прежде чем перейти к свойствам тензора кривизны
    пространства $\mathbb CV^6$, рассмотрим следующую теорему.
\vspace{1cm}
\begin{theoremr}
\label{theorem1r}
    Классификацию битензора, обладающего свойствами
\begin{equation}
\label{re11}
    R_{\alpha \beta \gamma \delta}=R_{\left[ \alpha \beta \right]
    \left[\gamma \delta] \right.} \sps
    R_{\alpha \beta \gamma \delta}=R_{\gamma \delta \alpha \beta}\sps
    R_{\alpha \beta \gamma \delta}+R_{\alpha \delta \beta \gamma}+
    R_{\alpha \gamma \delta \beta}=0
\end{equation}
    и принадлежащего касательному расслоению $\tau^\mathbb C(\mathbb CV^6)$ над
    шестимерным аналитическим римановым пространством
    $\mathbb CV^6$, можно свести к классификации тензора $R_a{}^b{}_c{}^d$
    4-мерного комплексного векторного пространства $\mathbb C^4$ такого, что
\begin{equation}
\label{re12}
    R_{\alpha \beta \gamma \delta}= A_{\alpha \beta d}{}^c
    A_{\gamma \delta r}{}^sR_c{}^d{}_s{}^r\spsd
\end{equation}
    Кроме того, выполнены следующие соотношения
\begin{equation}
\label{re13}
    R_k{}^k{}_s{}^r=R_s{}^r{}_k{}^k=0 \sps  R_c{}^d{}_s{}^r=
    R_s{}^r{}_c{}^d\spsd
\end{equation}
    Разложение
\begin{equation}
\label{re14}
    R_c{}^d{}_s{}^r=C_c{}^d{}_s{}^r-P_{cs}{}^{dr}-\frac{1}{40}\cdot
    R(3\delta_s{}^d \delta_c{}^r-2\delta_s{}^r \delta_c{}^d)
\end{equation}
    соответствует разложению тензора $R_{\alpha\beta}{}^{\gamma\delta}$
\begin{equation}
\label{re14a}
    R_{\alpha\beta}{}^{\gamma\delta}=C_{\alpha\beta}{}^{\gamma\delta}+
    R_{\left[\alpha \right.}{}^{\left[ \gamma \right.}
    g_{\left.\beta \right]}{}^{\left. \delta \right]}-
    1/10Rg_{\left[\alpha \right.}{}^{\left[ \gamma \right.}
    g_{\left.\beta \right]}{}^{\left. \delta \right]}
\end{equation}
    на неприводимые ортогональными  преобразованиями компоненты,
    которые будут удовлетворять следующим соотношениям
\begin{equation}
\label{re15}
    P_{cs}{}^{rd}=-4(R_{\left[c \right.}{}^{\left[r \right.}
    {}_{\left. s \right]}{}^{\left. d\right]}+
    R_k{}^{\left[r \right.}{}_{\left[c \right.}{}^{\left|k\right|}
    \delta_{\left. s\right]}{}^{\left.d\right]})\sps
\end{equation}
\begin{equation}
\label{re16}
    C_c{}^d{}_s{}^r=R_{\left(c \right.}{}^{\left(d \right.}{}_
    {\left. s\right)}{}^{\left. r\right)}+
    \frac{1}{40}\cdot R\delta_{\left(s \right.}{}^d\delta_{\left. c\right)}{}^r\sps
    C_c{}^d{}_s{}^r=C_{\left(c \right.}{}^{\left(d \right.}{}_
    {\left. s\right)}{}^{\left.r\right)}\sps
\end{equation}
\begin{equation}
\label{re17}
    R=R_\beta{}^\beta=-2\cdot R_k{}^r{}_r{}^k\sps
    P_{kc}{}^{kd}=1/2\cdot R\delta_c{}^d\sps
\end{equation}
\begin{equation}
\label{re18}
    R_l{}^d{}_s{}^l=-\frac{1}{8}\cdot R\delta_s{}^d\sps
\end{equation}
    последнее из которых является эквивалентом тождества
    Бианки (\ref{re11}).
\end{theoremr}

\begin{proof}
    На основании (\ref{re1}) верно следующее равенство
\begin{equation}
\label{re22}
    R_{\alpha\beta\gamma\delta}=1/16\cdot\eta_{\alpha}{}^{aa_1}
    \eta_{\beta}{}^{bb_1}\eta_{\gamma}{}^{cc_1}
    \eta_{\delta}{}^{dd_1}R_{aa_1bb_1cc_1dd_1}\spsd
\end{equation}
    Положим
\begin{equation}
\label{re23}
    R_c{}^d{}_s{}^r:=\frac{1}{4}R_{ck}{}^{dk}{}_{st}{}^{rt}
    \sps  R_{\beta\gamma}=\frac{1}{4}\cdot\eta_{\beta}{}^{cs}\eta_{\gamma}{}_{rd}
    \cdot P_{cs}{}^{rd}\spsd
\end{equation}
    Из этого с учетом (\ref{re4}) вытекает формула (\ref{re12})
\begin{equation}
\label{re12a}
    R_{\alpha \beta \gamma \delta}= A_{\alpha \beta d}{}^c
    A_{\gamma \delta r}{}^sR_c{}^d{}_s{}^r\spsd
\end{equation}
    Отсюда следует, что
\begin{equation}
\label{re7p}
       \begin{array}{c}
       R_{\beta\delta}=R_{\alpha\beta}{}^{\alpha}{}_\delta=
       A_{\alpha\beta d}{}^cA^\alpha{}_{\delta r}{}^s
       R_c{}^d{}_s{}^r=
       (\eta_\beta{}^{cs}\eta_{\delta rd}+
       \eta_\beta{}^{ck}\eta_{\delta kr}\delta_d{}^s)
       R_c{}^d{}_s{}^r=\\ \\ \\
       =\frac{1}{4}
       \eta_\beta{}^{cs}\eta_{\delta rd}\cdot 4
       (R_{\left[\right.c}{}^{\left[\right.d}
       {}_{s\left.\right]}{}^{r\left.\right]}-
       R_{\left[\right.c}{}^k
       {}_{|k|}{}^{\left[\right.r}
       \delta_{s\left.\right]}{}^{d\left.\right]})\spsd
       \end{array}
\end{equation}
       Положим
\begin{equation}
\label{re8p}
       \begin{array}{c}
       P_{cs}{}^{rd}:=
       -4(R_{\left[\right.c}{}^{\left[\right.r}
       {}_{s\left.\right]}{}^{d\left.\right]}-
       R_{\left[\right.c}{}^k
       {}_{|k|}{}^{\left[\right.r}
       \delta_{s\left.\right]}{}^{d\left.\right]})\sps
       \end{array}
\end{equation}
       тогда
\begin{equation}
\label{re9p}
       R_{\beta\delta}=\frac{1}{4}
       \eta_\beta{}^{cs}\eta_{\delta rd}P_{cs}{}^{rd}\sps
\end{equation}
       чем и доказана формула  (\ref{re15}).
       Поскольку скалярная кривизна имеет вид
\begin{equation}
\label{re10p}
       \begin{array}{c}
       R=R_\beta{}^\beta=\frac{1}{4}
       \eta_\beta{}^{cc_1}\eta^\beta{}_{aa_1}P_{cc_1}{}^{aa_1}=
       \frac{1}{4}\varepsilon_{aa_1}{}^{cc_1}P_{cc_1}{}^{aa_1}=
       \frac{1}{2}P_{aa_1}{}^{aa_1}=-2R_k{}^r{}_r{}^k
       \end{array}
\end{equation}
       и, кроме того, выполнено
\begin{equation}
\label{re11p}
       \begin{array}{c}
       P_{ks}{}^{kd}=
       -4(R_{\left[\right.k}{}^{\left[\right.k}
       {}_{s\left.\right]}{}^{r\left.\right]}+
       R_{\left[\right.r}{}^k
       {}_{|k|}{}^{\left[\right.r}
       \delta_{s\left.\right]}{}^{d\left.\right]})=\\ \\ \\
       =-4(-\frac{1}{2}R_k{}^d{}_s{}^k+
       \frac{1}{4}(R_k{}^r{}_r{}^k\delta_s{}^d+
       4R_k{}^d{}_s{}^k-2R_k{}^d{}_s{}^k))=
       -R_k{}^r{}_r{}^k\delta_s{}^d=\frac{1}{2}R\delta_s{}^d\sps
       \end{array}
\end{equation}
       то формулы (\ref{re17}) действительно будут верны.\\

       \indent Тождества Бианки (\ref{re11}) можно переписать
       следующим образом
\begin{equation}
\label{re26}
    (A_{\alpha\beta d}{}^cA_{\gamma\delta}{}_r{}^s+
    A_{\alpha\gamma d}{}^cA_{\delta\beta}{}_r{}^s+
    A_{\alpha\delta d}{}^cA_{\beta\gamma}{}_r{}^s)\cdot
    R_c{}^d{}_s{}^r=0\spsd
\end{equation}
    Свернув это уравнение с $A^{\alpha\beta}{}_t{}^lA^{\gamma\delta}{}_m{}^n$,
    получим, принимая во внимание (\ref{re21}),
    $$
    4R_k{}^l{}_m{}^k\delta_t{}^n+
    4R_r{}^n{}_t{}^r\delta_m{}^l-
    2R_k{}^l{}_t{}^k\delta_m{}^n-
    2R_k{}^n{}_m{}^k\delta_t{}^l-
    $$
\begin{equation}
\label{re27}
    -2R_k{}^r{}_r{}^k\delta_t{}^n\delta_m{}^l+
    R_r{}^k{}_k{}^r\delta_m{}^n\delta_t{}^l=0\spsd
\end{equation}
    Свертка этого уравнения c $\delta_n{}^t$ и
    приведет нас к формуле (\ref{re18}).
    При этом все 15 существенных  уравнений  сохранены.
    (Все выкладки выполнены в приложении -  формулы
    (\ref{re16p}) - (\ref{re18p})).\\

    \indent Положим
    $$
    C_{\alpha\beta}{}^{\gamma\delta}:=
    A_{\alpha\beta d}{}^cA^{\gamma\delta}{}_r{}^sC_c{}^d{}_s{}^r\sps
    $$
\begin{equation}
\label{re24}
    C_{\alpha\beta}{}^{\gamma\delta}:=
    R_{\alpha\beta}{}^{\gamma\delta}-
    R_{\left[\alpha \right.}{}^{\left[ \gamma \right.}
    g_{\left.\beta \right]}{}^{\left. \delta \right]}+
    1/10Rg_{\left[\alpha \right.}{}^{\left[ \gamma \right.}
    g_{\left.\beta \right]}{}^{\left. \delta \right]}\spsd
\end{equation}
    Из (\ref{re4}), (\ref{re15}), (\ref{re17}) следует
    $$
    R_{\left[\alpha \right.}{}^{\left[ \gamma \right.}
    g_{\left.\beta \right]}{}^{\left. \delta \right]}=
    A_{\alpha\beta d}{}^cA^{\gamma\delta}{}_r{}^s
    \cdot\frac{1}{4}(P_{sc}{}^{dr}-1/2R\delta_s{}^d\delta_c{}^r+
    \frac{1}{4}R\delta_s{}^r\delta_c{}^d)\sps
    $$
\begin{equation}
\label{re25}
    g_{\left[\alpha \right.}{}^{\left[ \gamma \right.}
    g_{\left.\beta \right]}{}^{\left. \delta \right]}=
    A_{\alpha\beta d}{}^cA^{\gamma\delta}{}_r{}^s
    \cdot\frac{1}{4}(1/2\delta_s{}^r\delta_c{}^d-2\delta_s{}^d\delta_c{}^r)\sps
\end{equation}
    откуда получим разложения (\ref{re14}), (\ref{re16}).
    (Выкладки находятся в приложении - формулы
    (\ref{re12p}) - (\ref{re20p})).\\
\end{proof}
\vspace{1cm}

\subsubsection{Следствия теоремы}
\begin{corollaryr}
\label{corollary1r}
\begin{enumerate}
    \item
    Условия простоты бивектора 6-мерного пространства $\mathbb C\mathbb R_6$ записываются
    в следующем виде
\begin{equation}
\label{re28a}
    p^{\left[\alpha\beta\right.}p^{\gamma\delta\left.\right]}=0\spsd
\end{equation}
    Координатам такого бивектора можно сопоставить бесследовую комплексную
    матрицу $4\times 4$ с нулевым следом такую, что выполнено следующее
    условие
\begin{equation}
\label{re28}
    p_l{}^dp_s{}^l-1/4(p_l{}^kp_k{}^l)\delta_s{}^d=0\spsd
\end{equation}
    \item
    Простому (выполнен пункт 1). этого следствия)
    изотропному ($p^{\alpha\beta}p_{\alpha\beta}=0$) бивектору
    пространства $\mathbb C\mathbb R^6$ можно сопоставить вырожденную
    нуль-пару Розенфельда: ковектор и
    вектор пространства $\mathbb C^4$, свертка которых есть нуль.
    При этом указанные вектор и ковектор
    определятся с точностью до комплексного множителя.\\
\end{enumerate}
\end{corollaryr}
\vspace{1cm}

\begin{proof}
    \indent 1). Бивектор прост тогда и только тогда, когда имеет
    место разложение
\begin{equation}
\label{re1www}
    p^{\alpha\beta}=X^\alpha Y^\beta-Y^\alpha X^\beta\spsd
\end{equation}
    Поэтому, если выполнено (\ref{re1www}), то будет
    верна формула (\ref{re28a}).\\
    \indent Обратно, если выполнены условия (\ref{re28a}), то
    их можно расписать следующим образом
\begin{equation}
\label{re2www}
    p^{\alpha\beta}p^{\gamma\delta}-
    p^{\alpha\gamma}p^{\beta\delta}+
    p^{\beta\gamma}p^{\alpha\delta}=0\spsd
\end{equation}
    Свернем это уравнение с такими ненулевыми ковекторами
    $T_\delta$  и $Z_\gamma$, что
    $p^{\gamma\delta}Z_\gamma T_\delta \ne 0$
\begin{equation}
\label{re3www}
    p^{\alpha\beta}=\frac{1}{p^{\lambda\mu}Z_\lambda T_\mu}
    (p^{\alpha\gamma}Z_\gamma p^{\beta\delta}T_\delta-
     p^{\beta\gamma}Z_\gamma p^{\alpha\delta}T_\delta)\spsd
\end{equation}
    Положим
\begin{equation}
\label{re4www}
    X^\alpha:=\frac{1}{p^{\lambda\mu}Z_\lambda T_\mu}p^{\alpha\gamma}Z_\gamma\sps
    Y^\beta:=\frac{1}{p^{\lambda\mu}Z_\lambda T_\mu}p^{\beta\delta}T_\delta\sps
\end{equation}
    откуда и будут следовать условия (\ref{re1www}).
    Поскольку тензор $R_{\alpha\beta\gamma\delta}=p_{\alpha\beta}
    p_{\gamma\delta}$ удовлетворяет условиям теоремы \ref{theorem1r},
    то формула (\ref{re28}) есть прямое следствие тождеств
    Бианки (\ref{re18}).\\

    \indent 2). В условиях первого пункта добавится условие
    изотропности
\begin{equation}
\label{re5www}
    p^{\alpha\beta}p_{\alpha\beta}=0\sps
\end{equation}
    которое ввиду формул (\ref{re21}) примет вид
\begin{equation}
\label{re6www}
    \begin{array}{c}
    A^{\alpha\beta}{}_a{}^b A_{\alpha\beta c}{}^d
    p_b{}^a p_d{}^c=0\sps\\ \\
    p_b{}^a p_a{}^b=0\spsd
    \end{array}
\end{equation}
    Отсюда следует, что  существуют такие
    ненулевые $X^a$ и $Y_b$, что
\begin{equation}
\label{re7www}
    p_a{}^b=X^aY_b\sps X^aY_a=0\spsd
\end{equation}
    Эту формулу можно рассматривать и как следствие леммы \ref{lemma1r}
    второй главы (для этого достаточно рассмотреть
    бивектор $p^{\alpha\beta}=r_1{}^{\left[\right.\alpha}
    r_2{}^{\beta\left.\right]}$,
\newpage
    \noindent где  $r_1{}^\alpha$ и $r_2{}^\alpha$ те же, что и в условии леммы). При этом $X^a$ и $Y_b$ определены
    с точностью до преобразования
\begin{equation}
\label{re8www}
    X^a\longmapsto e^\phi X^a\sps
    Y_b\longmapsto e^{-\phi} Y_b\spsd
\end{equation}
\end{proof}
\vspace{1cm}

    Отметим, что пара $(X^a,Y_b)$ будет является нуль-парой
    Розенфельда. В пространстве $\mathbb C\mathbb P^4='\mathbb C^4/'\mathbb C$
    (где $'\mathbb C^s=\mathbb C^s/{0}$) $X^a$ определит точку, а $Y_b$ -
    плоскость с условием инцидентности
\begin{equation}
\label{re9www}
     X^aY_a=0\spsd
\end{equation}
    Поэтому можно определить пространство $\mathbb C$П${}^4='\mathbb C^*{}^4/'\mathbb C$,
    двойственное пространству  $\mathbb C\mathbb P^4$. Тогда пространство
    $\mathbb C\mathbb P^4\times \mathbb C$П${}^4$ будет пространством нуль-пар Розенфельда.
    Следует отметить, что такие пространства изучались впервые
    Синцовым \cite{Sintsov1r} и Котельниковым \cite{Kotelnikov1r}.\\
\vspace{1cm}

\begin{corollaryr}
\label{corollary2r}
    В случае действительности битензора из теоремы \ref{theorem1r}
    на соответствующий тензор накладывается условие
\begin{equation}
\label{re29}
    R_{ab'cd'}=\bar R_{b'ad'c}
\end{equation}
    для метрики четного индекса и
\begin{equation}
\label{re30}
    R_a{}^{b'}{}_c{}^{d'}=
    \bar R_a{}^{b'}{}_c{}^{d'}
\end{equation}
    для метрики нечетного индекса.
\end{corollaryr}

\begin{proof}
    Оно основано на свойствах тензора вложения $s_{...}{}^{...}$.\\
\end{proof}
\newpage

\subsection{Основные свойства и тождества тензора кривизны}$ $
    \indent В качестве примера рассмотрим основные свойства
    тензора кривизны риманова пространства $\mathbb CV^6$.
    Поскольку в некотором неголономном базисе
    операторы $A_{\alpha\beta a}{}^b$ являются константами,
    то все свойства тензора кривизны мы можем получить, рассматривая
    тензор $R_a{}^b{}_c{}^d$. Тензор кривизны пространства $\mathbb CV_6$ в  некоторой
    окрестности U удовлетворяет  соотношениям теоремы \ref{theorem1r}.
    Положим
\begin{equation}
\label{re21p}
       \begin{array}{c}
       \Box_a{}^d:=\frac{1}{2}
       (\nabla_{ak}\nabla^{dk}-
       \nabla^{dk}\nabla_{ak})\sps\\ \\
       \Box_{\alpha\beta}:=
       2\nabla_{\left[\right.\alpha}\nabla_{\beta\left.\right]}\spsd
       \end{array}
\end{equation}
       Ввиду ковариантного постоянства обобщенных операторов
       Нордена будем иметь
\begin{equation}
\label{re22p}
       \begin{array}{c}
       \nabla_{\left[\right.\alpha}\nabla_{\beta\left.\right]}=
       \frac{1}{4}
       \eta_{\left[\right.\alpha}{}^{aa_1}\eta_{\beta\left.\right]}{}^{bb_1}
       \nabla_{aa_1}\nabla_{bb_1}=\\ \\
       =\frac{1}{4}
       \eta_\alpha{}^{aa_1}\eta_\beta{}^{bb_1}\cdot\frac{3}{2}
       (\nabla_{a\left[\right.a_1}\nabla_{bb_1\left.\right]}-
       \nabla_{\left[\right.bb_1}\nabla_{a\left.\right]a_1})=\\ \\
       =\frac{1}{4}
       \eta_\alpha{}^{aa_1}\eta_\beta{}^{bb_1}\cdot\frac{3}{2}
       (\delta_{\left[\right.a_1}{}^k\delta_b{}^n
       \delta_{b_1\left.\right]}{}^{n_1}\nabla_{ak}\nabla_{nn_1}-
       \delta_{\left[\right.b}{}^n\delta_{b_1}{}^{n_1}
       \delta_{a\left.\right]}{}^k\nabla_{nn_1}\nabla_{ka_1})=\\ \\
       =\frac{1}{4}
       \eta_\alpha{}^{aa_1}\eta_\beta{}^{bb_1}\cdot\frac{1}{4}
       (\varepsilon_{a_1bb_1d}\varepsilon^{knn_1d}\nabla_{ak}\nabla_{nn_1}-
       \varepsilon_{bb_1ad}\varepsilon^{nn_1kd}\nabla_{nn_1}\nabla_{ka_1})=\\ \\
       =\frac{1}{4}
       \eta_\alpha{}^{aa_1}\eta_\beta{}^{bb_1}\cdot\frac{1}{4}
       (\varepsilon_{a_1bb_1d}\varepsilon^{kd}{}_{nn_1}\nabla_{ak}\nabla^{nn_1}+
       2\varepsilon_{bb_1a_1d}\nabla^{kd}\nabla_{ka})=\\ \\
       =\frac{1}{4}
       \eta_\alpha{}^{aa_1}\eta_\beta{}^{bb_1}\cdot\frac{1}{4}
       \varepsilon_{a_1bb_1d}(2\delta_{\left[\right.n}{}^k
       \delta_{n_1\left.\right]}{}^d\nabla_{ak}\nabla^{nn_1}+
       2\nabla^{kd}\nabla_{ka})=\\ \\
       =\frac{1}{4}
       \eta_\alpha{}^{aa_1}\eta_{\beta a_1d}
       (\nabla_{an}\nabla^{nd}+
       \nabla^{kd}\nabla_{ka})=
       A_{\alpha\beta d}{}^a\cdot\frac{1}{4}
       (\nabla_{ak}\nabla^{dk}-
       \nabla^{dk}\nabla_{ak})\spsd
       \end{array}
\end{equation}
       Поэтому
\begin{equation}
\label{re23p}
       \begin{array}{c}
       \Box_{\alpha\beta}=A_{\alpha\beta d}{}^a\Box_a{}^d\spsd
       \end{array}
\end{equation}
     Сформулируем несколько основных утверждений,
     касающихся оператора $\Box_a{}^d$:
\begin{enumerate}
    \item
     Из тождества Риччи
\begin{equation}
\label{re46s}
     \Box_{\alpha\beta}k^{\gamma\delta}=
     R_{\alpha\beta\lambda}{}^\gamma k^{\lambda\delta}+
     R_{\alpha\beta\lambda}{}^\delta k^{\gamma\lambda}\sps
     \Box_{\alpha\beta}r^\gamma=
     R_{\alpha\beta\lambda}{}^\gamma r^{\lambda}
\end{equation}
\newpage
     \noindent будут следовать тождества ($k^{\alpha\beta}=-k^{\beta\alpha}$)
\begin{equation}
\label{re47s}
     \Box_a{}^bk_d{}^c=R_a{}^b{}_m{}^ck_d{}^m-R_a{}^b{}_d{}^nk_n{}^c\sps
     \Box_a{}^br^{cc_1}=R_a{}^b{}_m{}^cr^{mc_1}+R_a{}^b{}_m{}^{c_1}r^{cm}\spsd
\end{equation}
     А уже из них окончательно получим
\begin{equation}
\label{re47}
    \Box_a{}^bX^c=R_a{}^b{}_m{}^cX^m\sps
    \Box_a{}^bX_c=-R_a{}^b{}_c{}^mX_m
\end{equation}
    (доказательство в приложении: (\ref{re25p}) - (\ref{re42p})).
    \item
    Дифференциальные тождества Бианки
\begin{equation}
\label{re48s}
    \nabla_{\left[\right.\alpha}R_{\beta\gamma\left.\right]\delta\lambda}=0
\end{equation}
    примут вид
\begin{equation}
\label{re48}
    \nabla_{\left[cm\right.}R_{t\left.\right]}{}^k{}_r{}^s=
    \delta_{\left[m\right.}{}^k\nabla_{c\left|n\right|}R_{t\left.
    \right]}{}^n{}_r{}^s
\end{equation}
    (Доказательство в приложении: (\ref{re43p}) - (\ref{re50p})).
    \item
    Свернем (\ref{re48}) с $\delta_k{}^c$ и получим
\begin{equation}
\label{re49}
    \nabla_{c\left(\right.m}R_{t\left.\right)}{}^c{}_r{}^s=0\sps
\end{equation}
    а свертка последнего с $\delta_s{}^m$ даст
\begin{equation}
\label{re50}
    \nabla_{cm}R_t{}^c{}_r{}^m=1/8\nabla_{rt}R\sps
\end{equation}
     что эквивалентно известному уравнению
\begin{equation}
\label{re50s}
    \nabla^\alpha(R_{\alpha\beta}-1/2Rg_{\alpha\beta})=0\spsd
\end{equation}
\end{enumerate}

\subsection{\texorpdfstring{Каноническая форма бивекторов
         6-мерных $\mbox{(псевдо-)}$ евклидовых пространств $\mathbb R^6_{(p,q)}$
         с метрикой четного индекса q}{Каноническая форма бивекторов}}
\begin{theoremr}
\label{theorem7r}
    (О канонической форме бивектора).\\
    Для пространства $\mathbb R^6_{(p,q)}$ с метрикой четного индекса q=0,6
    невырожденная кососимметрическая билинейная форма может
    быть приведена в некотором базисе к каноническому виду
\begin{equation}
\label{re1g}
    \frac{1}{2}R_{\alpha\beta}X^\alpha Y^\beta=
    R_{16}X^{\left[\right.1} Y^{6\left.\right]}+
    R_{23}X^{\left[\right.2} Y^{3\left.\right]}+
    R_{45}X^{\left[\right.4} Y^{5\left.\right]}\spsd
\end{equation}
\end{theoremr}

\begin{proof}
    Поскольку в случае пространства $\mathbb R^6_{(p,q)}$
    с метрикой четного индекса q верно
\begin{equation}
\label{re2g}
    R_{\alpha\beta}=A_{\alpha\beta a}{}^b R_b{}^a\sps
    R_b{}^a=-\bar R^a{}_b\sps R_a{}^a=0\sps
\end{equation}
    что означает эрмитову симметрию  $iR_b{}^a$ в
    случае q=0,6. Ввиду этого матрица тензора $R_b{}^a$
    приводится к диагональному виду с помощью преобразований
    из некоторой группы, изоморфной $SU(4)$. Этим преобразованиям
    соответствуют преобразования из ортогональной группы $SO^e(6,\mathbb R)$.
    Отсюда следует, что матрица тензора $R_b{}^a$ в специальном
    базисе имеет вид
\begin{equation}
\label{re3g}
    \begin{array}{c}
    R_b{}^a=
    \left(
    \begin{array}{cccc}
    \lambda_1 & 0 & 0 & 0\\
    0 & \lambda_2  & 0 & 0\\
    0 & 0 &\lambda_3  & 0\\
    0 & 0 & 0 & \lambda_4 \\
    \end{array}
    \right)\sps\\ \\
    \begin{array}{ccc}
    q=0,6, & \lambda_1 ,\lambda_2, \lambda_3, \lambda_4 \in iR, &
    \lambda_1+\lambda_2+\lambda_3+\lambda_4=0\spsd\\
    \end{array}
    \end{array}
\end{equation}
    При этом выполнены 2 равенства
\begin{equation}
\label{re4g}
    \tilde R_b{}^a=S_b{}^cR_c{}^d\bar S^a{}_d\sps
    S_a{}^b\bar S^c{}_b=\delta_a{}^c\spsd
\end{equation}
\begin{proof}
    Будем рассматривать преобразования $K_\alpha{}^\beta$ из
    связной компоненты $SO^e(6,\mathbb R)$ и его соответствующее
    спинорное представление из группы $SU(4)$ вида $S_a{}^b$
\begin{equation}
\label{re5g}
    \begin{array}{c}
    K_\alpha{}^\beta K_\gamma{}^\delta R_{\beta\delta}=
    -A^{\beta\delta}{}_a{}^b K_{\left[\right.\alpha}{}^{ak}
    K_{\beta\left.\right]bk}A_{\beta\delta}{}_c{}^d R_d{}^c=\\ \\
    =-\frac{1}{2}(\frac{1}{2}\delta_a{}^b\delta_c{}^d-
    2\delta_a{}^d\delta_c{}^b)K_{\left[\right.\alpha}{}^{ak}
    K_{\beta\left.\right]bk}R_d{}^c=\\ \\
    =K_{\left[\right.\alpha}{}^{ak}K_{\beta\left.\right]bk}R_a{}^b=
    \frac{1}{2}A_{\alpha\beta c}{}^dK_{dr}{}^{ak}K^{cr}{}_{bk}R_a{}^b=\\ \\
    =\frac{1}{8}K_{dr}{}^{ak}K_{mn}{}^{sl}\varepsilon_{slbk}\varepsilon^{mncr}
    R_a{}^b A_{\alpha\beta c}{}^d=\tilde R_{\alpha\beta}=
    A_{\alpha\beta c}{}^d \tilde R_d{}^c\spsd
    \end{array}
\end{equation}
    Умножим обе части (\ref{re5g}) на $A^{\alpha\beta}{}_ p{}^t$
    и получим
\begin{equation}
\label{re6g}
    \begin{array}{c}
    \frac{1}{8}K_{pr}{}^{ak}K_{mn}{}^{sl}\varepsilon_{slbk}\varepsilon^{mntr}
    R_a{}^b=\tilde R_p{}^t\sps\\ \\
    \frac{1}{2}S_{\left[\right.p}{}^aS_{r\left.\right]}{}^k
    S_{\left[\right.m}{}^sS_{n\left.\right]}{}^l
    \varepsilon_{slbk}\varepsilon^{mntr}R_a{}^b=
    \frac{1}{2}S_{\left[\right.p}{}^a(S_{r\left.\right]}{}^k
    S_m{}^sS_n{}^l\varepsilon_{kslb})\varepsilon^{mntr}R_a{}^b=\\ \\
    =\frac{1}{2}S_{\left[\right.p}{}^a((S^{-1})_{|b|}{}^q
    \varepsilon_{r\left.\right]mnq})\varepsilon^{mntr}R_a{}^b=\\ \\
    =\frac{1}{4}(6S_p{}^a(S^{-1}){}_b{}^t-4S_r{}^a(S^{-1}){}_b{}^q
    \delta_{\left[\right.q}{}^t\delta_{p\left.\right]}{}^rR_a{}^b)=\\ \\
    =S_p{}^aR_a{}^b(S^{-1})_b{}^t=S_p{}^aR_a{}^b\bar S^t{}_b=\tilde R_p{}^t\spsd
    \end{array}
\end{equation}
\end{proof}
    Используя специальный базис, найдем в случае q=0,6 соответствующие
    координаты кобивектора из $\mathbb R^6_{(p,q)}$
\begin{equation}
\label{re7g}
    \begin{array}{c}
    R_{16}=A_{16a}{}^bR_b{}^a= -R_{61} \sps \\
    R_{23}=A_{23a}{}^bR_b{}^a= -R_{32} \sps \\
    R_{45}=A_{45a}{}^bR_b{}^a= -R_{54} \spsd \\
    \end{array}
\end{equation}
\end{proof}
    Отметим, что похожее утверждение можно было бы
    сформулировать и для случая q=2,4. Однако здесь
    возникнут некоторые сложности, связанные с
    проблемой диагонализации, поскольку в этом
    случае матрица тензора эрмитового поляритета в
    специальном базисе будет отлична от единичной.
\vspace{1cm}

\subsection{\texorpdfstring{Геометрическое представление твистора в $\mathbb R^6_{(2,4)}$}{Геометрическое представление твистора}}
\subsubsection{Стереографическая проекция}$ $

    \indent В этой части определяется понятие протяженности изотропного
    вектора пространства $\mathbb R^6_{(2,4)}$ с метрикой индекса 4.
    Ниже будет показано как выбрать вектор единичной протяженности.
    Тогда векторы, коллинеарные такому вектору, будут отличаться
    от последнего на некоторый  действительный множитель r - "протяженность
    флагштока".

    Метрика пространства $\mathbb R^6_{(2,4)}$ имеет вид
\begin{equation}
\label{re78}
    dS^2=dT^2+dV^2-dW^2-dX^2-dY^2-dZ^2\spsd
\end{equation}
    Пусть, кроме того, задано сечение светового конуса $K_6$
\begin{equation}
\label{re79}
    T^2+V^2-W^2-X^2-Y^2-Z^2=0
\end{equation}
    плоскостью V+W=1. Рассмотрим стереографическую проекцию
    этого сечения    на    плоскость     (V=0,W=1)    с     полюсом
    $N(0,\frac{1}{2},\frac{1}{2},0,0,0)$ так,  что  точке  P(T,V,W,X,Y,Z)
    соответствует $p(t,0,1,x,y,z)$ на плоскости V=0. Тогда
    выполнено
\begin{equation}
\label{re80}
    T/t=X/x=Y/y=Z/z=-\frac{(V-\frac{1}{2})}{\frac{1}{2}}\spsd
\end{equation}
    Сделаем замену
\begin{equation}
\label{re82}
    \varsigma=-ix+y \sps
    \omega=-i(t+z)  \sps
    \eta=i(z-t)
\end{equation}
    и получим
\begin{equation}
\label{re83}
    \varsigma =\frac{-iX+Y}{2V-1}\sps
    \eta=\frac{-i(T+Z)}{2V-1}    \sps
    \omega=\frac{i(Z-T)}{2V-1}   \spsd
\end{equation}
    Поэтому индуцированная метрика имеет вид
\begin{equation}
\label{re84}
    d s^2:=dT^2-dX^2-dY^2-dZ^2=-
    \frac{ d\varsigma d\bar\varsigma+d\omega d\eta}
    {(\varsigma\bar\varsigma+\eta\omega)^2}
\end{equation}
    (Доказательство этого факта вынесено в приложение
    (\ref{re81p}) - (\ref{re86p})).
    Положим
\begin{equation}
\label{re85}
    X:=
    \left(
    \begin{array}{ccc}
    \omega & \varsigma \\
    -\bar\varsigma & \eta
    \end{array}
    \right) \sps
    dX:=
    \left(
    \begin{array}{ccc}
    d\omega & d\varsigma \\
    -d\bar\varsigma & d\eta
    \end{array}
    \right) \sps
    \frac{\partial}{\partial X}:=
    \left(
    \begin{array}{ccc}
    \frac{\partial}{\partial\omega}   &
    \frac{\partial}{\partial\varsigma} \\
    -\frac{\partial}{\partial\bar\varsigma} &
    \frac{\partial}{\partial\eta}
    \end{array}
    \right)
\end{equation}
    Тогда (\ref{re84}) примет вид
\begin{equation}
\label{re86}
    d s^2=-\frac{det(dX)}{(det(X))^2}\sps
    \bar X^T+X=0\spsd
\end{equation}
\newpage
    \noindent Рассмотрим группу дробно-линейных преобразований L
\begin{equation}
\label{re87}
    {\tilde X}=(AX+B)(CX+D)^{-1} \sps
    S:=
    \left(
    \begin{array}{ccc}
     A & B \\
     C & D
    \end{array}
    \right) \sps
    detS=1\spsd
\end{equation}
    Условие действительности, накладываемое на X ($X^*+X=0)$,
    даст подгруппу унитарных дробно-линейных преобразований
    LU(2,2) так, что матрица S из (\ref{re87}) удовлетворяет условию
\begin{equation}
\label{re87a}
    S^*\hat E S=\hat E\sps
    \hat E:=\left(
    \begin{array}{cc}
    0 & E \\
    E & 0
    \end{array}
    \right)\spsd
\end{equation}
    Определим далее в специальном базисе согласно (\ref{re33v})
    \cite[т. 2, с. 83, (6.2.18), с. 361, (9.3.7)]{Penrose1r}
\begin{equation}
\label{re89}
    \begin{array}{ccccc}
    R^{12} & = & 1/\sqrt{2}(V+W)  & = &  \omega^0\xi^1-\omega^1\xi^0 \sps\\
    R^{34} & = & 1/\sqrt{2}(V-W)  & = &  \bar\pi^0\bar\eta^1-\bar\pi^1\bar\eta^0\sps\\
    R^{14} & = & i/\sqrt{2}(T+Z)  & = &  \omega^0\bar\eta^0-\xi^0\bar\pi^0\sps\\
    R^{23} & = & i/\sqrt{2}(Z-T)  & = &  \xi^1\bar\pi^1-\omega^1\bar\eta^1\sps\\
    R^{24} & = & 1/\sqrt{2}(Y+iX) & = & \omega^1\bar\eta^0-\xi^1\bar\pi^0\sps\\
    R^{13} & = & 1/\sqrt{2}(Y-iX) & = & \xi^0\bar\pi^1-\omega^0\bar\eta^1\spsd
    \end{array}
\end{equation}
    Эта формула примечательна тем, что в ней показано
    выражение бивектора $R^{ab}$ через его спинорные
    компоненты. Положим
\begin{equation}
\label{re88}
    \begin{array}{c}
    X:=YZ^{-1} \sps
    {\tilde Y}=AY+BZ \sps
    {\tilde Z}=CY+DZ \sps\\ \\
    Y=\left(
    \begin{array}{cc}
    \omega^0 & \xi^0 \\
    \omega^1 & \xi^1
    \end{array}
    \right)\sps
    Z=\left(
    \begin{array}{cc}
    \bar\pi_0 & \bar \eta_0 \\
    \bar\pi_1 & \bar \eta_1
    \end{array}
    \right)\sps
    \end{array}
\end{equation}
    тогда из (\ref{re89}) будет следовать
\begin{equation}
\label{re99}
    \begin{array}{c}
    R:=\parallel R^{ab}\parallel=
    \left(
    \begin{array}{ccc}
    (detY)J & YZ^{-1}(detZ)J \\
    -(YZ^{-1}(detZ)J)^T & (detZ)J
    \end{array}
    \right)=\\ \\
    =\left(
    \begin{array}{c}
    Y\\
    Z
    \end{array}
    \right)J\left(
    \begin{array}{cc}
    Y^T & Z^T
    \end{array}
    \right)\sps\\ \\
    J:=\left(
    \begin{array}{cc}
    0  & E \\
    -E & 0
    \end{array}
    \right)\sps
    \tilde R=SRS^T\spsd
    \end{array}
\end{equation}
\newpage
    \noindent Положим
\begin{equation}
\label{re100}
    \tilde S:=\tilde IS\tilde I^{-1}\sps
    \tilde I:=\frac{1}{\sqrt 2}
    \left(
    \begin{array}{ccc}
    E & E  \\
    -E & E
    \end{array}
    \right) \sps
\end{equation}
    тогда получим
\begin{equation}
\label{re101}
    \tilde S^*\tilde E\tilde S=\tilde E\spsd
\end{equation}
    Матрицы S образуют группу $\tilde {SU}(2,2)$, поэтому из
    (\ref{re101}) следует, что матрицы $\tilde S$ образуют группу
    SU(2,2) и (\ref{re100}) устанавливает
    изоморфизм этих групп. Назовем преобразования из группы LU(2,2)
    твисторными преобразованиями.  Ввиду двулистности
    накрытия связной компоненты единицы группы $SO(2,4)$
    (которая обозначается через $SO^e(2,4)$) группой SU(2,2) и
    двулистности накрытия группы конформных преобразований
    $C^{\uparrow 4}_+(1,3)$ (\cite[т. 2, с. 359, (9.2.10)]{Penrose1r}) группой
    $SO^e(2,4)$, следует существование цепочки изоморфизмов
\begin{equation}
\label{re102}
    \begin{array}{c}
    SU(2,2)/\{\pm 1 ;\pm i\}
    \cong LU(2,2)\cong C^{\uparrow 4}_+(1,3) \cong
    SO^e(2,4)/\{\pm 1\}\spsd
    \end{array}
\end{equation}
    Это означает, что группа LU(2,2) исчерпывает все
    конформные преобразования из группы $C^{\uparrow 4}_+(1,3)$.
    При этом матрица S из (\ref{re87}) восстанавливается с
    точностью до множителя $\lambda$ такого, что $\lambda^4=1$
    (det(S)=1), откуда и появляется указанная неоднозначность.
    На основании того, что верны тождества
\begin{equation}
\label{re90}
    Y=AX+B\ \ \Rightarrow \ \ dX=AdY \sps
    Y=X^{-1}\ \ \Rightarrow \ \ dX=-X^{-1}dXX^{-1}\sps
\end{equation}
    где A и B - некоторые постоянные матрицы, и
    используя условия  (\ref{re86}) - (\ref{re99}), имеем
\begin{equation}
\label{re92}
    \tilde Z^*\ d\tilde X\ \tilde Z=
    Z^*\ dX\ Z\spsd
\end{equation}
    Это уравнение инвариантно относительно преобразований
    из группы LU(2,2).
    (Доказательство этого факта рассмотрено в приложении
    (\ref{re87p}) - (\ref{re100p})).
    Другой инвариант можно получить,
    рассматривая тождества
\begin{equation}
\label{re93}
    Y=AX+B\ \ \Rightarrow \ \
    \frac{\partial}{\partial X}=A^T\frac{\partial}{\partial Y} \sps
    Y=X^{-1}\ \ \Rightarrow \ \ \frac{\partial}{\partial X}
    =-Y^T\frac{\partial}{\partial Y}Y^T\sps
\end{equation}
    где A и B - тоже некоторые постоянные матрицы.
    Он будет иметь вид
\begin{equation}
\label{re95}
    \tilde Z^{-1}\frac{\partial}{\partial
    \tilde X^T}\tilde Z^{*\ -1}=
     Z^{-1}\frac{\partial}{\partial X^T} Z^{*\ -1}
\end{equation}
    (Доказательство можно найти в приложении
    (\ref{re101p}) - (\ref{re114p})).
    Это означает, что существует действительный касательный
    вектор $\tilde L$ к гиперболоиду, полученному сечением конуса
    $K_6$ плоскостью V+W=1, инвариантный относительно
    преобразований базиса из группы LU(2,2) (т.е.
    координатно-независимый на касательном
    пространстве к данному гиперболоиду) и
    однозначно определенный матрицей
\begin{equation}
\label{re96}
    \begin{array}{c}
    \hat L:=\frac{1}{\sqrt{2}}
    (Z^{-1}\frac{\partial}{\partial X^T} Z^{*\ -1}-
    \bar Z^{-1}\frac{\partial}{\partial X} Z^{T\ -1})=\\ \\
    =\left(
    \begin{array}{cc}
    0  & 1 \\
    -1 & 0
    \end{array}
    \right)
    (\frac{\partial}{\partial\omega}(-\bar\eta^0\pi^0+\eta^0\bar\pi^0)+
    \frac{\partial}{\partial\eta}(-\bar\eta^1\pi^1+\eta^1\bar\pi^1)+\\ \\
    +\frac{\partial}{\partial\xi}(-\bar\eta^1\pi^0+\eta^0\bar\pi^1)+
    \frac{\partial}{\partial\bar\xi}(\bar\eta^0\pi^1-\eta^1\bar\pi^0))
    \cdot \frac{1}{(det(Z))^2 \sqrt{2}}:=
    \left(
    \begin{array}{cc}
    0  & 1 \\
    -1 & 0
    \end{array}
    \right)\tilde L\spsd
    \end{array}
\end{equation}
    Норма этого вектора в метрике (\ref{re86}) будет такой
\begin{equation}
\label{re96a}
    \parallel\tilde L\parallel =-\frac{1}{2(det(Y))^2}=-\frac{1}{(V+W)^2}\spsd
\end{equation}
    Назовем изотропный вектор k вектором, имеющим единичную
    протяженность первого типа (сравн. \cite[т. 1, с. 57, (1.4.16)]{Penrose1r}), в том
    случае, когда k будет задавать точку на изотропном конусе,
    принадлежащую сечению плоскостью V+W=1.
    Тогда $\parallel\tilde L\parallel=-1$ и любой изотропной вектор K,
    коллинеарный k, определится как
\begin{equation}
\label{re96aa}
    K=(-\parallel\tilde L\parallel)^\frac{1}{2}k\spsd
\end{equation}
    Однако, при V=-W получаются вектора с бесконечной протяженностью
    первого типа.
    Чтобы научиться их различать можно задавать сечение
    $K_6$ не плоскостью V+W=1, а T+Z=1 и ввести подобным
    образом некоторый вектор $\tilde{\tilde L}$ с нормой
\begin{equation}
\label{re96ab}
    \parallel\tilde{\tilde L}\parallel =-\frac{1}{(T+Z)^2}\spsd
\end{equation}
    Назовем изотропный вектор k вектором, имеющим единичную
    протяженность второго типа в том
    случае, когда k будет задавать точку на изотропном конусе,
    принадлежащую сечению плоскостью T+Z=1 и протяженность
    первого типа не будет конечной.
    Определим протяженность вектора К как конечную протяженность
    первого типа, а если такой не существует, то
    как протяженность второго типа.
    Отметим, что вектор $\tilde L$  не является координатно-независимым
    в пространстве $\mathbb R^6_{(2,4)}$, хотя и является инвариантом
    касательного пространства к гиперболоиду, полученному сечением
    конуса $K_6$ плоскостью V+W=1. Следующей нашей задачей
    и будет нахождение инварианта пространства $\mathbb R^6_{(2,4)}$.

\subsubsection{Геометрическое изображение твистора в 6-мерном пространстве}$ $

    \indent Теперь появилась возможность наглядно изобразить
    твистор в пространстве $\mathbb R^6_{(2,4)}$.
    Рассмотрим пару векторов из $\mathbb R^6_{(2,4)}$ равной протяженности
\begin{equation}
\label{re59}
K^\alpha=\eta^\alpha{}_{ab}\ iT^{\left[\right. a}X^{b\left.\right]}\sps
N^\alpha=\eta^\alpha{}_{ab}\ T^{\left[\right. a}Z^{b\left.\right]}\spsd
\end{equation}
    Из леммы \ref{lemma1r} второй главы следует, что
\begin{equation}
\label{re61}
    K^\alpha K_\alpha=0\sps
    N^\alpha K_\alpha=0\sps
    N^\alpha N_\alpha=0\spsd
\end{equation}
    Выберем вектор $Y^a$ таким образом, чтобы были выполнены
    условия
\begin{equation}
\label{re62}
    Y^aY_a=0\sps
    Y^aX_a=0\sps
    Y^aZ_a=0\sps
\end{equation}
\begin{equation}
\label{re66}
    \varepsilon^{abcd}X_aY_bZ_cT_d=X^cZ_cY^dT_d=1\sps
    \varepsilon^{abcd}=24X^{\left[\right.a}Y^bZ^cT^{d\left.\right]}\spsd
\end{equation}
    Таким образом получится базис из векторов $X^a,Y^a,Z^a,T^a$
    следующего вида (напомним $Y_a=s_{aa'}Y^{a'}$)
\begin{equation}
\label{re63}
    \begin{array}{c}
    Y^aY_a=0\sps
    Y^aX_a=0\sps
    Y^aZ_a=0\sps
    X^aX_a=0\sps
    X^aT_a=0\sps\\
    Z^aZ_a=0\sps
    Z^aT_a=0\sps
    T^aT_a=0\spsd
    \end{array}
\end{equation}
    Откуда
\begin{equation}
\label{re67}
    \varepsilon_{abcd}T^{\left[\right.c}Y^{d\left.\right]}=
    -2X_{\left[\right.a}Z_{b\left.\right]}\spsd
\end{equation}
    Поэтому векторы
\begin{equation}
\label{re68}
    L^\alpha=\eta^\alpha{}_{ab}(-T^{\left[\right.a}Y^{b\left.\right]}+
                                X^{\left[\right.a}Z^{b\left.\right]})\sps
    M^\alpha=\eta^\alpha{}_{ab}(-i)(T^{\left[\right.a}Y^{b\left.\right]}+
                                X^{\left[\right.a}Z^{b\left.\right]})
\end{equation}
    удовлетворяют следующим соотношениям
    $$
    L^\alpha=\bar L^\alpha\sps
    M^\alpha=\bar M^\alpha\sps
    $$
    $$
    L^\alpha K_\alpha=0\sps
    L^\alpha M_\alpha=0\sps
    M^\alpha K_\alpha=0\sps
    L^\alpha N_\alpha=0\sps
    M^\alpha N_\alpha=0\sps
    $$
\begin{equation}
\label{re69}
    L^\alpha L_\alpha=-2\sps
    M^\alpha M_\alpha=-2\spsd
\end{equation}
    Вот теперь мы можем построить тривектор
\begin{equation}
\label{re70}
    P^{\alpha\beta\gamma}=6K^{\left[\right.\alpha}
                           N^\beta L^{\gamma\left.\right]}\spsd
\end{equation}
    Зная $K^\alpha$, мы знаем $T^a$ и $X^a$  с точностью до
\begin{equation}
\label{re71}
    X^a\longmapsto\lambda_1 X^a+\mu_1 T^a\sps
    T^a\longmapsto\nu_1 X^a+\xi_1 T^a\sps
    det\left(
    \begin{array}{cc}
    \lambda_1 & \mu_1 \\
    \nu_1  & \xi_1
    \end{array}
    \right)=1\spsd
\end{equation}
    А если нам известен $N^\alpha$, то произвол в выборе
    $T^a$ и $Z^a$ таков
\begin{equation}
\label{re72}
    Z^a\longmapsto\lambda_2 Z^a+\mu_2 T^a\sps
    T^a\longmapsto\nu_2 Z^a+\xi_2  T^a\sps
    det\left(
    \begin{array}{cc}
    \lambda_2 & \mu_2 \\
    \nu_2  & \xi_2
    \end{array}
    \right)=1\spsd
\end{equation}
    Поэтому $\nu_1=\nu_2=0$ и $\xi_2=\xi_1$.
    Для $Y^a$ получим
\begin{equation}
\label{re73}
    Y^a  \longmapsto  \alpha X^a+\beta Y^a+\gamma Z^a+\delta T^a\spsd
\end{equation}
    Если теперь потребовать сохранение (\ref{re63}) (сохранение базиса вида (\ref{re63}) на самом деле
    означает, что два таких базиса связаны преобразованием
    из группы LU(2,2)), то получим
    $$
    \begin{array}{cclccl}
    X^a & \longmapsto & \tau^{-1} X^a+\mu T^a\ , \ &
    T^a & \longmapsto & \tau T^a\sps\\
    Z^a & \longmapsto & \tau^{-1} Z^a+\chi T^a\ , \ &
    Y^a & \longmapsto & -\bar\chi X^a+\tau Y^a-\bar\mu Z^a+\delta T^a\sps
    \end{array}
    $$
\begin{equation}
\label{re74}
    \bar\chi\mu+\bar\mu\chi+\tau\bar\delta+\bar\tau\delta=0\sps
    \tau\bar\tau=1\spsd
\end{equation}
    Откуда
\begin{equation}
\label{re75}
    \begin{array}{c}
    \begin{array}{cclcccl}
    X^{\left[\right.a}T^{b\left.\right]}& \longmapsto &
    X^{\left[\right.a}T^{b\left.\right]}& \Leftrightarrow &
    K^\alpha & \longmapsto &  K^\alpha \sps\\
    Z^{\left[\right.a}T^{b\left.\right]}& \longmapsto &
    Z^{\left[\right.a}T^{b\left.\right]}& \Leftrightarrow &
    N^\alpha & \longmapsto &  N^\alpha \sps
    \end{array}\\ \\
    \begin{array}{ccl}
    T^{\left[\right.a}Y^{b\left.\right]} & \longmapsto &
    -\tau\bar\chi T^{\left[\right.a}X^{b\left.\right]}+
    \tau^2 T^{\left[\right.a}Y^{b\left.\right]}-
    \bar\mu\tau T^{\left[\right.a}Z^{b\left.\right]}\sps\\
    X^{\left[\right.a}Z^{b\left.\right]} & \longmapsto &
    \tau^{-2} X^{\left[\right.a}Z^{b\left.\right]}+
    \tau^{-1}\chi X^{\left[\right.a}T^{b\left.\right]}+
    \tau^{-1}\mu T^{\left[\right.a}Z^{b\left.\right]}\spsd
    \end{array}
    \end{array}
\end{equation}
    Положим
\begin{equation}
\label{re76}
    \tau=:e^{i\Theta}\sps
\end{equation}
    тогда
    $$
    L^\alpha \psps \longmapsto \psps
    L^\alpha \cos(2\Theta) +M^\alpha\sin(2\Theta)-
    $$
    $$
    -i(\bar\chi\tau-\bar\tau\chi)K^\alpha+
    (\mu\bar\tau+\tau\bar\mu)N^\alpha)\sps
    $$
    $$
    M^\alpha \psps \longmapsto \psps
    M^\alpha \cos(2\Theta) -L^\alpha\sin(2\Theta)+
    $$
\begin{equation}
\label{re77}
    +(\bar\chi\tau+\bar\tau\chi)K^\alpha-
    i(\mu\bar\tau-\tau\bar\mu)N^\alpha)\spsd
\end{equation}
    Таким образом 3-полуплоскость натянутая на вектора
    $K^\alpha,N^\alpha,L^\alpha$ будет координатно-независима
    в пространстве $\mathbb R^6_{(2,4)}$.
    Итак, нашу конструкцию можно представить в следующем виде.
    Протяженность векторов $K^\alpha$ и $N^\alpha$ должна быть
    одинакова.
    $K^\alpha$ и $N^\alpha$ определяют 2-плоскость, множество
    векторов которой с протяженностью, равной протяженности вектора
    $K^\alpha$ и началом, совпадающим с началом вектора $K^\alpha$,
    назовем флагштоком.
    $K^\alpha$,$N^\alpha$,$L^\alpha$ определят 3-полуплоскость,
    которую назовем полотнищем флага.
    Таким образом, зная $K^\alpha$ и $N^\alpha$, мы знаем
    твистор $T^a$ с точностью до фазы $\Theta$. В свою очередь
    $2\Theta$ - это угол поворота полотнища флага - 3-полуплоскости
    $P^{\alpha\beta\gamma}$ - в 2-плоскости $(L^\alpha,M^\alpha)$
    вокруг флагштока - 2-плоскости $(N^\alpha,K^\alpha)$. Поэтому, поворот
    флага на $2\pi$ приведет к твистору $-T^a$, и только
    поворот на $4\pi$ вернет нашу конструкцию к исходному
    состоянию.
    Кроме того, коллинеарные твисторы различаются протяженностью
    вектора $K^\alpha$ так, что при преобразовании
    $T^a\longmapsto rT^a,\ Y^a\longmapsto r^{-1}Y^a\ (r\in R\backslash \{0\})$
    флагшток умножается на r, а полотнище остается неизменным.
    Следует, наконец, отметить тот факт, что указанная
    геометрическая структура однозначно
    восстанавливается по твистору $T^a$.

\newpage
\section{Теорема о двух квадриках}$ $

    \indent В этой главе исследуются  решения битвисторного
    уравнения
    $$
    \left\{
    \begin{array}{lcl}
    X^a & = & \dot X^a -ir^{ab}\dot Y_b\sps \\
    Y_b & = & \dot Y_b\sps
    \end{array}
    \right.
    $$
    приводящие к нуль-парам Розенфельда
    $$
    X^A:=(X^a,Y_b)\spsd
    $$
    За $\dot X^a,\dot Y_b$ принимаются  некоторые частные
    решения битвисторного уравнения (\ref{re53aa}). Нас
    будет интересовать геометрическое место точек,
    определенное уравнением
    $$
    X^a=0\spsd
    $$
    Показывается, что решения такого уравнения приводят
    к 2 квадрикам, для которых справедлив модифицированный
    принцип тройственности. Доказано, что модифицированный
    принцип тройственности является обобщением принципа тройственности
    Картана и соответствия Кляйна, что позволяет реализовать
    его в явном виде с помощью некоторых операторов
    $\eta_A{}^{KL}$, которые являются обобщением операторов
    Нордена $\eta_\alpha{}^{ab}$ и удовлетворяют уравнению
    Клиффорда, что приводит к числам Кэли. Доказательство
    этой теоремы является примером приложения 6-мерного
    спинорного формализма, развитого выше, и тесно
    связано с 4-мерным спинорным формализмом \cite{Penrose1r}.

\subsection{Решения битвисторного уравнения}$ $

    \indent Рассмотрим расслоение $A^\mathbb C$ со слоями, изоморфными $\mathbb C^4$,
    и базой $\mathbb CV^6$ - аналитическим комплексным пространством
    с квадратичной метрикой. Уравнение
\begin{equation}
\label{re143}
    \nabla^{a\left(b\right.}X^{c\left.\right)}=0\psps
    (a,b,...=\overline{1,4})
\end{equation}
    названо нами битвисторным уравнением ($X^c$ - аналитичны).
    По доказанному выше битвисторное уравнение будет
    конформно-инвариантным.
    Кроме того, условие интегрируемости
    уравнения (\ref{re143}) имеет вид
\begin{equation}
\label{re154}
    \frac{1}{2}\varepsilon_{akmn}\nabla^{m\left(\right.n}
    \nabla^{d|k|}X^{c\left.\right)}=\frac{5}{6}C_a{}^c{}_l{}^d
    X^l=0
\end{equation}
    (доказательство в приложении (\ref{re2500}) - (\ref{re2503})).
    Ограничимся  случаем  конформно-плоского
    пространства (смотри (\ref{re24}))
\begin{equation}
\label{re155}
    C_a{}^c{}_l{}^d=0\spsd
\end{equation}
    Это означает, что пространство $\mathbb CV^6$ конформно пространству $\mathbb C\mathbb R^6$,
    которое, не ограничивая общности, будем рассматривать далее.
    Если $X^c$ - решение (\ref{re143}), то величина $\nabla^{ab}\nabla^{cd}X^r$
    антисимметрична по $rd$, она же антисимметрична
    по $br$ ввиду того, что пространство плоское
    и производные коммутируют.
    Кроме того, имеется антисимметрия по парам $ab,cd,rc,ra$.
    Это означает, что $\nabla^{ab}\nabla^{cd}X^r$ антисимметрична по $abcdr$ и,
    следовательно, равна нулю.
    Фиксируем в $\mathbb C\mathbb R^6$ точку O - начало координат. Все остальные точки опишем
    векторами $r^a$ с началом в точке O, тогда для $r^\alpha\ne 0$ имеем
\begin{equation}
\label{re156}
    \nabla_\alpha r^\beta=\delta_\alpha{}^\beta\spsd
\end{equation}
    Поэтому $\nabla^{ab}X^c$ есть постоянная величина,
    антисимметричная по $abc$, что следует из (\ref{re154}).
    Положим
\begin{equation}
\label{re157}
    \nabla^{cd}X^a=
    -i\varepsilon^{cdab}Y_b\spsd
\end{equation}
    Проинтегрируем это уравнение, что даст решение
\begin{equation}
\label{re158}
    \left\{
    \begin{array}{lcl}
    X^a & = & \dot X^a -ir^{ab}\dot Y_b\sps \\
    Y_b & = & \dot Y_b\spsd
    \end{array}
    \right.
\end{equation}
    Здесь $r^{ab}$ - бивектор из формулы (\ref{re1}), причем множитель i
    выбран для удобства (это выяснится при рассмотрении вещественного
    случая). За $\dot X^a$ принимается постоянное векторное поле,
    значение которого совпадает со значением поля $X^a$ в начале
    координат O. При этом, в случае пространства $\mathbb R^6_{(2,4)}$,
    например
\begin{equation}
\label{re159}
    Y_k=s_{km'}\bar Y^{m'} \sps \bar Y^{m'}=\overline{Y^m}\spsd
\end{equation}
    Кроме того, радиус-вектор $r^\gamma$($=\frac{1}{2}\eta^\gamma{}_{ab}R^{ab}$)
    удовлетворяет  следующим соотношениям
\begin{equation}
\label{re160}
    \frac{1}{2}r^{ab}r_{ab}=pf(r)\sps
    r^{ab}r_{bc}=-\frac{1}{2}pf(r)\delta_c{}^a\spsd
\end{equation}

\subsection{Нуль-пары Розенфельда} $ $

    \indent Обозначим через $A^{\mathbb C*}$ пространство, двойственное $A^\mathbb C$ и
    образуем 8-мерное комплексное пространство $\mathbb T^2$ как
    прямую сумму $A^\mathbb C\oplus A^{\mathbb C*}$. То есть, если $X^a$ - координаты
    вектора в $A^\mathbb C$, а $Y_b$ - координаты ковектора в $A^{*\mathbb C}$, то
\begin{equation}
\label{re161}
    X^A:=(X^a,Y_b)\psps (A,B,...=\overline{1,8})
\end{equation}
    будут координатами вектора из $\mathbb T^2$.
    Преобразование (\ref{re158}) является линейным преобразованием,
    однако, не сохраняющим структуру упомянутой прямой суммы.
    Будем рассматривать координаты бивектора $r^{ab}$ как
    координаты точки аффинного пространства $\mathbb C\mathbb A^6$. Нас
    будет интересовать множество точек, заданных
    уравнением
\begin{equation}
\label{re162}
    X^a=0\psps \Leftrightarrow \psps \dot X^a=ir^{ab}\dot Y_b\spsd
\end{equation}
    Это есть система из 4 линейных уравнений с 6 неизвестными.
    Для выяснения ее ранга рассмотрим однородное уравнение
\begin{equation}
\label{re163}
    r^{ab}\dot Y_b=0\sps
\end{equation}
    которое имеет ненулевые решения тогда и только тогда,
    когда бивектор простой
\begin{equation}
\label{re164}
    \varepsilon_{abfc}r^{ab}r^{cd}=-pf(r)\delta_f{}^d=0
\end{equation}
     и, следовательно, представим в виде
\begin{equation}
\label{re165}
     r^{ab}_{\mbox{\scriptsize{однородное}}}=\dot P^a\dot Q^b-\dot P^b\dot Q^a\sps
\end{equation}
     причем $\dot P^a$ и $\dot Q^a$ определены с точностью до их
     линейных комбинаций. Из этого следует, что
\begin{equation}
\label{re166}
     \dot P^a\dot Y_a=0\sps
     \dot Q^a\dot Y_a=0\spsd
\end{equation}
     Обозначим через $\dot X^a$,$\dot S^a$,$\dot Z^a$ все
     решения уравнения
\begin{equation}
\label{re167}
      \dot X^a\dot Y_a=0\sps
\end{equation}
     которые образуют базис пространства, определенного (\ref{re167}).
     Тогда уравнение (\ref{re165}) примет вид
\begin{equation}
\label{re168}
    r^{ab}_{\mbox{\scriptsize{однородное}}}=
    \lambda_1\dot S^{\left[\right.a}\dot X^{b\left.\right]}+
    \lambda_2\dot X^{\left[\right.a}\dot Z^{b\left.\right]}+
    \lambda_3\dot S^{\left[\right.a}\dot Z^{b\left.\right]}
\end{equation}
    и, следовательно, определит в пространстве бивекторов
    3-мерное подпространство.
    Отсюда получается общее решение уравнения (\ref{re162}) в виде
\begin{equation}
\label{re169}
    r^{ab}=
    r^{ab}_{\mbox{\scriptsize{частное}}}+
    \lambda_1\dot S^{\left[\right.a}\dot X^{b\left.\right]}+
    \lambda_2\dot X^{\left[\right.a}\dot Z^{b\left.\right]}+
    \lambda_3\dot S^{\left[\right.a}\dot Z^{b\left.\right]}\sps
\end{equation}
    где $r^{ab}_{\mbox{\scriptsize{частное}}}$-произвольный бивектор, являющийся
    частным решением (\ref{re162}).\\

\subsection{\texorpdfstring{Построение квадрик $\mathbb CQ_6$ и $\mathbb C\tilde Q_6$}{Построение двух квадрик}}$ $

    \indent Определенное ранее пространство $\mathbb T^2$ является комплексным
    евклидовым пространством, в котором скалярный квадрат вектора определится
    квадратичной формой
\begin{equation}
\label{re170}
    \varepsilon_{AB}X^AX^B=2X^aY_a
\end{equation}
    в смысле определения (\ref{re161}) так, что
    матрица тензора $\varepsilon_{AB}$ имеет вид
\begin{equation}
\label{re172}
    \parallel \varepsilon_{AB}\parallel=\left(
    \begin{array}{cc}
    0 & \delta_a{}^c \\
    \delta^b{}_d  & 0
    \end{array}
    \right)\spsd
\end{equation}
    Форма (\ref{re170}) будет инвариантна
    по отношению к преобразованию  (\ref{re158})
\begin{equation}
\label{re171}
    X^aY_a=(\dot X^a-ir^{ab}\dot Y_b)\dot Y_a=\dot X^a\dot Y_b\spsd
\end{equation}
    При фиксированном $r^{ab}$ уравнение (\ref{re162})
    определит в $\mathbb T^2$ 4-мерное пространство, которое будет
    являться 4-мерной плоской образующей конуса
\begin{equation}
\label{re173}
    \varepsilon_{AB}X^AX^B=0\spsd
\end{equation}
    Таким образом мы можем рассматривать квадрику $\mathbb CQ_6$,
    задаваемую уравнением (\ref{re173}), в проективном
    пространстве $\mathbb C\mathbb P_7$. Ее 4 базисные точки будут
    удовлетворять условию
\begin{equation}
\label{re174}
    \varepsilon_{AB}X^A_iX^B_j=0\psps (i,j,...=\overline{1,4})\spsd
\end{equation}
    Положим
\begin{equation}
\label{re175}
    X^A_1:=(\dot X^a,\dot Y_b)\sps
    X^A_2:=(\dot Z^a,\dot T_b)\sps
    X^A_3:=(\dot L^a,\dot N_b)\sps
    X^A_4:=(\dot K^a,\dot M_b)\spsd
\end{equation}
    На основании (\ref{re169}) каждой точке квадрики $\mathbb CQ_6$
    мы можем поставить в соответствие 3-мерную изотропную плоскость
    пространства $\mathbb C\mathbb A_6$. Точку пространства $\mathbb C\mathbb A_6$ (t,v,w,x,y,z)
    можно представить прямой ($\lambda T,\lambda V,\lambda U,\lambda S,
    \lambda W,\lambda X,\lambda Y,\lambda Z$) пространства $\mathbb C\mathbb R^8$,
    обладающего метрикой
\begin{equation}
\label{re176}
    dL^2=dT^2+dV^2+dU^2+dS^2+dW^2+dX^2+dY^2+dZ^2\spsd
\end{equation}
    Эти прямые будут образующими изотропного
    конуса $\mathbb CK_8$
\begin{equation}
\label{re177}
    T^2+V^2+U^2+S^2+W^2+X^2+Y^2+Z^2=0\spsd
\end{equation}
    Далее, пересечение 7-плоскости
\begin{equation}
\label{re1771}
    U-iS=1
\end{equation}
    с указанным конусом $\mathbb CK_8$ обладает индуцированной метрикой
\begin{equation}
\label{re1772}
    d\tilde L^2=dT^2+dV^2+dW^2+dX^2+dY^2+dZ^2\spsd
\end{equation}
    Это пространство имеет вид параболоида на $\mathbb CK_8$ и тождественно
    пространству $\mathbb C\mathbb R^6$
\begin{equation}
\label{re1773}
    U=1+iS=\frac{1}{2}(1-T^2-V^2-W^2-X^2-Y^2-Z^2)\spsd
\end{equation}
    Всякая образующая этого конуса (множество точек, принадлежащих $\mathbb CK_8$
    с постоянным отношением $T:V:U:S:W:X:Y:Z$), не лежащая на гиперплоскости
    $U=iS$, пересекает параболоид в единственной точке. Образующим
    конуса, лежащим на гиперплоскости $U=iS$, соответствуют точки,
    принадлежащие бесконечности пространства $\mathbb C\mathbb R^6$. Таким образом,
    прямым $\mathbb C\mathbb R^8$, проходящим через начало $\mathbb C\mathbb R^8$, соответствуют
    точки проективного пространства $\mathbb C\mathbb P_7$.
    Стереографическая проекция указанного
    сечения на плоскость $U=0$ с полюсом $N(0,0,\frac{1}{2},\frac{i}{2},0,0,0,0)$
    отображает точку $P(T,V,U,S,W,X,Y,Z)$ гиперболоида на точку $p(t,v,1,0,w,x,y,z)$
    плоскости $U=1,S=0$
\begin{equation}
\label{re17731}
    \begin{array}{c}
    \lambda T=t\sps
    \lambda V=v\sps
    \lambda W=w\sps
    \lambda X=x\sps
    \lambda Y=y\sps
    \lambda Z=z\sps\\
    \lambda =\frac{1}{U+iS}\sps
    \lambda U=\frac{1}{2}(1-t^2-v^2-w^2-x^2-y^2-z^2)=-\lambda iS+1\sps\\
    pf(r)=-\frac{U-iS}{U+iS}\spsd
    \end{array}
\end{equation}
    Образующим же конуса $\mathbb CK_8$ соответствует квадрика $\mathbb C\tilde Q_6$ в проективном
    пространстве $\mathbb C\mathbb P_7$
\begin{equation}
\label{re179}
    G_{AB}R^AR^B=0\spsd
\end{equation}

\subsection{\texorpdfstring{Соответствие $\mathbb CQ_6\longmapsto \mathbb C\tilde Q_6$}{Соответствие: CQ в C'Q}}$ $

    \indent На основании (\ref{re169})
\begin{equation}
\label{re178}
    r^{ab}=r^{ab}_{\mbox{\scriptsize{частное}}}+r^{ab}_{\mbox{\scriptsize{однородное}}}=
    r^{ab}_{\mbox{\scriptsize{частное}}}+
    \lambda_1\dot S^{\left[\right.a}\dot X^{b\left.\right]}+
    \lambda_2\dot X^{\left[\right.a}\dot Z^{b\left.\right]}+
    \lambda_3\dot S^{\left[\right.a}\dot Z^{b\left.\right]}\sps
\end{equation}
    поэтому (\ref{re169}) определит 4-мерную плоскую образующую конуса $\mathbb CK_8$.
    Уравнения (\ref{re174}),(\ref{re175}) определят систему
\begin{equation}
\label{re182}
    \left\{
    \begin{array}{lcl}
    ir^{ab}\dot Y_b & = & \dot X^a\sps\\
    ir^{ab}\dot T_b & = & \dot Z^a\sps\\
    ir^{ab}\dot N_b & = & \dot L^a\sps\\
    ir^{ab}\dot M_b & = & \dot K^a
    \end{array}
    \right.
\end{equation}
    с условиями
\begin{equation}
\label{re183}
    \begin{array}{c}
    \dot X^a\dot Y_a=0\sps
    \dot Z^a\dot T_a=0\sps
    \dot L^a\dot N_a=0\sps
    \dot K^a\dot M_a=0\sps\\
    \dot X^a\dot T_a=-\dot Z^a\dot Y_a\sps
    \dot X^a\dot N_a=-\dot L^a\dot Y_a\sps
    \dot X^a\dot M_a=-\dot K^a\dot Y_a\sps\\
    \dot Z^a\dot N_a=-\dot L^a\dot T_a\sps
    \dot Z^a\dot M_a=-\dot K^a\dot T_a\sps
    \dot K^a\dot N_a=-\dot L^a\dot M_a\spsd
    \end{array}
\end{equation}
    Таким образом из 16 уравнений с 6 неизвестными
    $r^{ab}$ существенными будут только 6 уравнений (10 условий связи (\ref{re183})),
    что определит точку $\mathbb C\mathbb A_6$, а значит и точку
    квадрики $\mathbb C\tilde Q_6$.\\

    \indent Если из системы (\ref{re182}) нам известно одно
    уравнение
\begin{equation}
\label{re184}
    ir^{ab}\dot Y_b=\dot X^a
\end{equation}
    с условием
\begin{equation}
\label{re185}
    \dot X^a\dot Y_a=0\sps
\end{equation}
    то из 4 уравнений существенными будут лишь 3 (одно условие связи).
    Это означает, что точке квадрики  $\mathbb CQ_6$
    будет соответствовать плоская 3-мерная
    образующая $\mathbb C\mathbb P_3$, принадлежащая квадрике $\mathbb C\tilde Q_6$.
    Это следует из (\ref{re169}).\\

    \indent Если из системы (\ref{re182}) нам известны два
    уравнения
\begin{equation}
\label{re186}
    \left\{
    \begin{array}{lcl}
    ir^{ab}\dot Y_b & = & \dot X^a\sps\\
    ir^{ab}\dot T_b & = & \dot Z^a
    \end{array}
    \right.
\end{equation}
    с условиями
\begin{equation}
\label{re187}
    \begin{array}{c}
    \dot X^a\dot Y_a=0\sps
    \dot Z^a\dot T_a=0\sps
    \dot X^a\dot T_a=-\dot Z^a\dot Y_a\sps
    \end{array}
\end{equation}
    то из 8 уравнений существенными будут лишь 5 (неизвестных же 6 и
    3 условия связи).
    Это означает, что прямолинейной образующей $\mathbb C\mathbb P_1$ квадрики  $\mathbb CQ_6$
    будет соответствовать прямолинейная образующая $\mathbb C\mathbb P_1$,
    принадлежащая квадрике $\mathbb C\tilde Q_6$. При этом
    многообразие образующих $\mathbb C\mathbb P_1(\mathbb CQ_6)$, принадлежащих
    одной и той же образующей $\mathbb C\mathbb P_3(\mathbb CQ_6)$, определит
    пучок образующих $\mathbb C\mathbb P_1(\mathbb C\tilde Q_6)$, принадлежащий
    квадрике  $\mathbb C\tilde Q_6$ (этот пучок является на самом деле
    конусом).
    Центр пучка определится системой (\ref{re182}).\\

    \indent Если из системы (\ref{re182}) нам известны три
    уравнения
\begin{equation}
\label{re188}
    \left\{
    \begin{array}{lcl}
    ir^{ab}\dot Y_b & = & \dot X^a\sps\\
    ir^{ab}\dot T_b & = & \dot Z^a\sps\\
    ir^{ab}\dot N_b & = & \dot L^a
    \end{array}
    \right.
\end{equation}
    с условиями
\begin{equation}
\label{re189}
    \begin{array}{c}
    \dot X^a\dot Y_a=0\sps
    \dot Z^a\dot T_a=0\sps
    \dot L^a\dot N_a=0\sps\\
    \dot X^a\dot T_a=-\dot Z^a\dot Y_a\sps
    \dot X^a\dot N_a=-\dot L^a\dot Y_a\sps
    \dot Z^a\dot N_a=-\dot L^a\dot T_a\sps
    \end{array}
\end{equation}
    то из 12 уравнений существенными будут лишь 6 (и неизвестных 6
    с 6 условиями связи).
    Это означает, что 2-мерной образующей $\mathbb C\mathbb P_2$ квадрики  $\mathbb CQ_6$
    будет соответствовать точка квадрики $\mathbb C\tilde Q_6$. При этом
    многообразие образующих $\mathbb C\mathbb P_2(\mathbb CQ_6)$, принадлежащих
    одной и той же образующей $\mathbb C\mathbb P_3(\mathbb CQ_6)$, определит
    единственную точку квадрики $\mathbb C\tilde Q_6$. Эта точка
    определится системой (\ref{re182}).

\subsection{\texorpdfstring{Связующие операторы $\eta^A{}_{KL}$}{Связующие операторы для n=8}}$ $
    \indent Исходя из вышесказанного, рассмотрим прямолинейную
    образующую квадрики $\mathbb CQ_6$, определенную бивектором
\begin{equation}
\label{re190}
    \begin{array}{c}
    \hat R^{AB}=X_1^{\left[\right.A}X_2^{B\left.\right]}=
    \left(
    \begin{array}{cc}
    \dot X^a\dot Z^b-\dot X^b\dot Z^a & \dot X^a\dot T_d-\dot Y_d\dot Z^a \\
    \dot Y_c\dot Z^b-\dot T_c\dot X^b & \dot Y_c\dot T_d-\dot Y_d\dot T_c
    \end{array}
    \right)=\\ \\
    =\left(
    \begin{array}{cc}
    -2r^{\left[\right.a|k|}r^{b\left.\right]r}\dot Y_k\dot T_r &
    2ir^{ar}\dot Y_{\left[\right.r}\dot T_{b\left.\right]} \\
    i\varepsilon^{kbmn}r_{ck}\dot Y_{\left[\right.m}\dot T_{n\left.\right]}+
    \delta_c{}^b\dot X^k\dot T_k &
    2Y_{\left[\right.c}\dot T_{d\left.\right]} \\
    \end{array}
    \right)=\\ \\
    =\dot T_l\dot X^l\left(
    \begin{array}{cc}
    -\frac{1}{2}i\delta^a{}_kr^\gamma r_\gamma & r^{ar} \\
    r_{ck} & -i\delta_c{}^r
    \end{array}
     \right)\cdot\frac{1}{\dot T_l\dot X^l}\left(
    \begin{array}{cc}
    i\varepsilon^{kbmn}\dot Y_m\dot T_n  &
    0 \\
    -\delta_r{}^b r^{mn}\dot Y_m\dot T_n  &
    2iY_{\left[\right.r}T_{d\left.\right]}
    \end{array}
    \right):=\\ \\=R^A{}_KP^{KB}\spsd
    \end{array}
\end{equation}
    Положим
\begin{equation}
\label{re191}
    R^{AB}:=\varepsilon^{BC}R^A{}_C=\dot T_l\dot X^l\left(
    \begin{array}{cc}
    r^{an} & -\frac{1}{2}i\delta^a{}_kr^\gamma r_\gamma \\
    -i\delta_c{}^n & r_{ck}
    \end{array}
    \right)\sps
\end{equation}
\begin{equation}
\label{re191a}
    R_{AB}:=\varepsilon_{AC}R^C{}_B=\dot T_l\dot X^l\left(
    \begin{array}{cc}
    r_{ck} & -i\delta_c{}^n \\
    -\frac{1}{2}i\delta^a{}_kr^\gamma r_\gamma & r^{an}
    \end{array}
    \right)\spsd
\end{equation}
    При этом будет верно уравнение
\begin{equation}
\label{re192}
    R_A{}^C\hat R^{AB}=0\spsd
\end{equation}
    Очевидно, что любой тензор (\ref{re190}), представляющий
    образующую $\mathbb C\mathbb P_1(\mathbb CQ_6)$, будет содержать один и тот же
    тензор $R^A{}_K$ в своем разложении, при этом второй
    тензор разложения $P^{KB}$ будет отвечать
    за положение $\mathbb C\mathbb P_1$ в $\mathbb C\mathbb P_3$. Поэтому есть резон
    поставить в соответствие точке квадрики $\mathbb C\tilde Q_6$
    матрицу (\ref{re191}), которой она определится однозначно.
    При переходе к пространству $\mathbb C\mathbb R^8$,  исходя из
\begin{equation}
\label{re193a}
    \begin{array}{ccccccccc}
    r^{12} & = & \frac{1}{\sqrt 2}(v+iw)\sps &
    r^{13} & = & \frac{1}{\sqrt 2}(x+iy)\sps &
    r^{14} & = & \frac{i}{\sqrt 2}(t+iz)\sps \\
    r^{23} & = & \frac{i}{\sqrt 2}(iz-t)\sps &
    r^{24} & = & \frac{1}{\sqrt 2}(-x+iy)\sps&
    r^{34} & = & \frac{1}{\sqrt 2}(v-iw) \sps
    \end{array}
\end{equation}
    определим однородные координаты $\mathbb C\mathbb R^8$ следующим образом
\begin{equation}
\label{re193}
    \lambda=\left\{\begin{array}{ccccccccccccccl}
    R^{12} &:&  R^{13}  &:&  R^{14}  &:&  R^{23} &:& R^{24} &:& R^{34} &:&
    R^{15} &:& R^{51}\sps \\
    r^{12} &:& r^{13} &:& r^{14} &:& r^{23} &:& r^{24} &:& r^{34} &:&
    -\frac{1}{2}ir^\gamma r_\gamma &:& -i\sps
    \end{array}
    \right.
\end{equation}
\begin{equation}
\label{re194}
    \begin{array}{cccccc}
    R^{12} & = & \frac{1}{\sqrt 2}(V+iW)\sps &
    R^{13} & = & \frac{1}{\sqrt 2}(iY+X)\sps \\
    R^{14} & = & \frac{1}{\sqrt 2}(iT-Z)\sps  &
    R^{23} & = & \frac{1}{\sqrt 2}(-Z-iT)\sps \\
    R^{24} & = & \frac{1}{\sqrt 2}(-X+iY)\sps &
    R^{34} & = & \frac{1}{\sqrt 2}(V-iW)\sps \\
    R^{51} & = & S-iU\sps &
    R^{15} & = & \frac{1}{2}(iU+S)\sps
    \end{array}
\end{equation}$ $\\
    $$
    \begin{array}{cccccccc}
    R^{12} & = & -R^{21} & = & R^{78} & = & -R^{87}\sps \\
    R^{13} & = & -R^{31} & = & R^{86} & = & -R^{68}\sps \\
    R^{14} & = & -R^{41} & = & R^{67} & = & -R^{76}\sps \\
    R^{23} & = & -R^{32} & = & R^{58} & = & -R^{85}\sps \\
    R^{24} & = & -R^{42} & = & R^{75} & = & -R^{57}\sps \\
    R^{34} & = & -R^{43} & = & R^{56} & = & -R^{65}\sps \\
    R^{15} & = &  R^{26} & = & R^{37} & = &  R^{48}\sps \\
    R^{51} & = &  R^{62} & = & R^{73} & = &  R^{84}\sps
    \end{array}
    $$ $ $\\
\begin{equation}
\label{re195}
    \begin{array}{c}
    R^{AB}R_{AB}=8(R^{12}R^{34}-R^{13}R^{24}+R^{14}R^{23}+
    R^{15}R^{51})=\\ \\
    =4(T^2+V^2+U^2+S^2+W^2+X^2+Y^2+Z^2):=4nf(R)\\ \\ \\
    R^{AB}R_{CB}=\frac{1}{2}nf(R)\delta_C{}^A
    \ \ \Rightarrow \ \ R^{AB}R_{AB}=4nf(R)\spsd
    \end{array}
\end{equation}
    Для того, чтобы $(R_i)^A$ определяли образующую квадрики
    $\mathbb C\tilde Q_6$ необходимо и достаточно выполнение условий
\begin{equation}
\label{re197}
    G_{AB}R^A_iR^B_j=0\spsd
\end{equation}
    Определим некоторые связующие операторы $\eta_A{}^{BC}$
    так, чтобы выполнялись условия
\begin{equation}
\label{re198}
     R_A=\frac{1}{4}\eta_A{}^{BC}R_{BC}\sps
     R^A=\frac{1}{4}\eta^A{}_{BC}R^{BC}
\end{equation}
     и было верно тождество
\begin{equation}
\label{re199}
     G_{AB}\delta_K{}^L=\eta_{AK}{}^R\eta_B{}^L{}_R+
     \eta_{BK}{}^R\eta_A{}^L{}_R\spsd
\end{equation}
     Поэтому можно определить некоторые операторы $\gamma_A$
\begin{equation}
\label{re2o}
     \gamma_A:=\sqrt{2}\left(
     \begin{array}{cc}
     0 & \sigma_A \\
     \eta_A & 0
     \end{array}
     \right)\sps
     \eta_A:=\eta_A{}^{KR}\sps
     \sigma_A:=(\eta_A)^T{}_{RL}\spsd
\end{equation}
     Тогда $\gamma_A$ будут удовлетворять уравнению Клиффорда
\begin{equation}
\label{re3o}
     \gamma_A\gamma_B+\gamma_B\gamma_A=2G_{AB} I\spsd
\end{equation}
     При этом спуск и подъем одиночных индексов
     производится с помощью метрического тензора $\varepsilon_{AB}$,
     определенного выше. Определим
\begin{equation}
\label{re4o}
     \varepsilon_{PQRT}:=\eta^A{}_{PQ}\eta^B{}_{RT}G_{AB}\sps
     \varepsilon_{PQRT}=\varepsilon_{RTPQ}\sps
\end{equation}
     что даст еще один метрический тензор $\varepsilon_{PQRT}$,
     с помощью которого можно поднимать и опускать парные индексы.
     Действительно, если (\ref{re199})  свернуть $\delta_L{}^K$,
     то получим
\begin{equation}
\label{re5o}
     G_{AB}=\frac{1}{4}\eta_A{}^{PQ}\eta_{BPQ}\spsd
\end{equation}
     Свернем  (\ref{re4o})  с $\eta_C{}^{PQ}$, то
\begin{equation}
\label{re6o}
     \eta_{CRT}=\frac{1}{4}\eta_C{}^{PQ}\varepsilon_{PQRT}\spsd
\end{equation}
     Теперь тождество (\ref{re199}) можно переписать в виде
     (свернув с $\eta^A{}_{ST}\eta^B{}_{PQ}$)
\begin{equation}
\label{re7o}
     \varepsilon_{STPQ}\delta_K{}^L=
     \varepsilon_{STK}{}^R\varepsilon_{PQ}{}^L{}_R+
     \varepsilon_{PQK}{}^R\varepsilon_{ST}{}^L{}_R\spsd
\end{equation}
     Свертка же (\ref{re7o}) c $\delta_L{}^K$ даст
\begin{equation}
\label{re8o}
     \varepsilon_{STPQ}=
     \frac{1}{4}\varepsilon_{ST}{}^{KR}\varepsilon_{PQKR}\spsd
\end{equation}
     При  этом результат применения двух метрических тензоров
     должен быть одинаков
\begin{equation}
\label{re9o}
     \varepsilon_{PQ}=\frac{1}{4}\varepsilon_{PQRT}\varepsilon^{RT}\spsd
\end{equation}
     Таким образом, при наличии 3 метрических тензоров $G_{AB},
     \varepsilon_{PQRT},\varepsilon_{RT}$, можно от уравнения
     Клиффорда (\ref{re3o}) прийти к уравнению (\ref{re199}), и
     $\eta_A{}^{BC}$ являются образующими
     соответствующей алгебры Клиффорда. Опустим теперь индекс L
     в  (\ref{re7o}) и свернем его с $\varepsilon^{ST}$, то получим
\begin{equation}
\label{re10o}
     \varepsilon_{PQ(KL)}=\frac{1}{2}\varepsilon_{PQ}\varepsilon_{KL}\sps
     \varepsilon_{[PQ](KL)}=0\sps
\end{equation}
\newpage
    \noindent что приведет к тождеству
\begin{equation}
\label{re11o}
    \eta_A{}^{(MN)}=\frac{1}{8}\eta_A{}^{KL}\varepsilon_{KL}\varepsilon^{MN}\spsd
\end{equation}
    Если же тензоры $G_{AB}$  и $\varepsilon_{KL}$  имеют вид
\begin{equation}
\label{re12o}
    \begin{array}{c}
    \parallel G_{AB} \parallel=\left(
    \begin{array}{cccccccc}
    1 & 0 & 0 & 0 & 0 & 0 & 0 & 0\\
    0 & 1 & 0 & 0 & 0 & 0 & 0 & 0\\
    0 & 0 & 1 & 0 & 0 & 0 & 0 & 0\\
    0 & 0 & 0 & 1 & 0 & 0 & 0 & 0\\
    0 & 0 & 0 & 0 & 1 & 0 & 0 & 0\\
    0 & 0 & 0 & 0 & 0 & 1 & 0 & 0\\
    0 & 0 & 0 & 0 & 0 & 0 & 1 & 0\\
    0 & 0 & 0 & 0 & 0 & 0 & 0 & 1\\
    \end{array}
    \right)\sps\\ \\
    \parallel\varepsilon_{KL}\parallel=\left(
    \begin{array}{cccccccc}
    0 & 0 & 0 & 0 & 1 & 0 & 0 & 0\\
    0 & 0 & 0 & 0 & 0 & 1 & 0 & 0\\
    0 & 0 & 0 & 0 & 0 & 0 & 1 & 0\\
    0 & 0 & 0 & 0 & 0 & 0 & 0 & 1\\
    1 & 0 & 0 & 0 & 0 & 0 & 0 & 0\\
    0 & 1 & 0 & 0 & 0 & 0 & 0 & 0\\
    0 & 0 & 1 & 0 & 0 & 0 & 0 & 0\\
    0 & 0 & 0 & 1 & 0 & 0 & 0 & 0\\
    \end{array}
    \right)\sps
    \end{array}
\end{equation}
    то существенные координаты операторов $\eta^A{}_{KL}$ в некотором базисе
    будут такими
\begin{equation}
\label{re13o}
    \begin{array}{llll}
    \eta^2{}_{12}=\frac{1}{\sqrt{2}}\sps  &
    \eta^2{}_{34}=\frac{1}{\sqrt{2}}\sps  &
    \eta^5{}_{12}=-\frac{i}{\sqrt{2}}\sps &
    \eta^5{}_{34}=\frac{i}{\sqrt{2}}\sps  \\
    \eta^2{}_{78}=\frac{1}{\sqrt{2}}\sps  &
    \eta^2{}_{56}=\frac{1}{\sqrt{2}}\sps  &
    \eta^5{}_{78}=-\frac{i}{\sqrt{2}}\sps &
    \eta^5{}_{56}=\frac{i}{\sqrt{2}}\sps  \\
    \eta^1{}_{14}=-\frac{i}{\sqrt{2}}\sps &
    \eta^1{}_{23}=\frac{i}{\sqrt{2}}\sps  &
    \eta^8{}_{14}=-\frac{1}{\sqrt{2}}\sps &
    \eta^8{}_{23}=-\frac{1}{\sqrt{2}}\sps \\
    \eta^1{}_{67}=-\frac{i}{\sqrt{2}}\sps &
    \eta^1{}_{58}=\frac{i}{\sqrt{2}}\sps  &
    \eta^8{}_{67}=-\frac{1}{\sqrt{2}}\sps &
    \eta^8{}_{58}=-\frac{1}{\sqrt{2}}\sps\\
    \eta^7{}_{13}=-\frac{i}{\sqrt{2}}\sps &
    \eta^7{}_{24}=-\frac{i}{\sqrt{2}}\sps &
    \eta^6{}_{13}=\frac{1}{\sqrt{2}}\sps  &
    \eta^6{}_{24}=-\frac{1}{\sqrt{2}}\sps \\
    \eta^7{}_{68}=\frac{i}{\sqrt{2}}\sps  &
    \eta^7{}_{57}=\frac{i}{\sqrt{2}}\sps  &
    \eta^6{}_{68}=-\frac{1}{\sqrt{2}}\sps &
    \eta^6{}_{57}= \frac{1}{\sqrt{2}}\sps \\
    \eta^4{}_{15}=\frac{1}{\sqrt{2}}\sps  &
    \eta^4{}_{51}=\frac{1}{\sqrt{2}}\sps  &
    \eta^3{}_{15}=-\frac{i}{\sqrt{2}}\sps &
    \eta^3{}_{51}=\frac{i}{\sqrt{2}}\sps  \\
    \eta^4{}_{26}=\frac{1}{\sqrt{2}}\sps  &
    \eta^4{}_{62}=\frac{1}{\sqrt{2}}\sps  &
    \eta^3{}_{26}=-\frac{i}{\sqrt{2}}\sps &
    \eta^3{}_{62}=\frac{i}{\sqrt{2}}\sps  \\
    \eta^4{}_{37}=\frac{1}{\sqrt{2}}\sps  &
    \eta^4{}_{73}=\frac{1}{\sqrt{2}}\sps  &
    \eta^3{}_{37}=-\frac{i}{\sqrt{2}}\sps &
    \eta^3{}_{73}=\frac{i}{\sqrt{2}}\sps  \\
    \eta^4{}_{48}=\frac{1}{\sqrt{2}}\sps  &
    \eta^4{}_{84}=\frac{1}{\sqrt{2}}\sps  &
    \eta^3{}_{48}=-\frac{i}{\sqrt{2}}\sps &
    \eta^3{}_{84}=\frac{i}{\sqrt{2}}
    \end{array}
\end{equation}
    так, что в сокращенном виде можно записать
\begin{equation}
\label{re14o}
    \eta_A{}^{MN}=\left(
    \begin{array}{cc}
    \eta_\alpha{}^{ab} & \lambda\delta^a{}_d \\
    \mu\delta_c{}^b & \eta_{\alpha cd}
    \end{array}
    \right)\sps
    \eta_{\alpha cd}=\frac{1}{2}\varepsilon_{abcd}\eta_\alpha{}^{ab}\sps
\end{equation}
    где $\varepsilon_{abcd}$ - квадривектор из  (\ref{re2}).
    По сути, здесь использован тот же базис, что и в формуле  (\ref{re194}).\\
    \indent Рассмотрим далее бивектор $\hat R^{AB}$  вида (\ref{re190})
    такой, что его составляющие вектора $X_1{}^A,X_2{}^A$,
    определенные формулой (\ref{re175}),
    будут удовлетворять системе (\ref{re182}). Этим
    определится цепочка тождеств
\begin{equation}
\label{re15o}
    \begin{array}{c}
    ir^{ab}\dot Y_b=\dot X^a\sps
    ir_{ab}\dot Z^b=pf(r)\dot T_a\sps\\[2ex]
    \frac{1}{2}ir_{cd}\varepsilon^{abcd}\dot Y_b\varepsilon_{aklm}=
    \dot X^a\varepsilon_{aklm}\sps\\[2ex]
    3ir_{\left[\right. kl}\dot Y_{m\left.\right]}=
    \dot X^a\varepsilon_{aklm}\sps\\[2ex]
    ir_{kl}\dot Y_m+2ir_{m\left[\right. k}\dot Y_{l\left.\right]}=
    \dot X^a\varepsilon_{aklm}\spsd
    \end{array}
\end{equation}
    Свернем последнее тождество с $\dot Z^m$ и будем иметь
\begin{equation}
\label{re16o}
    ir_{kl}\dot Y_m\dot Z^m=\dot X^a\dot Z^a\varepsilon_{klmn}+
    pf(r)\cdot 2\dot T_{\left[\right. k}\dot Y_{l\left.\right]}\spsd
\end{equation}
    Таким образом, бивектор  $\hat R^{AB}$ определит прямолинейную
    образующую $\mathbb C\mathbb P_1(\mathbb CQ_6)$, принадлежащую некоторой плоской
    образующей $\mathbb C\mathbb P_3(\mathbb CQ_6)$, которая определит некоторую точку
    с координатами  $r^{ab}$, принадлежащую квадрике $\mathbb C\tilde Q_6$. Мы можем
    определить некоторый тензор $\hat{\hat R}^{AB}$
\begin{equation}
\label{re17o}
    \hat{\hat R}^{AB}:=\hat R^{AK} \hat P_K{}^B\sps
    \hat P_K{}^B:=\left(
    \begin{array}{cc}
    2\delta_m{}^k & 0\\
    0 & -2pf(r)i \delta^n{}_r
    \end{array}
    \right)\spsd
\end{equation}
    Тензор  $\hat{\hat R}^{AB}$ по прежнему будет представлять
    прямолинейную образующую $\mathbb CP_1(\mathbb CQ_6)$, принадлежащую
    некоторой плоской образующей $\mathbb CP_3(\mathbb CQ_6)$. Применим операторы
    $\eta^A{}_{KL}$ к тензору $\hat{\hat R}^{AB}$ и получим
    тождество
\begin{equation}
\label{re18o}
    \eta^A{}_{KL}\hat{\hat R}^{KL}=
    \eta^A{}_{KL}\dot Y_m\dot Z^m
    \left(
    \begin{array}{cc}
    r^{an} & -\frac{1}{2}i\delta^a{}_kr^\gamma r_\gamma \\
    -i\delta_c{}^n & r_{ck}
    \end{array}
    \right)=
    \eta^A{}_{KL}R^{KL}\sps
\end{equation}
\newpage
    \noindent где $\dot Y_m\dot Z^m$ нормировано так, что выполнено
    (\ref{re194}), а тензор $R^{KL}$ удовлетворяет выражению  (\ref{re191}).
    Таким образом, операторы $\eta^A{}_{KL}$ осуществляют факторизацию
    прямолинейных образующих $\mathbb C\mathbb P_1(\mathbb CQ_6)$ по принадлежности к
    одной плоской образующей $\mathbb C\mathbb P_3(\mathbb CQ_6)$, и это определит
    точку квадрики $\mathbb C\tilde Q_6$. В однородных же координатах
    тензор $R^{KL}$ определит координаты точки R пространства
    $\mathbb C\mathbb R^8$.\\

    \indent Свернем с $\delta_L{}^B$  тождество  (\ref{re199})
\begin{equation}
\label{re19o}
    G_{AK}=S_A{}^M\varepsilon_{KM}\sps
    S_A{}^M:=\eta_A{}^{MR}\eta_L{}^L{}_R+\eta_A{}^L{}_R\eta_L{}^{MR}\spsd
\end{equation}
    Поскольку $G_{AB},\varepsilon_{KL}$ имеют вид (\ref{re12o}), то
    $\parallel S_A{}^M\parallel$ будет такой
\begin{equation}
\label{re20o}
    \parallel S_A{}^M \parallel=\left(
    \begin{array}{cccccccc}
    0 & 0 & 0 & 0 & 1 & 0 & 0 & 0\\
    0 & 0 & 0 & 0 & 0 & 1 & 0 & 0\\
    0 & 0 & 0 & 0 & 0 & 0 & 1 & 0\\
    0 & 0 & 0 & 0 & 0 & 0 & 0 & 1\\
    1 & 0 & 0 & 0 & 0 & 0 & 0 & 0\\
    0 & 1 & 0 & 0 & 0 & 0 & 0 & 0\\
    0 & 0 & 1 & 0 & 0 & 0 & 0 & 0\\
    0 & 0 & 0 & 1 & 0 & 0 & 0 & 0\\
    \end{array}
    \right)\spsd
\end{equation}
    Поэтому $S_A{}^M$ будет тензором инволюции, а квадрики -
    B-цилиндрами.\\

    \indent Выясним, какому семейству принадлежат рассматриваемые
    выше образующие $\mathbb C\mathbb P_3(\mathbb CQ_6)$. Для этого рассмотрим условия
\begin{equation}
\label{re21o}
    \begin{array}{c}
    \varepsilon_{AB}X_i{}^AX_j{}^B=0\sps\\[2ex]
    X^{ABCD}:=\varepsilon^{ijkl}X_i{}^AX_j{}^BX_k{}^CX_l{}^D\sps
    \end{array}
\end{equation}
    где $\varepsilon^{ijkl}$ - квадривектор, кососимметричный по всем
    индексам. Кроме того, рассмотрим 8-вектор $e_{ABCDKLMN}$, тоже
    кососимметричный по всем индексам. Тогда, если в условии
\begin{equation}
\label{re22o}
    \frac{1}{24}e_{ABCDKLMN}X^{ABCD}=\rho\varepsilon_{KR}\varepsilon_{LT}\varepsilon_{MU}
    \varepsilon_{NV}X^{RTUV}\sps \rho^2=1
\end{equation}
    $\rho=1$, то будем говорить, что плоские образующие $\mathbb C\mathbb P_3(\mathbb CQ_6)$
    принадлежат I семейству, а если $\rho=-1$, то - II семейству.
    В нашем случае
\begin{equation}
\label{re23o}
    X_i{}^A=(\dot X_i{}^a,\dot Y_{ib})\sps
\end{equation}
    и (\ref{re22o}) даст такое, например, выражение
\begin{equation}
\label{re24o}
    \varepsilon^{ijkl}X_i{}^1X_j{}^2X_k{}^3X_l{}^4=\rho
    \varepsilon^{ijkl}X_i{}^1X_j{}^2X_k{}^3X_l{}^4\spsd
\end{equation}
    При этом тензор $\varepsilon_{KL}$  имеет вид (\ref{re12o}).
    Откуда $\rho=1$. Это означает, что наши образующие
    необходимо принадлежат  I семейству.\\
    \indent Кроме того, существует тензор $\tilde S_K{}^L$
\begin{equation}
\label{re25o}
    \parallel \tilde S_K{}^L \parallel=\frac{1}{\sqrt{2}}\left(
    \begin{array}{cccccccc}
    i & 0 & 0 & 0 &-i & 0 & 0 & 0\\
    0 & i & 0 & 0 & 0 &-i & 0 & 0\\
    0 & 0 & i & 0 & 0 & 0 &-i & 0\\
    0 & 0 & 0 & i & 0 & 0 & 0 &-i\\
    1 & 0 & 0 & 0 & 1 & 0 & 0 & 0\\
    0 & 1 & 0 & 0 & 0 & 1 & 0 & 0\\
    0 & 0 & 1 & 0 & 0 & 0 & 1 & 0\\
    0 & 0 & 0 & 1 & 0 & 0 & 0 & 1\\
    \end{array}
    \right)\sps
     det \parallel\tilde S_K{}^L\parallel=1\spsd
\end{equation}
    такой, что выполнено
\begin{equation}
\label{re26o}
    \varepsilon_{AB}\tilde S_K{}^A\tilde S_L{}^B=G_{KL}\spsd
\end{equation}
    Поэтому можно положить
\begin{equation}
\label{re27o}
    X_i{}^A=\tilde S_K{}^AR_i{}^K\sps
\end{equation}
    и (\ref{re21o}) перепишется так
\begin{equation}
\label{re28o}
    \begin{array}{c}
    G_{AB}R_i{}^AR_i{}^K=0\sps\\[2ex]
    X^{ABCD}=\tilde S_K{}^A\tilde S_L{}^B\tilde S_N{}^C\tilde S_M{}^DR^{KLNM}\spsd
    \end{array}
\end{equation}
    Поскольку $\rho=1$, то будем иметь цепочку тождеств
\begin{equation}
\label{re29o}
    \begin{array}{c}
    \frac{1}{24}e_{ABCDKLMN}X^{ABCD}=\varepsilon_{KR}\varepsilon_{LT}\varepsilon_{MU}
    \varepsilon_{NV}X^{RTUV}\sps\\[2ex]
    \frac{1}{24}e_{ABCDKLMN}\tilde S_U{}^A \tilde S_S{}^B\tilde S_T{}^C\tilde S_V{}^D
    R^{USTV}=\\
    =\varepsilon_{KP}\varepsilon_{LQ}\varepsilon_{MX}
    \varepsilon_{NY}\tilde S_G{}^P \tilde S_H{}^Q\tilde S_Z{}^X\tilde S_W{}^Y
    R^{GHZW}\\[2ex]
    \frac{1}{24}e_{ABCDKLMN}\tilde S_U{}^A \tilde S_S{}^B\tilde S_T{}^C\tilde S_V{}^D
    \tilde S_P{}^K \tilde S_Q{}^L\tilde S_X{}^M\tilde S_Y{}^NR^{USTV}=\\
    =G_{PG}G_{QM}G_{XZ}G_{YW}R^{GHZW}\sps\\[2ex]
    \frac{1}{24}det \parallel S_P{}^A \parallel e_{PQXYUSTV}R^{USTV}=G_{PG}G_{QM}G_{XZ}G_{YW}R^{GHZW}\sps\\[2ex]
    \frac{1}{24}e_{PQXYUSTV}R^{USTV}=G_{PG}G_{QM}G_{XZ}G_{YW}R^{GHZW}\spsd
    \end{array}
\end{equation}
    Отсюда видно, что и образующие $\mathbb C\mathbb P_3(\mathbb C\tilde Q_6)$ необходимо
    принадлежат I семейству. Для того, чтобы получить такой же
    результат для образующих II семейства, в качестве метрического
    тензора, с помощью которого поднимаются и опускаются одиночные
    индексы, следовало бы выбрать тензор
\begin{equation}
\label{re30o}
    \tilde\varepsilon_{KL}:=\sqrt{i}\varepsilon_{KL}\spsd
\end{equation}

\subsection{\texorpdfstring{Соответствие $\mathbb C\tilde Q_6\longmapsto \mathbb CQ_6$}{Соответствие C'Q в CQ}}$ $
     \indent Применяя операторы $\eta_A{}^{KL}$ к  (\ref{re197}) получим
\begin{equation}
\label{re199a}
      \begin{array}{c}
      R_i^{AB}R_{j\ AB}=0\ \ \Leftrightarrow\ \
      ((R_i){}^{AB}-(R_j){}^{AB})((R_k){}_{AB}-(R_l){}_{AB})=0
      \Leftrightarrow\\[2ex]
      ((r_i){}^{ab}-(r_j){}^{ab})((r_k){}_{ab}-(r_l){}_{ab})=0\spsd
      \end{array}
\end{equation}
      Здесь i,j, как обычно, номера базисных точек.
      Этим определится система
\begin{equation}
\label{re200}
    \left\{
    \begin{array}{lcl}
    i(r_1){}^{ab}\dot Y_b & = & \dot X^a\sps\\
    i(r_2){}^{ab}\dot Y_b & = & \dot X^a\sps\\
    i(r_3){}^{ab}\dot Y_b & = & \dot X^a\sps\\
    i(r_4){}^{ab}\dot Y_b & = & \dot X^a\sps
    \end{array}
    \right.
    \psps\Leftrightarrow\psps
    \left\{
    \begin{array}{lcl}
    i((r_1){}^{ab}-(r_2){}^{ab})\dot Y_b & = & 0\sps\\
    i((r_1){}^{ab}-(r_3){}^{ab})\dot Y_b & = & 0\sps\\
    i((r_3){}^{ab}-(r_4){}^{ab})\dot Y_b & = & 0\sps\\
    i(r_1){}^{ab}\dot Y_b & = & \dot X^a\spsd
    \end{array}
    \right.
\end{equation}
     Далее мы будем рассматривать только правую систему.
     Она строится следующим образом. Всегда существует
     такой ковектор $Y_a$, который обнуляет 3 различных
     простых бивектора. Это утверждение сводится к
     существованию ковектора, ортогонального данным
     трем векторам, поскольку каждый из простых бивекторов
     раскладывается по формуле $r^{ab}=2P^{\left[\right.a}
     Q^{b\left.\right]}$. По четвертому же уравнению
     определится некоторый вектор $X^a$. Поэтому такая
     система всегда определена.
     С другой стороны, на основании того, что все $r_i{}^{ab}$ имеют вид
     (\ref{re169}) (при фиксированных $\lambda_1,\lambda_2,
     \lambda_3$), верно равенство
\begin{equation}
\label{re201}
     ((r_i){}^{ab}-(r_j){}^{ab})((r_k){}_{ab}-(r_l){}_{ab})=0\spsd
\end{equation}
     \indent Итак, пусть нам известно последнее уравнение системы
     (\ref{re200})
\begin{equation}
\label{re202}
     i{r_1}^{ab}\dot Y_b  =  \dot X^a\sps
\end{equation}
     тогда мы имеем 4 уравнения, которые все будут существенными.
     Поскольку у нас 8 неизвестных $(\dot X^a,\dot Y_b)$ при
     фиксированных $(r_1){}^{ab}$, то точка квадрики
     $\mathbb C\tilde Q_6$ определит 3-мерную плоскую образующую
     $\mathbb C\mathbb P_3(\mathbb CQ_6)$.\\

     \indent Если нам известны все уравнения системы
     (\ref{re200}) с условиями
\begin{equation}
\label{re204}
     \begin{array}{c}
     ((r_1){}^{ab}-(r_2){}^{ab})((r_1){}_{ab}-(r_2){}_{ab})=0\sps
     ((r_1){}^{ab}-(r_3){}^{ab})((r_1){}_{ab}-(r_3){}_{ab})=0\sps \\
     ((r_3){}^{ab}-(r_4){}^{ab})((r_3){}_{ab}-(r_4){}_{ab})=0\sps \\
      (r_1){}^{ab}(r_2){}_{ab}=0\sps
      (r_1){}^{ab}(r_3){}_{ab}=0\sps
      (r_1){}^{ab}(r_4){}_{ab}=0\sps \\
      (r_2){}^{ab}(r_3){}_{ab}=0\sps
      (r_2){}^{ab}(r_4){}_{ab}=0\sps
      (r_3){}^{ab}(r_4){}_{ab}=0\sps
     \end{array}
\end{equation}
     то из 16 уравнений существенными будут 7 (8 неизвестных и
     9 условий связи).
     Поэтому образующей $\mathbb C\mathbb P_3(\mathbb C\tilde Q_6)$ будет
     соответствовать точка квадрики $\mathbb CQ_6$.\\

     \indent Если нам известны 3 уравнения системы
     (\ref{re200})
\begin{equation}
\label{re205}
    \left\{
    \begin{array}{lcl}
    i((r_1){}^{ab}-(r_2){}^{ab})\dot Y_b & = & 0\sps\\
    i((r_1){}^{ab}-(r_3){}^{ab})\dot Y_b & = & 0\sps\\
    i(r_1){}^{ab}\dot Y_b & = & \dot X^a
    \end{array}
    \right.
\end{equation}
      с условиями
\begin{equation}
\label{re206}
     \begin{array}{c}
     ((r_1){}^{ab}-(r_2){}^{ab})((r_1){}_{ab}-(r_2){}_{ab})=0\sps
     ((r_1){}^{ab}-(r_3){}^{ab})((r_1){}_{ab}-(r_3){}_{ab})=0\sps\\
      (r_1){}^{ab}(r_2){}_{ab}=0\sps
      (r_1){}^{ab}(r_3){}_{ab}=0\sps
      (r_2){}^{ab}(r_3){}_{ab}=0\sps
     \end{array}
\end{equation}
     то из 12 уравнений существенными будут тоже 7 уравнений
     (8 неизвестных и 5 условий связи).
     Это означает, что образующей $\mathbb C\mathbb P_2(\mathbb C\tilde Q_6)$ будет
     соответствовать точка квадрики $\mathbb CQ_6$. При этом
     многообразие образующих $\mathbb C\mathbb P_2(\mathbb C\tilde Q_6)$, принадлежащих
     одной образующей $\mathbb C\mathbb P_3(\mathbb C\tilde Q_6)$, определит единственную
     точку квадрики $\mathbb CQ_6$.\\

     \indent Если нам известны только 2 уравнения системы
     (\ref{re200})
\begin{equation}
\label{re207}
    \left\{
    \begin{array}{lcl}
    i((r_1){}^{ab}-(r_2){}^{ab})\dot Y_b & = & 0\sps\\
    i(r_1){}^{ab}\dot Y_b & = & \dot X^a\\
    \end{array}
    \right.
\end{equation}
      с условиями
\begin{equation}
\label{re208}
     ((r_1){}^{ab}-(r_2^{ab})((r_1){}_{ab}-(r_2){}_{ab})=0\sps
      (r_1){}^{ab}(r_2){}_{ab}=0\sps
\end{equation}
\newpage
     \noindent то из 8 уравнений существенными будет только 6 уравнений
     (8 неизвестных и 2 условия связи).
     Поэтому образующей $\mathbb C\mathbb P_1(\mathbb C\tilde Q_6)$ будет
     соответствовать прямолинейная образующая  $\mathbb C\mathbb P_1(\mathbb CQ_6)$. При этом
     многообразие образующих $\mathbb C\mathbb P_1(\mathbb C\tilde Q_6)$, принадлежащих
     одной образующей $\mathbb C\mathbb P_3(\mathbb C\tilde Q_6)$, определит пучок прямых
     $\mathbb C\mathbb P_1(\mathbb CQ_6)$, принадлежащий квадрике $\mathbb CQ_6$.
     Центр пучка определится системой (\ref{re200}).

\subsection{Теорема о двух квадриках}$ $

     \indent Таким образом доказана теорема:
\begin{theoremr}
\label{theorem9r}(Принцип тройственности для двух B- цилиндров).\\
     В проективном пространстве $\mathbb C\mathbb P_7$ существуют
     две квадрики (два B - цилиндра), обладающие
     следующими общими свойствами:
     \begin{enumerate}
     \item Плоская образующая $\mathbb C\mathbb P_3$ одной квадрики
           взаимооднозначно определит точку R другой.
     \item Плоская образующая $\mathbb C\mathbb P_2$ одной квадрики
           однозначно определит точку R другой. Но точке
           R можно сопоставить многообразие плоских образующих
           $\mathbb C\mathbb P_2$, принадлежащих одной плоской образующей $\mathbb C\mathbb P_3$
           второй квадрики.
     \item Прямолинейная образующая $\mathbb C\mathbb P_1$ одной квадрики
           взаимооднозначно определит прямолинейную образующую
           $\mathbb C\mathbb P_1$ из другой. Причем все прямолинейные образующие,
           принадлежащие одной плоской образующей $\mathbb C\mathbb P_3$ первой квадрики,
           определят пучок с центром в точке R, принадлежащий
           второй квадрике.
     \end{enumerate}
\end{theoremr}
     Эта теорема на самом деле является обобщением соответствия
     Кляйна. Докажем это.
\vspace{1cm}
\begin{proof}
     Рассмотрим на квадрике $\mathbb CQ_6$ только те образующие,
     которые имеют вид
\begin{equation}
\label{re209}
     X^A=(0,Y_b)\spsd
\end{equation}
     Многообразие таких образующих диффеоморфно $\mathbb C\mathbb P_3$. При
     этом каждой такой образующей можно поставить в соответствие
     точку квадрики $\mathbb CQ_4\subset \mathbb C\tilde Q_6$. Согласно системе
     (\ref{re182}) первое ее уравнение будет иметь вид
\begin{equation}
\label{re210}
      r^{ab}\dot Y_b=0\spsd
\end{equation}
     До конца доказательства положим $\bf{A,B,A',B',...}=\overline{1,2}$.
     Кроме того, рассмотрим спинорное представление твисторов
     согласно \cite[т.2, c.63, (6.1.24) и c. 83, (6.2.18)]{Penrose1r}
\begin{equation}
\label{re211}
     \begin{array}{c}
     \dot Y_b=(\dot \pi_{\mbox{\scriptsize\bf B}},
     \dot{\bar\omega}^{{\mbox{\scriptsize\bf B}}'})\sps\\ \\
     r^{ab}=const\left(
     \begin{array}{cc}
     -\frac{1}{2}\varepsilon^{\mbox{\scriptsize\bf AB}}r_cr^c &
     ir^{\mbox{\scriptsize\bf A}}{}_{\mbox{\scriptsize\bf B}'} \\[1.2ex]
     -i\bar r_{{\mbox{\scriptsize\bf A}}'}{}^{\mbox{\scriptsize\bf B}} &
     \bar\varepsilon_{{\mbox{\scriptsize\bf A}}'{\mbox{\scriptsize\bf B}}'}
     \end{array}
     \right)\spsd
     \end{array}
\end{equation}
     Поэтому первое уравнение перепишется в виде
\begin{equation}
\label{re212}
     \left\{
     \begin{array}{ccc}
     -\frac{1}{2}\varepsilon^{\mbox{\scriptsize\bf AB}}r_cr^c
     \dot \pi_{\mbox{\scriptsize\bf B}}+
     ir^{\mbox{\scriptsize\bf A}}{}_{{\mbox{\scriptsize\bf B}}'}
     \dot{\bar\omega}^{{\mbox{\scriptsize\bf B}}'} & = & 0\sps \\[1.2ex]
     -i\bar r_{{\mbox{\scriptsize\bf A}}'}{}^{\mbox{\scriptsize\bf B}}
     \dot \pi_{\mbox{\scriptsize\bf B}}+
     \bar\varepsilon_{{\mbox{\scriptsize\bf A}}'{\mbox{\scriptsize\bf B}}'}
     \dot{\bar\omega}^{{\mbox{\scriptsize\bf B}}'} & = & 0\sps
     \end{array}
     \right.
\end{equation}
     из которых будет существенным только одно
\begin{equation}
\label{re213}
     ir^{{\mbox{\scriptsize\bf AA}}'}
     \dot{\bar\pi}_{{\mbox{\scriptsize\bf A}}'}=
     \dot \omega^{\mbox{\scriptsize\bf A}}\spsd
\end{equation}
     (Здесь мы воспользовались операцией сопряжения и метрическими
      спинорами $\varepsilon_{\mbox{\scriptsize\bf A'B'}},
      \varepsilon_{\mbox{\scriptsize\bf AB}}$, с помощью
      которых подымаются и опускаются спинорные индексы и
      которые при сопряжении переходят друг в друга.)
      Поэтому система (\ref{re182}) определит систему
\begin{equation}
\label{re215}
     \begin{array}{c}
     \left\{
     \begin{array}{ccc}
     ir^{\mbox{\scriptsize\bf AA}'}
     \dot{\bar\pi}_{\mbox{\scriptsize\bf A}'} & = &
     \dot \omega^{\mbox{\scriptsize\bf A}}\sps\\[1.2ex]
     ir^{\mbox{\scriptsize\bf AA}'}
     \dot{\bar\eta}_{\mbox{\scriptsize\bf A}'} & = &
     \dot \xi^{\mbox{\scriptsize\bf A}}\sps
     \end{array}
     \right.
     \dot Y_b=(\dot \pi_{\mbox{\scriptsize\bf B}},
     \dot{\bar\omega}^{\mbox{\scriptsize\bf B}'})\sps
     \dot T_b=(\dot \eta_{\mbox{\scriptsize\bf B}},
     \dot{\bar\xi}^{\mbox{\scriptsize\bf B}'})\spsd
     \end{array}
\end{equation}
     Эта система совпадает с системой (6.2.14), которая, в
     свою очередь, приводит к соответствию Кляйна согласно
     \cite{Penrose1r}.
\end{proof}

     Следует в заключение отметить, что из этой
     теоремы следует принцип тройственности Картана:
     существует 3 диффеоморфных многообразия -
     многообразие точек квадрики и 2 многообразия
     плоских образующих  I  и II семейств.
     Это действительно так, поскольку две построенные
     квадрики можно отождествить, например, с помощью тензора
     $\tilde S_K{}^L$. При этом многообразие точек квадрики
     будет диффеоморфно многообразию плоских образующих
     I семейства.\\
     \indent Кроме того, поскольку принцип тройственности Картана
     выполнен, то операторы $\eta^A{}_{KL}$ определят
     алгебру октав, поскольку они удовлетворяют
     уравнению Клиффорда. Это утверждение основано на результатах,
     приведенных в монографии \cite[т.2, c.543-544]{Penrose1r}, где
     рассматриваются структурные константы этой алгебры.

\subsection{Заключение}
     На защиту выносятся следующие основные положения:
\begin{enumerate}
     \item  Вложение $\mathbb R^6_{(p,q)}\subset \mathbb C\mathbb R^6$ осуществляется
            с помощью операторов $H_i{}^\alpha$ таким образом,
            что спинорное представление инволюции имеет вид
            $$
            S_\alpha{}^{\beta '}=\frac{1}{4}\eta_\alpha{}^{ab}
            \bar \eta^{\beta '}{}_{c'd'}\cdot 2s_a{}^{c'}s_b{}^{d '}\sps
            s_a{}^{c '}\bar s_{c'}{}^d=\pm\delta_a{}^d
            $$
            при q - нечетном и
            $$
            S_\alpha{}^{\beta '}=\frac{1}{4}\eta_\alpha{}^{ab}
            \bar \eta^{\beta '}{}_{c'd'}\cdot s^{kc'}s^{nd'}\varepsilon_{knab}\sps
            s_{ac'}=\pm\bar s_{c'a}
            $$
            при q - четном.
     \item  Найдены в явном виде операторы $A_{\alpha\beta a}{}^b$,
            с помощью которых определяется соответствие между
            бивекторами пространства $\mathbb C\mathbb R^6$ и бесследовыми
            операторами пространства $\mathbb C^4$. Это позволяет
            изучать алгебраическую структуру тензора кривизны
            пространства $\mathbb C\mathbb R^6$
            $R_{\alpha\beta\gamma\delta}$ по его спинорному
            образу - тензору $R_a{}^b{}_c{}^d$.
     \item  Доказано, что простой изотропный бивектор пространства
            $\Lambda^2\mathbb C^6$ определит вырожденную нуль-пару Розенфельда -
            вектор и ковектор пространства $\mathbb C^4$, свертка которых есть нуль -
            c точностью до комплексного множителя.
     \item  Утверждается, что бивектор пространства $\Lambda^2\mathbb R^6_{(p,q)}$
            при четном q может быть приведен в некотором базисе к
            каноническому виду.
     \item  Доказан обобщенный принцип тройственности для
            пары B-цилиндров.
     \item  Определены операторы $\eta_A{}^{KL}$, удовлетворяющие
            уравнению Клиффорда и отвечающие за соответствие
            между прямолинейными образующими указанных
            B-цилиндров. Кроме того, эти операторы определяют
            структурные константы алгебры октав.
\end{enumerate}

\newpage
\section{ПРИЛОЖЕНИЕ}
\subsection{Доказательство формул второй главы}
\subsubsection{\texorpdfstring{Доказательство формул, содержащих
               оператор $A_{\alpha\beta a}{}^b$}{Доказательство формул, содержащих
               оператор A}}$ $

       \indent Определим $(\alpha,\beta,...=1,2,3,4,5,6;
       a_1,b_1,a,b,e,f,k,l,m,n,...=1,2,3,4)$.
       Поскольку
\begin{equation}
\label{re1p}
       A_{\alpha\beta d}{}^c=
       \eta_{\left[\right.\alpha}{}^{ca}
       \eta_{\beta\left.\right]}{}_{da}\sps
\end{equation}
       то
\begin{equation}
\label{re2p}
\begin{array}{c}
       A_{\alpha\beta d}{}^cA_\gamma{}^\beta{}_r{}^s=
       \frac{1}{4}
       (\eta_\alpha{}^{ca}\eta_{\beta ad}-
       \eta_{\alpha ad}\eta_\beta{}^{ca})
       (\eta_\gamma{}^{sa_1}\eta^\beta{}_{a_1r}-
       \eta_{\gamma a_1r}\eta^{\beta sa_1})=\\[2ex]
       =\frac{1}{4}
       (\eta_\alpha{}^{ca}\eta_\gamma{}^{sa_1}\varepsilon_{a_1rad}-
       \eta_{\alpha ad}\eta_\gamma{}^{sa_1}\varepsilon_{a_1r}{}^{ca}-
       \eta_\alpha{}^{ca}\eta_{\gamma a_1r}\varepsilon_{ad}{}^{sa_1}+
       \eta_{\alpha ad}\eta_{\gamma a_1r}\varepsilon^{casa_1})=\\[2ex]
       =\frac{1}{4}
       (\frac{1}{2}\eta_\alpha{}^{ca}\eta_{\gamma kl}
       \varepsilon^{klsa_1}\varepsilon_{a_1rad}-
       \eta_{\alpha ad}\eta_\gamma{}^{sa_1}
       (\delta_{a_1}{}^c\delta_r{}^a-\delta_{a_1}{}^a\delta_r{}^c)-\\[2ex]
       -\eta_\alpha{}^{ca}\eta_{\gamma a_1r}
       (\delta_a{}^s\delta_d{}^{a_1}-\delta_a{}^{a_1}\delta_d{}^s)+
       \frac{1}{2}\eta_\alpha{}^{kl}\eta_{\gamma a_1r}
       \varepsilon_{klad}\varepsilon^{casa_1})=\\[2ex]
       =\frac{1}{4}
       (-3\eta_\alpha{}^{ca}\eta_{\gamma \left[\right.ra}\delta_{d\left.\right]}{}^s-
       \eta_{\alpha rd}\eta_\gamma{}^{sc}+
       \eta_{\alpha ad}\eta_\gamma{}^{sa}\delta_r{}^c-\\[2ex]
       -\eta_\alpha{}^{cs}\eta_{\gamma dr}+
       \eta_\alpha{}^{ca}\eta_{\gamma ar}\delta_d{}^s-
       3\eta_\alpha{}^{\left[\right.cs}\delta_d{}^{a_1\left.\right]}\eta_{\gamma a_1r}=\\[2ex]
       =\frac{1}{4}
       (-\eta_\alpha{}^{ck}\eta_{\gamma rk}\delta_d{}^s+
       \eta_\alpha{}^{cs}\eta_{\gamma rd}-
       \eta_\alpha{}^{ck}\eta_{\gamma kd}\delta_r{}^s-
       \eta_{\alpha rd}\eta_\gamma{}^{sc}+
       \eta_{\alpha ad}\eta_\gamma{}^{sa}\delta_r{}^c-\\[2ex]
       -\eta_\alpha{}^{cs}\eta_{\gamma dr}+
       \eta_\alpha{}^{ca}\eta_{\gamma ar}\delta_d{}^s-
       \eta_\alpha{}^{cs}\eta_{\gamma dr}+
       \eta_\alpha{}^{ck}\eta_{\gamma kr}\delta_d{}^s-
       \eta_\alpha{}^{sk}\eta_{\gamma kr}\delta_d{}^c)=\\
\end{array}
\end{equation}
       $$
       \left|
       \begin{array}{c}
       \eta_{\alpha dr}\eta_\gamma{}^{sc}+
       \eta_{\alpha kd}\eta_\gamma{}^{sk}\delta_r{}^c=\\[2ex]
       =\frac{1}{4}
       (\eta_\alpha{}^{mm_1}\eta_{\gamma ll_1}
       \varepsilon_{mm_1dr}\varepsilon^{ll_1sc}+
       \eta_\alpha{}^{mm_1}\eta_{\gamma ll_1}
       \varepsilon_{mm_1kd}\varepsilon^{ll_1sk}\delta_r{}^c)=\\[2ex]
       =6\eta_\alpha{}^{mm_1}\eta_{\gamma ll_1}
       \delta_{\left[\right. m}{}^l\delta_{m_1}{}^{l_1}
       \delta_d{}^s\delta_{r\left.\right]}{}^c-
       \frac{3}{2}\eta_\alpha{}^{mm_1}\eta_{\gamma ll_1}
       \delta_{\left[\right. m}{}^l\delta_{m_1}{}^{l_1}
       \delta_{d\left.\right]}{}^s\delta_r{}^c=\\[2ex]
       =-g_{\alpha\gamma}\delta_r{}^s\delta_d{}^c+
       \eta_\alpha{}^{kc}\eta_{\gamma kd}\delta_r{}^s+
       \eta_\alpha{}^{ck}\eta_{\gamma kr}\delta_d{}^s+
       \eta_\alpha{}^{ks}\eta_{\gamma kr}\delta_d{}^c+
       \eta_\alpha{}^{sc}\eta_{\gamma dr}\spsd\\[2ex]
       \end{array}
       \right|
       $$
       $$
       \begin{array}{c}
       =\eta_\alpha{}^{cs}\eta_{\gamma rd}+
       \eta_\alpha{}^{ck}\eta_{\gamma kr}\delta_d{}^s+
       \frac{1}{2}(
       \eta_\alpha{}^{ck}\eta_{\gamma dk}\delta_r{}^s+
       \eta_\alpha{}^{sk}\eta_{\gamma rk}\delta_d{}^c)-
       \frac{1}{4}g_{\alpha\gamma}\delta_r{}^s\delta_d{}^c\spsd
       \end{array}
       $$
       Свернем ее с $g^{\alpha\gamma}$
\begin{equation}
\label{re3p}
       \begin{array}{c}
       A_{\alpha\beta d}{}^cA^{\alpha\beta}{}_r{}^s=\\[3ex]
       =\varepsilon_{rd}{}^{cs}+\varepsilon_{kr}{}^{ck}\delta_d{}^s+
       \frac{1}{2}(\varepsilon^{ck}{}_{dk}\delta_r{}^s+
       \varepsilon^{sk}{}_{rk}\delta_d{}^c)-
       \frac{3}{2}\delta_r{}^s\delta_d{}^c=\\[3ex]
       =\delta_r{}^c\delta_d{}^s-\delta_r{}^s\delta_d{}^c-
       3\delta_r{}^c\delta_d{}^s+\frac{3}{2}\delta_d{}^c\delta_r{}^s+
       \frac{3}{2}\delta_d{}^c\delta_r{}^s-
       \frac{3}{2}\delta_r{}^s\delta_d{}^c=\\[3ex]
       =\frac{1}{2}\delta_r{}^s\delta_d{}^c-
       2\delta_r{}^c\delta_d{}^s\spsd
       \end{array}
\end{equation}
       Кроме того
\begin{equation}
\label{re4p}
       \begin{array}{c}
       A_{\alpha\beta d}{}^cA^{\lambda\mu}{}_c{}^d=\\[3ex]
       =\frac{1}{2}\eta_{\left[\right.\alpha}{}^{ca}
       \eta_{\beta\left.\right]}{}^{nn_1}
       \varepsilon_{adnn_1}\cdot
       \frac{1}{2}
       \eta^{\left[\right.\lambda}{}_{ll_1}
       \eta^{\mu\left.\right]}{}_{a_1c}\varepsilon^{da_1ll_1}=\\[3ex]
       =\frac{1}{4}\eta_{\left[\right.\alpha}{}^{ca}
       \eta_{\beta\left.\right]}{}^{nn_1}
       \eta^{\left[\right.\lambda}{}_{ll_1}
       \eta^{\mu\left.\right]}{}_{a_1c}\cdot
       (-6\delta_{\left[\right. a}{}^{a_1}\delta_n{}^l
       \delta_{n_1\left.\right]}{}^{l_1})=\\[3ex]
       =-\frac{1}{2}(\eta_{\left[\right.\alpha}{}^{ca}
       \eta_{\beta\left.\right]}{}^{nn_1}
       \eta^{\left[\right.\lambda}{}_{nn_1}
       \eta^{\mu\left.\right]}{}_{ac}-
       2\eta_{\left[\right.\alpha}{}^{ca}
       \eta_{\beta\left.\right]}{}^{a_1k_1}
       \eta^{\left[\right.\lambda}{}_{ak_1}
       \eta^{\mu\left.\right]}{}_{a_1c})=\\[3ex]
       =-2\delta_{\left[\right. \alpha}{}^{\lambda}
       \delta_{\beta\left.\right]}{}^\mu+
       \frac{1}{2}(\eta_\alpha{}^{ca}
       \eta_\beta{}^{a_1k_1}
       \eta^{\left[\right.\lambda}{}_{ak_1}
       \eta^{\mu\left.\right]}{}_{a_1c}-
       \eta_\beta{}^{ca}
       \eta_\alpha{}^{a_1k_1}
       \eta^{\left[\right.\lambda}{}_{ak_1}
       \eta^{\mu\left.\right]}{}_{a_1c})=\\[3ex]
       =2\delta_{\left[\right. \alpha}{}^\mu
       \delta_{\beta\left.\right]}{}^\lambda+
       \frac{1}{2}(\eta_\alpha{}^{ca}
       \eta_\beta{}^{a_1k_1}
       \eta^{\left[\right.\lambda}{}_{ak_1}
       \eta^{\mu\left.\right]}{}_{a_1c}-
       \eta_\beta{}^{a_1k_1}
       \eta_\alpha{}^{ca}
       \eta^{\left[\right.\lambda}{}_{k_1a}
       \eta^{\mu\left.\right]}{}_{ca_1})=\\[3ex]
       =2\delta_{\left[\right. \alpha}{}^\mu
       \delta_{\beta\left.\right]}{}^\lambda\spsd
       \end{array}
\end{equation}
       И, наконец
\begin{equation}
\label{re5p}
       \begin{array}{c}
       A^{\alpha\beta}{}_m{}^nT_{\alpha\beta}=
       A^{\alpha\beta}{}_m{}^nA_{\alpha\beta k}{}^lT_l{}^k=\\[3ex]
       =(\frac{1}{2}\delta_m{}^n\delta_k{}^l-2\delta_k{}^n\delta_m{}^l)
       T_l{}^k=-2T_m{}^n\sps T_{\alpha\beta}=-T_{\beta\alpha}\spsd
       \end{array}
\end{equation}

\subsubsection{\texorpdfstring{Доказательство формул о 4-векторе
               $e_{\alpha\beta\gamma\delta}$}{Доказательство формул о 4-векторе}}$ $

       \indent Пусть
\begin{equation}
\label{re56p}
       %
\end{equation}
\vspace{1cm}
\subsubsection{\texorpdfstring{Доказательство формул о 6-векторе
               $e_{\alpha\beta\gamma\delta\rho\sigma}$}{Доказательство формул о 6-векторе}}$ $

       \indent Пусть
\begin{equation}
\label{re59p}
       %
       $$
       Откуда
\begin{equation}
\label{re68p}
       \tilde e=\frac{1}{8}i\spsd
\end{equation}

\subsubsection{\texorpdfstring{Доказательство формул, содержащих оператор $N_a{}^f$}{Доказательство формул, содержащих оператор N}}$ $

        \indent Пусть теперь
\begin{equation}
\label{re69p}
       \begin{array}{c}
       N_{\left[\right.a}{}^bN_{c\left.\right]}{}^d=
       \delta_{\left[\right.a}{}^b\delta_{c\left.\right]}{}^d\sps
       \end{array}
\end{equation}
       тогда свертка этого тождества с $\delta_b{}^a$ даст
\begin{equation}
\label{re70p}
       \begin{array}{c}
       N_a{}^aN_c{}^d-N_a{}^dN_c{}^a=3\delta_c{}^d\sps\\[2ex]
       N_a{}^dN_c{}^a=N_a{}^aN_c{}^d-3\delta_c{}^d\spsd
       \end{array}
\end{equation}
       Домножим (\ref{re69p}) на $N_d{}^f$
\begin{equation}
\label{re71p}
       \begin{array}{c}
       N_{\left[\right.a}{}^bN_{c\left.\right]}{}^dN_d{}^f=
       \delta_{\left[\right.a}{}^bN_{c\left.\right]}{}^f\spsd
       \end{array}
\end{equation}
       Подставим (\ref{re70p})  в (\ref{re71p})
\begin{equation}
\label{re72p}
       \begin{array}{c}
       N_{\left[\right.a}{}^bN_{c\left.\right]}{}^fN_r{}^r-
       3N_{\left[\right.a}{}^b\delta_{c\left.\right]}{}^f=
       \delta_{\left[\right.a}{}^bN_{c\left.\right]}{}^f\spsd
       \end{array}
\end{equation}
       Свернем  (\ref{re72p}) с $\delta_f{}^c$, то
\begin{equation}
\label{re73p}
       \begin{array}{c}
       3\delta_a{}^bN_r{}^r=N_r{}^r\delta_a{}^b+8N_a{}^b\sps\\[2ex]
       N_a{}^b=\frac{1}{4}N_r{}^r\delta_a{}^b\spsd
       \end{array}
\end{equation}
       Поэтому
\begin{equation}
\label{re74p}
       \begin{array}{c}
       N_a{}^b=n\delta_a{}^b\sps n^2=1\spsd
       \end{array}
\end{equation}

\newpage
\subsection{Доказательство формул четвертой главы}
\subsubsection{Доказательство тождества Бианки}
       Тождество Бианки имеет вид
\begin{equation}
\label{re16p}
       \begin{array}{c}
       R_{\alpha\beta\gamma\delta}+R_{\alpha\gamma\delta\beta}+
       R_{\alpha\delta\beta\gamma}=0\sps\\[3ex]
       (A_{\alpha\beta d}{}^cA_{\gamma\delta r}{}^s+
       A_{\alpha\gamma d}{}^cA_{\delta\beta r}{}^s+
       A_{\alpha\delta d}{}^cA_{\beta\gamma r}{}^s)
       R_c{}^d{}_s{}^r=0\spsd
       \end{array}%
\end{equation}
       чем доказана формула ~(\ref{re18}).

\subsubsection{Доказательство тождеств, касающихся тензора Вейля}$ $

       \indent Верно следующее соотношение
\begin{equation}
\label{re12p}
       \begin{array}{c}
       R_{\left[\right.\alpha}{}^{\left[\right.\gamma}
       g_{\beta\left.\right]}{}^{\delta\left.\right]}=\\[2ex]
       A_{\alpha\beta d}{}^cA^{\gamma\delta}{}_r{}^s
       \frac{1}{16}
       (P^{kdnr}\varepsilon_{kcns}-P_{kc}{}^{nr}\varepsilon^{kd}{}_{ns}+
       P_{kcns}\varepsilon^{kdnr}-P^{kd}{}_{ns}\varepsilon_{kc}{}^{nr})=\\[2ex]
       =A_{\alpha\beta d}{}^cA^{\gamma\delta}{}_r{}^s
       \frac{1}{16}
       (P^{kd}{}_{mm_1}\frac{1}{2}\varepsilon^{mm_1nr}\varepsilon_{kcns}+
       P_{kc}{}^{mm_1}\frac{1}{2}\varepsilon_{mm_1ns}\varepsilon^{kdnr}-\\[2ex]
       -P_{kc}{}^{nr}(\delta_n{}^k\delta_s{}^d-\delta_n{}^d\delta_s{}^k)-
       P^{kd}{}_{ns}(\delta_k{}^n\delta_c{}^r-\delta_k{}^r\delta_c{}^n))=\\[2ex]
       =A_{\alpha\beta d}{}^cA^{\gamma\delta}{}_r{}^s
       \frac{1}{16}
       (P^{kd}{}_{kc}\delta_s{}^r-P^{kd}{}_{ks}\delta_c{}^r+P^{rd}{}_{cs}+\\[2ex]
       +P_{kc}{}^{kd}\delta_s{}^r-P_{kc}{}^{kr}\delta_s{}^d+P_{sc}{}^{dr}-
       P_{kc}{}^{kr}\delta_s{}^d+P_{sc}{}^{dr}-
       P^{kd}{}_{ks}\delta_c{}^r+P^{rd}{}_{cs})=\\[2ex]
       =A_{\alpha\beta d}{}^cA^{\gamma\delta}{}_r{}^s
       \frac{1}{4}
       (P_{sc}{}^{dr}-\frac{1}{2}P^{kd}{}_{ks}\delta_c{}^r+
       \frac{1}{2}P_{kc}{}^{kd}\delta_s{}^r-
       \frac{1}{2}P_{kc}{}^{kr}\delta_s{}^d)=\\[2ex]
       =A_{\alpha\beta d}{}^cA^{\gamma\delta}{}_r{}^s
       \frac{1}{4}
       (P_{sc}{}^{dr}-\frac{1}{2}R\delta_s{}^d\delta_c{}^r+
       \frac{1}{4}R\delta_s{}^r\delta_c{}^d)\spsd
       \end{array}
\end{equation}
       Аналогично
\begin{equation}
\label{re13p}
       \begin{array}{c}
       g_{\left[\right.\alpha}{}^{\left[\right.\gamma}
       g_{\beta\left.\right]}{}^{\delta\left.\right]}=
       A_{\alpha\beta d}{}^cA^{\gamma\delta}{}_r{}^s
       \frac{1}{4}
       (\varepsilon_{sc}{}^{dr}-
       \frac{1}{2}\varepsilon^{kd}{}_{ks}\delta_c{}^r+
       \frac{1}{2}\varepsilon^{kd}{}_{kc}\delta_s{}^r-
       \frac{1}{2}\varepsilon_{kc}{}^{kr}\delta_s{}^d)=\\[3ex]
       =A_{\alpha\beta d}{}^cA^{\gamma\delta}{}_r{}^s
       \frac{1}{4}
       (\frac{1}{2}\delta_s{}^r\delta_c{}^d-2\delta_s{}^d\delta_c{}^r)\spsd
       \end{array}
\end{equation}
       Таким образом, из
\begin{equation}
\label{re14p}
       \begin{array}{c}
       C_{\alpha\beta}{}^{\gamma\delta}:=
       R_{\alpha\beta}{}^{\gamma\delta}-
       R_{\left[\alpha \right.}{}^{\left[ \gamma \right.}
       g_{\left.\beta \right]}{}^{\left. \delta \right]}+
       1/10Rg_{\left[\alpha \right.}{}^{\left[ \gamma \right.}
       g_{\left.\beta \right]}{}^{\left. \delta \right]}
       \end{array}%
\end{equation}
        что завершает доказательство формул ~(\ref{re16}).
\newpage

\subsubsection{Доказательство тождеств Риччи}$ $

       \indent По определению
\begin{equation}
\label{re2100p}
       \begin{array}{c}
       \Box_a{}^d:=\frac{1}{2}(\nabla_{ak}\nabla^{dk}-\nabla^{dk}\nabla_{ak})\sps\\[3ex]
       \Box_{\alpha\beta}:=2\nabla_{\left[\right.\alpha}
       \nabla_{\beta\left.\right]}\sps
       \Box_{\alpha\beta}=A_{\alpha\beta d}{}^a\Box_a{}^d\spsd
       \end{array}
\end{equation}
       Естественным образом предполагается ковариантное постоянство
       следующих величин
\begin{equation}
\label{re2101p}
       \begin{array}{c}
       \nabla_\alpha \eta_\gamma{}^{ab}=0\sps
       \nabla_\alpha A_{\beta\gamma d}{}^a=0\sps\\[3ex]
       \nabla_\alpha g_{\alpha\beta}=0\sps
       \nabla_\alpha \varepsilon_{abcd}=0\spsd
       \end{array}
\end{equation}
       Тождество Риччи имеет вид
\begin{equation}
\label{re25p}
       \begin{array}{c}
       \Box_{\alpha\beta}k^{\gamma\delta}=
       R_{\alpha\beta\lambda}{}^\gamma k^{\lambda\delta}+
       R_{\alpha\beta\lambda}{}^\delta k^{\gamma\lambda}\spsd
       \end{array}%
\end{equation}
       Умножим обе части на $\frac{1}{2}\eta_\gamma{}^{mm_1}
       A^{\beta\alpha}{}_a{}^b$
\begin{equation}
\label{re29p}
       %
\end{equation}
       Итак, пусть $k^{\alpha\beta}$ - простой изотропный бивектор,
       то согласно следствию \ref{corollary113r} из третьей главы он имеет
       бесследовый образ $k_a{}^b$ вида
\begin{equation}
\label{re32p}
       %
\end{equation}
       Если ~(\ref{re33p}) свернуть с произвольным ковектором $S_c$
       таким, что $S_cP^c=0$, то будет верно следующее равенство
\begin{equation}
\label{re34p}
       %
\end{equation}
       Рассмотрим теперь изотропный вектор $r^\gamma$ такой,
       что он имеет вид
\begin{equation}
\label{re35p}
       %
\end{equation}
       Свернем это уравнение с $Z_m$ таким, что $X^mZ_m=0$ и
       $Y^mZ_m \ne 0$, то с учетом ~(\ref{re34p})
       (это значит, что $Y^{m_1}Z_m\Box_a{}^bX^m=
       R_a{}^b{}_k{}^m Z_mX^kY^{m_1}$)  получим
\begin{equation}
\label{re37p}
       %
\end{equation}
\newpage
       \noindent Для произвольного $T_d$, спиноры $X^m$, $Y^m$ и $Z_m$ всегда можно выбрать так, что
       $X^mT_m=0$, $Y^mT_m=0$, $X^mZ_m=0$, $Y^mZ_m \ne 0$.
       Домножим  ~(\ref{re37p}) на $T_d$ и получим
\begin{equation}
\label{re38p}
       %
\end{equation}
       Вычтем из ~(\ref{re38p}) следующие 2 тождества, полученные из
       ~(\ref{re33p}),
\begin{equation}
\label{re39p}
       %
\end{equation}

\subsubsection{Доказательство дифференциальных тождеств Бианки}$ $
       \indent Дифференциальные тождества Бианки имеют вид
\begin{equation}
\label{re43p}
       %
\end{equation}
       Из них следует
\begin{equation}
\label{re44p}
       %
       $$
       Свернем это уравнение с $\eta^\alpha{}_{tp}$
\begin{equation}
\label{re46p}
       %
\end{equation}
       Свернем это уравнение с $\delta_n{}^m$
\begin{equation}
\label{re47p}
       %
\end{equation}
       Тогда тождество Бианки примет вид
\begin{equation}
\label{re49p}
       %
\end{equation}

\subsection{\texorpdfstring{Доказательство формул, связанных с метрикой,
            индуцированной в сечении конуса $K_6$}{Доказательство формул, связанных с метрикой,
            индуцированной в сечении конуса K}}$ $
     \indent Пусть задано сечение конуса $K_6$
\begin{equation}
\label{re81p}
     T^2+V^2-W^2-X^2-Y^2-Z^2=0
\end{equation}
     плоскостью V+W=1. Сделаем стереографическую проекцию
     полученного гиперболоида на плоскость V=0
\begin{equation}
\label{re82p}
%
\end{equation}
     Поэтому есть резон положить
    $$
    X:=
    \left(
    %
    \right)\spsd
\end{equation}
    Откуда
\begin{equation}
\label{re86p}
    d s^2=-\frac{det(dX)}{(det(X))^2}\sps \bar X^T+X=0\spsd
\end{equation}
\vspace{1cm}

\subsubsection{Доказательство формул о первом инварианте}$ $

    \indent Рассмотрим группу дробно-линейных преобразований L
\begin{equation}
\label{re87p}
    {\tilde X}=(AX+B)(CX+D)^{-1} \sps
    S:=
    \left(
    %
    \right) \sps
    detS=1\spsd
\end{equation}
    Пусть имеется два последовательных преобразования
\begin{equation}
\label{re88p}
    \tilde X=(AX+B)(CX+D)^{-1} \sps
    \tilde{\tilde X}=(\tilde A\tilde X+\tilde B)(\tilde C\tilde X+\tilde D)^{-1}\sps
\end{equation}
    то
\begin{equation}
\label{re89p}
%
\end{equation}
    Далее, для унитарных дробно-линейных преобразований имеем
\begin{equation}
\label{re90p}
%
\end{equation}
    Поэтому верны тождества
\begin{equation}
\label{re95p}
%
\end{equation}
    Доказательство второго таково
\begin{equation}
\label{re96p}
%
\end{equation}
    Поэтому получим цепочку тождеств
\begin{equation}
\label{re97p}
%
\end{equation}
    Домножим обе части на $CdX$
\begin{equation}
\label{re98p}
%
\end{equation}
    Таким образом получается первый инвариант.

\subsubsection{Доказательство формул о втором инварианте}$ $

    \indent Пусть
\begin{equation}
\label{re101p}
    \hat X=AX+B\sps
\end{equation}
    или в покомпонентной записи
\begin{equation}
\label{re102p}
    %
      $$
      Откуда
\begin{equation}
\label{re110p}
       \tilde Z^{-1}\frac{\partial}{\partial \tilde X^T}\tilde Z^{*-1}=
       Z^{-1}\frac{\partial}{\partial X^T} Z^{*-1}\spsd
\end{equation}
       Если, теперь, det(C)=0. Положим
\begin{equation}
\label{re111p}
%
\end{equation}
       Если все-таки det(A)=0, то можно положить
\begin{equation}
\label{re113p}
%
\end{equation}
        Очевидно, что $\tilde D$ всегда можно выбрать так, что
        $det(\tilde{\tilde{\tilde C}})\ne 0$. Поэтому
\begin{equation}
\label{re114p}
%
\end{equation}
         что и даст второй инвариант.

\newpage
\subsection{Доказательство формул пятой главы.}
\subsubsection{Доказательство условий интегрируемости
               битвисторного уравнения}$ $

     \indent По определению
\begin{equation}
\label{re2500}
%
\end{equation}
       Если  $\nabla^{a\left(\right.b}X^{c\left)\right.}=0$, то условия
       интегрируемости этого уравнения примут вид
\begin{equation}
\label{re2502}
%
\end{equation}
        где $C_a{}^c{}_l{}^d$ - аналог тензора Вейля.

\newpage
\References{
\bibitem{Besse1r}
{\sc Бессе А.}
Многообразия Эйнштейна. т. 1. Москва, Мир, 1990. На с. 269 дан перечень основных изоморфизмов алгебр Ли для $n\le 6$.
\bibitem{Dubrovin1r}
{\sc Дубровин В.А., Новиков С.П., Фоменко А.Т.}
 Современная геометрия. Москва, Наука, 1986.
\bibitem{Cartan1r}
{\sc Картан Э.}
Теория спиноров. Платон, Москва, 1997.
\bibitem{Koboyasi1r}
{\sc Кобояси Ш., Номидзу К.}
Основы дифференциальной геометрии. Т. 2, Москва, Наука, 1981. На с. 110-139 исследуются комплексные многообразия. Однако, в качестве вещественных вложений используются эрмитовы многообразия в отличие от данной диссертации.
\bibitem{Koboyasi2r}
{\sc Кобояси Ш.}
Группы преобразований в дифференциальной геометрии. Т. 2, Москва, Наука, 1981. На с. 6-55 рассказано об автоморфизмах G-структур, на с. 56-104 рассказано об изометриях римановых пространств, на с. 105-159 рассказано об автоморфизмах комплексных многообразий. 
\bibitem{Kotelnikov1r}
{\sc Котельников А.П.}
Винтовое счисление и некоторые приложения его к геометрии механники. Казань, 1895 г.
\bibitem{Lichnerowicz1r}
{\sc Лихнерович A.}
Теория связностей в целом и группы голономий. ИЛ, 1960. На с. 169-200 исследуются почти комплексные многообразия и связности на них. Однако, в качестве вещественных многообразий используются эрмитовы многообразия в отличие от данной диссертации.
\bibitem{Manin1r}
{\sc Манин Ю.И.}
Калибровочные поля и комплексная геометрия. Москва, Наука, 1996. На с. 15-72 исследуется пространство Минковского как многообразие вещественных точек большой клетки грассманиана комплексных плоскостей в пространстве твисторов.
\bibitem{Landau1r}
{\sc Ландау Л.Д., Лифшиц Е.М.}
Теория поля. Москва, Наука, 1988.
\bibitem{Neifeld1r}
{\sc Нейфельд Э.Г.}
Об инволюциях в комплексных пространствах. ТГС, Казань, 1989. Выпуск 19, с. 71-82.
\bibitem{Neifeld2r}
{\sc Нейфельд Э.Г.}
Геометрия поверхности в проективном пространстве над алгеброй. Геометрия обобщенных пространств, Уфа, 1982, с. 32-51.
\bibitem{Neifeld3r}
{\sc Нейфельд Э.Г.}
О внутренних геометриях поляризованных комплексных грассманианов. Известия ВУЗов, Казань, 1995, \No 5(396), с. 51-54.
\bibitem{Neifeld4r}
{\sc Нейфельд Э.Г.}
Нормализация комплексных грассманианов и квадрик. ТГС, Казань, 1990. Выпуск 20, с. 58-69.
\bibitem{Neifeld5r}
{\sc Нейфельд Э.Г.}
О внутренних геометриях нормализованного пенроузиана. ТГС, Казань, 1990. Выпуск 20, с. 70-73.
\bibitem{Neifeld6r}
{\sc Нейфельд Э.Г.}
Аффинные связности на нормализованном многообразии плоскостей проективного пространства. Известия ВУЗов, Казань 1976, \No 11(174), с. 48-55.
\bibitem{Neifeld7r}
{\sc Нейфельд Э.Г.}
О внутренних геометриях нуль-плоскостей максимальной размерности поляритетов второго порядка. ТГС, Казань, 1982. Выпуск 14, с. 50-55.
\bibitem{Norden1r}
{\sc Норден А.П.}
О комплексном представлении тензоров пространства Лоренца. Известия ВУЗов, Казань 1959, т. 8 \No 1, с. 156-164.
\bibitem{Norden2r}
{\sc Норден А.П.}
Обобщение основной теоремы теории нормализации. Известия ВУЗов, Казань, 1966, \No 2(51), c. 78-82.
\bibitem{Norden3r}
{\sc Норден А.П.}
О структуре связности на многообразии прямых неевклидового пространства. Известия ВУЗов, Казань, 1972, \No 12(127), с.~84-94.
\bibitem{Norden4r}
{\sc Норден А.П.}
Аффинная связность на поверхностях проективного пространства. Математический Сборник, Москва, 1947, \No 20(62), с. 263-286
\bibitem{Norden5r}
{\sc Норден А.П.}
Теория нормализации и векторные расслоения. ТГС, Казань, 1976. Выпуск 9, с. 68-77.
\bibitem{Norden6r}
{\sc Норден А.П.}
Пространства аффинной связности. Москва, Наука, 1976.
\bibitem{Penrose1r}
{\sc Пенроуз Р. Риндлер В.}
Спиноры и пространство-время. Т. 1, Москва, Мир, 1987, Т. 2, Москва, Мир, 1988.
\bibitem{Penrose2r}
{\sc Пенроуз Р.}
Твисторная программа // Твисторы и калибровочные поля. Москва, Мир, 1983. Дано определение пенроузиана на с. 13-24.
\bibitem{Penrose3r}
{\sc Пенроуз Р.}
Структура пространства-времени. Москва, Мир, 1972.
\bibitem{Penrose4r}
{\sc Пенроуз Р.}
Спинорная классификация тензора энергии. В книге: Гравитация. Проблемы и перспективы, Киев, Наука Думка, 1972, с. 203.
\bibitem{Petrov1r}
{\sc Петров А.З.}
Классификация пространств, определяемых гравитационными полями. Уч.зап. КГУ, 114:55(1954).
\bibitem{Petrov2r}
{\sc Петров А.З.}
Пространства Эйнштейна. Москва, Физматгиз, 1961.
\bibitem{Shevale1r}
{\sc Шевале К.}
Теория групп Ли. Т.1-3, Москва, ИЛ, 1948-58.
\bibitem{Postnikov1r}
{\sc Постников М.М.}
Группы и алгебры Ли. Москва, Наука, 1982. В лекциях 13-16 изложены основные идеи построения гиперкомплексных чисел на основании периодичности Ботта.
\bibitem{Rosenfeld2r}
{\sc Розенфельд Б.А.}
Неевклидовы геометрии. Москва, ГИТО, 1955. На с. 534 рассмотрен принцип тройственности Картана.
\bibitem{Rosenfeld1r}
{\sc Розенфельд Б.А.}
Многомерные пространства. Москва, Наука, 1966. Дано определение m-пар на с. 384. В нашем случае m=0.
\bibitem{Rosenfeld3r}
{\sc Хуа Ло-ген Розенфельд Б.А.}
Геометрия прямоугольных матриц и ее применение к вещественной проективной и неевклидовой геометрии. Известия ВУЗов, Казань, 1957, \No 1, с. 233-247.
\bibitem{Sintsov1r}
{\sc Д.М.Синцов}
Теория коннексов в пространстве в связи с теорией дифференциальных уравнений в частных производных первого порядка. Казань, 1894.
\bibitem{Fadeev1r}
{\sc Фадеев Д.К.}
Лекции по алгебре. Москва, Наука, 1984.
\bibitem{Hirtseburh1r}
{\sc Хирцебурх Ф.}
Топологические методы алгебраической геометрии. Мир, Москва, 1973.
\bibitem{Hodzh1r}
{\sc Ходж В.Д., Пидо Д.}
Методы алгебраической геометрии. Москва, ИЛ, 1954.
\bibitem{Einshtein1r}
{\sc Эйнштейн А.}
Сборник научных трудов. Т.1, Наука, 1966.
\bibitem{Adams_1r}
{\sc Adams J.F.}
Spin(8),triality,$F_4$ and all that in : Superspace and Supergravity, Cambridge University Press, Cambridge, 1981.
\bibitem{Brauer_1r}
{\sc Brauer R. Weye H.}
Spinors in n dimensions Amer J. Math 57, 425(1935).
\bibitem{Chevalley_1r}
{\sc Chevalley C.}
The Algebraic Theory of Spinors. Columbia University Press, New York, 1954.
\bibitem{Frank_1r}
{\sc Frank W.Warner}
Foundation off Differentiable Manifolds and Lie Groups. New York, Berlin, Heidelberg, Tokyo, Springer-verlag, 1983.
\bibitem{Hughston_1r}
{\sc Hughston L.P.}
Applications of SO(8) spinors in : Gravitation and Geometry. Robinson Festschrift volume, 1986.
\bibitem{LeBrun_1r}
{\sc Le Brun C.R.}
ambitwistors and Einstein's equations. Class. and Quantum Grav, 1985.
\bibitem{Penrose_1r}
{\sc Penrose R.}
Twistor algebra. J.Math.Phys, 354:8(1967)345-366.
\bibitem{Penrose_2r}
{\sc Penrose R.}
Twistor theory : its aims and achievements in : Quantum Gravity an Oxford Symposium. Oxford University Press, Oxford, 1975.
\bibitem{Penrose_3r}
{\sc Penrose R.}
On the  origins of twistor theory in: Gravitation and Geometry. Robinson Festschrift volume, 1986.
\bibitem{Penrose_4r}
{\sc Penrose R.}
Relativistic symmetry groups in : Group Theory in Non-linear Problems. 1974.
\bibitem{Hochschild_1r}
{\sc Hochschild G.}
The Structure of Lie Groups. Holden-Day, San Francisco, 1965.
\bibitem{Dirac_1r}
{\sc Dirac P.A.}
Wave equations in conformal space. Ann. of Math. 37, 429(1936).
\bibitem{Dirac_2r}
{\sc Dirac P.A.}
Relativistic wave equations. Proc. Roy. Soc. (London), 447(1936).
\bibitem{Klein_1r}
{\sc Klein F.}
Zur theorie der Liniercomplexe des ersten und zweiten Grades, Math. Ann. 2, 198(1870).
\bibitem{Klein_2r}
{\sc Klein F.}
Vorlesungen \"uber h\"ohere Geometrie. Berlin, Springer-Verlag, 1926, S. 80, 262.
}

\label{origine}
\newpage
\normalsize
\pagestyle{plain}
\tableofcontents
\listoftables
\listoffigures
\end{document}